\documentclass[11pt]{article}
\usepackage[margin=1in]{geometry}
\usepackage{amsmath}
\usepackage{graphicx}
\usepackage{enumerate}
\usepackage{enumitem}
\usepackage{natbib}
\usepackage{url} 

\usepackage{algorithmic} 
\usepackage{algorithm} 
\usepackage{setspace} 
\usepackage{rotating} 

\usepackage{amsfonts, amssymb} 
\usepackage{amsthm}
\usepackage[T1]{fontenc}
\usepackage{lmodern}

\usepackage{epstopdf} 
\usepackage{hyperref} 
\usepackage{bm} 
\usepackage{bbm}
\usepackage{multirow} 
\usepackage{latexsym} 
\usepackage{rotating} 
\usepackage{titlesec} 
\usepackage{caption} 
\usepackage{mathtools} 
\usepackage{color} 
\usepackage{diagbox}
\usepackage{pict2e}
\usepackage{etoolbox}
\makeatletter
\@ifundefined{c@lofdepth}{}{\let\c@lofdepth\undefined}
\@ifundefined{c@lotdepth}{}{\let\c@lotdepth\undefined}
\makeatother

\usepackage{tocloft}

\usepackage[english]{babel}
\usepackage{empheq}

\usepackage{cases}
\usepackage{tikz}
\usetikzlibrary{arrows.meta,shapes,shapes.arrows,
shapes.geometric,shapes.multipart,decorations.pathmorphing,
positioning,shapes.swigs,}

\usepackage{lato}

\newtheorem{theorem}{Theorem}[section]
\newtheorem{lemma}[theorem]{Lemma}

\theoremstyle{definition}

\definecolor{mycolor}{RGB}{0,200,200} 

\newcommand{\indep}{\rotatebox[origin=c]{90}{$\models$}\,}
\newcommand{\nindep}{\not \hspace{-0.1cm} \rotatebox[origin=c]{90}{$\models$}\,}
\DeclareMathOperator*{\argmax}{arg\,max}
\DeclareMathOperator*{\argmin}{arg\,min}

\DeclareMathOperator*{\arginf}{arg\,inf}
\DeclareMathOperator*{\median}{median}
\newcommand{\EXP}{\mathbbm{E}}
\newcommand{\VAR}{\text{Var}}
\newcommand{\AVER}{\mathbbm{P}}

\newcommand{\cond}{\, | \,}
\newcommand{\con}{ ; }
\newcommand{\R}{\mathbbm{R}}
\newcommand{\ind}{\mathbbm{1}}
\newcommand{\T}{{^\intercal}}
\newcommand{\InfFt}{\texttt{IF}}
\newcommand{\uncInfFt}{\texttt{uIF}}
\newcommand{\augInfFt}{\texttt{Aug}}

\newcommand{\bV}{\bm{V}}
\newcommand{\bv}{\bm{v}}
\newcommand{\bX}{\bm{X}}
\newcommand{\bx}{\bm{x}}
\newcommand{\bL}{\bm{L}} 
\newcommand{\bO}{\bm{O}}
 
\newcommand{\potY}[1]{Y^{(#1)}}
\newcommand{\suppZ}{\mathcal{Z}}
\newcommand{\suppZb}{\{0,1\}}
\newcommand{\suppZinf}{\suppZ_{\inf,\bX}}
\newcommand{\suppX}{\mathcal{X}}
\newcommand{\suppYX}{\mathcal{Y}_{\bX}}

\newcommand{\model}{\mathcal{M}}
\newcommand{\SUPP}{{Supplementary Material}}

\newcommand{\odds}{\pi}

\newcommand{\ETA}{_{\eta}}
\newcommand{\alphahat}{\widehat{\alpha}\LSS}
\newcommand{\pihat}{\widehat{\pi}\LSS}
\newcommand{\betahat}{\widehat{\beta}\LSS}
\newcommand{\gammahat}{\widehat{\gamma}\LSS}
\newcommand{\muhat}{\widehat{\mu}\LSS}

\newcommand{\what}{\widehat{\omega}\LSS}

\newcommand{\EXPhat}{\widehat{\EXP}\LSS}

\newcommand{\ftrue}[1]{f_{#1}^*}
\newcommand{\ff}[1]{f_{#1}}
\newcommand{\fhat}[1]{\widehat{f}_{#1}}

\newcommand{\fyhat}{\fhat{Y}}
\newcommand{\fy}{\ff{Y}}

\newcommand{\fzhat}{\fhat{Z}}
\newcommand{\fz}{\ff{Z}}

\newcommand{\fahat}{\fhat{A}}
\newcommand{\fa}{\ff{A}}

\newcommand{\II}{\mathcal{I}}
\newcommand{\LSS}{^{(-k)}}
\newcommand{\LSSstar}{^{(-k)\star}}
\newcommand{\SSS}{^{(k)}}
\newcommand{\SSk}{\mathcal{I}_k}
\newcommand{\AVERk}{\AVER_{\SSk}}
\newcommand{\EXPk}{\EXP^{(-k)}}
\newcommand{\EMPk}{\mathbbm{G}_{\SSk}^{(-k)}}

\newcommand{\norm}{\mathcal{N}}
\newcommand{\normhat}{\widehat{\mathcal{N}}^{(-k)}}

\newcommand{\HL}[1]{\hyperlink{(#1)}{(#1)}}
\newcommand{\HT}[1]{\hypertarget{(#1)}{(#1)}}

\newcommand\numeq{\addtocounter{equation}{1}\tag{\theequation}}

\setlist[itemize]{
 leftmargin=!,       
  labelwidth=1cm,   
  labelsep=0.5em,     
  itemsep=0cm,
  topsep=0.1cm,
  parsep=0cm,
  align=left           
  }  
\setlist[enumerate]{
 leftmargin=!,       
  labelwidth=1cm,   
  labelsep=0.5em,     
  itemsep=0cm,
  topsep=0.1cm,
  parsep=0cm,
  align=left           
  }  

\allowdisplaybreaks
\hypersetup{
colorlinks=true,
linkcolor=blue,
filecolor=blue,
urlcolor=blue,
citecolor=blue
}

\graphicspath{ {plot/} }

\definecolor{red1}{RGB}{255,64,64}
\definecolor{blue1}{RGB}{128,255,255}
\definecolor{green1}{RGB}{0,205,0}

\begin{document}

\title{Nonparametric Inference with an Instrumental Variable under a Separable Binary Treatment Choice Model}
\author{
Chan Park$^{a}$, Eric J. Tchetgen Tchetgen$^{b}$\\[0.25cm]
\makebox[1cm][c]{{\footnotesize $^{a}$Department of Statistics, University of Illinois Urbana-Champaign, Champaign, IL 61820, U.S.A.}}\\
\makebox[1cm][c]{{\footnotesize $^{b}$Department of Statistics and Data Science, University of Pennsylvania, Philadelphia, PA 19104, U.S.A.}}
}
 \date{}
  \maketitle
\begin{abstract}
Instrumental variable (IV) methods are widely used to infer treatment effects in the presence of unmeasured confounding. In this paper, we study nonparametric inference with an IV under a \emph{separable binary treatment choice model}, which posits that the odds of the probability of taking the treatment, conditional on the instrument and the treatment-free potential outcome, factor into separable components for each variable. While nonparametric identification of smooth functionals of the treatment-free potential outcome among the treated, such as the average treatment effect on the treated, has been established under this model, corresponding nonparametric efficient estimation has proven elusive due to variationally dependent nuisance parameters defined in terms of counterfactual quantities. To address this challenge, we introduce a new variationally independent parameterization based on nuisance functions defined directly from the observed data. This parameterization, coupled with a novel fixed-point argument, enables the use of modern machine learning methods for nuisance function estimation. We characterize the semiparametric efficiency bound for any smooth functional of the treatment-free potential outcome among the treated and construct a corresponding semiparametric efficient estimator without imposing any unnecessary restriction on nuisance functions. Furthermore, we describe a straightforward generative model justifying our identifying assumptions and characterize empirically falsifiable implications of the framework to evaluate our assumptions in practical settings. Our approach seamlessly extends to nonlinear treatment effects, population-level effects, and nonignorable missing data settings. We illustrate our methods through simulation studies and an application to the Job Corps study.

\end{abstract}
\noindent%
{\it Keywords:}  Falsification; Job Corps; Missing Data; Odds Ratio; Semiparametric Efficiency

\newpage

\section{Introduction} \label{sec-introduction}

In recent years, causal inference has provided several rigorous methods for inferring the impact of hypothetical interventions from observational studies and randomized trials with imperfect adherence to the study protocol. In such studies, the validity of causal inference relies on adequately accounting for potential confounders of the treatment-outcome relationship of primary interest. In practice, confounding is typically addressed by adjusting for measured pre-treatment covariates, with the hope that they capture all relevant sources of confounding. However, in observational studies and randomized trials subject to protocol violations, unobserved confounding is almost always a concern, as it may be impossible to rule out the presence of hidden common causes of the treatment and outcome variables in such settings; see Section \ref{sec-application} for a real-world example.

To date, a rich literature has also been developed to address unmeasured confounding. One of the most prominent approaches is the instrumental variable (IV)  approach, which has been broadly applied across social and health sciences \citep{LATE1994, Angrist1996, JAMA1996, Card2001, Angrist2001, Abadie2003, Vansteelandt2003, Newey2003, HR2006, Tan2006, Tan2010, Vansteelandt2011, TTV2013, Ogburn2014, Wang2017IV, Kennedy2018, Sun2018, Kennedy2020, Liu2020, YeShaoKang2021, TT2024_NATE, Liu2025MIV, Sun2025}. Formally, the IV framework relies on the availability of a valid instrument which satisfies three core conditions, namely \HL{IV1}: an IV exclusion restriction condition; \HL{IV2}: an IV independence condition; and \HL{IV3}: an IV relevance condition; see Section \ref{sec-assumption} for formal definitions of these three core assumptions. In addition, for point identification of a treatment effect, an additional fourth IV assumption is typically required. The specific form of this assumption determines which causal estimand can be identified. Commonly studied estimands include the local average treatment effect \citep{LATE1994, Angrist1996}, the population average treatment effect \citep{Wang2017IV}, and the average treatment effect on the treated (ATT) \citep{Robins1994SNMM, HR2006, Liu2020, Liu2025MIV}; see \citet{Levis2024} for a recent review comparing these estimands and the associated IV assumptions.

In this paper, we study inference with an IV which assumes conditions \HL{IV1}-\HL{IV3}, and \HL{IV4}: a certain \emph{separable binary treatment choice model}, which posits that the odds of the extended propensity score---the probability of taking the treatment conditional on the instrument and the potential outcome under no treatment---can be written as a product of separable components, one for each variable; see Section \ref{sec-assumption} for its formal definition. When combined with the IV independence condition \HL{IV2}, this model assumes that the exogenous instrument and the endogenous hidden confounding exert separable, multiplicative effects on the treatment odds. Equivalently, the model implies a \textit{logit-separable treatment mechanism}, whereby the logit of the extended propensity score is additive in the instrument and the hidden confounding, represented through the potential outcome under no treatment. For this reason, we use ``separable binary treatment choice model'' and ``logit-separable treatment mechanism'' interchangeably throughout the paper. Accordingly, we refer to the IV framework satisfying \HL{IV1}-\HL{IV4} as the \emph{logit-separable IV model}.

In recent related work, \citet{Liu2020} established nonparametric identification of the law of the treatment-free potential outcome under \HL{IV1}-\HL{IV4}. Consequently, any smooth functional of this counterfactual distribution (e.g., ATT) is also nonparametrically identified. \citet{Sun2018} also established related results in a nonignorable missing data setting; see Section \ref{sec-identification of the ATT} for details. Importantly, both studies highlighted that identification generally fails if the separability condition is violated. For instance, if the logit of the extended propensity score cannot be expressed as a sum of components that depend only on the instrument and the potential outcome---implying non-separability on the logit scale---the model is no longer identified in general. While identification may still be possible without logit-separability, it typically requires separability to hold on an alternative scale, such as the probit scale. This example illustrates the unique and essential role of the separable treatment mechanism in achieving nonparametric identification.

Despite their foundational identification results, \citet{Sun2018}, \citet{Liu2020}, and more recently \citet{Sun2025} face several important limitations in estimation and inference. First, the parameterization proposed by \citet{Sun2018} and \citet{Liu2020} is not variationally independent, as infinite-dimensional nuisance parameters functionally constrain each other. In an effort to address this limitation, \citet{Sun2025} introduced a variationally independent parameterization; however, like the earlier approaches, their proposed re-parameterization intrinsically relies on counterfactual quantities which are not directly available from the observed data. Together, the lack of variational independence and the reliance on counterfactual parameters complicate estimation and inference. The previous papers attempted to resolve these complications by relying on parametric specification of certain nuisance functions. However, variational dependence makes model specification challenging, and can easily lead to incoherence among nuisance parameters, further compromising robustness of estimation and inference. Consequently, the theoretical guarantees of existing methods are limited in terms of model flexibility. Details of previously proposed parameterizations for the logit-separable IV model are given in Section \ref{sec-previous parameterization}.  In addition, previous studies lack a well-articulated generative model for the logit-separable treatment mechanism and do not provide empirical conditions to falsify the model when it does not hold. These limitations make it difficult for investigators to assess the model plausibility and apply the existing method to real-world settings. Accordingly, a comprehensive framework for principled nonparametric estimation and inference is currently lacking for the logit-separable IV model.

The overarching goal of this paper is to establish a comprehensive nonparametric IV framework under the logit-separable treatment mechanism with a binary instrument. Specifically, we extend and substantially generalize prior works \citep{Sun2018, Liu2020, Sun2025} by systematically addressing its fundamental limitations. As we argue below, our nonparametric treatment of the logit-separable IV model makes multiple novel contributions to the causal inference literature, as it necessitates the development of several new methodological techniques and technical arguments that cannot easily be deduced from the existing literature. Our main contributions are summarized as follows.

\vspace*{0.2cm}
\noindent (\textit{Contribution 1: A New Parameterization})  \\
    In Section \ref{sec-parameterization}, we propose an alternative parameterization that is variationally independent and relies solely on quantities directly available from the observed data. In other words, our model parameters are defined in a way that avoids any dependence on counterfactual quantities. Its construction involves introducing an intermediate quantity via a fixed-point equation in a function space, requiring a rigorous and nontrivial analysis. We further show that this equation can be efficiently solved using an iterative procedure that converges to its unique solution from any arbitrary initial values. Notably, the algebraic techniques developed herein are novel and may be of independent interest beyond the existing causal inference literature.

\vspace*{0.2cm}
\noindent (\textit{Contribution 2: Semiparametric Efficiency Theory})  \\
    In Section \ref{sec-theory}, we characterize the full set of influence functions for the ATT in a nonparametric model for the observed data law, which allows all nuisance functions to be unrestricted. This result stands in sharp contrast to previous work, which relies on parametric specifications for certain model parameters. Furthermore, we establish that the influence function for the ATT is unique and thus attains the semiparametric efficiency bound for the model, and is available in closed-form expression. Notably, deriving this expression presents several technical challenges, which we carefully address, as it necessitates solving a complex integral equation. 

\vspace*{0.2cm}
\noindent (\textit{Contribution 3: A Semiparametric Efficient Estimator})  \\
    In Section \ref{sec-estimation}, we construct an estimator for the ATT based on the efficient influence function introduced in Contribution 2. Importantly, all nuisance functions are estimated nonparametrically, leveraging the alternative parameterization introduced in Contribution 1. The proposed estimator is consistent, asymptotically normal, and semiparametric efficient, provided that some, but not necessarily all, nuisance functions are estimated at a sufficiently fast rate. We further compare the conditions required for asymptotic normality with those in previous work, considering both scenarios in which nuisance functions are estimated parametrically and nonparametrically. Simulation studies in Section \ref{sec-simulation} demonstrate that the finite-sample performance of the proposed estimator aligns well with the established asymptotic properties.

\vspace*{0.2cm}
\noindent (\textit{Contribution 4: A Generative Model and Falsification Implications})  \\
    We also provide results that facilitate the interpretation and assessment of the logit-separable IV model. Specifically, in Section \ref{sec-generative model}, we introduce a generative model for the logit-separable IV model motivated by standard discrete choice theory widely studied in economics and related empirical social sciences; in Section \ref{sec-falsification}, we examine the falsification implications of the IV assumptions. Together, these results help justify and interpret the IV model, thereby enhancing their practical relevance. Both the generative model and the falsification implications are illustrated through a real-world application in Section \ref{sec-application}.

\vspace*{0.2cm}
\noindent (\textit{Contribution 5: Extensions})  \\
    Lastly, as byproducts of our approach, we extend the method to infer nonlinear treatment effects, such as the quantile treatment effect on the treated (QTT), and treatment effects defined over the entire population, such as the average treatment effect, beyond causal effects among treated units. We also discuss how our IV framework can be readily adapted to settings with missing data under a nonignorable missing mechanism. Finally, we provide additional results in the case of categorical or continuous IV, where the model is no longer saturated, and causal parameters become over-identified. For example, we provide a closed-form characterization for the set of influence functions in the case of a categorical IV. Additionally, the new parameterization introduced in Contribution 1 extends to general $Z$ with only minor modifications to the fixed-point equation. Due to space constraints, these extensions are presented in {\SUPP} \ref{sec-supp-extension}.

\section{Preliminary} \label{sec-preliminary}
\subsection{Setup}

Suppose that we observe $N$ independent and identically distributed (i.i.d.) units, indexed by the subscript $i \in \{1,\ldots,N\}$. We suppress the subscript $i$ unless necessary. For each unit, the observed data are $\bO \equiv (Y, A, Z, \bX)$; here, $\bX \in \R^d$ denotes a $d$-dimensional pre-treatment covariate, $Z \in \{0,1\}$ is a candidate binary IV, $A \in \{0,1\}$ denotes treatment assignment ($A=1$ for treated, $A=0$ for control), and $Y \in \R$ is the outcome of interest. Let $\suppX \subseteq \R^d$ denote the support of $\bX$. Let $\potY{a,z} \in \R$ denote the potential outcome that one would have observed had, possibly contrary to the fact, $(A,Z)$ been set to $(a,z)$. Similarly, let $\potY{a} \in \R$ be the potential outcome that one would have observed had, possibly contrary to the fact, $A$ been set to $a$, regardless of
$Z$. In this paper, we mainly focus on the ATT, defined as $\tau^* = \EXP \{ \potY{1} - \potY{0} \cond A=1 \}$. In {\SUPP} \ref{sec-supp-extension}, we discuss how our method can be used to infer nonlinear treatment effects (e.g., QTT), population-level effects, (e.g., average treatment effect), and outcomes subject to a nonignorable missingness mechanism.

Throughout the paper, we impose the following assumptions, which are standard in the causal inference literature:
\begin{itemize}
    \item[\HT{A1}] (Consistency) $Y = \potY{A}$ almost surely;
    \item[\HT{A2}] (Overlap) For each $\bX \in \suppX$, the support of $\potY{0} \cond (A=a,Z=z,\bX)$ is identical for all values of $a \in \{0,1\}$ and $z \in \suppZb$, which is denoted by $\suppYX \subseteq \R$. 
\end{itemize}
Under these assumptions, the ATT can be expressed as $\tau^* = \EXP \{ Y - \potY{0} \cond A=1 \}$. Consequently, identification of the ATT depends on the ability to recover the distribution of $\potY{0} \cond (A=1)$ from the observed data. 

To facilitate the derivation of identification results and for later use, we introduce some additional notation. Let $f^*(\cdot \cond \cdot)$ denote the true (conditional) density function; for example, $f^*(Z \cond \bX)$ represents the true conditional density of $Z$ given $\bX$. Without an asterisk, $f (\cdot \cond \cdot)$ denotes a working model for a (conditional) density function. For each $\bX \in \suppX$, let $y_R \in \suppYX$ be a fixed outcome reference value, and define the functions $\odds^*$, $\alpha^*$, $\beta^*$, and $\mu^*$ by
\begin{align}
    &
    \odds^*(y,z,\bX) = 
    \frac{ f^*(A=1 \cond \potY{0}=y,Z=z,\bX ) }{ f^* (A=0 \cond \potY{0}=y,Z=z,\bX) } 
    \ , 
    \label{eq-nuis-pi}
    \\
    &
    \alpha^*(y,z,\bX)
    =
    \frac{\odds^*(y,z,\bX)}{\odds^*(y_R,z,\bX)} \ ,
    &&
    \hspace*{-1.75cm}
    \beta^*(z,\bX) = \odds^*(y_R,z,\bX)
    \label{eq-nuis-alpha beta}
    \ ,  \\ 
    &
    \mu^*(z,\bX)
    =
    \EXP \{ 
        \potY{0} \cond A=1,Z=z,\bX
    \} \ . 
    \label{eq-nuis-mu}
\end{align}
In words, $\odds^*$ is the odds function of $A$ given $(\potY{0}, Z,\bX)$, $\alpha^*$ is the (generalized) odds ratio function relating $A$ and $\potY{0}$ given $(Z,\bX)$ \citep{Chen2007, TTRR2010}, and $\beta^*$ is the baseline odds function. By definition, we have $\pi^* = \alpha^* \beta^* $; moreover, under \HL{A2}, both  $\alpha^*$ and $\beta^*$ take values in $(0,\infty)$. We remark that $\alpha^*$ can be interpreted as a scale quantifying the degree of confounding bias in the association between $A$ and $\potY{0}$. In particular, if $\alpha^*(y,Z,\bX)=1$ for all $y \in \suppYX$ given $(Z,\bX)$, then $A$ and $\potY{0}$ have no association, and are in fact conditionally independent, given $(Z,\bX)$; that is, there is no confounding bias conditional on $(Z,\bX)$. In addition, $\alpha^*$ satisfies the boundary condition that $\alpha^*(y_R,Z,\bX)=1$ for all $(Z,\bX)$. The choice of reference value $y_R$ does not affect the results established below.

Next, 
$\mu^*$ is the counterfactual outcome regression of $\potY{0}$ given $(A=1,Z,\bX)$. For any integrable function $g(\potY{0}, Z, \bX)$, we define $\mu^*(Z,\bX \con g) = \EXP \{ g(\potY{0},Z,\bX) \cond A=1,Z,\bX \}$. By straightforward algebra, it follows that
\begin{align*} 
    \mu^*(Z,\bX \con g)
    =
    \frac{ \EXP \{ g(Y,Z,\bX) \alpha^*(Y,Z,\bX) \cond A=0,Z,\bX \} }{
    \EXP \{ \alpha^*(Y,Z,\bX) \cond A=0,Z,\bX \}
    }
    \ .
    \numeq
    \label{eq-nuis-mu alpha}
\end{align*}

Finally, we introduce auxiliary notation used throughout the paper.
For a function $h(\bV)$, where $\bV \subseteq \{\potY{1},\potY{0},A,Z,\bX\}$, let $\|h \|_{\infty} = \sup_{\bv} | h(\bv) |$ and $\|h\|_{P,2} = [ \EXP \{ h^2(\bV) \} ]^{1/2}$ denote the supremum norm and $L^2(P)$-norm of $h$ with respect to the distribution of $\bV$.  With a slight abuse of notation, we write the integral of a generic function $g(\potY{a}, A, Z, \bX)$ with respect to the outcome $\potY{a}$ as $\int g(y,A,Z,\bX) \, dy $, which can be interpreted either as a summation or as an integral with respect to an appropriate dominating measure. For a sequence of random variables $\bV_N$ and a sequence of numbers $r_N$ , we write $\bV_N$ = $O_P(r_N)$ if $\bV_N /r_N$ is stochastically bounded, and $\bV_N=o_P(r_N)$ if $\bV_N /r_N$ converges to zero in probability as $N \rightarrow \infty$. Let $\bV_N \stackrel{D}{\rightarrow} \bV$ denote weak convergence of $\bV_N$ to $\bV$ as $N \rightarrow \infty$. Lastly, $\indep$ denotes statistical independence.

\subsection{Instrumental Variable Assumptions} \label{sec-assumption}

We next introduce the conditions required for the candidate IV $Z$ to be valid and sufficient for identifying the counterfactual distribution $\potY{0} \cond (A=1)$:
\begin{itemize}
    \item[\HT{IV1}] (Exclusion Restriction) $\potY{a,z}=\potY{a}$ almost surely for all $a \in \{0,1\}$ and $z \in \suppZb$;
    \item[\HT{IV2}] (Independence/Unconfoundedness) $\potY{0} \indep Z \cond \bX$;
    \item[\HT{IV3}] (Relevance) $A \nindep Z \cond \bX$;
    \item[\HT{IV4}] (Logit-separable Treatment Mechanism) $\alpha^*(y,Z=1,\bX) = \alpha^*(y,Z=0,\bX)$ for all $y \in \suppYX$ and $\bX \in \suppX$. Consequently, in slight notational abuse, we denote $\alpha^*(y,\bX)$, omitting its dependence on $Z$.
\end{itemize}
Condition \HL{IV1} states that $Z$ has no direct effect on $Y$, which allows us to focus solely on $\potY{a}$ without considering $Z$. Condition \HL{IV2} implies that $\bX$ includes all common causes of $Z$ and $\potY{0}$; in particular, it is automatically satisfied when $Z$ is randomized. Condition \HL{IV3} requires that $Z$ be associated with $A$ conditional on $\bX$, and this association may be either causal or non-causal. When \HL{IV1}-\HL{IV3} hold, $Z$ is referred to as a valid IV. However, these three conditions alone are insufficient for point identification of a treatment effect.

Condition \HL{IV4} is invoked to achieve point identification of the ATT and was first proposed by \citet{Sun2018} and \citet{Liu2020} in the missing data and causal inference contexts, respectively. Specifically, it requires that the odds ratio function $\alpha^*$ does not depend on the level of the IV. Therefore, under \HL{IV4}, the odds function admits the factorization $\odds^*(y,z,\bX) = \alpha^*(y,\bX) \beta^*(z,\bX)$, which implies that, on the logit scale, the treatment mechanism $A \cond (\potY{0}, Z, \bX)$ can be decomposed into components that depend separately on $(\potY{0},\bX)$ and $(Z,\bX)$. Consequently, \HL{IV4} can be interpreted as the logit-separable treatment mechanism. In addition, since $\alpha^*$ is a scale for measuring the degree of confounding, \HL{IV4} implies that this degree, measured on the odds ratio scale, remains constant across all levels of $Z$. Therefore, this condition can also be viewed as an odds ratio equi-confounding assumption, first introduced by \citet{UDID2022_Stat} and \citet{UDID2024_Epi} in a distributional generalization of difference-in-differences.

\subsection{A Generative Model} \label{sec-generative model}

It is instructive to consider a generative model that satisfies the four IV assumptions, as it helps illustrate how these assumptions can be justified and interpreted in real-world applications. In brief, the model is motivated by classical econometric frameworks based on the binary choice model \citep{McFadden1973, CiC2006, ImbensNewey2009, Train2009}. 

Specifically, we begin by assuming that each unit has an unobserved characteristic $U$, observed covariates $\bX$, and a candidate IV $Z$, with $U \indep Z \mid \bX$, i.e., the instrument is exogenous with respect to the unobserved characteristic. Following \citet{CiC2006}, let the potential outcome under treatment $A=a$ be $\potY{a}=h_a(U,\bX)$ where $h_a(u,\bX)$ is strictly monotone in $u$; in the econometrics literature, $h_a$ is often referred to as a production function. 

We further assume that a unit’s utility under treatment $A=a$ is $\mathcal{U}_a = \overline{h}_a(U,\bX) - c_a(Z,\bX) + \epsilon_a$. Here, $\overline{h}_a$ denotes the unit’s subjective belief regarding her production function $h_a$. Notably, $\overline{h}_a$ may differ from $h_a$, reflecting imperfect knowledge of the potential outcome under the corresponding treatment. The function $c_a$ represents the cost associated with choosing $A=a$ given observed information $(Z,\bX)$. We assume that $c_1(1,\bX) - c_1(0,\bX) \neq c_0(1,\bX) - c_0(0,\bX)$ holds for every $\bX$. In words, $Z$ directly instruments the cost associated with a treatment choice, and as such, treatment and control costs must respond differently to changes in $Z$. Finally, $\epsilon_0$ and $\epsilon_1$ are i.i.d. treatment-specific idiosyncratic error terms drawn from a type-I extreme value distribution (also known as the Gumbel distribution) with location 0 and scale 1,  a standard specification in the rich discrete choice literature \citep{McFadden1973, Train2009}.

In summary, a unit's utility under each treatment condition consists of perceived production $\overline{h}_{a}$ minus the cost $c_{a}$, plus an idiosyncratic error $\epsilon_{a}$. A unit chooses the treatment by comparing the two utilities; specifically, it takes the treatment if the utility under treatment exceeds that under control, i.e., $A = \ind (\mathcal{U}_1 > \mathcal{U}_0)$. Under this model, the IV conditions \HL{IV1}-\HL{IV4} are satisfied; technical details are provided in {\SUPP} \ref{sec-supp-generative model}.  The binary choice generative model is illustrated in Figure \ref{fig-Generative Model} using a single world intervention graph \citep{SWIG2013}.

\begin{figure}[!htp]
\centering
\scalebox{0.85}{
    \begin{tikzpicture}
\tikzset{line width=1pt, outer sep=1pt,
ell/.style={draw,fill=white, inner sep=2pt, line width=0.75pt},
elldot/.style={draw, dotted, fill=white, inner sep=1pt, line width=1pt},
swig vsplit={gap=3pt, inner line width right=0.75pt, line width right=1.75pt}};
\node[name=V, ell,  shape=ellipse] at (3,0) {$U$} ;
\node[name=Y0, ell,  shape=ellipse] at (12,0) {$\potY{a} = h_{a}(U)$} ;
\node[name=U0, ell,  shape=ellipse] at (3,1.25) {$\mathcal{U}_0 = \overline{h}_0(U) - c_0(Z) + \epsilon_0$}  ;
\node[name=U1, ell,  shape=ellipse] at (3,-1.25) {$\mathcal{U}_1 = \overline{h}_1(U) - c_1(Z) + \epsilon_1$}  ;
\node[name=Z, ell, shape=ellipse] at (0,0) {$Z$} ;
\node[name=A, shape=swig vsplit]  at (7,0) { \nodepart{left}{$A = \ind (\mathcal{U}_1 > \mathcal{U}_0)$} \nodepart{right}{$a$} };


\begin{scope}[>={Stealth[black]},
      every edge/.style={draw=black}]
\path [->] (V) edge[bend right=20] (Y0);
\path [->] (V) edge (U0);
\path [->] (V) edge (U1);
\path [->] (Z) edge (U0);
\path [->] (Z) edge (U1);
\path [->] (U0) edge (A);
\path [->] (U1) edge (A);
\path [->] (A) edge (Y0);
\end{scope} 

\end{tikzpicture} }%
\caption{A graphical illustration of the described generative model compatible with \protect\HL{IV1}-\protect\HL{IV4}. The measured covariates $\bX$ are suppressed for brevity.}
\label{fig-Generative Model}
\end{figure}
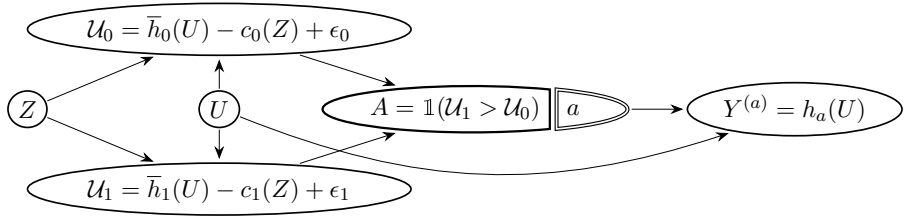

 
While the separability of the utility function might appear to be motivated primarily for analytical convenience---specifically to simplify the extended propensity score---it is, in fact, a necessary condition for identification.  To illustrate, suppose that the utility difference admits an interaction between $U$ and $Z$, say $\mathcal{U}_1 - \mathcal{U}_0 = \{ \overline{h}_1(U,\bX) - \overline{h}_0(U,\bX)\} - \{c_{1}(Z,\bX) - c_{0}(Z,\bX)\} + \imath(U,Z,\bX) + (\epsilon_{1}-\epsilon_{0})$, where $\imath$ captures the interaction between $U$ and $Z$. When $\imath (U,Z,\bX) \neq 0$, identification of the counterfactual data $(\potY{0},A,Z,X)$ generally fails, provided that the main effects are non-null (i.e., $\overline{h}_1(U,\bX)- \overline{h}_0(U,\bX)\neq 0$ and $c_{1}(Z,\bX) -  c_{0}(Z,\bX) \neq 0$); in {\SUPP} \ref{sec-supp-generative model}, we present a concrete example demonstrating this failure. Given the fundamental role of the separable structure in establishing identification, it is important to assess its plausibility in practice. To this end, we propose an empirical falsification test in Section \ref{sec-falsification}. Section \ref{sec-application} further demonstrates how this model can be justified in a real-world application.

We briefly illustrate the generative model across several domains to further highlight its interpretability. Several studies in the social sciences (e.g., \citep{Card2001}) have examined the causal effect of education on labor market earnings by using variables that influence schooling choices as instruments, such as geographic proximity to college. Omitting observed covariates, suppose that an individual’s utility from attending college $(a=1)$ or not $(a=0)$ is given by the following simple linear model, which governs a person's college attendance decision:
\begin{align*}
    &
    \mathcal{U}_a = b_{\text{int},a} + b_{h,a} U - b_{c,a} Z + \epsilon_a \ , 
    &&
    a \in \{0,1\} \ ,
    \numeq 
    \label{eq-linear utility}
    \\
    \Rightarrow
    \quad 
    &
    A = \ind( b_{\text{int}} + b_{h} U - b_{c} Z + \epsilon_1-\epsilon_0 > 0 ) \ , 
    &&
    b_{\star}\equiv b_{\star,1}-b_{\star,0} \ , \quad 
    \star \in \{\text{int},h,c\} \ ,
\end{align*}
where $U$ reflects the degree of innate self-motivation, $Z$ captures the geographic proximity to college, and $b_{c} \neq 0$ according to the requirement on the cost-function. Then, the rate $| b_{c} / b_{h}| $ characterizes the individual’s marginal rate of substitution between motivation and distance in the college attendance decision. Specifically, it quantifies the additional distance an individual is willing to travel for a one-unit increase in self-motivation on the utility scale.

As another example in the health sciences, \citet{JAMA1996} studied the causal effect of admission to a neonatal intensive care unit (NICU) versus non-admission on infant mortality. In this setting, let an individual’s utility from NICU admission follow the same linear model in \eqref{eq-linear utility}, where $U$ and $Z$ now reflect the severity of the infant’s illness and geographic proximity to the nearest NICU, respectively. As in the previous example, the ratio $| b_{c} / b_{h}| $ can be interpreted as the increase in distance an individual is willing to tolerate for a one-unit increase in illness severity.

Although this simple linear utility model facilitates interpretation and often aligns with subject-matter intuition, it is restrictive because it imposes a constant marginal rate of substitution and does not appropriately incorporate the possibility for diminishing returns on various forms of personal capital investments. In the college attendance example, the marginal impact of additional self-motivation on college attendance may diminish at higher levels of motivation, while reductions in distance may matter less once a college is sufficiently nearby. Similar diminishing returns might be expected in the NICU example. These considerations highlight the importance of accommodating flexible functional forms for the utility function---a feature incorporated in the proposed generative model---as it is essential for capturing more nuanced behaviors, such as the law of diminishing returns, and ensuring the model remains realistic.

\subsection{Identification} \label{sec-identification of the ATT}

Under \HL{A1}-\HL{A2} and \HL{IV1}-\HL{IV4},  \citet{Liu2020} and \citet{Sun2018} established that the joint distribution of $(\potY{0},A,Z,\bX)$ can be identified from the observed data; see Section \ref{sec-previous parameterization} for a brief review of their results. Consequently, the nuisance functions in \eqref{eq-nuis-pi}-\eqref{eq-nuis-mu} are also identified. In turn, the ATT is nonparametrically identified as:
\begin{align}
    \tau^*
    &
    =
    \EXP[ \{ A - (1-A) \alpha^*(Y,\bX) \beta^*(Z,\bX) \} Y ] / \EXP(A)
    \label{eq-IPW}
    \\
    &
    =
    \EXP[ A \{ Y - \mu^*(Z,\bX) \} ] / \EXP(A)
    \label{eq-OR}
    \\
    &
    =
    \EXP[ \{ A - (1-A) \alpha^*(Y,\bX) \beta^*(Z,\bX) \} \{ Y - \mu^*(Z,\bX) \} ] / \EXP(A) \ .
    \label{eq-AIPW}
\end{align}
The first representation involves only the treatment odds function $\odds^* = \alpha^* \beta^*$; the second representation relies solely on the outcome regression $\mu^*$; and the last representation incorporates both functions. Accordingly, we refer to these representations as the inverse probability-weighted (IPW), outcome regression-based, and augmented IPW (AIPW) representations of the ATT. Similar representations have been derived in a difference-in-differences setting \citep{Park2022JRSSA}.

For estimation and inference, one may use one of the representations \eqref{eq-IPW}-\eqref{eq-AIPW} as a basis. In particular, \citet{Robins2000_Sensitivity} showed that \eqref{eq-AIPW} corresponds to the efficient influence function (EIF) for the ATT in a semiparametric model where the odds ratio function $\alpha^*$ is known a priori, while the other nuisance functions of the observed data distribution are unrestricted and Assumptions \HL{IV1}-\HL{IV4} are dropped. This facilitates the construction of an estimator with robust properties in this semiparametric model, as formalized by semiparametric efficiency theory \citep{BKRW1998, Tsiatis2006, Robins2008_HOIF}. In practice, however, $\alpha^*$ is rarely known and must be estimated; in this case, robust inference attainable under known $\alpha^*$ is not guaranteed when $\alpha^*$ is treated as unknown and identified with the IV assumptions \HL{IV1}-\HL{IV4}.

For inference in the logit-separable IV model, \citet{Liu2020}, \citet{Sun2018}, and \citet{Sun2025} instead assumed that $\alpha^*$ belongs to a parametric model. They derived the set of influence functions (IFs) for the ATT under this semiparametric model and proposed locally semiparametric efficient estimators in which all nuisance functions are estimated parametrically.  However, these approaches do not accommodate an unrestricted $\alpha^*$ or fully nonparametric estimation of the nuisance functions, due to inherent limitations of their choice of parameterization; see Section \ref{sec-previous parameterization} for details.

In contrast, our new parameterization accommodates both an unrestricted $\alpha^*$ and nonparametric estimation of the nuisance functions. In the remainder of the paper, we first introduce this new parameterization in Section \ref{sec-parameterization} and then present the set of IFs for the ATT under the proposed model with unrestricted $\alpha^*$ in Section \ref{sec-theory}. By leveraging the binary nature of $Z$, we further show that the IF for the ATT is unique (and therefore is the EIF) and we provide its closed-form expression. Based on this EIF, we construct a globally efficient ATT estimator in Section \ref{sec-estimation} where all nuisance functions are estimated nonparametrically. 

\section{Parameterization} \label{sec-parameterization}

In this Section, we consider various potential choices of parameterization of the distribution of $(\potY{0}, A ,  Z) \cond \bX$ satisfying Assumptions \HL{A1}-\HL{A2} and \HL{IV1}-\HL{IV4}, including both existing approaches and our new proposal. Nuisance functions written without the asterisk $(^*)$ denote generic parameters that may differ from the true functions.

\subsection{Existing Parameterizations} \label{sec-previous parameterization}

We begin by reviewing parameterizations used in previous works. \citet{Sun2018} and \citet{Liu2020} considered a parameterization in terms of parameters $\theta_{\text{LS}} \equiv \{ \fz \equiv f(Z \cond \bX), \fy \equiv f(Y \cond A=0,Z,\bX), \alpha \equiv \alpha(\potY{0},\bX), \beta \equiv \beta(Z,\bX) \}$. Specifically, the density $f(\potY{0},A, Z  \cond \bX)$ is expressed as
\begin{align}
    \label{eq-Liu parameterization}
    &
    f(\potY{0}=y,A=a, Z \cond \bX)
    =
    \frac{ 
    \{ \alpha(y,\bX) \beta(Z,\bX) \}^{a}
    f(Y=y \cond A=0,Z,\bX) f(Z \cond \bX) }
    { \int \{ 1+ \alpha(t,\bX) \beta(Z,\bX) \}
    f(Y=t \cond A=0,Z,\bX) \, dt } 
    \ .
\end{align}
This representation is a direct consequence of \citet{Chen2007}; also see \citet{TTRR2010}. 

Therefore, identification of the distribution of $(\potY{0},A,Z) \cond \bX$ reduces to determining whether $\theta_{\text{LS}}$ can be identified from the observed data. The parameters $\fy$ and $\fz$ can be identified nonparametrically: $\fy$ by regressing $Y$ on $(Z, \bX)$ among units with $A=0$, and $\fz$ by regressing $Z$ on $\bX$. By contrast, it is not immediately obvious how $\alpha$ and $\beta$ can be identified from the observed data, as their definitions involve counterfactual quantities. Nonetheless, these quantities can be recovered using the following moment conditions \citep{Sun2018,Liu2020}:
\begin{align*}
    &  
    \EXP_{\theta_{\text{LS}}}
    \!
    \left[  
    \bigg\{
    \frac{1-A}
    {
    f ( A=0 \cond Y,Z,\bX; \alpha,\beta )
    }  
    -
    1
    \bigg\} 
    \times 
    b_1(\bX)
    \right] \! 
    = 0  \ ,      
    \\
    &  
    \EXP_{\theta_{\text{LS}}}
    \!
    \left[ 
    \!
    \left[
    \begin{array}{l}  
    \displaystyle{ \frac{
    (1-A)
    \{ b_2(Y,\bX) - \mu(Z,\bX \con b_2,\alpha,\fy) \}
    }
    {
    f ( A=0 \cond Y,Z,\bX; \alpha,\beta )
    }  
    }
    \\[0.3cm] 
    +
    \mu(Z,\bX \con b_2,\alpha,\fy)
    \end{array}
    \right]  
    \!
    \times     
    \!
    \left[
    \begin{array}{l}
    b_3(Z,\bX)
    \\[0.1cm] 
    -
    \EXP_{\theta_{\text{LS}}} \{ b_3(Z,\bX) \cond \bX \} 
    \end{array}
    \right]  
    \!
    \right] 
    = 0  \ ,     \! 
\end{align*}
where $(b_1, b_2, b_3)$ are arbitrary functions, $\EXP_{\theta_{\text{LS}}}$ denotes expectation under the observed data law parameterized by $\theta_{\text{LS}}$, and $\mu(Z,\bX \con b_2,\alpha,\fy)$ is given by the expression in \eqref{eq-nuis-mu alpha}. 

Unfortunately, the parameterization based on $\theta_{\text{LS}}$ has important drawbacks. In particular, the parameters are not variationally independent, as also noted by \citet{Sun2025}. To illustrate this issue, we consider the case of binary $Y$ and binary $Z$, suppressing covariates for simplicity. Based on the observed data, the following five parameters associated with $(\potY{0}, A, Z)$ can be identified: $\{ f(Z=1), f(A=1 \cond Z=0), f(A=1 \cond Z=1), f(Y=1\cond A=0,Z=0), f(Y=1\cond A=0,Z=1) \}$. However, the parameter set $\theta_{\text{LS}}$ comprises six quantities: $\{f(Z=1), f(Y=1|A=0,Z=0), f(Y=1|A=0,Z=1), \alpha(1), \beta(0), \beta(1) \}$, where $\alpha(0)$ is omitted because it is fixed by the boundary condition $\alpha(0) = 1$. This discrepancy implies that a restriction must be imposed on the model expressed in terms of $\theta_{\text{LS}}$. One such restriction is \HL{IV2}, which requires that the conditional density $f(\potY{0} \cond Z,\bX)$, obtained from \eqref{eq-Liu parameterization} by marginalizing over $A$, be independent of $Z$ conditional on $\bX$, i.e.,  
\begin{align*}
    &
    f(\potY{0}=y \cond \bX)
    \stackrel{\text{\HL{IV2} and \eqref{eq-Liu parameterization}}}{=}
    \frac{ 
    \{ 1 + \alpha(y,\bX) \beta(Z,\bX) \}
    f(Y=y \cond A=0,Z,\bX) f(Z \cond \bX) }
    { \int \{ 1+ \alpha(t,\bX) \beta(Z,\bX) \}
    f(Y=t \cond A=0,Z,\bX) \, dt }  \  .
\end{align*}
This result clearly implies a complex relationship among the components of $\theta_{\text{LS}}$, as the right-hand side of the above expression must be independent of $Z$. Consequently, $\theta_{\text{LS}}$ contains a redundant component, rendering the parameterization overidentified.

To address this issue of variational dependence, \citet{Sun2025} proposed an alternative parameterization based on $\theta_{\text{SMW}} \equiv \{ \fz, f_{\potY{0}} \equiv f( \potY{0} \cond \bX), \alpha, \beta \}$, where the density $f(\potY{0},A,Z \cond \bX)$ is given by
\begin{align*}
    f(\potY{0}=y,A=a,Z \cond \bX)
    =
    \underbrace{
    \frac{ \{ \alpha(y,\bX) \beta(Z,\bX) \}^{a} }{1+ \alpha(y,\bX) \beta(Z,\bX) } 
    }_{=f(A=a \cond \potY{0}=y,Z,\bX; \alpha,\beta)}
    f( \potY{0}=y \cond \bX) 
    f(Z \cond \bX)
    \ . 
    \numeq 
    \label{eq-Sun parameterization}
\end{align*}
In the toy example where both $Y$ and $Z$ are binary and $\bX$ is absent, $\theta_{\text{SMW}}$ clearly consists of five quantities: $\{ f(Z=1), f(\potY{0}=1), \alpha(1), \beta(0), \beta(1) \}$. This matches the number of parameters for $(\potY{0},A,Z)$ that are identifiable from the observed data, confirming that the parameterization under $\theta_{\text{SMW}}$ is variationally independent.

Again, identifying the distribution of $(\potY{0}, A, Z) \cond \bX$ reduces to determining whether $\theta_{\text{SMW}}$ can be recovered from the observed data. The $\fz$ component can be identified nonparametrically, and the remaining components $f_{\potY{0}}$, $\alpha$, and $\beta$ can be identified based on the following moment conditions \citep{Sun2025}:
\begin{align*}
    &  
    \EXP_{\theta_{\text{SMW}}}
    \left[  
    \bigg\{
    \frac{1-A}
    {
    f ( A=0 \cond Y,Z,\bX; \alpha,\beta )
    }  
    -
    1
    \bigg\} 
    \times 
    b_1(\bX)
    \right] \! 
    = 0  \ ,     
    \numeq
    \label{eq-Sun moment 1}
    \\
    &  
    \EXP_{\theta_{\text{SMW}}}
    \left[  
    \frac{
    1-A
    }
    {
    f ( A=0 \cond Y,Z,\bX; \alpha,\beta )
    } 
    \! \times  
    \! 
    \big[
    b_2(Y,\bX)
    -
    \EXP_{\theta_{\text{SMW}}} \{ b_2(\potY{0},\bX) \cond \bX \}
    \big]
    \right] \! 
    = 0  \ ,     
    \numeq
    \label{eq-Sun moment 2}
    \\
    &  
    \EXP_{\theta_{\text{SMW}}}
    \left[  
    \frac{
    1-A
    }
    {
    f( A=0 \cond Y,Z,\bX; \alpha,\beta )
    } 
    \! \times  
    \! 
    \left[
    \begin{array}{l} 
    \big[
    b_2(Y,\bX)
    -
    \EXP_{\theta_{\text{SMW}}} \{ b_2(\potY{0},\bX) \cond \bX \}
    \big] \! 
    \\[0.1cm] 
    \times 
    \big[
    b_3(Z,\bX)
    -
    \EXP_{\theta_{\text{SMW}}} \{ b_3(Z,\bX) \cond \bX \}
    \big]
    \end{array}
    \right]  
    \right] \! 
    = 0  \ ,     
    \numeq
    \label{eq-Sun moment 3}
\end{align*}
where $(b_1, b_2, b_3)$ are arbitrary functions and $\EXP_{\theta_{\text{SMW}}}$ denotes expectation under the observed data law parameterized by $\theta_{\text{SMW}}$.

To estimate each component of $\theta_{\text{SMW}}$, \citet{Sun2025} proposed two approaches:
\begin{itemize}
    \item[(i)] jointly estimate all three nuisance functions by specifying parametric models and fitting them using \eqref{eq-Sun moment 1}-\eqref{eq-Sun moment 3}; or
    \item[(ii)] posit a parametric model for $\alpha$, say $\alpha( \cdot \con \xi_{\alpha})$, then define $\beta=\beta( \cdot \con \xi_{\alpha})$ and $f_{\potY{0}} = f_{\potY{0}}( \cdot \con \xi_{\alpha})$ as the solutions to \eqref{eq-Sun moment 1} and \eqref{eq-Sun moment 2} (which may be estimated nonparametrically), and finally determine $\xi_{\alpha}$ from \eqref{eq-Sun moment 3}. This yields nuisance functions that simultaneously satisfy \eqref{eq-Sun moment 1}-\eqref{eq-Sun moment 3}.
\end{itemize}
In principle, one may attempt fully nonparametric estimation of these nuisance functions by adopting strategies analogous to (i) and (ii) above, but this approach is challenging. First, for the analogue of (i), jointly estimating all three components nonparametrically from \eqref{eq-Sun moment 1}-\eqref{eq-Sun moment 3} is arduous because they enter the equation in a highly intertwined manner. Specifically, $\alpha$ and $\beta$ are multiplicatively coupled, and $f_{\potY{0}}$ must be estimated through the mean embedding $\EXP_{\theta_{\text{SMW}}} \{ b_2(\potY{0},\bX) \cond \bX \}$, which must be verified for any admissible choice of $b_2$. Second, although a sequential strategy mirroring (ii) is conceivable, it necessitates specifying a nonparametric model for $\alpha$ to initialize the procedure. This sequential approach can be computationally intensive and may fail to yield solutions for \eqref{eq-Sun moment 1}-\eqref{eq-Sun moment 3}, even after exploring numerous candidate models for $\alpha$. 

In short, the parameterizations in previous works \citep{Sun2018, Liu2020, Sun2025} suffer from a lack of variational independence and rely on parameters defined in terms of counterfactual quantities. To partially address these challenges, the authors impose parametric models for the various nuisance functions, including the odds ratio function $\alpha$. Consequently, while the availability of an instrument in principle allows for $\alpha$ to be unrestricted, existing works assume it has a known functional form with an unknown finite-dimensional parameter. As a result, both the theoretical guarantees and the practical estimation procedures of these methods are confined to such parametric settings. Crucially, consistency and asymptotic normality of the existing estimators of the ATT require a correctly specified model for $\alpha$; see Section \ref{sec-bias comparison}. However, to the best of our knowledge, goodness-of-fit tests to evaluate possible model misspecification of $\alpha$ are currently lacking. Taken together, these issues limit the practical applicability of these existing methods and may compromise the validity of the resulting estimator.

\subsection{A New Parameterization} \label{sec-our parameterization}

To overcome the limitations of the existing standard parameterizations, we consider an alternative parameterization based on $\theta \equiv \{ \fz \equiv f(Z \cond \bX), \fa \equiv f(A \cond Z, \bX), \fy \equiv f(Y \cond A=0,Z,\bX) \}$. This parameterization has two distinctive features compared to those based on $\theta_{\text{LS}}$ and $\theta_{\text{SMW}}$. First, the parameters are variationally independent. To illustrate, consider the running toy example in which both $Y$ and $Z$ are binary and $\bX$ is absent. Then, $\theta$ consists of five quantities: $\{ f(Z=1), f(A=1 \cond Z=0), f(A=1 \cond Z=1), f(Y=1 \cond A=0, Z=0), f(Y=1 \cond A=0, Z=1) \}$, which matches the number of parameters for $(\potY{0},A,Z)$ identifiable from the observed data. Second, the parameters are directly defined in terms of the observed data and do not rely on any counterfactual quantities. Consequently, all components of $\theta$ can be readily estimated using standard nonparametric methods as well as highly adaptive modern machine learning methods; see Section \ref{sec-estimation} for details.

It is straightforward to see that $\theta$ characterizes the distribution of $(\potY{0}, A=0, Z) \cond \bX$ via $ f(\potY{0}=y, A=0, Z \cond \bX)    = f(Y=y \cond A=0, Z, \bX) f(A=0 \cond Z, \bX) f(Z \cond \bX)$. However, it is far less obvious how the counterfactual distribution $(\potY{0}, A=1, Z ) \cond \bX$ can be expressed in terms of $\theta$, since none of its components is directly defined in terms of the counterfactual data. 

We present the first main result of this paper, which shows how $\alpha$ and $\beta$---functions of counterfactual quantities---can be characterized from $\theta$, which are defined purely in terms of the observed data. To streamline the exposition, we introduce the following notation. For a given $\bX$, let $z_{m} = \argmin_{z\in\{0,1\}} f(A=1 \cond Z=z,\bX)$ and $z_{M} = \argmax_{z\in\{0,1\}} f(A=1 \cond Z=z,\bX)$; note that $z_{m} \neq z_{M}$ under \HL{IV3}. Although $z_{m}$ and $z_{M}$ depend on $\bX$, we suppress this dependence for notational brevity. Additionally, for $\bX \in \suppX$, define $\mathcal{L}_+^{\infty} (\mathcal{Y}_{\bX}) = \{ h \cond h(y,\bX) \in [0,\infty), \ \forall y  \in \mathcal{Y}_{\bX} \}$. We are now ready to state the Theorem.

\begin{theorem} \label{thm-Psi-binary}

Suppose that Assumptions \HL{A1}-\HL{A2} and \HL{IV1}-\HL{IV4} hold for $(\potY{0},A,Z,\bX)$, with corresponding parameters $\theta = \{ \fz,\fa,\fy \}$. For a fixed $\bX \in \suppX$, define the mapping $\Psi: \mathcal{L}_+^{\infty} (\mathcal{Y}_{\bX}) \rightarrow \mathcal{L}_+^{\infty} (\mathcal{Y}_{\bX})$ by
\begin{align*}
    &
    \!\!\!
    \Psi\big( g(y,\bX) \con \theta \big)
    \\
    &
    \!\!\!
    =
    \frac{
    \left\{
    \begin{array}{l}
    \ff{}(A=1 \cond Z=z_m,\bX)
    g(y,\bX)
    \\
    + \ff{}(A=0 \cond Z=z_m,\bX)
    \ff{}(Y=y \cond A=0,Z=z_m,\bX)
    \\
    -
    \ff{}(A=0 \cond Z=z_M,\bX)
    \ff{}(Y=y \cond A=0,Z=z_M,\bX)
    \end{array} 
    \right\}
    \!
    \displaystyle{
    \int \!
    g(t,\bX)
    R_{Y}(t,\bX) \, d t
    }
    }
    {
    \ff{}(A=1 \cond Z=z_M,\bX)
    R_{Y}(y,\bX)
    }   ,
    \numeq \label{eq-Psi-binary}
\end{align*}
where 
\begin{align}
 R_{Y}(y,\bX) = \frac{ f(Y=y \cond A=0,Z=z_M,\bX)}{f(Y=y \cond A=0,Z=z_m,\bX)} \ .
 \numeq 
 \label{eq-Ry-binary}
\end{align}

Then, the following results hold:
\begin{itemize}
    \item[(i)] Let $ g^\star(y,\bX \con \theta)$ denote the fixed point of $\Psi$, i.e., 
    \begin{align*}
        \Psi \big( g^\star(y,\bX \con \theta) \con \theta \big) = g^\star(y,\bX \con \theta) \ .
        \numeq 
        \label{eq-fixed point equation-binary}
    \end{align*}
    Then, $g^\star$ uniquely exists.

    \item[(ii)] The odds ratio function $\alpha$ and the baseline odds function $\beta$ in \eqref{eq-nuis-alpha beta} are expressed in terms of $g^\star$ as follows:
    \begin{align*}
    & 
    \alpha(y,\bX \con \theta)
    =
    \frac{ g^\star (y,\bX \con \theta) }{ g^\star  (y_R,\bX \con \theta) }
    \frac{ \ff{}(Y=y_R \cond A=0,Z=z_m,\bX) }{\ff{}(Y=y \cond A=0,Z=z_m,\bX)} \ , 
    && 
    y \in \suppYX \ ,
    \numeq \label{eq-alpha g-binary}
    \\
    & 
    \beta(z,\bX \con \theta) 
    =
    \frac{f(A=1 \cond Z=z,\bX) / f(A=0 \cond Z=z,\bX)}{ \int  \alpha(y,\bX \con \theta)  f(Y=y \cond A=0,Z=z,\bX) \, dy }  \ , 
    && 
    z \in \suppZb \ .
    \numeq \label{eq-beta g-binary}
    \end{align*}

\end{itemize} 
    
\end{theorem}

Theorem \ref{thm-Psi-binary} establishes that the counterfactual quantities $\alpha$ and $\beta$ are uniquely determined by the fixed point of the mapping $\Psi$. This mapping is defined entirely in terms of $\theta$, which in turn is derived solely from the observed data. Consequently, $\alpha$ and $\beta$ are uniquely identified through a deterministic relationship with observable quantities, without any direct reference to counterfactual data. In turn, the counterfactual distribution $(\potY{0}, A=1, Z) \cond \bX$ can be expressed entirely in terms of $\theta$ via \eqref{eq-Liu parameterization}, using the induced $\alpha$ and $\beta$.

This new parameterization stands in stark contrast to the previous approaches discussed in Section \ref{sec-previous parameterization}, which employ counterfactual quantities as model parameters. A direct consequence of those parameterizations is that parametric estimation becomes practically unavoidable, since certain nuisance functions, most notably the odds ratio function $\alpha$, cannot easily be estimated nonparametrically from the observed data; see Section \ref{sec-previous parameterization} for details. In contrast, our parameterization is entirely based on quantities readily available from the observed data, eliminating the need for ad hoc parametric restrictions. As a result, estimating all nuisance functions reduces to standard regression and density learning tasks; see Section \ref{sec-estimation} for the details on the estimation procedure.

While Theorem \ref{thm-Psi-binary} establishes the theoretical foundation for expressing $\alpha$ and $\beta$ in terms of $\theta$, its practical use may appear limited, since the fixed-point equation \eqref{eq-fixed point equation-binary} seems challenging to solve directly. However, the following Theorem shows that this equation can be efficiently solved using iterative updates, provided suitable boundedness conditions are satisfied. We first introduce the required assumptions and then present the Theorem.

\begin{itemize}

\item[\HT{A3}] (Boundedness for $\theta$) 
There exist constants  $0< c_Z < 1$, $0<c_A<1$, $0<c_Y \leq C_Y < \infty$, and $0< c_g \leq C_g < \infty$ such that the parameter $\theta = \{\fy, \fa, \fz\}$ and the corresponding  $g^\star(\cdot \con \theta)$ in \eqref{eq-fixed point equation-binary} satisfy $f(Z=z \cond \bX) \in [c_Z,1-c_Z]$, $f(A=a \cond Z=z,\bX) \in [c_A,1-c_A]$,  $ f(Y=y \cond A=0,Z=z,\bX) \in [c_Y,C_Y]$, and $g^\star(y,\bX \con \theta) \in [c_g,C_g]$ for all $y \in \suppYX$, $a \in \{0,1\}$, $z \in \suppZb$, and $\bX \in \suppX$. 
\end{itemize}

Assumption \HL{A3} requires that all nuisance functions, as well as the resulting $g^\star$ function, are uniformly bounded away from 0 and infinity. This condition guarantees that the iterative procedure introduced below is both well-defined and numerically stable.

\begin{theorem} 
\label{thm-global contraction-binary}
    Suppose that Assumptions \HL{A1}-\HL{A2} and \HL{IV1}-\HL{IV4} hold for $(\potY{0},A,Z,\bX)$, with corresponding parameters $\theta = \{ \fz,\fa,\fy \}$. Further suppose that \HL{A3} holds for $\theta$. 
    For a fixed $\bX \in \suppX$, define the mapping $\overline{\Psi}: \mathcal{L}_+^{\infty} (\mathcal{Y}_{\bX}) \rightarrow \mathcal{L}_+^{\infty} (\mathcal{Y}_{\bX})$ as the normalized version of $\Psi$ in \eqref{eq-Psi-binary}:
    \begin{align*}
    &
    \overline{\Psi}
    \big( g(y,\bX) \con \theta \big)
    =
    \frac{ \Psi\big( g(y,\bX)  \con \theta \big) }{\int \Psi\big( g(t,\bX)  \con \theta \big) \, dt} 
     \ .
    \end{align*}
    Consider the iterative update defined by
    \begin{align*}
        \overline{h}^{(j+1)}(y,\bX \con \theta) = \overline{\Psi}(\overline{h}^{(j)}(y,\bX \con \theta) \con \theta) \ , 
        \quad j \in \{0,1,\ldots\} \ ,
        \numeq 
        \label{eq-iterative update-binary}
    \end{align*}
     with an arbitrary initial function $\overline{h}^{(0)} \in \mathcal{L}_+^{\infty} (\mathcal{Y}_{\bX})$ satisfying $\int \overline{h}^{(0)}(y,\bX) \, dy = 1$. Then, as $j \rightarrow \infty$, $\overline{h}^{(j)}(y,\bX \con \theta)$ converges to $g^{\star}(y,\bX \con \theta)$ in \eqref{eq-fixed point equation-binary} for all $y \in \suppYX$, i.e., 
    \begin{align*}
        \sup_{y \in \suppYX}    
        \big| \overline{h}^{(j)}(y,\bX \con \theta) - g^{\star}(y,\bX \con \theta)
        \big|
        \rightarrow 0
        \quad \text{ as } \quad j \rightarrow \infty \ .
    \end{align*}    
    Moreover, the convergence of the iterative update  \eqref{eq-iterative update-binary} is exponentially fast with respect to a certain divergence function; see {\SUPP} \ref{sec-supp-global contraction} for details.     
\end{theorem}

Theorem \ref{thm-global contraction-binary} states that the fixed-point equation \eqref{eq-fixed point equation-binary} can be solved using the iterative update \eqref{eq-iterative update-binary}, with the sequence converging at an exponential rate to the solution. This exponential convergence rate occurs because the mapping $\overline{\Psi}$ is a global contraction toward $g^\star$, which we formally establish in {\SUPP} \ref{sec-supp-global contraction}. As a result, for any given $\theta$, $g^\star(y,\bX \con \theta)$ can be computed efficiently for any user-specified $(y,\bX)$, provided the boundedness condition \HL{A3} holds. This, in turn, enables the computation of $\alpha$ and $\beta$ from \eqref{eq-alpha g-binary} and \eqref{eq-beta g-binary}. These results are crucial for constructing the ATT estimator in the following Section, which relies not only on estimates of $\theta$ but also on the corresponding estimates of $\alpha$ and $\beta$.

\section{Semiparametric Efficiency Theory} \label{sec-theory}

In this Section, we characterize IFs for the ATT in model $\model$, a semiparametric model for the observed data law satisfying Assumptions \HL{A1}-\HL{A2} and \HL{IV1}-\HL{IV4}, without imposing any restrictions on the nuisance functions. The following Theorem states the result:

\begin{theorem} \label{thm-IF}
Suppose that Assumptions \HL{A1}-\HL{A2} and \HL{IV1}-\HL{IV4} hold. Let $\omega(\potY{0},Z,\bX)$ be a function satisfying the following conditions: \\
\makebox[1.2cm][l]{\HT{$\omega$-i}} $\omega = \nu - \EXP \{ \nu \cond \potY{0},\bX \}
    -
    \EXP ( \nu \cond Z,\bX )
    +
    \EXP ( \nu \cond \bX )$ for some function $\nu(\potY{0},Z,\bX)$;\\
\makebox[1.2cm][l]{\HT{$\omega$-ii}} $\omega$ solves the following integral equation:  
    \begin{align} \label{eq-w IE}
    &
    \EXP \big[ \big\{ \potY{0} - \mu^*(Z,\bX) 
    \big\}  - \big\{  \omega (\potY{0},Z,\bX) - \mu^*(Z,\bX \con \omega) \big\} \cond A=1,\potY{0},\bX \big] = 0 \ .
    \end{align} 
Then, the following function is an influence function for $\tau^*$ in model $\model$:
  \begin{align*}
  \InfFt^*(\bO)
  =
  \frac{ 
   \left[ 
    \begin{array}{l}
    \{ A - (1-A) \alpha^*(Y,\bX) \beta^*(Z,\bX) \}
     \big\{ Y - \mu^*(Z,\bX) \big\}     
     - A \tau^*
     \\
       +
       (1-A) \alpha^*(Y,\bX) \beta^*(Z,\bX) \big\{ \omega(Y,Z,\bX) - \mu^*(Z,\bX \con \omega) \big\}
  \\
  +
    A
    \mu^*(Z,\bX \con \omega)
    + 
    (1-A) \omega(Y,Z,\bX) 
    \end{array}
    \right]
    }{\EXP(A)}
     \ .
     \numeq 
     \label{eq-IF}
\end{align*}
In addition, every influence function for $\tau^*$ in model $\model$ is of the form \eqref{eq-IF}.  
\end{theorem}
The IF in \eqref{eq-IF} contains two sets of terms. The first line on the right-hand side is precisely the AIPW representation in \eqref{eq-AIPW}. This corresponds to the EIF for the ATT in a semiparametric model where the odds ratio function $\alpha^*$ is known a priori, while all other components of the observed data distribution remain unrestricted \citep{Robins2000_Sensitivity}. The second and third lines capture the additional variability induced by allowing $\alpha^*$ to be unrestricted. To the best of our knowledge, this result is entirely new in the IV literature and provides a substantial improvement over previous works that rely on $\alpha^*$ either being known or known up to a finite-dimensional parameter \citep{Robins2000_Sensitivity, Sun2018, Liu2020, Sun2025}.  

We note that directly solving the integral equation in \eqref{eq-w IE} may be challenging in practice, as it involves multiple nested dependencies within the counterfactual data distribution. However, we can strengthen the result of Theorem \ref{thm-IF} by leveraging the binary nature of $Z$, as stated in Theorem \ref{thm-IF binary Z}.

\begin{theorem} \label{thm-IF binary Z}

Suppose that Assumptions \HL{A1}-\HL{A2} and \HL{IV1}-\HL{IV4} hold. Let $
        \omega^*(\potY{0},Z,\bX) 
        =
    [ L^*(\potY{0},\bX) - \EXP \{ L^*(\potY{0},\bX) \cond \bX \} ]
    \{ Z - \EXP ( Z \cond \bX ) \}$, where $L^*$ is given by 
    \begin{align*}
    L^*(\potY{0},\bX) =  \frac{\potY{0} + \EXP\{ Z \cond A=1,\potY{0},\bX\} B^* (\bX) + C^* (\bX)}{ \EXP\{ Z \cond A=1,\potY{0}, \bX \} - \EXP (Z \cond \bX)}  \ .
\end{align*} 
Closed-form expressions for $B^*$ and $C^*$ are given in equation \eqref{eq-proof-BC} in {\SUPP} \ref{sec-proof-IF binary Z}. Then, $\omega^*$ is the unique function satisfying conditions \HL{$\omega$-i} and \HL{$\omega$-ii} in Theorem \ref{thm-IF}. Therefore, the influence function in \eqref{eq-IF} evaluated at $\omega^*$ is the efficient influence function for $\tau^*$ in model $\model$.
\end{theorem}

Theorem \ref{thm-IF binary Z} provides two additional results compared with Theorem \ref{thm-IF} by leveraging the binary nature of $Z$. First, it provides a closed-form expression for $\omega^*$ that satisfies conditions \HL{$\omega$-i} and \HL{$\omega$-ii}. Second, it establishes the uniqueness of $\omega^*$, which in turn implies that the IF for the ATT is unique, and therefore indeed the EIF. Intuitively, this result arises because, when $Z$ is binary, the integral equation in \eqref{eq-w IE} can be written as a Fredholm integral equation of the second kind with a separable kernel, which can be solved as a linear system \citep{Kress2014}. Furthermore, this linear system is well-posed under the relevance assumption \HL{IV3}, and the functions $B^*$ and $C^*$ can be obtained as its solution. Despite this simplification, the algebraic steps remain nontrivial, as reflected in the complexity of the resulting expressions. These derivations are provided in {\SUPP} \ref{sec-proof-IF binary Z}.

\section{Estimation} \label{sec-estimation}

\subsection{Semiparametric Efficient Estimator for the ATT}

This Section focuses on constructing an estimator for the ATT based on the EIF in Theorem \ref{thm-IF binary Z}. The proposed estimator employs a cross-fitting approach \citep{Schick1986, DDML2018}. Specifically, we randomly partition the $N$ units, denoted by $\mathcal{I} = \{1,\ldots,N\}$, into $K$ non-overlapping folds $\{ \II_1,\ldots, \II_K \}$, and define the complement of each fold as $\II_k^c = \mathcal{I} \setminus \II_k$ for $k  \in \{1,\ldots,K\}$. For each fold $k$, we proceed as follows: 
\begin{itemize}[labelwidth=1.5cm]
    \item[(Step 1)] Using units in $\II_k^c$, estimate the nuisance functions $\theta^* \equiv \{ \fz^* \equiv f^*(Z \cond \bX), \fa^* \equiv f^*(A \cond Z,\bX), \fy^* \equiv f^*(Y \cond A=0,Z,\bX) \}$ defined in Section \ref{sec-our parameterization}; 
    
    \item[(Step 2)] Obtain estimates of the odds ratio function $\alpha^*$ and baseline odds function $\beta^*$ based on the estimates from Step 1;
    
    \item[(Step 3)] Obtain an estimate of the EIF from Theorem \ref{thm-IF binary Z} using the estimates from Steps 1 and 2.
\end{itemize}
Finally, we aggregate the $K$ fold-specific EIF estimates from Step 3 to obtain an ATT estimate. Further details of the procedure are provided below.

\vspace*{0.2cm}

\noindent (\textit{Step 1: Estimation of the Nuisance Functions})

\noindent For each $k \in \{1,\ldots,K\}$, we estimate the nuisance functions using observations in $\II_k^c$. First, to estimate the conditional probability $\fz^*$, we model $Z$ as the response and $\bX$ as the regressors, employing probabilistic machine learning methods and their ensemble via the superlearner method \citep{vvLaan2007}. The conditional probability $\fa^*$ is estimated similarly, with $A$ as the response and $(Z,\bX)$ as the regressors. For the outcome densities $\fy^*$, the estimation approach depends on the type of $Y$. For categorical $Y$, a superlearner-based approach can be again used, similar to the estimation of the other nuisance functions, taking $Y$ as the response and $(Z,\bX)$ as the regressors, using only units with $A=0$. For continuous $Y$, we apply nonparametric conditional density estimation methods, including kernel regression \citep{HallRacineLi2004}, smoothing splines \citep{gss2014}, and ensembles of these approaches, again using only units with $A=0$. Additional details on the base machine learning algorithms in the superlearner library and software for conditional density estimation are provided in {\SUPP} \ref{sec-supp-estimation}. We denote the resulting estimates by $\widehat{\theta}\LSS \equiv \{ \fzhat\LSS, \fahat\LSS, \fyhat\LSS \}$, respectively, where the superscript $\LSS$ indicates that the $k$th fold $\II_k$ is excluded from estimation.

These nuisance estimates will be used in the subsequent steps. In particular, the density ratio of the outcome regression, $R_Y^*$ in \eqref{eq-Ry-binary}, plays a central role in estimating $\alpha^*$ and $\beta^*$. However, for continuous $Y$, the naive density ratio estimate obtained as the ratio of two density estimates often produces unstable results, as it focuses on individual densities rather than their ratio. To overcome this challenge, one can employ strategies that target the density ratio directly, yielding more stable and reliable estimates. Some popular nonparametric density ratio estimation methods include the  \textit{unconstrained least-squares importance fitting} and the \textit{Kullback-Leibler importance estimation procedure} \citep{DR2009_uLSIF, DensityRatio2007_1, DensityRatio2007_2}. In essence, these methods embed the density ratio within a reproducing kernel Hilbert space and estimate it by minimizing either the mean squared error or the Kullback-Leibler divergence. Further technical details for these targeted approaches can be found in {\SUPP} \ref{sec-supp-estimation}.

\vspace*{0.2cm}

\noindent (\textit{Step 2: Estimation of $\alpha^*$ and $\beta^*$ via Iterative Updates})

\noindent Given the estimates $\widehat{\theta}\LSS$, we now obtain estimates of the odds ratio function $\alpha^*$ and the baseline odds function $\beta^*$. In brief, once $\widehat{\theta}\LSS$ is available, the corresponding estimates of $\alpha^*$ and $\beta^*$ can be directly obtained from Theorems \ref{thm-Psi-binary} and \ref{thm-global contraction-binary}, provided that $\widehat{\theta}\LSS$ satisfies the boundedness condition \HL{A3}, which we formally introduce as an assumption later; see Theorem \ref{thm-AN} below. Specifically, for each $\bX_{i}$ with $i \in \II_k\LSS$, we implement the iterative update \eqref{eq-iterative update-binary} with $\widehat{\theta}\LSS$ substituted in. By Theorem \ref{thm-global contraction-binary}, the resulting sequence converges to the fixed point. The resulting estimates of $\alpha^*$ and $\beta^*$ are then obtained via \eqref{eq-alpha g-binary} and \eqref{eq-beta g-binary}, and are denoted by $\widehat{\alpha}\LSS \equiv \alpha(\cdot \con \widehat{\theta}\LSS)$ and $\widehat{\beta}\LSS \equiv \beta(\cdot \con \widehat{\theta}\LSS)$, respectively. For binary or polytomous $Y$, estimation of $\alpha^*$ and $\beta^*$ is straightforward, as the iterative updates only require evaluation over the finite set of outcome values for each $\bX_{i}$. For continuous $Y$, however, direct implementation of the iterative updates is infeasible because the support of $Y$ is uncountable. In such cases, we recommend applying the iterative updates over a sufficiently fine grid of outcome values to numerically approximate $\widehat{\alpha}\LSS$ and $\widehat{\beta}\LSS$.

\vspace*{0.2cm}

\noindent (\textit{Step 3: Estimation of the EIF})

\noindent Given $\widehat{\theta}\LSS$, $\widehat{\alpha}\LSS$, and $\widehat{\beta}\LSS$  from the previous steps, we can now construct an estimate of the density $f^*(\potY{0},A, Z  \cond \bX)$ using the relationship in \eqref{eq-Liu parameterization}. We denote this estimate by $\widehat{f}\LSS(\potY{0},A,Z \cond \bX) \equiv f(\potY{0},A,Z \cond \bX \con \widehat{\theta}\LSS )$. Substituting $\widehat{f}\LSS$ into the nuisance functions of the EIF in \eqref{eq-IF} yields plug-in estimators. For example, for a generic function $b(\potY{0},\bX)$, $\mu^*(Z,\bX \con b)$ can be estimated by 
\begin{align*}
    \widehat{\mu}\LSS(Z,\bX \con b ) = \int b(y,\bX) \widehat{f}\LSS(\potY{0}=y \cond A=1,Z,\bX) \, dy \ .
\end{align*}
The other components of the EIF can be obtained through analogous substitution procedures, and we use similar notation for their estimators, e.g., $\widehat{\omega}\LSS$. Using these estimates, we define $\widehat{\phi}\LSS$ as an estimate of the key part of the EIF in \eqref{eq-IF}:
\begin{align*}
    &
    \widehat{\phi}\LSS(\bO)  
    \\
    &
    =
    \left[ 
    \begin{array}{l}
    \{ A - (1-A) \widehat{\alpha}\LSS(Y,\bX) \widehat{\beta}\LSS (Z,\bX) \}
     \big\{ Y - \widehat{\mu}\LSS(Z,\bX) \big\}      
     \\
       +
       (1-A) \widehat{\alpha}\LSS(Y,\bX) \widehat{\beta}\LSS(Z,\bX) \big\{ \widehat{\omega}\LSS(Y,Z,\bX) - \widehat{\mu}\LSS(Z,\bX \con \widehat{\omega}\LSS) \big\}
  \\
  +
    A
    \widehat{\mu}\LSS(Z,\bX \con \widehat{\omega}\LSS)
    + 
    (1-A) \widehat{\omega}\LSS(Y,Z,\bX) 
    \end{array}
    \right] \ .  
\end{align*}

Finally, given $\widehat{\phi}\LSS$ for $k \in \{ 1,\ldots,K \}$, one can construct an ATT estimate based on the EIF as follows:
\begin{align*}
    \widehat{\tau}
    =
    \frac{  
    \sum_{k=1}^{K}
    \sum_{i \in \II_k}
    \widehat{\phi}\LSS(\bO_i) 
    }
    { 
    \sum_{i=1}^{N} A_i 
    }
    \ .
    \numeq
    \label{eq-tauhat}
\end{align*}

It remains to establish statistical properties of $\widehat{\tau}$. To this end, we introduce additional conditions and define the following notation.  For $k \in \{1,\ldots,K\}$, let $r_{Y}\LSS$, $r_{A}\LSS$, and $r_{Z}\LSS$ denote the estimation errors of the nuisance functions measured in the $L^2(P)$-norm, i.e., 
\begin{align*}
& r_{Y}\LSS = \| \widehat{f}\LSS (\potY{0} \cond A=0,Z,\bX) - f^*(\potY{0} \cond A=0,Z,\bX) \|_{P,2} \ ,
\\
&
r_{A}\LSS = \| \widehat{f}\LSS  (A=1 \cond Z,\bX) - f^*(A=1 \cond Z,\bX) \|_{P,2} \ ,
\numeq 
\label{eq-error}
\\
&
r_{Z}\LSS = \| \widehat{f}\LSS(Z=1 \cond \bX) - f^*(Z=1 \cond \bX) \|_{P,2} \ .
\end{align*}

With this notation in place, we now state the assumptions required to establish the asymptotic normality of $\widehat{\tau}$:
\begin{itemize}

\item[\HT{A4}] (Bounded $Y$) There exists a constant $B_{Y}>0$ such that $| Y | \leq B_{Y}$ almost surely.

\item[\HT{A5}] (Strong Relevance) For $k \in \{1,\ldots,K\}$ and binary $Z$, there exists a constant $B_{\beta}>0$ such that $| \beta^*(1,\bX) - \beta^*(0,\bX) | \geq B_{\beta}$ and $| \widehat{\beta}\LSS (1,\bX)  - \widehat{\beta}\LSS(0,\bX) | \geq B_{\beta}$ for all $\bX \in \suppX$.

\item[\HT{A6}] (Consistent Estimation) For $k \in \{1,\ldots,K\}$, we have $r_Y\LSS + r_A\LSS+r_Z\LSS = o_P(1)$. 

\item[\HT{A7}] (Second-order Estimation Error) For $k \in \{1,\ldots,K\}$, we have
\begin{align*}
\{ r_Y\LSS \}^{2} + \{ r_A\LSS \}^2 + r_Y\LSS r_A\LSS +  r_Y\LSS r_Z\LSS + r_A\LSS r_Z\LSS = o_P(N^{-1/2}) \ .
\numeq
\label{eq-bias structure}
\end{align*}

\end{itemize}

Assumption \HL{A4} states that the outcome is uniformly bounded. Together with the boundedness condition \HL{A3}, Assumption \HL{A5} guarantees that the difference in the extended propensity score, $| f^*(A=1 \cond \potY{0},Z=1,\bX) - f^*(A=1 \cond \potY{0},Z=0,\bX)|$, is bounded away from zero, which can be interpreted as a strengthening of the IV relevance condition, i.e., \HL{IV3}. An analogous requirement is imposed on the estimated extended propensity score. Assumption \HL{A6} states that the nuisance functions are consistently estimated in the $L^2(P)$-norm. Assumption \HL{A7} further requires that certain nuisance functions be estimated at a sufficiently fast rate. In particular, \HL{A7} requires that the outcome model $\fy^*$ and the propensity score model $\fa^*$ be estimated at an $o_P(N^{-1/4})$ rate, which typically necessitates appropriate smoothness conditions. By contrast, the instrument model $\fz^*$ may be estimated at a substantially slower rate, provided that its cross‐product estimation errors with the outcome and propensity score models, namely $r_{Z}\LSS r_{Y}\LSS$ and $r_{Z}\LSS r_{A}\LSS$, remain $o_P(N^{-1/2})$ in each case. If all nuisance functions are parametrically specified, Assumption \HL{A7} corresponds to the requirement that the working models $\fy$ and $\fa$ are correctly specified, while $\fz$ may be misspecified. We note that Assumption \HL{A7} differs from conditions in previous works \citep{Sun2018, Liu2020, Sun2025}, reflecting differences in parameterization; a detailed discussion is provided in Section \ref{sec-bias comparison}.

The following Theorem states that $\widehat{\tau}$ is an asymptotically normal and semiparametric efficient estimator of $\tau^*$ in model $\model$:
\begin{theorem} \label{thm-AN}
Suppose that Assumptions \HL{A1}-\HL{A2}, \HL{A4}-\HL{A7}, and \HL{IV1}-\HL{IV4} hold. Further suppose that \HL{A3} holds for the true nuisance functions $\theta^*$ and estimated nuisance functions $\widehat{\theta}\LSS$ for $k \in \{1,\ldots,K\}$. Then, the ATT estimator in \eqref{eq-tauhat} is asymptotically normal as $\sqrt{N} ( \widehat{\tau} - \tau^* )    \stackrel{D}{\rightarrow}    N ( 0, \sigma^2 )$, where $\sigma^2 \equiv \VAR \{ \InfFt^*(\bO) \}$ is the semiparametric efficiency bound for $\tau^*$ in model $\model$.  Moreover, a consistent estimator of the asymptotic variance of $\widehat{\tau}$ is given by
\begin{align*}
    \widehat{\sigma}^2
    =
    \frac{1}{N} 
    \sum_{k=1}^{K}
    \sum_{i \in \II_k}
    \bigg\{ \frac{ \widehat{\phi}\LSS(\bO_i) - A_i \widehat{\tau}  }{ \sum_{i=1}^{N} A_i / N } \bigg\}^2 \ .
\end{align*}
    
\end{theorem}
Theorem \ref{thm-AN} also implies that Wald-type confidence intervals for $\tau^*$ can be constructed. Specifically, let $z_{c}$ denote the 100$c$-th percentile of the standard normal distribution. Then, a $100(1-c)$\% confidence interval for $\tau^*$ is given by $(\widehat{\tau} + z_{c/2} \widehat{\sigma}/N^{1/2}, \widehat{\tau} + z_{1-c/2} \widehat{\sigma}/N^{1/2})$.

\subsection{Comparison of Required Estimation Errors Across Parameterizations} \label{sec-bias comparison}

We conclude this Section by detailing how Assumption \HL{A7} differs from conditions in previous works \citep{Sun2018, Liu2020, Sun2025}; we remind the reader that these parameterizations are reviewed in Section \ref{sec-previous parameterization}. Under the parameterization of \citet{Sun2018} and \citet{Liu2020}, a condition corresponding to \HL{A7} can be expressed as
\begin{align*}
    \{ r_{\alpha}\LSS \}^2
    &
    +
    r_{\alpha}\LSS r_{\beta} \LSS 
    +
    r_{\alpha}\LSS r_{Y} \LSS 
    +
    r_{\beta}\LSS r_{Y}\LSS 
    \\
    &
    +
    r_{\alpha}\LSS r_{Z}\LSS  
    +
    r_{\beta}\LSS r_{Z}\LSS  
    +
    r_{Y}\LSS r_{Z}\LSS  
    =
    o_P(N^{-1/2}) \ .
    \numeq
    \label{eq-Liu Bias}
\end{align*}
Here, $r_{\alpha}\LSS$ and $r_{\beta}\LSS$ denote the $L^2(P)$-norm estimation errors of $\alpha^*$ and $\beta^*$, respectively, defined analogously to the estimation errors in \eqref{eq-error}. The derivation of \eqref{eq-Liu Bias} is provided in {\SUPP} \ref{sec-supp-other para}.
Under this parameterization, $r_{\alpha}\LSS$ is required to attain an $o_P(N^{-1/4})$ rate, while weaker requirements are imposed on the other three error terms $r_{\beta}\LSS$, $r_{Y}\LSS$, and $r_{Z}\LSS$. Specifically, it suffices that the six cross-product errors in \eqref{eq-Liu Bias} are $o_P(N^{-1/2})$, rather than requiring each of $r_{\beta}\LSS$, $r_{Y}\LSS$, and $r_{Z}\LSS$ individually to be $o_P(N^{-1/4})$. When all nuisance functions are parametrically specified, \eqref{eq-Liu Bias} implies that the working model for $\alpha$ must be correctly specified, whereas it is sufficient for only two of the other three working models ($\beta$, $\fy$, and $\fz$) to be correctly specified. Of note, for the asymptotic normality of their estimators, \citet{Sun2018} and \citet{Liu2020} required both $\alpha$ and $\fz$ to be correctly specified, while only one of $\beta$ or $\fy$ needed to be correctly specified; however, the correct specification of $\fz$ can in fact be relaxed.

In contrast, under the parameterization of \citet{Sun2025}, a parallel condition to \HL{A7} takes the form of
\begin{align*}
    \{ r_{\alpha}\LSS \}^2
    +
    \{ r_{\beta} \LSS \}^2
    &
    +
    r_{\alpha}\LSS r_{\beta} \LSS 
    +
    r_{\alpha}\LSS r_{\potY{0}} \LSS 
    +
    r_{\beta}\LSS r_{\potY{0}}\LSS 
    \\
    &
    +
    r_{\alpha}\LSS r_{Z}\LSS  
    +
    r_{\beta}\LSS  r_{Z}\LSS  
    +
    r_{\potY{0}}\LSS  r_{Z}\LSS  
    =
    o_P(N^{-1/2}) \ ,
    \numeq 
    \label{eq-SMW Bias}
\end{align*}
where $r_{\potY{0}}\LSS$ denotes the $L^2(P)$-norm estimation error of $f_{\potY{0}}^*$; see {\SUPP} \ref{sec-supp-other para} for a detailed derivation of \eqref{eq-SMW Bias}. Unlike the previous case, this condition requires both $r_{\alpha}\LSS$ and $r_{\beta}\LSS$ to attain an $o_P(N^{-1/4})$ rate, while weaker requirements apply to $r_{\potY{0}}\LSS$ and $r_{Z}\LSS$. Specifically, it is sufficient that the six cross-product terms in \eqref{eq-SMW Bias} converge at rate $o_P(N^{-1/2})$, without requiring $r_{\potY{0}}\LSS$ and $r_{Z}\LSS$ themselves to achieve an $o_P(N^{-1/4})$ rate. Consequently, with parametric estimation, working models $\alpha$ and $\beta$ (and thus that for the extended propensity score $\pi$ in \eqref{eq-nuis-pi}) must be correctly specified, while only one of $f_{\potY{0}}$ or $\fz$ needs to be correctly specified.

While it is unsurprising that the requirements across these three parameterizations differ substantially, our approach provides a notable practical advantage when empirically assessing the necessary conditions. To illustrate this, consider the scenario in which all nuisance functions are parametrically estimated. The parameterizations by others require correct specification of $\alpha$; however, $\alpha$ is a counterfactual quantity that cannot be directly observed, making empirical assessment of goodness-of-fit challenging. In contrast, our parameterization avoids this issue by relying solely on densities derived from the observed data as model parameters. This allows goodness-of-fit to be, in principle, directly assessed using standard methods and facilitates the selection of appropriate parametric models based entirely on the observed data. Furthermore, our approach accommodates standard tuning and regularization for nuisance functions when using adaptive nonparametric or machine learning methods. Together, these features highlight a practical benefit of our parameterization.

\section{Falsification Implications} \label{sec-falsification}

While many key IV assumptions (e.g., monotonicity \citep{Angrist1996}) cannot be directly verified, they often yield conditions that can be refuted by the observed data \citep{Huber2015, Kitagawa2015, Mourifie2017, Keele2019, Danieli2025}. In other words, while these assumptions cannot be proven true, they can be empirically falsified. These falsifiable implications provide a practical diagnostic tool; specifically, if falsifiable implications appear violated, the IV approach is likely invalid for the data at hand, as the underlying IV assumptions are refuted. Conversely, if the implications hold, the IV approach is not strongly contradicted, providing supporting evidence for the credibility of the IV model. Therefore, by focusing on falsifiable implications, researchers can systematically assess the plausibility of IV assumptions and strengthen the robustness of their analysis.

Motivated by this principle, we investigate the falsifiability of the logit-separable IV model under assumptions \HL{IV1}-\HL{IV4}, an important aspect that has been largely overlooked in prior studies relying on the same IV assumptions \citep{Sun2018, Liu2020, Sun2025}. Notably, the logit-separable IV model entails falsifiable implications, which are formalized in the following Theorem:
\begin{theorem} \label{thm-falsification}
    Suppose that Assumptions \HL{A1}-\HL{A2} and \HL{IV1}-\HL{IV4} hold. Then, the following result holds for all $y\in \suppYX$ and $\bX \in \suppX$:
    \begin{align*}
    \frac{
        f^*(Y=y,A=0 \cond Z=1,\bX) - f^*(Y=y,A=0 \cond Z=0,\bX) }{ 
        f^*(A=0 \cond Z=1,\bX) - f^*(A=0 \cond Z=0,\bX)
        } \geq 0 \ .
        \numeq
        \label{eq-falsification}
    \end{align*}
\end{theorem}
 
Theorem \ref{thm-falsification} provides a basis for constructing a falsification test for the IV model. Specifically, if condition \eqref{eq-falsification} is violated for some $(y,\bX)$, then at least one of the IV assumptions must be violated provided that Assumptions \HL{A1} and \HL{A2} hold. Details on the implementation of the falsification test are provided in {\SUPP} \ref{sec-supp-falsification}.

Condition \eqref{eq-falsification} is closely related to the results of \citet{Kitagawa2015}, which investigated falsifiable implications of the IV model under the monotonicity assumption of \citet{Angrist1996}. Suppressing covariates for simplicity, \eqref{eq-falsification} can be expressed as
\begin{align*}
\begin{array}{l}
\text{if } f^*(A=1 \cond Z=0) > f^*(A=1 \cond Z=1) , \\
\hspace*{2cm} \text{then } f^*(Y=y, A=0 \cond Z=1) \geq f^*(Y=y, A=0 \cond Z=0) \ ; \\[0.25cm]
\text{if } f^*(A=1 \cond Z=0) < f^*(A=1 \cond Z=1), \\
\hspace*{2cm} \text{then } f^*(Y=y, A=0 \cond Z=0) \geq f^*(Y=y, A=0 \cond Z=1) \ .
\end{array}
\numeq
\label{eq-falsification binary}
\end{align*}
In contrast, the condition in \citet{Kitagawa2015} can be written as
\begin{align*}
\begin{array}{l}
f^*(Y=y, A=1 \cond Z=1) \geq f^*(Y=y, A=1 \cond Z=0) \ , \\[0.25cm]
f^*(Y=y, A=0 \cond Z=0) \geq f^*(Y=y, A=0 \cond Z=1) \ .
\end{array}
\numeq
\label{eq-falsification Kitagawa}
\end{align*}
Both \eqref{eq-falsification binary} and \eqref{eq-falsification Kitagawa} impose restrictions on the joint distribution of $(Y,A)$ conditional on $Z$. However, \eqref{eq-falsification binary} involves only the $A=0$ case, whereas \eqref{eq-falsification Kitagawa} imposes constraints for both $A=0$ and $A=1$ cases. Moreover, in \eqref{eq-falsification binary}, the direction of the inequality for $f(Y,A=0 \cond Z)$ depends on the observed association between $A$ and $Z$. In contrast, in \eqref{eq-falsification Kitagawa}, only the case $f^*(Y=y,A=0 \cond Z=0) \geq f^*(Y=y,A=0 \cond Z=1)$ is possible for the $A=0$ case. This difference arises because our IV framework does not impose a directional restriction on the relationship between $A$ and $Z$, whereas the monotonicity assumption does.

\section{Simulation} \label{sec-simulation}

We conducted a simulation study to evaluate the finite-sample performance of the proposed estimator. We considered two data-generating processes (DGPs), one with a continuous outcome and the other with a binary outcome. For each unit, we generated a three-dimensional vector of pre-treatment covariates $\bX=(X_{1},X_{3},X_{3})\T$, where each component was independently drawn from a standard normal distribution. Let $W = \sum_{j=1}^{3} X_{j}$. The instrument $Z$ was then generated from $Z \sim \text{Ber}(\text{expit}(0.1 W) )$ where $\text{expit}(v) = \{ 1 + \exp(-v) \}^{-1}$. The continuous and binary potential outcomes were generated as follows:
\begin{align*}
    &
    \text{Continuous $Y$:}
    &&
    \potY{0} \sim N ( 0.1 W, 1  )
    \ , 
    &&
    \potY{1} \sim N ( 1+0.1 W, 0.25^2 )
    \ , 
    \\
    &
    \text{Binary $Y$:}
    &&
    \potY{0} \sim \text{Ber} ( \text{expit}(0.2 W - 0.25 ) )
    \ , 
    &&
    \potY{1} \sim \text{Ber} ( \text{expit}(0.2 W + 0.25 ) )
    \ .
\end{align*}
Lastly, the treatment $A$ was generated from $A \cond \{ \potY{0},Z,\bX \} \sim \text{Ber}( \text{expit} ( \log \alpha^*(\potY{0},\bX) + \log \beta^*(Z,\bX) ) )$; here, $\log \beta^*(z,\bX) = -2 + 4z + 0.1 W$, and $\log \alpha^*(y,\bX) =  (0.4+0.05W) \mathfrak{a}^*(y) $. The function $\mathfrak{a}^*$ was specified according to the outcome type as follows:
\begin{align*}
    & \text{Continuous $Y$:}
    &&
    \mathfrak{a}^*(y)
    =
    \big\{ 1 - (y-1)^{-2}  \big\}
    \ind \big\{ y \notin (0,2) \big\}
    +
    y (y-2)
    \ind \big\{ y \in [0,2] \big\} 
    \ ,
    \\
    &
    \text{Binary $Y$:}
    &&
    \mathfrak{a}^*(y) = 2y \ .
\end{align*}
The corresponding ATTs for the continuous and binary outcome DGPs are 1 and 0.082, respectively.

For each DGP, we generated $N \in \{1000, 2000, 4000\}$ study units. Based on these units, the proposed estimator $\widehat{\tau}$ was obtained following the procedure described in Section \ref{sec-estimation}. For inference, we constructed 95\% Wald-type confidence intervals using the result of Theorem \ref{thm-AN}. Performance of the proposed estimators and the corresponding confidence intervals was evaluated based on 1000 repetitions for each value of $N$. To mitigate variability induced by random sample splitting, we adopted the median adjustment based on three repetitions of the cross-fitting procedure; see {\SUPP} \ref{sec-supp-median} for details.

As competing estimators, we considered the following three:
\begin{itemize}
    \item[(i)] $\widehat{\tau}_{\text{2SLS}}$: The coefficient of $\widehat{A}$ in the linear regression of $Y$ on $(\widehat{A},\bX)$, where $\widehat{A}$ is obtained from regressing $A$ on $(Z,\bX)$. 
    \item[(ii)] $\widehat{\tau}_{\text{Ign}}$: The estimator based on the EIF for the ATT under the assumption of no unmeasured confounding \citep{Hahn1998}, treating $(Z,\bX)$ as pre-treatment covariates.
    \item[(iii)] $\widehat{\tau}_{\text{LM}}$: The coefficient of $A$ in the linear regression model of $Y$ on $(A,Z,\bX)$, again treating $(Z,\bX)$ as pre-treatment covariates.
\end{itemize}
The first competitor, $\widehat{\tau}_{\text{2SLS}}$, commonly referred to as the two-stage least squares (2SLS) estimator, was included because it is a widely used IV method that can account for potential unmeasured confounding. However, it relies on the linearity assumption for consistency and may therefore yield biased estimates of the ATT under our DGPs. The other two competitors, $\widehat{\tau}_{\text{Ign}}$ and $\widehat{\tau}_{\text{LM}}$, do not account for potential unmeasured confounding and are thus expected to perform poorly. For $\widehat{\tau}_{\text{Ign}}$, we also implemented a cross-fitting procedure analogous to that used for our proposed estimator $\widehat{\tau}$; see {\SUPP} \ref{sec-supp-simulation} for implementation details.

The top panel of Figure \ref{fig-Simulation} summarizes the empirical distributions of the estimators. In both the continuous and binary $Y$ settings, the proposed estimator $\widehat{\tau}$ exhibits negligible bias across all sample sizes $N$, with variability decreasing as $N$ increases. In contrast, the competing estimators show notable bias, arising either from model misspecification or the failure to account for potential unmeasured confounding. The bottom panel presents numerical summaries for $\widehat{\tau}$. Consistent with the graphical results, both standard error estimates decrease as $N$ increases and are of comparable magnitude. The empirical coverage rates are close to the nominal 95\% level. Overall, the simulation results indicate that the finite-sample performance of the proposed estimator $\widehat{\tau}$ is consistent with the asymptotic properties established in Theorem \ref{thm-AN}.

\begin{figure}[!htp]
    \centering
    \includegraphics[width=1\linewidth]{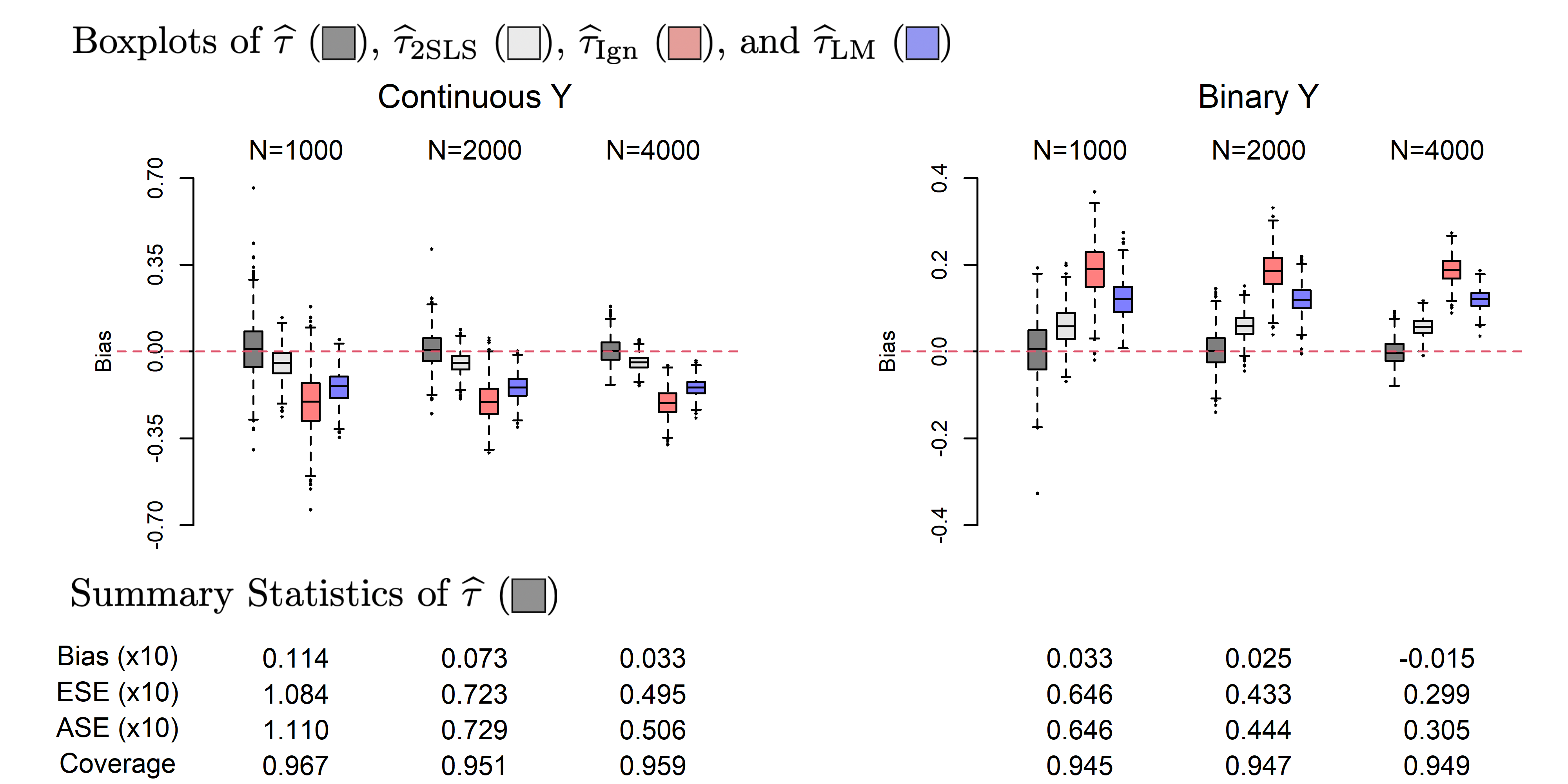}
    \caption{\footnotesize A Graphical Summary of the Simulation Results. The left and right panels correspond to the DGPs with continuous and binary $Y$, respectively. In the top panel, boxplots show the bias of each estimator for sample sizes $N \in \{1000,2000,4000\}$. The bottom panel presents numerical summaries for $\widehat{\tau}$. Each row shows the empirical bias, empirical standard error (ESE), asymptotic standard error based on the proposed variance estimator (ASE), and the empirical coverage rate of 95\% confidence intervals computed using the ASE. Biases and standard errors are scaled by a factor of 10.}
    \label{fig-Simulation}
\end{figure}

\section{Application} \label{sec-application}

We applied the proposed method to a real-world application based on the Job Corps study \citep{JobCorps2001, JobCorps2008}. Administered by the U.S. Department of Labor, Job Corps is the nation’s largest vocational training program for disadvantaged youth. As of 2023, it has served over two million participants and continues to support more than 40,000 young people annually across centers nationwide \citep{JobCorps2024, JobCorps2025}. The program offers intensive vocational training, academic education, and a range of support services aimed at improving employment outcomes. Given its scale and policy significance, there has been sustained interest in understanding the relationship between Job Corps participation and subsequent labor market outcomes. The most prominent effort in this regard is the National Job Corps Study, a nationally representative randomized experiment conducted in 2001 \citep{JobCorps2001, JobCorps2008}. While the study offers rich demographic data, unobserved factors, such as personality traits and latent characteristics, likely influence both program participation and outcomes, acting as unmeasured confounders. To address this, many studies have employed random assignment to the program group as an IV \citep{JobCorps2001, JobCorps2008, Blanco2013, JobCorps2019_Heterogeneous, Liu2025MIV}. Our analysis, detailed below, can be viewed as an extension of these works, conducted under a different IV assumption, specifically \HL{IV4}.

We used a dataset from the \texttt{causalweight} R package \citep{causalweightpackage}, which includes $N=9240$ individuals. The pre-treatment covariates $\bX$ comprised sex, age, race, ethnicity, education, English as a mother tongue, cohabitation or marital status, having children, and prior work experience. The candidate IV, $Z \in \{0,1\}$, indicates random assignment to the program group ($Z=1$) or control group ($Z=0$). Individuals assigned to the program group were eligible to participate in the Job Corps program, while those in the control group were ineligible for Job Corps for three years, although they could participate in other education or vocational training programs. Treatment receipt, $A \in \{0,1\}$, was defined as actual participation in any education or vocational training program (Job Corps or alternatives) within two years of random assignment, with $A=1$ indicating participation and $A=0$ indicating non-participation. The outcome, $Y \in [0,52]$, was the number of working weeks measured four years after random assignment.
 
In this context, the validity of the candidate IV $Z$ can be justified as follows. First, \HL{IV1} appears plausible, as the effect of random assignment to Job Corps eligibility on working weeks was expected to operate only indirectly through program participation. Second, \HL{IV2} holds because $Z$ was randomly assigned. Finally, \HL{IV3} is well-justified, as assignment to the program group would have substantially increased the likelihood of Job Corps enrollment due to the study design, and is confirmed empirically; see {\SUPP} \ref{sec-supp-application} for details. Therefore, $Z$ can be reasonably used as a valid IV for this analysis. 

The key assumption of our framework, \HL{IV4}, can be justified using the generative model described in Section \ref{sec-generative model}. Specifically, if each individual decides whether to participate in education or vocational training by comparing net utilities---perceived outcome minus cost plus an idiosyncratic error---then this selection mechanism is consistent with \HL{IV4}. Thus, establishing the plausibility of \HL{IV4} reduces to showing that the structural assumptions of the generative model are reasonable, as outlined below.

First, in the Job Corps context, the unobserved characteristic, denoted by $U$, likely represents latent traits such as intrinsic motivation, career ambition, or learning aptitude. Given $U$, the model defines the potential outcome as $\potY{a} = h_{a}(U, \bX)$, which is assumed to be strictly monotone in $U$. This assumption is plausible because higher levels of motivation or aptitude generally yield greater labor market returns. The model also incorporates $\overline{h}_a(U, \bX)$, representing an individual's subjective belief about their potential outcome. Crucially, $\overline{h}_a$ is allowed to differ from the true production function $h_{a}$ and does not require monotonicity, reflecting the possibility of imperfect self-assessment. It is further reasonable to assume that random assignment to Job Corps eligibility did not affect potential outcomes or individuals’ self-assessments, supporting the independence of both $h_a$ and $\overline{h}_a$ from $Z$. 

Another crucial component of the model is the cost function $c_{a}(Z, \bX)$, representing the cost of participating $(a=1)$ or not $(a=0)$ in education or vocational programs conditional on $(Z, \bX)$. In the Job Corps context, this cost can be interpreted as structural barriers to participation, such as administrative requirements, eligibility rules, or program awareness. Given this interpretation, the independence of $c_{a}$ from $U$ is reasonable for the Job Corps study, as an individual's intrinsic motivation or aptitude likely did not directly alter these structural costs. The model also assumes that this cost responds differentially to the instrument, i.e., $c_1(1, \bX) - c_1(0, \bX) \neq c_0(1, \bX) - c_0(0, \bX)$. This condition aligns with the study design. Note that assignment to the program group ($Z=1$) significantly lowered the cost of participation, as these individuals were eligible to participate in Job Corps. Conversely, assignment to the control group ($Z=0$) imposed a three-year embargo on Job Corps enrollment, forcing individuals to pursue alternative programs if they wished to participate and, consequently, face additional barriers. As a result, the Job Corps eligibility assignment likely substantially lowered the cost of participation, i.e., $c_1(1, \bX) \ll c_1(0, \bX)$, while the cost of not participating remained similar for both groups, i.e., $c_0(1, \bX) \approx c_0(0, \bX)$. This asymmetry justifies the required restriction on the cost function.
 
Under the linear utility model in \eqref{eq-linear utility}, it is reasonable to assume that individuals expect better labor market outcomes if they participate in an education or vocational program, implying $b_{h} \equiv b_{h,1} - b_{h,0} > 0$. At the same time, the coefficient on the cost component is expected to satisfy $b_{c} \equiv b_{c,1} - b_{c,0} \ll 0$, reflecting the asymmetry in participation costs described above.  The ratio $|b_h / b_c|$ can therefore be interpreted as the increase in motivation or aptitude required to offset the absence of Job Corps eligibility and yield the same likelihood of program participation. Nevertheless, the linear utility model rules out diminishing returns: the marginal effect of additional self-motivation on participation may decrease at higher levels of motivation. Since our generative model allows for nonlinear utility functions, more flexible specifications could be considered to capture these nuanced behavioral responses, although we omit such extensions here. Together, these elements provide a clear rationale and interpretation of the generative model. 

Using data from the Job Corps study, we estimated the ATT of participating in education or vocational training within two years of random assignment on the number of working weeks four years later. Consistent with our simulation studies, we compared four estimators: $\widehat{\tau}$, $\widehat{\tau}_{\text{2SLS}}$, $\widehat{\tau}_{\text{Ign}}$, and $\widehat{\tau}_{\text{LM}}$. For $\widehat{\tau}$ and $\widehat{\tau}_{\text{Ign}}$, we
applied the median adjustment by repeating cross-fitting 200 times; details are provided in {\SUPP} \ref{sec-supp-median}. We also empirically evaluated the support condition \HL{A2} and the falsification implication in Theorem \ref{thm-falsification}. Both conditions appear to hold, providing additional support for the validity of the results based on $\widehat{\tau}$ presented below. Further information can be found in {\SUPP} \ref{sec-supp-application}.

Table \ref{tab-ATT} reports the ATT estimates based on the four competing estimators. First, our estimator, $\widehat{\tau}$, indicates that education or vocational training increased participants’ annual work by 15.49 weeks, with the effect statistically significant at the 5\% level. Compared to the other IV estimator, $\widehat{\tau}_{\text{2SLS}}$, our estimator yields a slightly larger effect. This difference likely arises because the two estimators rely on different IV assumptions, and 2SLS requires the linearity assumption and homogeneous effects. Nevertheless, the estimates are not significantly different, as each lies within the 95\% confidence interval of the other. We also find that the two estimators developed under the assumption of no unmeasured confounding, $\widehat{\tau}_{\text{Ign}}$ and $\widehat{\tau}_{\text{LM}}$, show considerably smaller effect sizes compared to the two IV-based estimators. This stark contrast demonstrates that bias due to unmeasured confounding may be substantial.

\begin{table}[!htp]
\scriptsize
\renewcommand{\arraystretch}{1.05} \centering
\setlength{\tabcolsep}{10pt} 
\begin{tabular}{|c|c|c|c|c|}
\hline
Estimator & $\widehat{\tau}$ & $\widehat{\tau}_{\text{2SLS}}$ & $\widehat{\tau}_{\text{Ign}}$ & $\widehat{\tau}_{\text{LM}}$ \\ \hline
Estimate  & 15.49            & 11.23                          & 2.58                          & 2.44                         \\ \hline
ASE       & 4.88             & 2.62                           & 0.44                          & 0.43                         \\ \hline
95\% CI   & (5.92,25.05)     & (6.09,16.36)                   & (1.71,3.45)                   & (1.61,3.28)                  \\ \hline
\end{tabular}
\caption{Analysis Results of the Job Corps Study. The figures are reported in weeks per year.}
\label{tab-ATT}
\end{table}

In {\SUPP} \ref{sec-supp-application}, we further analyzed the QTTs of participation at the first, second, and third quartiles of the working weeks distribution using our IV approach. The results showed that all three QTTs are significantly positive at the 5\% level, indicating that participation increased working time across the distribution. Combined with the ATT results, these findings indicate that education or vocational training had consistently positive effects on participants' working time, both on average and across the distribution.

\section{Concluding Remarks} \label{sec-conclusion}

In this paper, we have focused on the IV framework under the logit-separable treatment mechanism assumption \HL{IV4}, also referred to as the separable binary treatment choice model. We first developed a variationally independent parameterization that depends only on quantities directly available from the observed data and does not involve counterfactual quantities. We then characterized the full class of IFs for the ATT in model $\model$, in which all nuisance functions are a priori unrestricted, and established that, due to the binary IV, it admits a unique closed-form representation and thus becomes the EIF. Building on this result, we constructed an EIF-based estimator that is consistent, asymptotically normal, and semiparametric efficient, provided that some, but not necessarily all, nuisance functions are estimated at a sufficiently fast rate. We further proposed a generative model for the considered IV framework and derived its falsification implications. Finally, we demonstrated the theoretical properties of the proposed estimator through simulation studies and applied our method to a real-world analysis evaluating the impact of education or vocational training on labor market outcomes.

While the focal estimand in the paper has been the ATT, our framework can be seamlessly extended to a broader class of estimands and other settings without imposing parametric assumptions. First, our approach can be adapted to infer nonlinear effects, such as the QTT, by extending the methods of \citet{Berger1994} and \citet{Lee2025}. Second, the logit-separable IV model can be generalized to infer treatment effects defined over the entire population, such as the average treatment effect, beyond causal effects among treated units. Third, it can accommodate settings with missing data under a missing-not-at-random mechanism, enabling inference on outcomes subject to nonignorable missingness. Detailed descriptions of these extensions are provided in {\SUPP} \ref{sec-supp-extension}.

Although this paper focused primarily on the canonical binary instrument case, several results readily extend to general instruments. In particular, Theorem \ref{thm-IF} holds for a general $Z$; however, solving the integral equation in \eqref{eq-w IE} for general $Z$ is practically challenging due to non unique closed-form solution and multiple nested dependencies within the counterfactual distribution. Nevertheless, we provide a closed-form expression for an IF in the case of categorical IVs. Additionally, our new parameterization in Section \ref{sec-our parameterization} remains applicable to general $Z$ following minor modifications to the fixed-point equation. These results are detailed in {\SUPP} \ref{sec-supp-extension}. That said, valid statistical inference for general $Z$ requires further technical development well-beyond the scope of this paper.

Several questions remain open under the current IV framework. First, this paper focuses exclusively on binary treatments; extending methods to more general treatment forms is an important direction for future research. Second, it would be of interest to extend the framework to longitudinal settings in which time-varying covariates, instruments, treatments, and outcomes are observed. Finally, developing additional diagnostic procedures, such as sensitivity analysis methods to assess violations of the logit-separable treatment mechanism \HL{IV4}, is an important direction for future research. These directions will be pursued elsewhere.

\newpage  
{\small
\newpage

\appendix

\renewcommand{\theequation}{S.\arabic{equation}}
\renewcommand{\thetable}{S\arabic{table}}
\renewcommand{\thefigure}{S\arabic{Figure}}
\setcounter{equation}{0}

\section*{Supplementary Material}

This document contains supplementary materials for ``Nonparametric Inference with an Instrumental Variable under a Separable Binary Treatment Choice Model.'' Section \ref{sec-supp-1} provides additional details and results that supplement the main paper. Section \ref{sec-supp-2} presents Lemmas used in this document. Section \ref{sec-supp-proof} contains proofs of the Theorems from the main paper.

\section{Details of the Main Paper} \label{sec-supp-1}



\subsection{Details of the Generative Model} \label{sec-supp-generative model}

\subsubsection{Verification of \protect\HL{IV1}-\protect\HL{IV4} under the Generative Model}

We provide algebraic details showing that the generative model in Section \ref{sec-generative model} of the main paper satisfies \HL{IV1}-\HL{IV4}. 

First, \HL{IV1} is trivially satisfied, since $\potY{a}$ does not depend on the value of $Z$. 

Second, \HL{IV2} can be shown as follows:
\begin{align*}
    U \indep Z \cond \bX
    \quad \Rightarrow \quad 
    h_{a}(U,\bX) \indep Z \cond \bX
    \quad \Rightarrow \quad 
    \potY{a} \indep Z \cond \bX \ .
\end{align*}

Third, we verify \HL{IV4}. Let $c(Z,\bX) = c_1(Z,\bX)-c_0(Z,\bX)$ and $\overline{h}(U,
\bX) = \overline{h}_1(U,
\bX) - \overline{h}_0(U,
\bX)$, and  $\epsilon = \epsilon_0 - \epsilon_1$. It is well-known that $\epsilon \sim \text{Logistic}(0,1)$, where $\text{Logistic}(m,s)$ denotes the logistic distribution with location $m$ and scale $s$. Since $h_{a}(U,\bX)$ is monotone in $U$ for each $\bX$, there exists an inverse function $h_{a}^{-1}(\cdot ,\bX)$ such that $U = h_{a}^{-1}(\potY{a},\bX)$. Consequently, the treatment assignment can be expressed as
\begin{align*}
    A &
    = \ind (\mathcal{U}_1 > \mathcal{U}_0)
    \\
    &
    = 
    \ind \big\{
     \overline{h}_1(U,\bX) - c_{1}(Z,\bX) + \epsilon_1 \geq \overline{h}_0(U,\bX) - c_{0}(Z,\bX) + \epsilon_0
     \big\}
     \\
    &
    = 
    \ind \big\{
     \overline{h}(U,\bX) - c(Z,\bX) \geq  \epsilon
     \big\}
     \\
    &
    = 
    \ind \big\{
    \underbrace{
     \overline{h}(h_{0}^{-1}(\potY{0},\bX) ,\bX) }_{\equiv \lambda(\potY{0},\bX)} - c(Z,\bX) \geq  \epsilon
     \big\} \ .
\end{align*}
This implies that
\begin{align*}
    &
    \Pr \big\{ A = 1 \cond \potY{0},Z,\bX \big\}
    =
    \frac{ 
    \exp \big\{ \lambda(\potY{0},\bX) - c(Z,\bX) \big\} 
    }{
    1+
    \exp \big\{ \lambda(\potY{0},\bX) - c(Z,\bX) \big\}
    }
    \\
    \Rightarrow
    \quad &
    \frac{ \Pr \big\{ A = 1 \cond \potY{0},Z,\bX \big\} } { \Pr \big\{ A = 0 \cond \potY{0},Z,\bX \big\} }
    = \exp \big\{ \lambda(\potY{0},\bX) \} \exp \big\{ - c(Z,\bX) \big\} \ .
\end{align*} 
Therefore, taking $\exp \big\{ \lambda(\potY{0},\bX) \}$ and $\exp \big\{ - c(Z,\bX) \big\}$ as $\alpha(\potY{0},\bX)$ and $\beta(Z, \bX)$, respectively, we conclude that \HL{IV4} is satisfied.

Lastly, we show that \HL{IV3} is satisfied. By the assumption, $c(z,\bX) \neq c(z',\bX)$ for some $z \neq z'$, which further implies $\beta(z,\bX) \neq \beta(z',\bX)$. Then, by Lemma \ref{lemma-beta and gamma}, it follows that $\Pr(A=1 \cond Z=z,\bX) \neq \Pr(A=1 \cond Z=z',\bX)$ for some $z \neq z'$ at each $\bX$, so $A \nindep Z \cond \bX$. 

Therefore, the generative model in Section \ref{sec-generative model} satisfies \HL{IV1}-\HL{IV4}.

\subsubsection{Failure of Identification in the Presence of an Interaction Term}

In this Section, we provide details on why the presence of an interaction term in the generative model generally makes the model unidentifiable. A similar argument can be found in Appendix B of \citet{Liu2020}; we reproduce it here for completeness. Suppose that  $\mathcal{U}_1-\mathcal{U}_0$ is equal to $\mathcal{U}_1-\mathcal{U}_0 = \overline{h}(U,\bX) - c(Z,\bX) + \imath(U,Z,\bX) - \epsilon$. Then,  we find
\begin{align*}
    A &
    = \ind (\mathcal{U}_1 - \mathcal{U}_0>0)
    \\
    &
    = 
    \ind \big\{
     \overline{h}(U,\bX) - c(Z,\bX) + \imath(U,Z,\bX) \geq  \epsilon
     \big\}
     \\
    &
    = 
    \ind \big\{
    \underbrace{
     \overline{h}(h_{0}^{-1}(\potY{0},\bX) ,\bX) }_{\equiv \lambda(\potY{0},\bX)} - c(Z,\bX) +
     \underbrace{ \imath(h_{0}^{-1}(\potY{0},\bX),Z,\bX) }_{\equiv \jmath (\potY{0},Z,\bX)}
     \geq  \epsilon
     \big\} \ .
\end{align*}
This implies 
\begin{align*}
	\log 
    \frac{ \Pr \big\{ A = 1 \cond \potY{0},Z,\bX \big\} } { \Pr \big\{ A = 0 \cond \potY{0},Z,\bX \big\} }
    = \lambda(\potY{0},\bX)  - c(Z,\bX)  + \jmath (\potY{0}, Z,\bX)
     \ .
\end{align*} 

Suppose, for simplicity, that $\potY{0}$ is binary and that $\bX$ is suppressed. Then, there exists a set of  parameters $\theta \equiv (\theta_1,\ldots,\theta_5)$ such that
\begin{align*}
&
\Pr \big\{ A = 0 \cond \potY{0},Z \con \theta \big\}
=
\frac{1}{ 1+ \exp \{ \theta_1 + \theta_2 Z + \theta_3 \potY{0} + \theta_{4} Z \potY{0} \} } \ , 
\\
&
\Pr \big\{ \potY{0}= 1 \con \theta \big\} = \exp(\theta_{5}) \ .
\end{align*} 
Condition 1 of \citet{Liu2020} states that, the joint law $(\potY{0},A,Z)$ is identified if and only if
\begin{align*}
	\frac{ \Pr\{ A=0 \cond \potY{0},Z \con \theta\} }
	{ \Pr\{A=0 \cond \potY{0},Z \con \theta'\} }
	\neq
	\frac{ \Pr\{ \potY{0} \con \theta' \} }{ \Pr\{ \potY{0}  \con \theta \} }
	\ , \quad
	\forall 
	\theta \neq \theta' \ .
	\numeq
	\label{eq-proof-liu id}
\end{align*}

To demonstrate the failure of identification, we specify distinct parameter vectors $\theta$ and $\theta'$ that violate \eqref{eq-proof-liu id}. Consider:
\begin{align*}
\theta &= (0, \log 6, \log 11, 0, \log 0.4), \
&&
\theta' = (\log(5/3), \log 5, \log 3, \log 1.3, \log 0.2).
\end{align*}  
Straightforward algebra confirms that these specifications yield the same observed data distribution, thereby violating the condition in \eqref{eq-proof-liu id}. This example illustrates that the non-separable extended propensity score model renders the joint law $(Y^{(0)}, A, Z)$ not identifiable.

\subsection{Details of the Nuisance Function Estimation} \label{sec-supp-estimation}

We provide additional details omitted from Section \ref{sec-estimation}. We first describe the base machine learning algorithms included in the superlearner method \citep{vvLaan2007}. Our implementation includes the following methods and their corresponding R packages: linear regression via \texttt{glm}, lasso/elastic net via \texttt{glmnet} \citep{glmnet}, splines via \texttt{earth} \citep{earth} and \texttt{polspline} \citep{polspline}, generalized additive models via \texttt{gam} \citep{gam}, boosting via \texttt{xgboost} \citep{xgboost} and \texttt{gbm} \citep{gbm}, random forests via \texttt{ranger} \citep{ranger}, and neural networks via \texttt{RSNNS} \citep{RSNNS}.

Next, for a fixed $z \in \{0,1\}$, the conditional density $Y \cond (A=0,Z=z,\bX)$ is estimated using kernel regression \citep{HallRacineLi2004} with the \texttt{np} R package \citep{np2008}, smoothing splines \citep{gss2014} with the \texttt{gss} R package \citep{gss2014}, and an ensemble of these two methods. Below, we describe the ensemble procedure in detail.

Let $\widehat{f}_{\text{np}|z}$ and $\widehat{f}_{\text{gss}|z}$ denote the conditional density estimates obtained from kernel regression and smoothing splines, respectively. For a weight $w \in [0,1]$, the ensemble density estimate is defined as $\widehat{f}_{w|z} = w \widehat{f}_{\text{np}|z} + (1-w) \widehat{f}_{\text{gss}|z}$. The ensemble weight is chosen by cross-validation, as detailed below. Let $\II_{\text{val}}$ denote a validation set. Over the validation set, the following two evaluation metrics are considered: 

\vspace*{0.2cm}

\noindent (i) Mean squared prediction error (MSPE): let $\text{MSPE}(w)$ be defined as 
\begin{align*}
        \text{MSPE}(w)
        =
        \frac{1}{| \II_{\text{val}} |}
        \sum_{i \in \II_{\text{val}}}
        \ind(A_i=0,Z_i=z)
        \bigg\{ 
        \begin{array}{l}
        \int \widehat{f}_{w|z}^2(Y=y \cond A=0,Z=z,\bX_i)
        \, dy
        \\
        - 
        2 \widehat{f}_{w|z}(Y_i\cond A=0,Z=z,\bX_i) 
        \end{array}
        \bigg\} \ .
    \end{align*}
    This metric is motivated by a well-known approximation to the integrated squared error. Specifically, for a density estimate $\widehat{f}(V=v)$ of the true density $f^*(V=v)$, we have
    \begin{align*}
    &
    \int \big\{ \widehat{f}(V=v) -  f^*(V=v ) \big\}^2 \, dv
    \\
    &
    =
    \int \widehat{f}^2(V=v) \, dv - 2 \int \widehat{f}(V=v) f^*(V=v) \, dv + \text{constant}
    \\
    &
    \simeq
    \int \widehat{f}^2(V=v) \, dv - \frac{2}{N} \sum_{i=1}^{N} \widehat{f}(V_i) + \text{constant} \quad \text{ for $V_1,\ldots,V_N \stackrel{i.i.d.}{\sim} f^*(V=v)$}
    \end{align*}

\vspace*{0.2cm}
    
\noindent (ii) Kullback-Leibler divergence (KL): let $\text{KL}(w)$ be defined as
\begin{align*}
        \text{KL}(w)
        =
        -
        \frac{1}{| \II_{\text{val}} |}
        \sum_{i \in \II_{\text{val}}}
        \ind(A_i=0,Z_i=z)
        \log \widehat{f}_{w|z}(Y_i \cond A=0,Z=z,\bX_i)
         \ .
    \end{align*}

Let $w_{\text{MSPE}}$ and $w_{\text{KL}}$ denote the weights that minimize the MSPE and KL criteria over the validation set, respectively. In the simulation studies and data analysis, we select the final ensemble weight as $w_{\text{ens}} = (w_{\text{MSPE}}+w_{\text{KL}})/2$. The ultimate ensemble estimate of the density is given by $\widehat{f}_{\text{ens}|z} = w_{\text{ens}} \widehat{f}_{\text{np}|z} + (1-w_{\text{ens}}) \widehat{f}_{\text{gss}|z}$.

Finally, we describe the density ratio estimation procedures. We employ unconstrained least-squares importance fitting (uLSIF) and the Kullback-Leibler importance estimation procedure (KLIEP) \citep{DR2009_uLSIF, DensityRatio2007_1, DensityRatio2007_2}, which are implemented in the \texttt{densratio} and \texttt{densityratio} R packages \citep{densratiopackage, densityratiopackage}. Below, we briefly review the uLSIF and KLIEP approaches. 

First, we fix $z \in \{0,1\}$. Accordingly, define $r_Y^*$ as
\begin{align*}
    r_Y^* (y,\bx) = \frac{ f^*(Y=y,\bX=\bx \cond A=0,Z=1-z) }{f^*(Y=y,\bX=\bx \cond A=0,Z=z)} \ .
\end{align*}
Note that $r_Y^*(y,\bX)$ is proportional to $R_Y^*(y,\bx) = f^*(Y=y \cond A=0,Z=1-z,\bX=\bx) / f^*(Y=y \cond A=0,Z=z,\bX=\bx)$, so it suffices to estimate $r_Y^*$. Let $\II_{\text{Est}}$ denote the dataset used for estimation consisting of observations with $A=0$, and let $\II_{\text{Est},z}$ be the subset consisting of observations with $(A=0,Z=z)$.

\vspace*{0.2cm}

\noindent (i) KLIEP: The key idea of the KL importance estimation procedure is to characterize $r_Y^*$ as the minimizer of the KL divergence between the numerator density $f^*(Y=y,\bX=\bx \cond A=0,Z=1-z)$ and the density induced by the denominator and the density ratio, namely $f^*(Y=y,\bX=\bx \cond A=0,Z=z) \cdot r_{Y}^*(y,\bx)$. Accordingly, an estimator of $r_Y^*$, denoted by  $\widehat{r}_Y$, can be obtained by solving the following constrained optimization problem: 
\begin{align}				\label{eq-supp-KL-DR1}
& 
\widehat{r}_Y
=
\argmax_{r \in \mathcal{H}} 
\EXP \big[  \log \{ r(Y,\bX) \} \cond A=0,Z=1-z \big]
\nonumber
\\
&
\text{subject to }
\EXP \big\{ r (Y,\bX) \cond A=0,Z=z \big\} = 1 \ ,
\end{align}
where $\mathcal{H}$ is a Reproducing Kernel Hilbert Space (RKHS) of $(Y,\bX)$ associated with a Gaussian kernel $\mathcal{K}$. Then, the KLIEP estimator $\widehat{r}_{Y,\text{KL}}$ is obtained from the empirical analogue of \eqref{eq-supp-KL-DR1} and takes the form
\begin{align*}
    \widehat{r}_{Y,\text{KL}} (y,\bx) =	\sum_{j \in \II_{\text{Est}} } \widehat{\delta}_{j} \mathcal{K} \big( (y,\bx), (y_j,\bx_j) \big)
    \ ;
\end{align*}
where the nonnegative coefficients $\widehat{\bm{\delta}} = \big\{ \widehat{\delta}_j \big\}_{j \in \II_{\text{Est}}}$ are obtained from
\begin{align*}	 
\widehat{\bm{\delta}} 
=
\argmax_{\bm{\delta}} \AVER_{\II_{\text{Est},1-z}} \big\{ \log \big( K_{11} \bm{\delta} \big) \big\} \text{ subject to } 
\AVER_{\II_{\text{Est},z}} \big( K_{01} \bm{\delta} \big) = 1 \ , \quad 
\bm{\delta} \geq 0 \ .
\end{align*} 
Here, $K_{zz'}$ is the gram matrix of which $(i,j)$th entry is $\mathcal{K} \big( (y_{i},\bx_{i}), (y_{j},\bx_{j} ) \big)$ for $i \in \II_{\text{Est},z}$ and $j \in \mathcal{I}_{\text{Est},z'}$. 

\vspace*{0.2cm}

\noindent (ii) uLSIF: Instead of the KL divergence, we can use the squared loss. Then, \eqref{eq-supp-KL-DR1} can be expressed as
\begin{align*}
\widehat{r}_Y =
\argmin_{r \in \mathcal{H}}
\iint \big\{ r(y,\bx) - r_Y^*(y,\bx) \big\}^2 f^*(Y=y,\bX=\bx \cond A=0,Z=z) \, d(y,\bx)  \ .
\end{align*}
The empirical counterpart of the solution to the equation is 
\begin{align*}
    \widehat{r}_{Y,\text{MSE}} (y,\bx) =	\sum_{j \in \II_{\text{Est}} } \widehat{\delta}_{j} \mathcal{K} \big( (y,\bx), (y_j,\bx_j) \big)
    \ ;
\end{align*} 
where the nonnegative coefficients $\widehat{\bm{\delta}} = \big\{ \widehat{\delta}_j \big\}_{j \in \II_{\text{Est}}}$ are obtained by solving the following optimization problem including $L_2$-regularization term:
\begin{align*}
&	
\widehat{\bm{\delta}}
=
\argmin_{\bm{\delta}}
\Big[
\AVER_{\mathcal{I}_{\text{Est},z}} \big\{ \big( K_{z \circ } \bm{\delta} \big)^2	\big\} 
- 2 \AVER_{\mathcal{I}_{\text{Est},1-z}} \big\{ \big( K_{1-z \circ}  \bm{\delta} \big)^2	\big\}
+
\lambda
\bm{\delta}\T K_{\circ\circ} \bm{\delta}
\Big]
\ , 
\\
&
K_{z\circ} 
=
\Big[ \mathcal{K} \big( (y_{i},\bx_{i}), (y_{j},\bx_{j} ) \big) \Big]_{ i \in \mathcal{I}_{\text{Est},z}, \ j \in \mathcal{I}_{\text{Est}}}  \ , 
K_{\circ\circ} 
=
\Big[ \mathcal{K} \big( (y_{i},\bx_{i}), (y_{j},\bx_{j} ) \big) \Big]_{ i, j \in \mathcal{I}_{\text{Est}} } \ ,
\end{align*}
where $\lambda$ is the regularization parameter that can be chosen from cross-validation.

\vspace*{0.2cm}

Once $\widehat{r}_{Y,\star}$ for $\star \in \{ \text{KL},\text{MSE} \}$ are obtained, define $\widehat{f}_{\star|1-z} \equiv \widehat{f}_{\star|1-z} (Y=y \cond A=0,Z=1-z,\bX)$ which is constructed as a function proportional to $\widehat{r}_{Y,\star}(y,\bX) \widehat{f}_{\text{ens}|z} (Y=y \cond A=0,Z=z,\bX)$, with the normalization constraint $\int \widehat{f}_{\star|1-z} (Y=y \cond A=0,Z=1-z,\bX) \, dy = 1$. One may further consider an ensemble of these estimates, say $ \widehat{f}_{\text{ens}|1-z} = w_{\text{ens}}'  \widehat{f}_{\text{KL}|1-z} + (1-w_{\text{ens}}') \widehat{f}_{\text{MSE}|1-z}$ where the optimal weight $w_{\text{ens}}'$ can be determined analogously to $w_{\text{ens}}$. Accordingly, the conditional density ratio estimator can be constructed as 
\begin{align*}
    \widehat{R}_Y(y,\bx) = 
    \frac{ \widehat{f}_{\text{ens}|1-z}(Y=y \cond A=0,Z=1-z,\bX=\bx) }{ \widehat{f}_{\text{ens}|z}(Y=y \cond A=0,Z=z,\bX=\bx)  } \ .
\end{align*}

All nuisance function estimators require the selection of hyperparameters, including tuning parameters in machine learning algorithms, bandwidths in kernel regression and RKHS-based estimators, and regularization parameters. We recommend selecting these hyperparameters via cross-validation.

\subsection{Bias Structure Under Different Parameterizations} \label{sec-supp-other para}

In this Section, we examine the bias structure of the ATT estimator under the parameterizations proposed by \citet{Sun2018}, \citet{Liu2020}, and \citet{Sun2025}. Throughout, we let $\gamma(Z,\bX) = f(A=1 \cond Z,\bX)/f(A=0 \cond Z, \bX)$ denote the treatment odds conditional on $(Z,\bX)$.

Our analysis builds on the intermediate results and notation introduced in Section \ref{sec-proof-AN}, which contains the proof of Theorem \ref{thm-AN}. As shown in \eqref{eq-proof-bias of ATT estimator}-\eqref{eq-proof-bias of ATT estimator4}, the bias of the ATT estimator admits the following characterization:
\begin{align*}
	&
	N^{-1/2}
	\big\|
	\text{Bias of $\widehat{\tau}$}
	\big\|
     \\
     &
     \lesssim
    \bigg\| \begin{array}{l}\alphahat (\potY{0},\bX) \\
    - \alpha^*(\potY{0},\bX)
    \end{array}  \bigg\|_{P,2}
    \times 
          \bigg\|
     \begin{array}{l} 
    \EXP \bigg[ 
          \begin{array}{l} \big\{ \mathcal{G}^{(0)} - \what \} 
          \\
     - \muhat(Z,\bX \con \mathcal{G}-\what) 
          \end{array} 
          \, \bigg| \,
          A=1,\potY{0},\bX \bigg]
     \end{array} 
          \bigg\|_{P,2}
     \\
     &
     \qquad +
     \bigg\| \begin{array}{l} \betahat (Z,\bX) \\
     - \beta^*(Z,\bX) \end{array} \bigg\|_{P,2}
    \times 
          \bigg\|
     \begin{array}{l} 
    \EXP \bigg[ 
          \begin{array}{l} \big\{ \mathcal{G}^{(0)} - \what \} 
          \\
     - \muhat(Z,\bX \con \mathcal{G}-\what) 
          \end{array} 
          \, \bigg| \,
          A=1,Z,\bX \bigg]
     \end{array} 
          \bigg\|_{P,2}
     \\
     &
     \qquad +
     \bigg\| \begin{array}{l}\alphahat (\potY{0},\bX) \\
    - \alpha^*(\potY{0},\bX)
    \end{array}  \bigg\|_{P,2}
     \bigg\| \begin{array}{l} \betahat (Z,\bX) \\
     - \beta^*(Z,\bX) \end{array} \bigg\|_{P,2} 
     \\
     &
     \qquad
     +
         \left[
        \begin{array}{l}
        \big\| \widehat{f}\LSS(A=1 \cond Z,\bX) - f^*(A=1 \cond Z,\bX)
        \big\|_{P,2}
        \\
        +
        \big\| \widehat{f}\LSS(\potY{0} \cond A=0,Z,\bX) - f^*(\potY{0} \cond A=0,Z,\bX)
        \big\|_{P,2}
        \end{array}
    \right]
    \\
    &
    \qquad\qquad
    \times
    \big\|  \widehat{f}\LSS(Z=1 \cond \bX) - f^*(Z=1 \cond \bX) \big\|_{P,2} 
     \ .
    \numeq 
    \label{eq-proof-general bias form}
\end{align*}
Note that, in \eqref{eq-proof-bias of ATT estimator2}, we set the normalization constant defined in \eqref{eq-proof-normalD L2P} to $\mathcal{N}=1$. Therefore, it suffices to characterize the above $L^2(P)$-errors under each respective parameterization. Throughout, we impose standard regularity conditions, such as uniformly bounded nuisance functions and their estimates, to avoid unnecessary technical complications.

We first consider the parameterization of \citet{Sun2018} and \citet{Liu2020}, with the model parameter as $\theta_{\text{LS}} \equiv \{ \fz \equiv f(Z \cond \bX), \fy \equiv f(Y \cond A=0,Z,\bX), \alpha \equiv \alpha(\potY{0},\bX) \}$.

From \eqref{eq-proof-gamma}, we find that $\widehat{\gamma}\LSS(Z,\bX) $ is equal to 
\begin{align*}
\widehat{\gamma}\LSS(Z,\bX)
=
\widehat{\beta}\LSS(Z,\bX)
\int \widehat{\alpha}\LSS(y,\bX) \widehat{f}\LSS(Y=y \cond A=0,Z,\bX) \, dy \ .
\end{align*}
Consequently, we find that the $L^2(P)$-norm of the estimation error of $\gamma^*$ satisfies
\begin{align*}
&
\big\|
\widehat{\gamma}\LSS(Z,\bX)
-
\gamma^*(Z,\bX)
\big\|_{P,2}
\\
&
\lesssim 
\left[ 
	\begin{array}{l}
\big\| \widehat{\beta}\LSS(Z,\bX) - \beta^*(Z,\bX) \big\|_{P,2}
\\
+
\big\| \widehat{\alpha}\LSS(\potY{0},\bX) - \alpha^*(\potY{0},\bX) \big\|_{P,2}
\\
+
\big\| \widehat{f}\LSS(\potY{0} \cond A=0,Z,\bX) - f^*(\potY{0} \cond A=0,Z,\bX) \big\|_{P,2} 
	\end{array}
\right]
\ .
\numeq
\label{eq-proof-LS-bias1}
\end{align*}
In addition, \eqref{eq-proof-gamma to fa} establishes that
\begin{align*}
\|\widehat{\gamma}\LSS(Z,\bX) - \gamma^*(Z,\bX) \|_{P,2} \asymp \| \widehat{f}\LSS(A=1 \cond Z,\bX) - f^*(A=1 \cond Z,\bX) \|_{P,2} \ .
\numeq
\label{eq-proof-LS-bias2}
\end{align*}
Consequently, combining \eqref{eq-proof-LS-bias1} and \eqref{eq-proof-LS-bias2}, we obtain
\begin{align*}
 \| \widehat{f}\LSS(A=1 \cond Z,\bX) - f^*(A=1 \cond Z,\bX) \|_{P,2}
 \lesssim 
 r_{\alpha}\LSS + r_{\beta}\LSS + r_{Y} \LSS 
 \ .
 \numeq
 \label{eq-proof-LS-bias3}
\end{align*}
In turn, the convergence rate in \eqref{eq-proof-GH convergence} is bounded by
\begin{align*}
&
              \bigg\|
     \begin{array}{l} 
    \EXP \bigg[ 
          \begin{array}{l} \big\{ \mathcal{G}^{(0)} - \what \} 
          \\
     - \muhat(Z,\bX \con \mathcal{G}-\what) 
          \end{array} 
          \, \bigg| \,
          A=1,\potY{0},\bX \bigg]
          \bigg\|_{P,2}
     \end{array} 
    \\
    &
    \lesssim 
    \left[
        \begin{array}{l}
        \big\| \widehat{f}\LSS(A=1 \cond Z,\bX) - f^*(A=1 \cond Z,\bX)
        \big\|_{P,2}
        \\
        +
        \big\| \widehat{f}\LSS(\potY{0} \cond A=0,Z,\bX) - f^*(\potY{0} \cond A=0,Z,\bX)
        \big\|_{P,2}
        \end{array}
    \right] 
    \\
    &
    \stackrel{\eqref{eq-proof-LS-bias3}}{
    \lesssim}
 r_{\alpha}\LSS + r_{\beta}\LSS + r_{Y} \LSS \ . 
\numeq
\label{eq-proof-LS-bias4}
\end{align*}

In addition, for a uniformly bounded function $v$, we have:
\begin{align*}
    & 
    \muhat(Z,\bX \con v) - \mu^*(Z,\bX \con v)
    \\
    &
    =  
\frac{
\widehat{\EXP}\LSS \big\{ v \times \alphahat(\potY{0},\bX) \cond A=0,Z,\bX \big\}
}{
\widehat{\EXP}\LSS \big\{  \alphahat(\potY{0},\bX) \cond A=0,Z,\bX \big\}
} -
\frac{ \EXP \big\{   v \times \alpha^*(\potY{0},\bX) \cond A=0, Z,\bX \big\} }{ \EXP \big\{  \alpha^*(\potY{0},\bX) \cond A=0, Z,\bX \big\} } 
\\
\Rightarrow \quad 
&
    \big\| \muhat(Z,\bX \con v) - \mu^*(Z,\bX \con v)
    \big\|_{P,2}
    \\  
&
\lesssim 
\Bigg\|
\begin{array}{l}
\EXP \big\{   v \times \alpha^*(\potY{0},\bX) \cond A=0, Z,\bX \big\}
\widehat{\EXP}\LSS \big\{ \alphahat(\potY{0},\bX) \cond A=0,Z,\bX \big\}
\\
-
\widehat{\EXP}\LSS \big\{ v \times  \alphahat(\potY{0},\bX) \cond A=0,Z,\bX \big\}
\EXP \big\{  \alpha^*(\potY{0},\bX) \cond A=0, Z,\bX \big\}
\end{array} 
\Bigg\|_{P,2}
\\
&
\lesssim 
\left[ 
\begin{array}{l}
\big\| 
\alphahat(\potY{0},\bX) 
- \alpha^*(\potY{0},\bX) 
\big\|_{P,2}
\\
+
\big\| \widehat{f}\LSS(\potY{0} \cond A=0,Z,\bX) - f^*(\potY{0} \cond A=0,Z,\bX) \big\|_{P,2}  
\end{array}
\right]
\ .
\end{align*}
Taking $v = \mathcal{G}^{(0)} - \what$, we obtain
\begin{align*}
    &
              \bigg\|
     \begin{array}{l} 
    \EXP \bigg[ 
          \begin{array}{l} \big\{ \mathcal{G}^{(0)} - \what \} 
          \\
     - \muhat(Z,\bX \con \mathcal{G}-\what) 
          \end{array} 
          \, \bigg| \,
          A=1,Z,\bX \bigg]
     \end{array} 
          \bigg\|_{P,2}
          \lesssim
          r_{\alpha}\LSS + r_{Y}\LSS \ .
          \numeq
\label{eq-proof-LS-bias5}
\end{align*}

Combining \eqref{eq-proof-general bias form}-\eqref{eq-proof-LS-bias5}, we obtain
\begin{align*}
    &
    N^{-1/2}
	\big\|
	\text{Bias of $\widehat{\tau}$}
	\big\|
    \\
    &
    \lesssim 
    r_{\alpha}\LSS 
    \big\{  r_{\alpha}\LSS + r_{\beta}\LSS + r_{Y} \LSS \big\} 
    +
    r_{\beta}\LSS  
    \big\{  r_{\alpha}\LSS + r_{Y}\LSS \big\}
    \\
    &
    \qquad
    + 
    \big\{  r_{\alpha}\LSS + r_{\beta}\LSS + r_{Y} \LSS \big\} r_{Z}\LSS
    \\
    &    
    \lesssim 
    \big\{ r_{\alpha}\LSS \big\}^2
    +
    r_{\alpha}\LSS r_{\beta} \LSS 
    +
    r_{\alpha}\LSS r_{Y} \LSS 
    +
    r_{\beta}\LSS r_{Y}\LSS 
    \\
    &
    \qquad
    +
    r_{\alpha}\LSS r_{Z}\LSS  
    +
    r_{\beta}\LSS r_{Z}\LSS  
    +
    r_{Y}\LSS r_{Z}\LSS  \  .
\end{align*}
This proves the result in \eqref{eq-Liu Bias} of the main paper.

Next, we consider the parameterization of \citet{Sun2025}, with the model parameter as $\theta_{\text{SMW}} \equiv \{ \fz \equiv f(Z \cond \bX), f_{\potY{0}} \equiv f(\potY{0} \cond \bX), \alpha \equiv \alpha(\potY{0},\bX), \beta \equiv \beta(Z,\bX) \}$. 

From \eqref{eq-proof-Y0AZ}, we find
\begin{align*} 
    &
    f^* (Y=y, A = 0 \cond  Z, \bX) 
    =
    \frac{ 
    f^* (\potY{0}=y \cond \bX)
    }{  1 + \alpha^*(y,\bX) \beta^*(Z,\bX)  }
    \\
    \Rightarrow
    \quad
    &
    f^*(A=0 \cond Z, \bX)
    =
    \int \frac{ f^* (\potY{0}=y \cond \bX)
    }{  1 + \alpha^*(y,\bX) \beta^*(Z,\bX)  }
    \, dy \ .
\end{align*} 
This implies that 
\begin{align*}
    &
    \big\|
    \widehat{f}\LSS (A = 0 \cond  Z, \bX) 
    -
    f^* (A = 0 \cond  Z, \bX) 
    \big\|^2
    \\
    &
    =
    \bigg\|
     \int
      \frac{ \widehat{f}\LSS (\potY{0}=y \cond \bX)
    }{  1 + \alphahat(y,\bX) \betahat(Z,\bX)  }
    -
    \frac{ f^* (\potY{0}=y \cond \bX)
    }{  1 + \alpha^*(y,\bX) \beta^*(Z,\bX)  }
    \, dy
    \bigg\|^2
    \\
    &
    \lesssim
     \int
      \left\| 
        \begin{array}{l}
      \widehat{f}\LSS (\potY{0}=y \cond \bX)
      \{  1 + \alpha^*(y,\bX) \beta^*(Z,\bX) \}
      \\
      -
      f^* (\potY{0}=y \cond \bX) \{ 1 + \alphahat(y,\bX) \betahat(Z,\bX) \}
        \end{array}
      \right\|^2
    \, dy 
    \\
    &
    \lesssim 
    \int \bigg[ 
    \begin{array}{l}
    \big\| 
    \widehat{f}\LSS (\potY{0}=y \cond \bX)
    -
    f^* (\potY{0}=y \cond \bX)
    \big\|^2 
    \\
    +
    \big\|
    \alphahat(y,\bX) - \alpha^*(y,\bX)
    \big\|^2
    \end{array}
    \bigg]
    f^* (\potY{0}=y \cond \bX) \, dy
    \\
    &
    \qquad + \big\| \betahat(Z,\bX) - \beta^*(Z,\bX) \big\|^2 \ .
\end{align*}
Consequently, we obtain
\begin{align*}
    &
    \big\|
    \widehat{f}\LSS (A = 0 \cond  Z, \bX) 
    -
    f^* (A = 0 \cond  Z, \bX) 
    \big\|_{P,2} 
    \lesssim
    r_{\potY{0}}\LSS + r_{\alpha}\LSS + r_{\beta}\LSS \ .
    \numeq 
    \label{eq-proof-SMW-bias1}
\end{align*}

Again, from \eqref{eq-proof-Y0AZ}, we find
\begin{align*} 
    &
    f^* (Y=y \cond A=0, Z, \bX) 
    =
    \frac{1}{f^*(A=0 \cond Z,\bX) }
    \frac{ 
    f^* (\potY{0}=y \cond \bX)
    }{  1 + \alpha^*(y,\bX) \beta^*(Z,\bX)  }
    \\
    \Rightarrow
    \quad
    &
    \big\| \widehat{f}\LSS (\potY{0} \cond A=0, Z, \bX)  - f^* (\potY{0} \cond A=0, Z, \bX)  \big\|_{P,2}
    \\
    &
    \lesssim 
    \left[ 
        \begin{array}{l}
             \big\| \widehat{f}\LSS(A=1 \cond Z,\bX) - f^*(A=1 \cond Z,\bX) \big\|_{P,2}
             \\ 
             + 
             \big\|
    \widehat{f}\LSS (\potY{0} \cond \bX) 
    -
    f^* (\potY{0} \cond  \bX) 
    \big\|_{P,2} 
    \\
    + \big\| \alphahat(\potY{0},\bX) - \alpha^*(\potY{0},\bX) \big\|_{P,2} 
    \\
    + \big\| \betahat(Z,\bX) - \beta^*(Z,\bX) \big\|_{P,2} 
        \end{array}
    \right]
    \\
    &
    \stackrel{\eqref{eq-proof-SMW-bias1}}{
    \lesssim 
    }
r_{\potY{0}}\LSS + r_{\alpha}\LSS + r_{\beta}\LSS \ .
    \numeq 
    \label{eq-proof-SMW-bias2}
    \end{align*}
    
Therefore, we obtain
\begin{align*}
&
              \bigg\|
     \begin{array}{l} 
    \EXP \bigg[ 
          \begin{array}{l} \big\{ \mathcal{G}^{(0)} - \what \} 
          \\
     - \muhat(Z,\bX \con \mathcal{G}-\what) 
          \end{array} 
          \, \bigg| \,
          A=1,\potY{0},\bX \bigg]
          \bigg\|_{P,2}
     \end{array} 
    \\
    &
    \lesssim 
    \left[
        \begin{array}{l}
        \big\| \widehat{f}\LSS(A=1 \cond Z,\bX) - f^*(A=1 \cond Z,\bX)
        \big\|_{P,2}
        \\
        +
        \big\| \widehat{f}\LSS(\potY{0} \cond A=0,Z,\bX) - f^*(\potY{0} \cond A=0,Z,\bX)
        \big\|_{P,2}
        \end{array}
    \right] 
    \\
    &
    \stackrel{\eqref{eq-proof-SMW-bias1},\eqref{eq-proof-SMW-bias2}}{
    \lesssim}
r_{\potY{0}}\LSS + r_{\alpha}\LSS + r_{\beta}\LSS \ ,
\numeq
\label{eq-proof-SMW-bias3}
\end{align*}
and
\begin{align*}
    &
              \bigg\|
     \begin{array}{l} 
    \EXP \bigg[ 
          \begin{array}{l} \big\{ \mathcal{G}^{(0)} - \what \} 
          \\
     - \muhat(Z,\bX \con \mathcal{G}-\what) 
          \end{array} 
          \, \bigg| \,
          A=1,Z,\bX \bigg]
     \end{array} 
          \bigg\|_{P,2}
          \\
          &
          \stackrel{\eqref{eq-proof-LS-bias5}}{\lesssim}
          r_{\alpha}\LSS + r_{Y}\LSS 
          \\
          &
          \stackrel{\eqref{eq-proof-SMW-bias2}}{
    \lesssim}
r_{\potY{0}}\LSS + r_{\alpha}\LSS + r_{\beta}\LSS \ .
\numeq
\label{eq-proof-SMW-bias4}
\end{align*}
 Combining \eqref{eq-proof-general bias form}, \eqref{eq-proof-SMW-bias1}-\eqref{eq-proof-SMW-bias4}, we obtain
\begin{align*}
    &
    N^{-1/2}
	\big\|
	\text{Bias of $\widehat{\tau}$}
	\big\|
    \\
    &
    \lesssim 
    r_{\alpha}\LSS 
    \big\{  r_{\potY{0}}\LSS + r_{\alpha}\LSS + r_{\beta}\LSS \big\} 
    +
    r_{\beta}\LSS  
    \big\{  r_{\potY{0}}\LSS + r_{\alpha}\LSS + r_{\beta}\LSS \big\}
    \\
    &
    \qquad
    + 
    \big\{  r_{\potY{0}}\LSS + r_{\alpha}\LSS + r_{\beta}\LSS \big\} r_{Z}\LSS
    \\
    &    
    \lesssim 
    \big\{ r_{\alpha}\LSS \big\}^2
    +
    \big\{ r_{\beta}\LSS \big\}^2
    \\ 
    &
    \qquad 
    +
    r_{\alpha}\LSS r_{\beta} \LSS 
    +
    r_{\alpha}\LSS r_{\potY{0}} \LSS 
    +
    r_{\beta}\LSS r_{\potY{0}}\LSS 
    +
    r_{\alpha}\LSS r_{Z}\LSS  
    +
    r_{\beta}\LSS r_{Z}\LSS  
    +
    r_{\potY{0}}\LSS r_{Z}\LSS  \  .
\end{align*}
This proves the result in \eqref{eq-SMW Bias} of the main paper.

\subsection{Median Adjustment for Cross-fitting Estimators} \label{sec-supp-median}

 Cross-fitting estimators depend on a specific sample split and, therefore, may produce outlying estimates if some split samples do not represent the entire data. To resolve the issue, \citet{DDML2018} proposed to use the median adjustment from multiple cross-fitting estimates.  First, for $s \in \{1,\ldots,S\}$, let $\widehat{\tau}_{s}$  be the $s$th cross-fitting estimate with the corresponding variance estimate $\widehat{\sigma}_{s}^2$. Then, the median-adjusted cross-fitting estimate and its variance estimate are defined as follows: 
 \begin{align*}
 & \widehat{\tau}_{\median}
 :=
 \median_{s=1,\ldots,S} \widehat{\tau}_{s}
 \ , \quad \widehat{\sigma}_{\median}^2
 :=
 \median_{s=1,\ldots,S} \big\{ \sigma_s^2 + (\widehat{\tau}_s - \widehat{\tau}_{\median} )^2 \big\} \ .
 \end{align*}
 These estimates are more robust to the particular realization of sample partition.

\subsection{Details of the Falsification Test} \label{sec-supp-falsification}

In this Section, we provide details of the falsification test introduced in Section \ref{sec-falsification}. For simplicity, we restrict attention to the case of a binary $Z$. The construction of the falsification test follows closely that of \citet{Kitagawa2015}. 

First, define
\begin{align*}
    & 
    \kappa_0^*(A,Z,\bX)
    \\
    &
    =
        \frac{1-A}
        {
        f^*(A=0 \cond Z=1,\bX)
        -
        f^*(A=0 \cond Z=0,\bX)
        }
    \frac{
    Z - f^*(Z=1 \cond \bX)
    }{
    f^*(Z=1 \cond \bX)
    f^*(Z=0 \cond \bX)
    } \ .
\end{align*} 
For a generic function $g$, consider an expression 
$ \EXP    
    \big\{ 
    g(Y,\bX)
    \kappa_0^*(A,Z,\bX)
    \cond \bX
    \big\}$. 
From straightforward algebra, we obtain
\begin{align*}
    \EXP    
    \Bigg\{
    g(Y,\bX)
        \frac{
        (1-A) \ind(Z=z)}{f^*(Z=z \cond \bX)}
        \, \Bigg| \,
        \bX 
        \Bigg\}
    &
    =
    \EXP    
    \big[
        \EXP \big\{ g(Y,\bX)
        (1-A) 
        \cond Z=z, \bX \big\}
        \, \big| \,
        \bX 
        \big]
    \\
    &
    =
    \int g(y,\bX) f^*(Y=y,A=0 \cond Z=z,\bX) \, dy \ .
\end{align*}      
Consequently, for a generic function $g$, we have
\begin{align*}
    &
    \EXP    
    \big\{ 
    g(Y,\bX)
    \kappa_0^*(A,Z,\bX)
    \cond \bX
    \big\}
    \\
    &
    =
    \EXP    
    \left\{
    \frac{
    g(Y,\bX) (1-A)}
        {
        f^*(A=0 \cond Z=1,\bX)
        -
        f^*(A=0 \cond Z=0,\bX)
        }
    \frac{
    Z - f^*(Z=1 \cond \bX)
    }{
    f^*(Z=1 \cond \bX)
    f^*(Z=0 \cond \bX)
    }
        \, \Bigg| \, 
        \bX
    \right\}
    \\
    &
    =
    \EXP    
    \left\{ 
    g(Y,\bX)
        (1-A)  
    \frac{
    Z / f^*(Z=1 \cond \bX)
    - 
    (1-Z) / f^*(Z=0 \cond \bX)
    }
        {
        f^*(A=0 \cond Z=1,\bX)
        -
        f^*(A=0 \cond Z=0,\bX)
        }
        \, \Bigg| \, 
        \bX
    \right\}
    \\
    &
    =
    \int g(y,\bX) 
    \frac{ f^*(Y=y,A=0 \cond Z=1, \bX) - f^*(Y=y,A=0 \cond Z=0, \bX) 
    }{f^*(A=0 \cond Z=1,\bX) - f^*(A=0 \cond Z=0,\bX)} 
     \, dy \ .
\end{align*}
Therefore, following the argument in \citet{Kitagawa2015}, one can establish that
\begin{align*}
    & 
    \frac{ f^*(Y=y,A=0 \cond Z=1, \bX) - f^*(Y=y,A=0 \cond Z=0, \bX) 
    }{f^*(A=0 \cond Z=1,\bX) - f^*(A=0 \cond Z=0,\bX)}
    \geq 0 
    \\
    &
    \text{if and only if}
    \quad 
    \EXP    
    \big\{ 
    g(Y,\bX)
    \kappa_0^*(A,Z,\bX) 
    \big\}
    \geq 0
    \ , \quad 
    \forall g \in \mathcal{G}_{\text{ind}} \ ,
\end{align*}
where $\mathcal{G}_{\text{ind}}$ denotes the class of indicator functions for rectangular sets in the support of $(Y,\bX)$.

To test the null hypothesis
\begin{align*}
    H_0: 
    \EXP \big\{ g(Y,\bX)
    \kappa_0^*(A,Z,\bX) 
    \big\}
    \geq 0
    \ , \quad 
    \forall g \in \mathcal{G}_{\text{ind}} \ , 
\end{align*}
we consider a variance-weighted Kolmogorov-Smirnov-type statistic, which is a modification of the test proposed by \citet{Kitagawa2015}. The test statistic is defined as
\begin{align*}
    T_{N}
    =
    - 
    \frac{1}{
    \sqrt{N}}
    \inf_{g \in \mathcal{G}_{\text{ind}}}
    \frac{ \sum_{i=1}^{N}  g(Y_i,\bX_i) \widehat{\kappa}_0(A_i,Z_i,\bX_i) }{ \max(\xi, \widehat{\sigma}_0) } \ .
    \numeq 
    \label{eq-proof-KS statistic}
\end{align*}
Here, $\widehat{\kappa}_0$ is an estimate of $\kappa_0^*$, $\widehat{\sigma}_0^2$ is the sample variance of $g(Y_i,\bX_i) \widehat{\kappa}_0(A_i,Z_i,\bX_i) $, and $\xi$ is a positive constant specified by the investigator. Following \citet{Kitagawa2015}, the critical value is obtained via bootstrap, treating $\widehat{\kappa}_0$ as fixed. Specifically, for $b \in \{1,\ldots,B\}$, where $B$ denotes the number of bootstrap replications, let $\mathcal{I}^{(b)}$ be the index set of the $b$th bootstrap sample drawn from $\mathcal{I}=\{1,\ldots,N\}$. The bootstrap analogue of $T_N$ is then given by
\begin{align*}
    T_{N}^{(b)}
    =
   - 
    \frac{1}{
    \sqrt{N}}
    \inf_{g \in \mathcal{G}_{\text{ind}}} 
    \frac{ \sum_{i \in \II^{(b)} }  g(Y_i,\bX_i) \widehat{\kappa}_0(A_i,Z_i,\bX_i)  
      - \sum_{i=1}^{N} g(Y_i,\bX_i) \widehat{\kappa}_0(A_i,Z_i,\bX_i)   }{ \max(\xi, \widehat{\sigma}_0^{(b)}) } \ .
\end{align*}
The null hypothesis is rejected at significance level $c$ if $T_N$ exceeds the $100(1-c)$th percentile of the empirical distribution of ${T_N^{(1)},\ldots,T_N^{(B)}}$.

Implementing this test in practice may pose practical challenges. First, computing the infimum over $\mathcal{G}_{\text{ind}}$ is challenging, particularly when $Y$ is continuous. A practical remedy is to approximate $\mathcal{G}_{\text{ind}}$ by a finite subclass. In addition, $\widehat{\kappa}_0$ may be obtained using the nonparametric methods described in Section \ref{sec-supp-estimation}, but these estimators do not necessarily satisfy the empirical moment restriction. Specifically, for any function $g(\bX)$, $\kappa_0^*$ must satisfy $\EXP[ g(\bX) \{ \kappa_0^*(A,Z,\bX) -1 \} ] = 0$. However, since $\widehat{\kappa}_0$ is not estimated by enforcing this moment restriction---being instead constructed by plugging in estimates of the corresponding nuisance parameters---the empirical moment $N^{-1}\sum_{i=1}^{N} g(\bX_i) \{ \widehat{\kappa}_0 (A_i,Z_i,\bX_i)-1\}$ may deviate substantially from zero. In such cases, the falsification test based on $T_N$ may perform poorly in finite samples. 

As an alternative, one may directly evaluate condition \eqref{eq-falsification} for each value of $\bX$  using the estimated joint density $\widehat{f}(Y=y,A=0 \cond Z,\bX)$ over a fine grid of $y \in \suppYX$. This approach is particularly more practical when $Y$ is discrete.

\subsection{Extensions} \label{sec-supp-extension}

\subsubsection{Extensions to General $Z$ Cases} \label{sec-supp-Conti Z}

In this Section, we extend our results to the general $Z$ case. Specifically, we present theorems corresponding to Theorems \ref{thm-Psi-binary} and \ref{thm-global contraction-binary} for a general $Z$.

To streamline the exposition, we introduce the following notation. Let $\suppZ \subset \R$ be the support of $Z$. For each $z \in \suppZ$ and $\bX \in \suppX$, define $\mathcal{Z}_{z,\bX}^c = \{ z' \cond f(A=1 \cond Z=z',\bX) \neq f(A=1 \cond Z=z,\bX) \}$. Note that $\mathcal{Z}_{z,\bX}^c$ is nonempty for every $z \in \suppZ$ by \HL{IV3}. Additionally, for $\bX \in \suppX$, define $\mathcal{L}_+^{\infty} (\mathcal{Y}_{\bX}) = \{ h \cond h(y,\bX) \in [0,\infty), \ \forall y  \in \mathcal{Y}_{\bX} \}$. We are now ready to state the Theorems.

\begin{theorem} \label{thm-Psi}

Suppose that Assumptions \HL{A1}-\HL{A2} and \HL{IV1}-\HL{IV4} hold for $(\potY{0},A,Z,\bX)$, with corresponding parameters $\theta = \{ \fz,\fa,\fy \}$. For a fixed $z_0 \in \suppZ$ and $\bX \in \suppX$, select any $z_1 \in \suppZ_{z_0,\bX}^c$. Define the mapping $\Psi: \mathcal{L}_+^{\infty} (\mathcal{Y}_{\bX}) \rightarrow \mathcal{L}_+^{\infty} (\mathcal{Y}_{\bX})$ by
\begin{align*}
    &
    \!\!\!
    \Psi\big( g(y,\bX) \con \theta, z_1, z_0 \big)
    \\
    &
    \!\!\!
    =
    \frac{
    \left\{
    \begin{array}{l}
    \ff{}(A=1 \cond Z=z_0,\bX)
    g(y,\bX)
    \\
    + \ff{}(A=0 \cond Z=z_0,\bX)
    \ff{}(Y=y \cond A=0,Z=z_0,\bX)
    \\
    -
    \ff{}(A=0 \cond Z=z_1,\bX)
    \ff{}(Y=y \cond A=0,Z=z_1,\bX)
    \end{array} 
    \right\}
    \!
    \displaystyle{
    \int \!
    g(t,\bX)
    R_{Y}(t,\bX) \, d t
    }
    }
    {
    \ff{}(A=1 \cond Z=z_1,\bX)
    R_{Y}(y,\bX)
    }   ,
    \numeq \label{eq-Psi}
\end{align*}
where 
\begin{align}
 R_{Y}(y,\bX) = \frac{ f(Y=y \cond A=0,Z=z_1,\bX)}{f(Y=y \cond A=0,Z=z_0,\bX)} \ .
 \numeq 
 \label{eq-Ry}
\end{align}

Then, the following results hold:
\begin{itemize}
    \item[(i)] Let $ g^\star(y,\bX \con \theta, z_1, z_0)$ denote the fixed point of $\Psi$, i.e., 
    \begin{align*}
        \Psi \big( g^\star(y,\bX \con \theta, z_1, z_0) \con \theta, z_1, z_0 \big) = g^\star(y,\bX \con \theta, z_1, z_0) \ .
        \numeq 
        \label{eq-fixed point equation}
    \end{align*}
    Then, $g^\star$ uniquely exists and does not depend on the choice of $z_1 \in \suppZ_{z_0,\bX}^c$, i.e.,
    \begin{align*}
        g^\star(y,\bX \con \theta,z_1,z_0)
        =
        g^\star(y,\bX \con \theta,z_1',z_0) 
         \ , \quad \forall (z_1,z_1') \in \suppZ_{z_0,\bX}^c \ .
    \end{align*}
    Therefore, we denote the fixed point as $g^\star(y,\bX \con \theta, z_0)$, suppressing the dependence on $z_1$.

    \item[(ii)] The odds ratio function $\alpha$ and the baseline odds function $\beta$ in \eqref{eq-nuis-alpha beta} are expressed in terms of $g^\star$ as follows:
    \begin{align*}
    &
    \alpha(y,\bX \con \theta, z_0)
    =
    \frac{ g^\star (y,\bX \con \theta, z_0) }{ g^\star  (y_R,\bX \con \theta, z_0) }
    \frac{ \ff{}(Y=y_R \cond A=0,Z=z_0,\bX) }{\ff{}(Y=y \cond A=0,Z=z_0,\bX)} \ , \
    y \in \suppYX \ ,
    \numeq \label{eq-alpha g}
    \\
    &
    \beta(z,\bX \con \theta, z_0) 
    =
    \frac{f(A=1 \cond Z=z,\bX) / f(A=0 \cond Z=z,\bX)}{ \int  \alpha(y,\bX \con \theta, z_0)  f(Y=y \cond A=0,Z=z,\bX) \, dy }  \ , \
    z \in \suppZ \ .
    \numeq \label{eq-beta g}
    \end{align*}
    Furthermore, $\alpha$ and $\beta$ do not depend on the choice of $z_0 \in \suppZ$, i.e.,
    \begin{align*}
        \alpha(y,\bX \con \theta, z_0) = \alpha(y,\bX \con \theta, z_0')
        \ , \quad 
        \beta(z,\bX \con \theta, z_0) = \beta(z,\bX \con \theta, z_0')
        \ , \quad \forall (z_0,z_0') \in \suppZ \ .
    \end{align*}
    Therefore, we denote $\alpha(y,\bX \con \theta)$ and $\beta(z,\bX \con \theta)$, suppressing the dependence on $z_0$.

\end{itemize} 
    
\end{theorem} 

Next, we modify Assumption \HL{A3} in order to accommodate a general form of $Z$:
\begin{itemize}
\item[\HT{A3-General $Z$}] (Boundedness for $\theta$) \\
There exist constants  $0< c_Z < C_Z < \infty$, $0<c_A<1$, $0<c_Y \leq C_Y < \infty$, and $0< c_g \leq C_g < \infty$ such that the parameter $\theta = \{\fy, \fa, \fz\}$ and the corresponding  $g^\star(\cdot \con \theta)$ in \eqref{eq-fixed point equation-binary} satisfy $f(Z=z \cond \bX) \in [c_Z,C_Z]$, $f(A=a \cond Z=z,\bX) \in [c_A,1-c_A]$,  $ f(Y=y \cond A=0,Z=z,\bX) \in [c_Y,C_Y]$, and $g^\star(y,\bX \con \theta) \in [c_g,C_g]$ for all $y \in \suppYX$, $a \in \{0,1\}$, $z \in \suppZb$, and $\bX \in \suppX$. 
\end{itemize}

\begin{theorem} 
\label{thm-global contraction}
    Suppose that Assumptions \HL{A1}-\HL{A2} and \HL{IV1}-\HL{IV4} hold for $(\potY{0},A,Z,\bX)$, with corresponding parameters $\theta = \{ \fz,\fa,\fy \}$. Further suppose that \HL{A3-General $Z$} holds for $\theta$. 
    For a fixed $\bX \in \suppX$ and $z_0 \in \suppZinf \equiv \arginf_{z \in \suppZ} f(A=1 \cond Z=z,\bX) $, define the mapping $\overline{\Psi}: \mathcal{L}_+^{\infty} (\mathcal{Y}_{\bX}) \rightarrow \mathcal{L}_+^{\infty} (\mathcal{Y}_{\bX})$ as the normalized version of $\Psi$ in \eqref{eq-Psi}:
    \begin{align*}
    &
    \overline{\Psi}
    \big( g(y,\bX) \con \theta, z_0 \big)
    =
    \frac{ \Psi\big( g(y,\bX)  \con \theta, z_0 \big) }{\int \Psi\big( g(t,\bX)  \con \theta, z_0 \big) \, dt} 
     \ .
    \end{align*}
    Consider the iterative update defined by
    \begin{align*}
        \overline{h}^{(j+1)}(y,\bX \con \theta, z_0) = \overline{\Psi}(\overline{h}^{(j)}(y,\bX \con \theta, z_0) \con \theta, z_0) \ , 
        \quad j \in \{0,1,\ldots\} \ ,
        \numeq 
        \label{eq-iterative update}
    \end{align*}
     with an arbitrary initial function $\overline{h}^{(0)} \in \mathcal{L}_+^{\infty} (\mathcal{Y}_{\bX})$ satisfying $\int \overline{h}^{(0)}(y,\bX) \, dy = 1$. Then, as $j \rightarrow \infty$, $\overline{h}^{(j)}(y,\bX \con \theta, z_0)$ converges to $g^{\star}(y,\bX \con \theta, z_0)$ in \eqref{eq-fixed point equation} for all $y \in \suppYX$, i.e., 
    \begin{align*}
        \sup_{y \in \suppYX}    
        \big| \overline{h}^{(j)}(y,\bX \con \theta, z_0) - g^{\star}(y,\bX \con \theta, z_0)
        \big|
        \rightarrow 0
        \quad \text{ as } \quad j \rightarrow \infty \ .
    \end{align*}    
    Moreover, the convergence of the iterative update  \eqref{eq-iterative update} is exponentially fast with respect to a certain divergence function; see Section \ref{sec-supp-global contraction} for details.     
\end{theorem}

Note that Theorems \ref{thm-Psi-binary} and \ref{thm-global contraction-binary} are special cases of \ref{thm-Psi} and \ref{thm-global contraction}  with $\suppZ = \{0,1\}$ and $z_0 = \argmin_{z \in \{0,1\}} f(A=1 \cond Z=z,\bX)$ and $z_1=1-z_0$. 

In addition, Theorem \ref{thm-IF} continues to hold for general $Z$. We therefore omit its full statement here, except to clarify how the function $\omega$ is defined: \\
\makebox[1.2cm][l]{\HL{$\omega$-i}} $\omega = \nu - \EXP \{ \nu \cond \potY{0},\bX \}
    -
    \EXP ( \nu \cond Z,\bX )
    +
    \EXP ( \nu \cond \bX )$ for some function $\nu(\potY{0},Z,\bX)$;\\
\makebox[1.2cm][l]{\HL{$\omega$-ii}} $\omega$ solves the following integral equation:  
    \begin{align} \tag{\ref{eq-w IE}}
    &
    \EXP \big[ \big\{ \potY{0} - \mu^*(Z,\bX) 
    \big\}  - \big\{  \omega (\potY{0},Z,\bX) - \mu^*(Z,\bX \con \omega) \big\} \cond A=1,\potY{0},\bX \big] = 0 \ .
    \end{align} 
For general $Z$, directly solving the integral equation in \eqref{eq-w IE} may be challenging in practice, as it involves multiple nested dependencies within the counterfactual data distribution. If $Z$ is categorical with finite support, the integral equation admits a closed-form solution. However, when $Z$ has more than two levels, the equation generally admits multiple solutions; we discuss these in the next Section. The situation is even more complex when $Z$ is continuous, as the integral equation admits multiple solutions and no closed-form representation appears to exist. This particular type of challenge may be viewed as an instance of IV overidentification, where the complexity of $Z$ exceeds that of the binary treatment $A$ \citep[Chapter 15]{Wooldridge2016}.

\subsubsection{A Closed-Form Solution of the Integral Equation \eqref{eq-w IE} under Categorical $Z$} 
\label{sec-supp-ClosedFormSolution}

In this Section, we provide details on constructing a closed-form solution of the integral equation \eqref{eq-w IE} when $Z$ is categorical. To simplify notation, let $\mathfrak{p}_{z}(y) = f^*(Z=z \cond A=1,\potY{0}=y,\bX)$. Without loss of generality, let $\mathcal{Z} = \{0,1,\ldots,m\}$. 

By \eqref{eq-proof-Y0AZ}, we can obtain an alternative representation of $\mathfrak{p}_z$ for $z \in \suppZ$:
\begin{align*}
    \mathfrak{p}_{z}(\potY{0},\bX)   
    & 
    =  f^*(Z=z \cond A=1, \potY{0},\bX)
    \\
    & =
    \frac{ 
    \frac{\alpha^*(\potY{0}, \bX) \beta^*(z,\bX)}{ 1 + \alpha^*(\potY{0},\bX) \beta^*(z,\bX) }
    f^*(\potY{0} \cond \bX) f^*(Z=z \cond \bX)
    }{
    \sum_{c=0}^{m}
    \frac{\alpha^*(\potY{0}, \bX) \beta^*(c,\bX)}{ 1 + \alpha^*(\potY{0},\bX) \beta^*(c,\bX) }
    f^*(\potY{0} \cond \bX) f^*(Z=c \cond \bX)
    }
    \\  
    &
    \propto   \frac{  \beta^*(z,\bX) f^*(Z=z \cond \bX) 
    }{ 1 + \alpha^*(\potY{0},\bX) \beta^*(z,\bX) }  \ .
    \numeq \label{eq-proof-frakp1 poly}
\end{align*} 

From the proof of Lemma \ref{lemma-beta and gamma}, specifically the derivation in \eqref{eq-proof-vtilde L representation poly}, it follows that $\omega = \widetilde{\nu} = \nu - \nu^\dagger$ admits the representation:
\begin{align*}
	&
    \omega (\potY{0},Z=z,\bX) 
    \\
    &
    = 
    \big[ L_{z}(\potY{0},\bX) - \EXP \big\{ L_{z}(\potY{0},\bX) \cond \bX \big\} \big]
    \big\{ \ind(Z=z) - \Pr \big( Z=z \cond \bX \big) \big\} \ , 
    \numeq \label{eq-proof-omega L representation poly}
\end{align*}
for functions $L_{z}(\potY{0}, \bX) \in \mathcal{L}^2(\potY{0},\bX)$ for $z \in \{0,1,\ldots,m\}$. Note that $L_0$ is uniquely determined once $\{L_1,\ldots,L_{m} \}$ are fixed due to the constraint $\EXP \{ \nu^\dagger \cond \potY{0},\bX \} = 0$. 

Equation \eqref{eq-proof-omega L representation poly} implies that the following expressions hold for $ z \in \suppZ$:
\begin{align*}
    & 
    \mu^*(z,\bX \con \omega)
    \\
    &
    =
    \EXP \big\{ \omega(\potY{0},Z,\bX) \cond A=1,Z=z,\bX \big\}
    \\
    &
    = 
    \big\{ \ind(Z=z) - \Pr \big( Z=z \cond \bX \big) \big\}
    \big[ \mu^*(z,\bX \con L_z) - \EXP\{ L_z(\potY{0},\bX) \cond \bX\} \big] \ , 
\end{align*}
and
\begin{align*}
    &
    A
    \big\{
    \omega (\potY{0},Z=z,\bX)
    - 
    \mu^* (z,\bX \con \omega)
    \big\}
    \\
    &
    =
    A 
    \big\{ \ind(Z=z) - \Pr \big(Z=z \cond \bX \big) \big\}
    \big[ L_{z}(\potY{0},\bX) - \mu^*(z,\bX \con L_z) \big]  \ .
\end{align*}

Consequently, we obtain 
\begin{align*}
    &
    \EXP \big\{
    \omega(\potY{0},Z,\bX)
    - 
    \mu^*(z,\bX \con \omega)  \cond A=1,\potY{0},\bX \big\}
    \\
    &
    = 
    \sum_{z=0}^{m}
    \mathfrak{p}_{z} (\potY{0},\bX) f^*(Z \neq z \cond \bX) \big\{ L_z (\potY{0},\bX) - \mu^*(z,\bX \con L_z) \big\}
    \ .
\end{align*} 
Likewise, we have  
\begin{align*}
    &
    \EXP \big\{
    \potY{0}
    - \mu^*(Z,\bX) \cond A=1,\potY{0},\bX \big\}  
    = 
    \mathcal{G}^{(0)}
    - \sum_{z=0}^{m} \mathfrak{p}_z (\potY{0},\bX)
    \mu^*(z,\bX \con L_z) \ .
\end{align*} 

Therefore, $\omega$ satisfies \HL{$\omega$-ii} if and only if $\{L_0,\ldots,L_{m} \}$ satisfy the following condition:
\begin{align*}
&
 \potY{0}
    - \sum_{z=0}^{m} \mathfrak{p}_z (\potY{0},\bX) 
    \mu^*(z,\bX \con L_z)
\\
&
    =
    \sum_{z=0}^{m}
    \mathfrak{p}_{z} (\potY{0},\bX) f^*(Z \neq z \cond \bX) \big\{ L_z (\potY{0},\bX) - \mu^*(z,\bX \con L_z) \big\}  
    \\
    \Leftrightarrow
    \quad 
	&    
\sum_{z=0}^{m}
    \mathfrak{p}_{z} (\potY{0},\bX) f^*(Z \neq z \cond \bX) L_z (\potY{0},\bX)  
    \\
&
    =
 \potY{0}
 -
 \sum_{z=0}^{m}
 \mathfrak{p}_{z} (\potY{0},\bX) f^*(Z = z \cond \bX) \mu^*(z,\bX \con L_z)
    \numeq 
    \label{eq-proof-L equation poly}
\end{align*}
Note that \eqref{eq-proof-L equation poly} can be expressed as the following Fredholm integral equation of the second kind with a separable kernel; see \citet{Kress2014} for technical details on Fredholm integral equations. To make this explicit, let $\bm{L}(\potY{0},\bX) = (L_0(\potY{0},\bX),\ldots,L_m(\potY{0},\bX))\T \in \R^{m+1}$. Then, we find
\begin{align*}
	&	\sum_{z=0}^{m}
    \mathfrak{p}_{z} (\potY{0},\bX) f^*(Z \neq z \cond \bX) L_z (\potY{0},\bX)
     =
    \bm{P}\T (\potY{0},\bX) 
    \bm{L}(\potY{0},\bX) \ ,
\end{align*}
and
\begin{align*}
	&
	\sum_{z=0}^{m}
    \mathfrak{p}_{z} (\potY{0},\bX) f^*(Z = z \cond \bX) 
	\mu^*(z,\bX \con L_z)
    \\
	&
	= 
	\sum_{z=0}^{m}
    \mathfrak{p}_{z} (\potY{0},\bX) f^*(Z = z \cond \bX) 
    \int 
	L_{z} (t,\bX) f^*(\potY{0}=t \cond A=1,Z=z,\bX) \, dt
	\\
	&
	=
	\bm{Q}\T (\potY{0},\bX) 
	\int  \bm{L}(t,\bX) \circ \bm{f}(\potY{0}=t \cond A=1, \suppZ,\bX)  \, dt \ ,
\end{align*}  
where $\circ$ denotes the element-wise (Hadamard) product, and 
\begin{align*}
	&
	\bm{P} (\potY{0},\bX)
    \equiv 
    \left[ 
    	\begin{array}{c}
    \mathfrak{p}_{0} (\potY{0},\bX) f^*(Z \neq 0 \cond \bX)
	\\
    \vdots
    \\
    \mathfrak{p}_{m} (\potY{0},\bX) f^*(Z \neq m \cond \bX)
    	\end{array}
    \right]  
    \in \R^{m+1}
     \ , 
    \\
    &
\bm{Q}(\potY{0},\bX)
\equiv 
\left[ 
	\begin{array}{c}
	 \mathfrak{p}_{0} (\potY{0},\bX) f^*(Z = 0 \cond \bX)
	 \\
	 \vdots
	 \\
    \mathfrak{p}_{m} (\potY{0},\bX) f^*(Z  = m \cond \bX)
	\end{array}
\right] 
\in \R^{m+1}
    \ , 
    \\
    &
	\bm{f}(\potY{0}=t \cond A=1, \suppZ,\bX) 
    =
    \left[
    	\begin{array}{c}
    	f^*(\potY{0}=y \cond A=1,Z=0,\bX) 
    	\\
    	\vdots
    	\\
    	f^*(\potY{0}=y \cond A=1,Z=m,\bX) 
    	\end{array} \right] 
    	\in \R^{m+1}
    	\ .
\end{align*}
Consequently, we establish a Fredholm integral equation of the second kind as follows:
\begin{align*}
	&
	\bm{P}\T (\potY{0},\bX) 
    \bm{L}(\potY{0},\bX)
    \\
    &
    =
    \potY{0}
    -
   \bm{Q}\T (\potY{0},\bX) 
   \underbrace{
	\int  \bm{L}(t,\bX) \circ \bm{f}(\potY{0}=t \cond A=1,  \suppZ,\bX)  \, dt }_{ \equiv \bm{\mu} (\suppZ,\bX \con \bm{L}) }
	  \ .	
    \numeq 
    \label{eq-proof-L equation poly Fredholm}
\end{align*}
Note that \eqref{eq-proof-L equation poly Fredholm} defines a linear system and therefore admits a closed-form solution. However, the system is underdetermined. Specifically, it involves $m+1$ unknowns but only two restrictions---one from the identity \eqref{eq-proof-L equation poly Fredholm} and one from the constraint and the other from the constraint on $\bm{L}$ given by $\EXP \{ \nu^\dagger \cond \potY{0},\bX \} = 0$.  Therefore, the solution is generally not unique when $m>1$, i.e., non-binary $Z$.

Among multiple solutions, we may consider the minimum norm solution. Suppose that $\bm{L}$ has the following form:
\begin{align*}
	\bm{L}(\potY{0},\bX) 
	=
	\frac{ 
	\potY{0} - \bm{Q}\T(\potY{0},\bX) \bm{\mu}(\suppZ,\bX \con \bm{L})	
	}{ 
	\big\|	\bm{P} (\potY{0},\bX) \big\|_{2}^{2}
	 }
	 \bm{P} (\potY{0},\bX)  \ .
\end{align*}
Then, $\bm{\mu} (\suppZ,\bX \con \bm{L})$ must satisfy the following relationship: 
\begin{align*} \
&
\bm{\mu} (\suppZ,\bX \con \bm{L})
\\
&
=
\int  \bm{L}(t,\bX) \circ \bm{f}(\potY{0}=t \cond A=1,  \suppZ,\bX)  \, dt 
\\
&
=
\int 
	\frac{ 
	t - \bm{Q}\T(t,\bX) \bm{\mu}(\suppZ,\bX \con \bm{L})	
	}{ 
	\big\|	\bm{P} (t,\bX) \big\|_{2}^{2}
	 }
	 \bm{P} (t,\bX)
	 \circ \bm{f}(\potY{0}=t \cond A=1,  \suppZ,\bX)  \, dt 
\\
&
=
\underbrace{
\int 
	\frac{ 
	t \bm{P} (t,\bX)
	 \circ \bm{f}(\potY{0}=t \cond A=1,  \suppZ,\bX)  
	}{ 
	\big\|	\bm{P} (t,\bX) \big\|_{2}^{2}
	 }
	  \, dt }_{\equiv \bm{b}(\suppZ,\bX)}
	 \\
	 &
	 \qquad 
	 - 
\underbrace{
	 \bigg\{
	 \int 
	\frac{ 
	 \bm{P} (t,\bX)
	 \circ \bm{f}(\potY{0}=t \cond A=1,  \suppZ,\bX) 
	}{ 
	\big\|	\bm{P} (t,\bX) \big\|_{2}^{2}
	 } 
	\bm{Q}\T(t,\bX) 
	 \, dt  
	 \bigg\} }_{\equiv \bm{\Omega}(\suppZ,\bX)}
	 \bm{\mu}(\suppZ,\bX \con \bm{L})	
	 \\
	 &
	 \equiv
	 \bm{b}(\suppZ,\bX) - \bm{\Omega}(\suppZ,\bX)
	 \bm{\mu} (\suppZ,\bX \con \bm{L}) \ .
\end{align*}
Consequently, we find
\begin{align*}
 \bm{\mu} (\suppZ,\bX \con \bm{L})
=
\big\{ I + 
\bm{\Omega}(\suppZ,\bX) \big\}^{+} \bm{b}(\suppZ,\bX) \ ,
\end{align*}
where $B^{+}$ is the Moore-Penrose pseudo-inverse of a matrix $B$, which is used to address the rank deficiency induced by the constraint on $\bm{L}$ given by $\EXP \{ \nu^\dagger \cond \potY{0},\bX \} = 0$. This yields a closed-form representation of $\bm{L}(\potY{0},\bX)$:
\begin{align*}
	\bm{L}(\potY{0},\bX) 
	=
	\frac{ 
	\potY{0} - \bm{Q}\T(\potY{0},\bX) 
	\big\{ I + \bm{\Omega}(\suppZ,\bX) \big\}^{+} \bm{b}(\suppZ,\bX) 
	}{ 
	\big\|	\bm{P} (\potY{0},\bX) \big\|_{2}^{2}
	 }
	 \bm{P} (\potY{0},\bX)  \ .
\end{align*}

\subsubsection{Inference of the Quantile Treatment Effect on the Treated} \label{sec-supp-QTT}

In this Section, we provide a simple method for inferring the quantile treatment effect on the treated (QTT). Note that the QTT is defined by $\tau_{q}^* = \tau_{1,q}^* - \tau_{0,q}^*$ where $\tau_{a,q}^* = F_{\potY{a}|A=1}^{*-1}(q)$ for $a \in \{0,1\}$; here,  $F_{\potY{a} \cond A=1}^*$ is the cumulative distribution function (CDF) of $\potY{a} \cond (A=1)$ and $q \in (0,1)$ is a predetermined quantile of interest. We assume that the inverses of the CDFs are well-defined, at least in a local neighborhood around the true quantile. This condition is automatically satisfied when the outcome is continuous.

The proposed inference procedure is motivated by \citet{Berger1994} and \citet{Lee2025}. The ultimate goal is to obtain a $100(1-c)$\% confidence interval for $\tau_q^*$ for a given $c>0$. We first select a constant $c_{1} < c$; \citet{Berger1994} recommended choosing a small value, such as $c_{1}=10^{-3}$ or $c_{1}=10^{-4}$. Then, from the standard bootstrap, one can obtain a $100(1-c_{1})$\% confidence interval for $\tau_{1,q}^*$, which is denoted by $\widehat{\mathcal{C}}_{1,c_1}$.  

Next, note that Theorem \ref{thm-AN} can be generalized for constructing an asymptotically normal estimator for $\rho^* \equiv \EXP \{ \mathcal{G}(\potY{0},\bX)  \cond A=1 \}$, where $\mathcal{G}(\cdot)$ is a fixed, uniformly bounded function; see Section \ref{sec-proof-AN} for this relaxation. For our purposes, we choose a specifically designed function $\mathcal{G}$. Specifically, for the candidate QTT value, denoted by $\tau_{q} \in \R$, and the candidate quantile of $Y \cond (A=1)$, denoted by $\tau_{1,q} \in \R$, define 
\begin{align*}
&
\mathcal{G}(\potY{0} \con \tau_{q}, \tau_{1,q},q) \equiv \ind \{ \potY{0} \leq \tau_{1,q} - \tau_{q} \} - q \ ,
\\
&
\rho^*(\tau_{q}, \tau_{1,q}, q)
\equiv 
\EXP \big\{ \mathcal{G}(\potY{0} \con \tau_{q}, \tau_{1,q},q) \cond A=1 \big\} \ .
\end{align*}
Note that $\tau_{1,q} - \tau_{q}$ corresponds to $\tau_{0,q}$, the candidate counterfactual quantile of $\potY{0} \cond (A=1)$. Consequently, given $( \tau_{q}, \tau_{1,q} , q)$, one can obtain a p-value for testing $H_0: \rho^*(\tau_q, \tau_{1,q},q) = 0$ by inverting the confidence interval obtained from Theorem \ref{thm-AN}. This p-value can be expressed as
\begin{align*}
    p(\tau_{q}, \tau_{1,q},q) = 2 \overline{\Phi} \bigg( \sqrt{N}
    \bigg|
    \frac{  \widehat{\rho}(\tau_{q},\tau_{1,q},q) }
    { \widehat{\sigma}_{\rho}(\tau_{q},\tau_{1,q},q) }
    \bigg|
    \bigg) \ , 
\end{align*}
where $\overline{\Phi}$ is the the survival function (complementary cumulative distribution function) of the standard normal distribution, $\widehat{\rho}(\tau_{q},\tau_{1,q},q)$ is an estimator for $\rho^*(\tau_{q},\tau_{1,q},q)$, and $\widehat{\sigma}_{\rho}^2(\tau_{q},\tau_{1,q},q)$ is the corresponding variance estimator; we remark that $\widehat{\rho}(\tau_{q},\tau_{1,q},q)$ and $\widehat{\sigma}_{\rho}^2(\tau_{q},\tau_{1,q},q)$ can be constructed following the procedure in Section \ref{sec-estimation}.

In turn, we consider the following adjusted p-value:
\begin{align*}
    \overline{p}(\tau_{q},q) \equiv 
    \sup
    \big\{ 
    p(\tau_{q}, \tau_{1,q},q)
    \cond \tau_{1,q} \in \widehat{\mathcal{C}}_{1,c_1} 
    \big\}
    +
    c_1 \ .
\end{align*}
This quantity can be interpreted as a p-value for testing the hypothesis $H_0: \rho^*(\tau_q, \tau_{1,q}^* ,q) = 0$, where $\tau_{1,q}$ is set to its true value $\tau_{1,q}^*$. Consequently, one can invert this p-value in order to obtain a $100(1-c)\%$ confidence interval for $\tau_{q}^*$, i.e.,
\begin{align*}
    \widehat{\mathcal{C}}_{c} \equiv \big\{ 
        \tau_{q} \cond \overline{p}(\tau_q,q) > c 
    \big\} \ .
\end{align*}
Under the condition for Theorem \ref{thm-AN}, one can establish that
\begin{align*}
    \lim_{N \rightarrow \infty }
    \Pr \big( \tau_q^* \in  \widehat{\mathcal{C}}_{c} \big) = 1- c \ . 
\end{align*}
The proof is analogous to that in \citet{Lee2025} and is omitted here.

\subsubsection{Inference of the Average Treatment Effect} \label{sec-supp-ATE}

In the main paper, all IV assumptions are stated with respect to $\potY{0}$. Suppose that an analogous set of assumptions also holds for $\potY{1}$, in addition to $\potY{0}$. Specifically, \HL{A2}, \HL{IV2}, and \HL{IV4} are replaced with the following \HL{A2'}, \HL{IV2'}, and \HL{IV4'}:
\begin{itemize}
    \item[\HT{A2'}] (Overlap) For each $a \in \{0,1\}$ and $\bX \in \suppX$, the support of $\potY{a} \cond (A=a',Z=z,\bX)$ is identical for all $a' \in \{0,1\}$ and $z \in \suppZ$, which is denoted by $\suppYX^{(a)}$; 
    \item[\HT{IV2'}] (Unconfoundedness) $\potY{a} \indep Z \cond \bX$ for both $a \in \{0,1\}$; 
    \item[\HT{IV4'}] (Logit-separable Treatment Mechanism) 
    For all $y \in \suppYX^{(a)}$, $a \in \{0,1\}$, $(z,z') \in \suppZ$, and $\bX \in \suppX$, we have
    \begin{align*} 
        \frac{ \Pr(A=1-a \cond \potY{a}=y,Z=z,\bX ) }{ \Pr(A=a \cond \potY{a}=y,Z=z,\bX ) }
        =
        \alpha_{a}^*(y,\bX) 
        \beta_{a}^*(z,\bX) \ .
    \end{align*}
\end{itemize}
Note that the generative model described in Section \ref{sec-generative model} is already compatible with \HL{A1}, \HL{A2'}, \HL{IV1}, \HL{IV2'}, \HL{IV3}, and \HL{IV4'}. Therefore, it can also serve as an illustrative example for justifying these assumptions.

Under these assumptions, a parallel set of results can be established for causal effects defined over the entire population, such as the average treatment effect (ATE), defined by $\tau_{\text{ATE}}^* = \EXP\{ \potY{1} - \potY{0} \}$. For example, the following result is an analogue of Theorem \ref{thm-IF} for the ATE:

    \begin{theorem} \label{thm-IF ATE}
Suppose that Assumptions \HL{A1}, \HL{A2'}, \HL{IV1}, \HL{IV2'}, \HL{IV3}, and \HL{IV4'} hold for $(\potY{a},A,Z,\bX)$ for $a \in \{0,1\}$, and let $\model_{\text{ATE}}$ be a semiparametric model for the observed data law of $\bO=(Y,A,Z,\bX)$ under these assumptions. Suppose further that there exist functions $\omega_{a}(\potY{a},Z,\bX)$ for $a \in \{0,1\}$ satisfying the following conditions: \\
\makebox[2.3cm][l]{\HT{$\omega$-i-ATE}} $\omega_{a} = \nu_{a} -  \EXP\{ \nu_{a} \cond \potY{a},\bX\} - \EXP(\nu_{a} \cond Z,\bX) + \EXP(\nu_{a} \cond \bX)$ \\
\makebox[2.3cm][l]{} for some function $\nu_{a}(\potY{a},Z,\bX)$; \\
\makebox[2.3cm][l]{\HT{$\omega$-ii-ATE}} $\omega_{a}$ solves the following integral equation:  
    \begin{align*}
    &
    \EXP \big[  
    \big\{ \potY{a} - \mu_{a}^*(Z,\bX)  \big\}  - \big\{  \omega_{a} (\potY{a},Z,\bX) - \mu_{a}^*(Z,\bX \con \omega_{a}) \big\}  
    \, \big| \, A=1-a,\potY{a},\bX \big] = 0 \ , 
    \end{align*}  
    where $\mu_{a}^*(Z,\bX \con h_{a}) = \EXP \{ h_{a}(\potY{a},Z, \bX) \cond A=1-a,Z,\bX \}$.
    
Then, the following function is an influence function for $\tau_{\text{ATE}}^*$ in model $\model_{\text{ATE}}$, 
  \begin{align*}
  &
  \InfFt_{\text{ATE}}^*(\bO)
  \\
  &
  =
  \left[ 
    \begin{array}{l}
    \{ A \alpha_{1}^*(Y,\bX) \beta_{1}^*(Z,\bX) - (1-A) \}
     \big\{ \mu_{1}^*(Z,\bX) - Y \big\}    
     \\
       -
       A \alpha_{1}^*(Y,\bX) \beta_{1}^*(Z,\bX) \big\{ \mu_{1}^*(Z,\bX \con \omega_{1}) - \omega_{1}(Y,Z,\bX)  \big\}
  \\
  -
    A
    \omega_{1}(Y,Z,\bX) 
    - 
    (1-A) 
    \mu_{1}^*(Z,\bX \con \omega_{1})
    \end{array}
    \right]
    \\
    &
    + 
   \left[ 
    \begin{array}{l}
    \{ A - (1-A) \alpha_{0}^*(Y,\bX) \beta_{0}^*(Z,\bX) \}
     \big\{ Y - \mu_{0}^*(Z,\bX) \big\}    
     \\
       +
       (1-A) \alpha_{0}^*(Y,\bX) \beta_{0}^*(Z,\bX) \big\{ \omega_{0}(Y,Z,\bX) - \mu_{0}^*(Z,\bX \con \omega_{0}) \big\}
  \\
  +
    A
    \mu_{0}^*(Z,\bX \con \omega_{0})
    + 
    (1-A) \omega_{0}(Y,Z,\bX) 
    \end{array}
    \right]
    -
    \tau_{\text{ATE}}^*
     \ .
     \numeq  
     \label{eq-IF ATE}
\end{align*}
In addition, every IF for $\tau_{\text{ATE}}^*$ in model $\model_{\text{ATE}}$ has a form of \eqref{eq-IF ATE}. 
\end{theorem}
The remaining results follow analogously. In particular, a consistent, asymptotically normal, and semiparametric efficient estimator can be constructed following the approach in Section \ref{sec-estimation} under binary $Z$, and inference for the population quantile treatment effect of the population $\tau_{q,\text{QTE}}^* \equiv F_{\potY{1}}^{*-1}(q) - F_{\potY{0}}^{*-1}(q)$, can be obtained analogously to Section \ref{sec-supp-QTT}. These results are omitted, as their derivation is straightforward.

\subsubsection{A Missing Data Setting} \label{sec-supp-MNAR}

We reinterpret our IV framework for a nonignorable missing data model. Consider a data structure $\{ (1-A)W, A, Z,\bX \}$ where $\bX$ is a collection of baseline covariates, $W$ is an outcome of
interest, $Z$ is an instrument, and $A$ is an indicator of whether $W$ is missing $(A = 1)$ or measured $(A = 0)$. Note that $A$ is defined in this way to maintain consistency with the notation used in the main paper.

For simplicity, let $\tau_W^* \equiv \EXP(W)$ represent the target estimand; more generally, one could consider $\EXP\{\mathcal{G}(W, \bX)\}$ as the target, where $\mathcal{G}(\cdot)$ is a fixed, uniformly bounded function. Note that $\EXP(W)
    =
    \EXP ( W \cond A=0 ) \Pr(A=0)
    +
    \EXP ( W \cond A=1 ) \Pr(A=1)$, and $\EXP ( W \cond A=0 )$ and $\Pr(A=1)$ are identified and can be estimated from the observed data based on their EIFs. Therefore, to identify and construct an IF-based estimator, it suffices to focus on $\EXP ( W \cond A=1 )$. 
    
    To this end, we make assumptions on the full data $(W,A,Z,\bX)$ where $W$ replaces $\potY{0}$ from the main paper. Under these assumptions, all results of the main paper extend to the missing data setting. For example, the following result is the analogue of Theorem~\ref{thm-IF} for the missing data setting.
    
    \begin{theorem} \label{thm-IF MNAR}
Suppose that Assumptions \HL{A1}-\HL{A2} and \HL{IV1}-\HL{IV4} hold for $(W,A,Z,\bX)$ where $W$ replaces $\potY{0}$, and let $\model_{W}$ be a semiparametric model for the observed data law of $\{ (1-A)W, A, Z,\bX \}$ under these assumptions. Suppose further that there exists a function $\omega(W,Z,\bX)$ satisfying the following conditions: \\
\makebox[2.3cm][l]{\HT{$\omega$-i-missing}} $\omega = \nu -  \EXP(\nu \cond W,\bX) - \EXP(\nu \cond Z,\bX) + \EXP(\nu \cond \bX)$ for some function $\nu(W,Z,\bX)$; \\
\makebox[2.3cm][l]{\HT{$\omega$-ii-missing}} $\omega$ solves the following integral equation:  
    \begin{align*}
    &
    \EXP \big[ \big\{ W - \mu^*(Z,\bX) 
    \big\}  - \big\{  \omega (W,Z,\bX) - \mu^*(Z,\bX \con \omega) \big\} \cond A=1,W,\bX \big] = 0 \ .
    \end{align*}
    
Then, the following function is an influence function for $\tau_W^* $ in model $\model_{W}$, 
  \begin{align*}
  \InfFt_{W}^*(\bO)
  =
   \left[ 
    \begin{array}{l}
    (1-A) W
    +
    (1-A) \alpha^*(W,\bX) \beta^*(Z,\bX)
     \big\{ W - \mu^*(Z,\bX) \big\}     
     \\
     +
     A
     \mu^*(Z,\bX)
     - \tau_{W}^*
     \\
       -
       (1-A) \alpha^*(Y,\bX) \beta^*(Z,\bX) \big\{ \omega(Y,Z,\bX) - \mu^*(Z,\bX \con \omega) \big\}
  \\
  -
    A
    \mu^*(Z,\bX \con \omega)
    - 
    (1-A) \omega(Y,Z,\bX) 
    \end{array}
    \right]
     \ .
     \numeq 
     \label{eq-IF MNAR}
\end{align*}
In addition, every IF for $\tau_{W}^*$ in model $\model_{W}$ has a form of \eqref{eq-IF MNAR}. 
\end{theorem}

The remaining results follow analogously. Finally, quantiles of the outcome subject to missingness can be inferred using the approaches described in Section \ref{sec-supp-QTT}. These results are omitted, as they can be readily deduced.

\subsection{Details of the Simulation} \label{sec-supp-simulation}

We provide details on the construction of the other three estimators. First,  $\widehat{\tau}_{\text{2SLS}}$ and $\widehat{\tau}_{\text{LS}}$ were obtained based on the following R code:
\begin{itemize}
    \item \makebox[1.2cm][l]{$\widehat{\tau}_{\text{2SLS}}$:} \texttt{ivreg::ivreg(Y $\sim$ A + X | Z + X)};
    \item \makebox[1.2cm][l]{$\widehat{\tau}_{\text{LM}}$:} \texttt{lm(Y $\sim$ A + Z + X)}.     
\end{itemize}
Here, the two-stage least squares (2SLS) estimator is implemented via the \texttt{ivreg} R package \citep{ivregpackage}. Second, $\widehat{\tau}_{\text{Ign}}$ was constructed as follows. \citet{Hahn1998} showed that the EIF for the ATT $\tau^*$ under the assumption of no unmeasured confounding is given by
\begin{align*}
    \frac{ \big\{ A - (1-A) \frac{ \Pr(A=1 \cond \bL) }{ \Pr(A=0 \cond \bL) } \big\} \{ Y - \EXP(Y \cond A=1 , \bL) \} - A \tau^* }
    { 
        \EXP(A)
    } \ ,
\end{align*}
where $\bL$ is the collection of confounders of $A$ and $Y$. Therefore, an EIF-based estimator can be constructed from the cross-fitting approach described in Section \ref{sec-estimation}. Following the main paper, we randomly partition the $N$ units, denoted by $\mathcal{I} = \{1,\ldots,N\}$, into $K$ non-overlapping folds $\{ \II_1,\ldots, \II_K \}$, and define the complement of each fold as $\II_k^c = \mathcal{I} \setminus \II_k$ for $k  \in \{1,\ldots,K\}$. Then, the EIF-based estimator $\widehat{\tau}_{\text{Ign}}$ is defined by
\begin{align*}
    \widehat{\tau}_{\text{Ign}}
    = 
    \frac{     
    \displaystyle{
    \sum_{k=1}^{K}
    \sum_{i \in \II_k}
    \bigg\{ A_i - (1-A_i) \frac{ \widehat{\Pr}\LSS (A_i=1 \cond \bL_i) }{ \widehat{\Pr}\LSS (A_i=0 \cond \bL_i) } 
    \bigg\}
    \big\{ Y_i - \widehat{\EXP}\LSS (Y_i \cond A_i=1 , \bL_i) \big\}
    }
    }{
    \sum_{i=1}^{N} A_i
    } \ ,
\end{align*}
where $\widehat{\Pr}\LSS$ and $\widehat{\EXP}\LSS$ are estimated from $\II_k^c$ using the machine learning approaches in Section \ref{sec-supp-estimation}. For both the simulation and data analysis, we set $\bL=(Z,\bX)$, treating $Z$ as a regular confounder.

\subsection{Details of the Application} \label{sec-supp-application}

In this Section, we provide additional details on the real-world application discussed in Section \ref{sec-application}.

\subsubsection{Assessment of the Relevance Assumption \protect\HL{IV3}}
 
We report the $t$-statistic from the first-stage regression of $A$ on $(Z,\bX)$ (including only main effects) to assess whether the coefficient of $Z$ differs significantly from zero. Under the null hypothesis that $A \indep Z \cond \bX$, the coefficient of $Z$ in this regression should be zero. The $t$-statistic for testing this coefficient follows a $t$-distribution with 9,228 degrees of freedom. In our analysis, the observed $t$-statistic was 15.98, corresponding to a p-value smaller than machine precision. We therefore conclude that \HL{IV3} does not appear to be violated.

\subsubsection{Assessment of the Overlap Assumption \protect\HL{A2}}

We indirectly assessed whether the overlap assumption \HL{A2} appears to be satisfied in the data. We note that \HL{A2} involves counterfactual data $(\potY{0},A=1,Z,\bX)$ and, therefore, cannot be verified using the observed data alone. Accordingly, the following results should be interpreted as indicative evidence rather than a direct assessment of \HL{A2}.

We first fix an observation $i \in \{1,\ldots,N\}$, with $N=9240$. Then, for each cross-fitting index $s \in \{1,\ldots,200\}$, let $\II_{s}^{(-i)}$ denote the split-sample fold that does not contain observation $i$, and  $\widehat{f}_{s}^{(-i)}$ denote the corresponding density estimator constructed using $\mathcal{I}_{s}^{(-i)}$. For $y \in \{0,1,\ldots,52\}$, $a \in \{0,1\}$, and $z \in \{0,1\}$, we considered the median across the 200 cross-fitted density estimates, denoted by $\widehat{f}_{\text{med}}$, i.e., 
\begin{align*}
    \widehat{f}_{\text{med}}(Y_i=y \cond A_i = a, Z_i=z, \bX_i)
    &
    =
    \median_{s=1,\ldots,200}
    \widehat{f}_{s}^{(-i)} (Y_i=y \cond A_i=a, Z_i=z, \bX_i) 
    \ .
\end{align*} 
For each observation $i$, we then defined $r_{i}(y,a,z)$ as 
\begin{align*}
    r_{i}(y,a,z) 
    =
    \frac{ \widehat{f}_{\text{med}}(Y_i=y \cond A_i = a, Z_i=z, \bX_i) }{
    \widehat{f}_{\text{med}}(Y_i=y \cond A_i = 0, Z_i=0, \bX_i)
    } \ ,
\end{align*}
which represents the ratio of the conditional density of $Y$ given $(A=a,Z=z,\bX)$ relative to its baseline value $(A=0,Z=0,\bX)$. 

Table \ref{tab-supp-Overlap} reports summary statistics of $r_i(y,a,z)$ across $y \in \{0, \ldots, 52\}$ and $i \in \{1, \ldots, N\}$. We observe that the distribution of $r_i(y,a,z)$ is close to 1, suggesting that the estimated densities are consistent with \HL{A2}.

\begin{table}[!htp]
\renewcommand{\arraystretch}{1.2} \centering
\setlength{\tabcolsep}{10pt} 
\begin{tabular}{|c|c|c|c|c|c|}
\hline
$(a,z)$ & Min    & Q1     & Q2     & Q3     & Max    \\ \hline
$(0,0)$ & 1.0000 & 1.0000 & 1.0000 & 1.0000 & 1.0000 \\ \hline
$(0,1)$ & 0.9993 & 1.0000 & 1.0001 & 1.0002 & 1.0004 \\ \hline
$(1,0)$ & 0.9833 & 0.9987 & 0.9992 & 1.0002 & 1.0222 \\ \hline
$(1,1)$ & 0.9835 & 0.9989 & 0.9993 & 1.0001 & 1.0217 \\ \hline
\end{tabular}
\caption{Summary Statistics of $r_{i}(y,a,z)$ across $y \in \{0, \ldots, 52\}$ and $i \in \{1, \ldots, N\}$, reported separately for each $(a,z)$ pair.}
\label{tab-supp-Overlap}
\end{table}

\subsubsection{Assessment of the Falsification Implication in Theorem \protect\ref{thm-falsification}} 

We evaluated whether the IV model is falsified, as characterized in Theorem \ref{thm-falsification}. Similar to the assessment of \HL{A2} above, we considered the median across the 200 cross-fitted density estimates. Specifically, we defined the following quantities for $y \in \{0,1,\ldots,52\}$ and $z \in \{0,1\}$:
\begin{align*}
    \widehat{f}_{\text{med}}(Y_i=y,A_i=0 \cond Z_i=z, \bX_i)
    &
    =
    \median_{s=1,\ldots,200}
    \widehat{f}_{s}^{(-i)} (Y_i=y,A_i=0 \cond Z_i=z, \bX_i) \ ,
    \\
    \widehat{f}_{\text{med}}(A_i=0 \cond Z_i=z, \bX_i)
    &
    =
    \median_{s=1,\ldots,200}
    \widehat{f}_{s}^{(-i)} (A_i=0 \cond Z_i=z, \bX_i) \ ,
    \\
    \widehat{f}_{\text{med}}(Z_i=z \cond \bX_i)
    &
    =
    \median_{s=1,\ldots,200}
    \widehat{f}_{s}^{(-i)} (Z_i=z \cond \bX_i)
    \ .
\end{align*} 

Our first approach was to implement the KS-type falsification test described in Section \ref{sec-supp-falsification}. In order to approximate $\mathcal{G}_{\text{ind}}$, we considered the following class:
\begin{align*}
    \widetilde{\mathcal{G}}_{\text{ind}}
    =
    \big\{ \ind\{Y=y, X_{j} = x_j \} \cond y \in \{0,1,\ldots,52\} ,\ x_j \in \text{supp}(X_j), \ j \in \{1,\ldots,d\}  \} \ .
\end{align*}
In words, $\widetilde{\mathcal{G}}_{\text{ind}}$ consists of indicator functions defined by $Y$ and the $j$th covariate. These functions are well-defined because both $Y$ and all covariates are discrete. In addition, we calculated the test statistic in \eqref{eq-proof-KS statistic} using $\widetilde{\mathcal{G}}_{\text{ind}}$ and the estimated $\widehat{\kappa}_0$ with the median-adjusted estimates $\widehat{f}_{\text{med}}$. We obtained a test statistic of $T_N = -0.5485$, while the bootstrap critical value at the 5\% level was $3.198$. Since $T_N$ is smaller than the critical value, there is insufficient evidence to refute the IV assumptions based on the KS-type falsification test.

The second approach was to use the direct assessment of the falsification implication (condition \eqref{eq-falsification} in the main paper).
For each observation $i$, we evaluated whether condition \eqref{eq-falsification} is violated by examining the following criterion:
\begin{align*} 
&
\frac{
\widehat{f}_{\text{med}}(Y_i=y, A_i=0 \mid Z_i=1, \bX_i)
-
\widehat{f}_{\text{med}}(Y_i=y, A_i=0 \mid Z_i=0, \bX_i)
}{
\widehat{f}_{\text{med}}(A_i=0 \mid Z_i=1, \bX_i)
-
\widehat{f}_{\text{med}}(A_i=0 \mid Z_i=0, \bX_i)
}
\geq 0 \ .
\numeq 
\label{eq-supp-falsification1}
\end{align*}
Finally, we computed the empirical proportion of observations for which this inequality holds.  We verified that condition \eqref{eq-supp-falsification1} holds for all $i \in \{1,\ldots,9240\}$ and $s \in \{1,\ldots,200\}$, indicating that no observation provides significant evidence against the IV assumptions. 

From these two approaches, we conclude that the IV assumptions are not refuted in this application.

\subsubsection{Inference of the Quantile Treatment Effect on the Treated} 

We applied the inference procedure outlined in Section \ref{sec-supp-QTT}. Specifically, we considered $q \in \{0.25,0.50,0.75\}$ to study the treatment effect on the treated at the first, second, and third quartiles. For $c_1$, we selected $c_1 \in \{0.001,0.005,0.010,0.025\}$. For each cross-fitting index $s \in \{1,\ldots,200\}$, we obtained 95\% confidence intervals for $\tau_{q}^*$, and we obtained the median of these 200  confidence intervals. 

Table \ref{tab-supp-QTT} summarizes the results. All three QTT estimates are significant at the 5\% level, indicating that participation in education or vocational training increased working time across these quantiles. Combined with the ATT analysis, we conclude that participation in education or vocational training resulted in a positive and broadly consistent effect on participants’ working time, both on average and across the distribution.

\begin{table}[!htp]
\renewcommand{\arraystretch}{1.2} \centering
\setlength{\tabcolsep}{10pt} 
\begin{tabular}{|lc|ccc|}
\hline
\multicolumn{2}{|c|}{\multirow{2}{*}{}}              & \multicolumn{3}{c|}{$q$}                                                  \\ \cline{3-5} 
\multicolumn{2}{|c|}{}                               & \multicolumn{1}{c|}{0.25}     & \multicolumn{1}{c|}{0.50}      & 0.75     \\ \hline
\multicolumn{1}{|l|}{\multirow{4}{*}{$c_1$}} & 0.001 & \multicolumn{1}{c|}{$[1,16]$} & \multicolumn{1}{c|}{$[10,40]$} & $[1,39]$ \\ \cline{2-5} 
\multicolumn{1}{|l|}{}                       & 0.005 & \multicolumn{1}{c|}{$[2,16]$} & \multicolumn{1}{c|}{$[11,40]$} & $[1,39]$ \\ \cline{2-5} 
\multicolumn{1}{|l|}{}                       & 0.010  & \multicolumn{1}{c|}{$[2,16]$} & \multicolumn{1}{c|}{$[11,40]$} & $[1,39]$ \\ \cline{2-5} 
\multicolumn{1}{|l|}{}                       & 0.025 & \multicolumn{1}{c|}{$[1,15]$} & \multicolumn{1}{c|}{$[11,39]$} & $[1,40]$ \\ \hline
\end{tabular}
\caption{95\% Confidence Intervals for the QTTs with $q \in \{0.25,0.50,0.75\}$.}
\label{tab-supp-QTT}
\end{table}

\newpage 

\section{Useful Lemmas} \label{sec-supp-2}

In this Section, we introduce Lemmas that will be used in Section \ref{sec-supp-proof} to present the proofs of the Theorems in the main paper and this supplementary material. We consider general $Z$ with support $\suppZ \subset \R$, unless the binary IV condition is specifically required. We also establish some notational conventions that will be used throughout. For nonnegative quantities $(Q_0,Q_1,\ldots,Q_J)$, we write $Q_0 \lesssim \sum_{j=1}^{J} Q_j$ to indicate that $Q_0 \leq \sum_{j=1}^{J} C_j Q_j$ for some nonnegative constants $C_1,\ldots,C_J$. Let $Q_1 \asymp Q_2$ indicate $Q_1 \lesssim Q_2$ and $Q_2 \lesssim Q_1$. Finally, we let $\gamma(Z,\bX) = f(A=1 \cond Z,\bX)/f(A=0 \cond Z, \bX)$ denote the treatment odds conditional on $(Z,\bX)$.

\subsection{Relationships between Nuisance Functions}

\begin{lemma} \label{lemma-proof-useful}
    Suppose that Assumptions \HL{A1}-\HL{A2} and \HL{IV1}-\HL{IV4}  hold for $(\potY{0},A,Z,\bX)$, with corresponding parameters $\theta = \{ \fz,\fa,\fy \}$. Then, the following relationships hold for nuisance functions:
    \begin{itemize}
        \item[(i)] 
        \begin{align*}  
        &
        \frac{ \ff{}(\potY{0}=y,A=1,Z,\bX) }
    { \ff{}(Y=y,A=0,Z,\bX) }
    =
    \alpha(y,\bX) \beta(Z,\bX)
    \ ,     
    \numeq \label{eq-proof-basis alpha beta}
    \\
    &
        \frac{\ff{}(\potY{0}=y \cond A=1,Z,\bX)}{\ff{}(Y=y \cond A=0,Z,\bX)} 
        =
        \frac{\alpha(y,\bX)
    \beta(Z,\bX)}{\gamma(Z,\bX)}  \ . 
    \numeq \label{eq-proof-basis gamma}
        \end{align*}

        \item[(ii)] 
        \begin{align*}           \gamma(Z,\bX) 
    = 
    \EXP \{ \alpha(Y,\bX) \cond A=0,Z,\bX \con \fy \}
    \beta(Z,\bX)  \ .  
    \numeq \label{eq-proof-gamma}
    \end{align*}

    \item[(iii)] 
    \begin{align*}
    &
    \ff{}(\potY{0}=y \cond \bX)
    \\
    &
    =
    \ff{}(\potY{0}=y \cond Z, \bX)
    \\
    &
    =
    \big\{ 1 +  \alpha(y,\bX) \beta(Z,\bX) \big\}
    \ff{} (Y = y \cond A=0, Z, \bX)
    \ff{}(A=0 \cond Z,\bX)
    \\
    &
    =
    \bigg[ 1 + \frac{ \alpha(y,\bX) \gamma(Z,\bX) }{\EXP \big\{ \alpha(Y,X) \cond A=0,Z,\bX \con \fy \big\} } \bigg]
    \frac{ \ff{} (Y=y \cond A=0, Z, \bX) }{ 1+ \gamma(Z,\bX) } \ .
    \numeq 
    \label{eq-proof-Y0Z-gamma}
\end{align*}

\item[(iv)] 
\begin{align*}
    & \ff{}(Y=y,A=0\cond Z, \bX) 
    =
    \frac{\ff{} (\potY{0}=y \cond \bX)
    }{ 1 + \alpha(y,\bX) \beta(Z,\bX) } 
    \ ,
    \numeq 
    \label{eq-proof-Y0AZ}
    \\
    & \ff{}(\potY{0}=y,A=1\cond Z, \bX) 
    =
    \frac{ \alpha(y,\bX) \beta(Z,\bX) \ff{} (\potY{0}=y \cond \bX)
    }{ 1 + \alpha(y,\bX) \beta(Z,\bX) }
    \ .
\end{align*}

    \end{itemize}

\begin{proof} \quad 
    \begin{itemize}
        \item[(i)] These results follow directly from the definitions of $(\alpha,\beta,\gamma)$.

        \item[(ii)] From \eqref{eq-proof-basis gamma}, one can multiply both sides by $\ff{}(Y=y \cond A=0,Z,\bX)$ and  then integrate with respect to $y$.

        \item[(iii)] The first line holds from \HL{IV2}: $\potY{0} \indep Z \cond \bX$. The second line holds from
        \begin{align*}
    &
    \ff{}(\potY{0}=y \cond Z, \bX)
    \\
    & 
    =
    \ff{}(\potY{0}=y, A=1\cond \bX)
    +
    \ff{}(\potY{0}=y, A=0\cond \bX)    
    \\
    &
    \stackrel{\eqref{eq-proof-basis alpha beta}}{=}
    \big\{ 1 + \alpha(y,\bX) \beta(z,\bX) \big\} 
    \ff{}(Y=y, A=0\cond Z=z, \bX) \ .
\end{align*}
The last line holds from \eqref{eq-proof-gamma} and $f(A=0 \cond Z,\bX) = \{1+\gamma(Z,\bX) \}^{-1}$.

\item[(iv)] The result directly follows from \eqref{eq-proof-basis alpha beta} and \HL{IV2}: $\potY{0} \indep Z \cond \bX$.

    \end{itemize}
\end{proof}

\end{lemma}

\subsection{Monotonicity Relationships of Nuisance Functions with Respect to $Z$}
 
\begin{lemma} \label{lemma-beta and gamma}
Suppose that Assumptions \HL{A1}-\HL{A2} and \HL{IV1}-\HL{IV4}  hold for $(\potY{0},A,Z,\bX)$, with corresponding parameters $\theta = \{ \fz,\fa,\fy \}$. Then, for $(z_0,z_1) \in \suppZ$, we have
\begin{align*}
& \beta(z_1 , \bX) > \beta(z_0, \bX)
\\
\Leftrightarrow 
\quad 
&
 \ff{}(Y=y,A=0 \cond Z=z_{0},\bX) >  \ff{}(Y=y,A=0 \cond  Z=z_{1},\bX) \ , \quad \forall y \in \suppYX
 \numeq \label{eq-proof-fya monotone}
\\
\Leftrightarrow 
\quad 
&
f(A=1 \cond Z=z_1, \bX) > f(A=1 \cond Z=z_0, \bX)
\\
\Leftrightarrow 
\quad 
&
\gamma(z_1, \bX) > \gamma(z_0, \bX) \ .    
\end{align*}

\end{lemma}

    \begin{proof}
    From \eqref{eq-proof-Y0AZ}, we find 
\begin{align*}
     &
     \frac{ \ff{}(Y=y,A=0 \cond Z=z_{1},\bX) }{ \ff{}(Y=y,A=0 \cond  Z=z_{0},\bX) }
     \stackrel{\eqref{eq-proof-Y0AZ}}{=}
     \frac{ 1 + \alpha(y,\bX) \beta(z_0,\bX) }
     { 1 + \alpha(y,\bX) \beta(z_1,\bX) }  \ .
\end{align*}
Therefore, 
\begin{align*}
    &
    \beta(z_1,\bX) > \beta(z_0,\bX)
    \\
     \Leftrightarrow
     \quad 
     &
     \ff{}(Y=y,A=0 \cond Z=z_{1},\bX) <  \ff{}(Y=y,A=0 \cond  Z=z_{0},\bX) \ , \quad \forall y \in \suppYX \ ,
     \numeq \label{eq-proof-implication 1}
     \\
     \Rightarrow
     \quad &
     \ff{}(A=0 \cond Z=z_{1},\bX) <  \ff{}(A=0 \cond  Z=z_{0},\bX)
     \numeq \label{eq-proof-implication 2}
     \\
     \Leftrightarrow
     \quad &
     \ff{}(A=1 \cond Z=z_{1},\bX) >  \ff{}(A=0 \cond  Z=z_{0},\bX)
     \\
     \Leftrightarrow
     \quad &
     \gamma(z_1,\bX) > \gamma(z_0,\bX) \ .
\end{align*} 
It remains to show that \eqref{eq-proof-implication 2} implies \eqref{eq-proof-implication 1}, which is established as follows:
\begin{align*}
    & \ff{}(A=0 \cond Z=z_{1},\bX) <  \ff{}(A=0 \cond  Z=z_{0},\bX)
    \\
    \Rightarrow \quad 
    & \ff{}(Y=y,A=0 \cond Z=z_{1},\bX) <  \ff{}(Y=y, A=0 \cond  Z=z_{0},\bX) \quad \text{ for some $y$}
    \\
    \stackrel{\eqref{eq-proof-Y0AZ}}{\Rightarrow} \quad 
    &  \frac{ 1 + \alpha(y,\bX) \beta(z_0,\bX) }
     { 1 + \alpha(y,\bX) \beta(z_1,\bX) } > 1  \quad \text{ for some $y \in \suppYX$}
     \\
     \Rightarrow \quad 
     &
     \beta(z_1,\bX) > \beta(z_0,\bX) \ .
\end{align*}

This completes the proof.

    \end{proof}

\subsection{Boundedness}

\begin{lemma} \label{lemma-boundedness}
Suppose that Assumptions \HL{A1}-\HL{A2} and \HL{IV1}-\HL{IV4}  hold for $(\potY{0},A,Z,\bX)$, with corresponding parameters $\theta = \{ \fz,\fa,\fy \}$. Further suppose that $\theta$ satisfies Assumption \HL{A3-General $Z$}. Then, there exist constants $0<c_\alpha \leq C_{\alpha} < \infty$ and $0<c_\beta \leq C_{\beta} < \infty$ so that the following results hold:
\begin{align*}
    &
    \alpha(y,\bX \con \theta) \in [c_{\alpha},C_{\alpha}]
    \ , 
    &&
    \beta(z,\bX \con \theta) \in [c_{\beta},C_{\beta}]
    \ , 
    &&
    \forall y \in \suppYX \ , \quad 
    \forall z \in \suppZ \ , \quad 
    \forall  \bX \in \suppX \ .
\end{align*}
Moreover, from \eqref{eq-Sun parameterization} of the main paper, \eqref{eq-proof-Y0Z-gamma}, and \eqref{eq-proof-Y0AZ}, there exist constants $0< c_O \leq C_{O} < \infty$ such that
\begin{align*}
    f(\potY{0}=y,A=a,Z=z \cond \bX) 
    \in [c_O, C_O] \ , 
    &&
    \forall y \in \suppYX \ , \
    \forall a \in \{0,1\} \ , \
    \forall z \in \suppZ \ , \
    \forall  \bX \in \suppX \ .
\end{align*} 

\end{lemma}

\begin{proof}
    We omit the dependence on $\theta$ in the notation for brevity.
    Let $z_0 = \arginf_{z \in \suppZ} f(A=1 \cond Z=z,\bX)$. From \eqref{eq-alpha g} of the main paper, we have \begin{align*}
    \alpha(y,\bX)
    &
    =
    \frac{ g^\star(y,\bX) }{ g^\star (y_R,\bX) }
    \frac{ \ff{}(Y=y_R \cond A=0,Z=z_0,\bX) }{\ff{}(Y=y \cond A=0,Z=z_0,\bX)} 
    \\
    &
    \in 
    \bigg[ 
        \frac{c_g c_Y }{C_g C_Y}
        ,
        \frac{C_g C_Y }{c_g c_Y}
    \bigg]
    \equiv 
    \big[ c_{\alpha}, C_{\alpha}     \big]
     \ .
\end{align*}
This further implies $ \EXP \{ \alpha(Y,\bX) \cond A=0,Z=z, \bX \con \fy \} \in [ c_{\alpha} , C_{\alpha} ]$.     Therefore, from \eqref{eq-beta g} of the main paper, we have
    \begin{align*}
        \beta(z,\bX )
    & =
    \frac{f(A=1 \cond Z=z,\bX) / f(A=0 \cond Z=z,\bX)}{ \EXP \{ \alpha(Y,\bX)  \cond A=0,Z=z,\bX \con \fy  \} }
    \\
    &
    \in 
    \bigg[ \frac{c_A}{(1-c_A) C_{\alpha}}\ , 
    \frac{1-c_A}{c_A c_{\alpha}}
    \bigg]
    \equiv 
    [c_{\beta},C_{\beta} ] \ . 
    \end{align*} 
    
    This concludes the proof.
\end{proof}

\subsection{Uniqueness of the Solution to the Integral Equation for Binary $Z$}

\begin{lemma} \label{lemma-unique integral equation solution}
    Suppose that Assumptions \HL{A1}-\HL{A2} and \HL{IV1}-\HL{IV4} holds. Further suppose that $\suppZ = \{0,1\}$. For a function $v(\potY{0},Z,\bX)$, let $\widetilde{v} = v - v^\dagger$ where $v^\dagger \equiv \EXP\{ v \cond \potY{0},\bX \} + \EXP (v \cond Z , \bX ) - \EXP(v \cond \bX)$. Then, $\widetilde{v}(\potY{0},Z,\bX) =  0$ is the unique function satisfying
    \begin{align*}
    &
    \EXP \{ \widetilde{v}(\potY{0},Z,\bX) - \mu^* (Z,\bX \con \widetilde{v}) \cond A=1, \potY{0},\bX\} = 0 \ .
    \end{align*}

\end{lemma}

\begin{proof}
    We first show that $\widetilde{v}$ admits a simplified structure when $Z$ is binary. Note that any function $v(\potY{0},Z,\bX)$ can be written as
\begin{align*}
v(\potY{0},Z,\bX)
=
v_1(\potY{0},\bX)Z
+
v_2(\potY{0},\bX)(1-Z) \ ,
\end{align*}
indicating that
\begin{align*} 
&
v^\dagger(\potY{0},Z,\bX)
\\
&
= 
v_1(\potY{0},\bX)\EXP ( Z \cond \bX)
+
v_0(\potY{0},\bX) \{ 1- \EXP ( Z \cond \bX) \}
\\
&
\quad 
+
Z \EXP \big\{ \nu_1 (\potY{0},\bX) \cond \bX \big\}
+
(1-Z) \EXP \big\{ \nu_0 (\potY{0},\bX) \cond \bX \big\}
\\
&
\quad 
-
\EXP(Z \cond \bX) \EXP \big\{ v_1 (\potY{0},\bX) \cond \bX \big\}
+
\{ 1-\EXP(Z \cond \bX)\} \EXP \big\{ v_0 (\potY{0},\bX) \cond \bX \big\} \ .
\end{align*}
Therefore, $\widetilde{v} = v - v^\dagger$ is given by \begin{align*}
    \widetilde{v} (\potY{0},Z,\bX) 
    &
    = 
    \big\{ v_1 - v_0 - \EXP \big( v_1 - v_0 \cond \bX \big) \big\}
    \big\{
    Z - \EXP \big( Z \cond \bX \big) \big\}
    \\
    &
    = 
    \big[ L(\potY{0},\bX) - \EXP \big\{ L(\potY{0},\bX) \cond \bX \big\} \big]
    \big\{ Z - \EXP \big( Z \cond \bX \big) \big\} \ , 
    \numeq \label{eq-proof-vtilde L representation}
\end{align*}
where $L(\potY{0},\bX) \equiv  v_1(\potY{0},\bX) -v_0(\potY{0},\bX)$. 

This structure can be generalized to a categorical  $Z$. Without loss of generality, let $\suppZ = \{0,1,\ldots,m\}$. Then, any function $v(\potY{0},Z,\bX)$ can be written as
\begin{align*}
v(\potY{0},Z,\bX)
=
\sum_{z=0}^{m}
v_z(\potY{0},\bX) \ind (Z=z) \ ,
\end{align*}
and $v^\dagger$ is expressed as
\begin{align*}
&
	v^\dagger(\potY{0},Z=z,\bX) 
	\\
	&
	= \bigg \{ v_z - \sum_{c \neq z} v_c - \EXP \bigg( v_z - \sum_{c \neq z} v_c \, \bigg| \, \bX \bigg) \bigg\} 
	\big\{ \ind(Z=z) - \Pr(Z=z \cond \bX) \big\} 
	\\
	&
	=
	\big[ L_{z}(\potY{0},\bX) - \EXP \big\{ L_z(\potY{0},\bX) \cond \bX \big\} \big]
	\big\{ \ind(Z=z) - \Pr(Z=z \cond \bX) \big\} \ .
    \numeq \label{eq-proof-vtilde L representation poly}
\end{align*}
Note that we have a constraint $\EXP \{ v^\dagger \cond \potY{0},\bX \} = 0$. Consequently, $L_z$ is determined once $\{ L_c \cond c \in \suppZ \setminus \{z\} \}$ are fixed.

Since $Z$ is binary, we find that
\begin{align*}
    &
    \widetilde{v}(\potY{0},Z,\bX)
    =
    \big[ L(\potY{0},\bX) - \EXP \big\{ L(\potY{0},\bX) \cond \bX \big\} \big] 
    \big\{ Z - \EXP(Z \cond \bX) \big\}
    \ , 
    \\
    &
    \widetilde{v}(\potY{0},Z,\bX) - \mu^*(Z,\bX \con \widetilde{v})
    =
    \big[ L - 
    \EXP ( L \cond A=1,Z,\bX )
    \big] 
    \big\{ Z - \EXP(Z \cond \bX) \big\} \ . 
\end{align*}
This indicates
\begin{align*}
    &
    \EXP \{ \widetilde{v}(\potY{0},Z,\bX) - \mu^*(Z,\bX \con \widetilde{v}) \cond A=1, \potY{0}=y,X \}
    =
    0
    \\[0.5cm]
    \Leftrightarrow \quad   
    &
    L(y,\bX)
    \bigg[
    \begin{array}{l}
    f^*(Z=0 \cond \bX) 
    f^*(Z=1 \cond A=1, \potY{0}=y, \bX) 
    \\
    -
    f^*(Z=1 \cond \bX) 
    f^*(Z=0 \cond A=1, \potY{0}=y, \bX) 
    \end{array}
    \bigg]
    \\
    &
    =
    \bigg[
    \begin{array}{l}\mu^*(1,\bX \con L)
    f^*(Z=0 \cond \bX) 
    f^*(Z=1 \cond A=1, \potY{0}=y,\bX)
    \\
    -
    \mu^*(0,\bX \con L)
    f^*(Z=1 \cond \bX) 
    f^*(Z=0 \cond A=1, \potY{0}=y,\bX)
    \end{array}
    \bigg] 
    \\[0.5cm]
    \stackrel{\eqref{eq-proof-Y0AZ}}{\Leftrightarrow}
    \quad   
    &
    \frac{
    L(y,\bX) 
    \{ \beta^*(1,\bX) - \beta^*(0,\bX) \}
    f^*(Z=1 \cond \bX)
    f^*(Z=0 \cond \bX) 
    }{
    \beta^*(1,\bX)
    f^*(Z=1 \cond \bX) 
    +
    \beta^*(0,\bX)
    f^*(Z=0 \cond \bX)
    + 
    \alpha^*(y,\bX) \beta^*(0,\bX)
    \beta^*(1,\bX)
    }
    \\
    &
    =
    \frac{
    \mu^*(1,\bX \con L)
    \big\{ 1 + \alpha^*(y,\bX) \beta^*(0,\bX) \big\} \beta^*(1,\bX)
    f^*(Z=1 \cond X)
    f^*(Z=0 \cond X)
    }{
    \beta^*(1,\bX)
    f^*(Z=1 \cond \bX)
    +
    \beta^*(0,\bX)
    f^*(Z=0 \cond \bX)
    +
    \alpha^*(y,\bX)
    \beta^*(0,\bX) 
    \beta^*(1,\bX) 
    }
    \\
    &
    \hspace*{0.5cm}
    -
    \frac{
    \mu^*(0,\bX \con L)
    \big\{ 1 + \alpha^*(y,\bX) \beta^*(1,\bX) \big\} \beta^*(0,\bX)
    f^*(Z=1 \cond \bX)
    f^*(Z=0 \cond \bX)
    }{
    \beta^*(1,\bX)
    f^*(Z=1 \cond \bX)
    +
    \beta^*(0,\bX)
    f^*(Z=0 \cond \bX)
    +
    \alpha^*(y,\bX)
    \beta^*(0,\bX) 
    \beta^*(1,\bX) 
    } 
    \\[0.5cm]
    \Leftrightarrow \quad   
    &
    L(y,\bX)  
    \{ \beta^*(1,\bX) - \beta^*(0,\bX) \} 
    \\
    &
    =
    \alpha^*(y,\bX) \beta^*(0,\bX) \beta^*(1,\bX)  
    \big\{ \mu^*(1,\bX \con L) - \mu^*(0,\bX \con L) \big\}
    \\
    &
    \hspace*{0.5cm}
    +
    \mu^*(1,\bX \con L) \beta^*(1,\bX)
    -
    \mu^*(0,\bX \con L) \beta^*(0,\bX) 
    \\[0.5cm] 
    \Leftrightarrow \quad   
    &
    L(y,\bX)  \\
    &
    =
    \mu^*(0,\bX \con L)
    +
    \frac{  \big\{
    \alpha^*(y,\bX) \beta^*(0,\bX) + 1
    \big\}
    \big\{ \mu^*(1,\bX \con L) - \mu^*(0,\bX \con L) \big\} \beta^*(1,\bX)  }{ \beta^*(1,\bX) - \beta^*(0,\bX) } \ .
\end{align*}

Suppose that $\mu^*(1,\bX \con L) \neq \mu^*(0,\bX \con L)$ at $\bX \in \suppX$. From Lemma \ref{lemma-beta and gamma}, $\beta^*(1,\bX) - \beta^*(0,\bX)$ is always positive (or negative) and $\alpha^*(y,\bX) > 0$ for all $y \in \suppYX$. These two results imply that 
\begin{align*}
    &
    L(y,\bX) - \mu^*(0,\bX \con L)
    \\
    &
    =
    \underbrace{ \big\{
    \alpha^*(y,\bX) \beta^*(0,\bX) + 1
    \big\}
     \beta^*(1,\bX) 
    }_{> 0}
    \underbrace{
    \frac{  \mu^*(1,\bX \con L) - \mu^*(0,\bX \con L)  }{ \beta^*(1,\bX) - \beta^*(0,\bX) }
    }_{\neq 0}
    \ ,
\end{align*}
which in turn shows that $L(y,\bX) > \mu^*(0,\bX \con L) $ or $L(y,\bX) < \mu^*(0,\bX \con L) $ for all $y \in \suppYX$. Consequently, we obtain
\begin{align*}
    & \EXP \{ L(\potY{0} , \bX) \cond A=1, Z=0, X \} > \mu^*(0,\bX \con L) 
    \quad \text{ or}
    \\
    & \EXP \{ L(\potY{0} , \bX) \cond A=1, Z=0, X \} < \mu^*(0,\bX \con L) \ . 
\end{align*}
However, this contradicts the definition of $\mu^*$, namely $ \mu^*(0,\bX \con L) \equiv \EXP \{ L(\potY{0} , \bX) \cond A=1, Z=0, \bX \}$. Therefore, we must have $\mu^*(1,\bX \con L)= \mu^*(0,\bX \con L)$, which in turn implies $L(y,\bX) = \mu^* (0,\bX \con L)$. In other words, $L$ does not depend on $\potY{0}$. Therefore, we establish $\widetilde{v}(\potY{0},Z,\bX) = 0$ as follows:
\begin{align*}
    \widetilde{v} (\potY{0},Z,\bX) 
    &
    \stackrel{\eqref{eq-proof-vtilde L representation}}{=}
    \big[ L(\potY{0},\bX) - \EXP \big\{ L(\potY{0},\bX) \cond \bX \big\} \big]
    \big\{ Z - \EXP \big( Z \cond \bX \big) \big\} 
    \\
    &
    = 
    \big[ L(\bX) - \EXP \big\{ L(\bX) \cond \bX \big\} \big]
    \big\{ Z - \EXP \big( Z \cond \bX \big) \big\} 
    \\
    &
    = 
    0 \ .
\end{align*}

This concludes the proof.
\end{proof}

\subsection{Rates of Convergence for Estimation Errors}
 
Before stating the Lemma, we introduce some notation. Recall that $\pi(\potY{0},Z,\bX) = \alpha(\potY{0},\bX) \beta(Z,\bX)$. For a generic function $g$, define the $L^2(\suppYX)$-norm of $g(\potY{0},Z,\bX)$ by
\begin{align*}
    \|g (\potY{0},Z,\bX) \|_{Y,2} = \bigg\{ \int_{\suppYX} g^2 (y,Z,\bX) \, dy \bigg\}^{1/2} \ .
\end{align*} 
For an estimated function $\widehat{g}\LSS$ from the estimation fold $\II_k$ $(k \in \{1,\ldots,K\})$ and a fixed function $g$, let 
\begin{align*}
    &
    \EXPk \{ \widehat{g}\LSS(V) \cond W \}
    =
    \int \widehat{g}\LSS(v) \ftrue{}(V=v \cond W) \, dv \ ,
    \\
    &
    \EXPhat\{ g(V) \cond W \}
    =
    \int g(v) \widehat{f}\LSS(V=v \cond W) \, dv \ .
\end{align*}

\begin{lemma} \label{lemma-supp-convergence} 
Suppose that the assumptions in Theorem \ref{thm-AN} hold.  Then, we have the following results for $k \in \{1,\ldots,K\}$:
    \begin{itemize}[leftmargin=0.75cm]
    \item[\HT{i}] For $z \in \{0,1\}$, we have
    \begin{align*}
        &
        \big\| \widehat{\gamma}\LSS(Z,\bX)  - \gamma^*(Z,\bX)
        \big\|_{P,2}
        \\
        &
        \asymp
        \big\| \widehat{f}\LSS(A=1 \cond Z,\bX) - f^*(A=1 \cond Z,\bX) \big\|_{P,2} \ .
        \numeq 
        \label{eq-proof-gamma to fa}
    \end{align*}

    \item[\HT{ii}] 
For $z \in \{0,1\}$, we have
\begin{align*}
    &
    \big\|
    \pihat(\potY{0},z,\bX) - \pi^*(\potY{0},z,\bX) 
    \big\|_{Y,2}
    \\
    &
    \lesssim
    \left[ 
    \begin{array}{l}            
    \| \widehat{f}\LSS(\potY{0} \cond A=0,Z=1,\bX)
    -
    f^*(\potY{0} \cond A=0,Z=1,\bX)\|_{Y,2}
    \\
    +
    \| \widehat{f}\LSS(\potY{0} \cond A=0,Z=0,\bX)
    -
    f^*(\potY{0} \cond A=0,Z=0,\bX) \|_{Y,2}
    \\
             + \| \widehat{f}\LSS(A=1 \cond Z=1,\bX) - f^*(A=1 \cond Z=1,\bX) \|
             \\
             + \| \widehat{f}\LSS(A=1 \cond Z=0,\bX) - f^*(A=1 \cond Z=0,\bX) \|
        \end{array}
    \right] \ .
    \numeq 
    \label{eq-proof-alpha beta wrt y}
\end{align*}

\item[\HT{iii}] 
\begin{align*}
    &
    \big\|
    \pihat(\potY{0},Z,\bX) - \pi^*(\potY{0},Z,\bX) 
    \big\|_{P,2}
    \\
    &
    \lesssim 
    \left[
        \begin{array}{l}
        \big\| \widehat{f}\LSS(A=1 \cond Z,\bX) - f^*(A=1 \cond Z,\bX)
        \big\|_{P,2}
        \\
        +
        \big\| \widehat{f}\LSS(\potY{0} \cond A=0,Z,\bX) - f^*(\potY{0} \cond A=0,Z,\bX)
        \big\|_{P,2}
        \end{array}
    \right] \ .
    \numeq 
    \label{eq-proof-alpha beta L2P}
\end{align*}

\item[\HT{iv}] 
    Let $\norm^*(Z,\bX)$ and $\normhat(Z,\bX)$ denote the normalizing constant and its estimator, respectively, which are defined as
\begin{align*}
    \norm^* (Z,\bX) 
    &
    = \EXP \big\{ \pi^*(Y,Z,\bX) \cond 
    A=0,Z,\bX \big\}  \ ,
    \\
    \normhat (Z,\bX) 
    &
    = \EXPhat \big\{ \pihat (Y,Z,\bX) \cond 
    A=0,Z,\bX \big\}     \ .
\end{align*}
Then, we have
\begin{align*}
    &
    \big\|
    \normhat (Z,\bX) 
    -
    \norm^* (Z,\bX)  
    \big\|_{P,2}
    \\
    &
    \lesssim 
    \left[
        \begin{array}{l}
        \big\| \widehat{f}\LSS(A=1 \cond Z,\bX) - f^*(A=1 \cond Z,\bX)
        \big\|_{P,2}
        \\
        +
        \big\| \widehat{f}\LSS(\potY{0} \cond A=0,Z,\bX) - f^*(\potY{0} \cond A=0,Z,\bX)
        \big\|_{P,2}
        \end{array}
    \right] \ .
    \numeq 
    \label{eq-proof-normalD L2P}
\end{align*}

\item[\HT{v}] Let $\alpha_{\norm}^*$, $\beta_{\norm}$, $\alphahat_{\norm}$, and $\betahat_{\norm}$ be:
\begin{align*}
    &
    \alpha_{\norm}^*(\potY{0},\bX)
    =
    \frac{\pi^*(\potY{0},Z,\bX)}{\norm^* (Z,\bX)}
    \ , 
    &&
    \beta_{\norm}^*(Z,\bX)
    =
    \frac{\pi^*(\potY{0},Z,\bX)}{\alpha_{\norm}^*(\potY{0},\bX)}
    \ ,
    \\
    &
    \alphahat_{\norm}(\potY{0},\bX)
    =
    \frac{\pihat(\potY{0},Z,\bX)}{\normhat (Z,\bX)}
    \ , 
    &&
    \betahat_{\norm}(Z,\bX)
    =
    \frac{\pihat(\potY{0},Z,\bX)}{\alphahat_{\norm}(\potY{0},\bX)} \ .
\end{align*}
Then, we have
\begin{align*}
    &
    \big\|
    \alphahat_{\norm}(\potY{0},\bX)
    -
    \alpha_{\norm}^*(\potY{0},\bX)
    \big\|_{P,2}
    \\
    &
    \lesssim 
    \left[
        \begin{array}{l}
        \big\| \widehat{f}\LSS(A=1 \cond Z,\bX) - f^*(A=1 \cond Z,\bX)
        \big\|_{P,2}
        \\
        +
        \big\| \widehat{f}\LSS(\potY{0} \cond A=0,Z,\bX) - f^*(\potY{0} \cond A=0,Z,\bX)
        \big\|_{P,2}
        \end{array}
    \right] \ ,
    \numeq 
    \label{eq-proof-alpha L2P}
    \\
    &
    \big\|
    \betahat_{\norm}(Z,\bX)
    -
    \beta_{\norm}^*(Z,\bX)
    \big\|_{P,2}
    \\
    &
    \lesssim 
    \left[
        \begin{array}{l}
        \big\| \widehat{f}\LSS(A=1 \cond Z,\bX) - f^*(A=1 \cond Z,\bX)
        \big\|_{P,2}
        \\
        +
        \big\| \widehat{f}\LSS(\potY{0} \cond A=0,Z,\bX) - f^*(\potY{0} \cond A=0,Z,\bX)
        \big\|_{P,2}
        \end{array}
    \right] \ .
    \numeq 
    \label{eq-proof-beta L2P}
\end{align*}

\item[\HT{vi}] For a uniformly bounded function $v(\potY{0},Z,\bX)$, we have
\begin{align*}
    &
    \big\|
    \widehat{\mu} \LSS (Z,\bX \con v) 
    -
    \mu^* (Z,\bX \con v)
    \big\|_{P,2}
    \\
    &
    \lesssim 
    \left[
        \begin{array}{l}
        \big\| \widehat{f}\LSS(A=1 \cond Z,\bX) - f^*(A=1 \cond Z,\bX)
        \big\|_{P,2}
        \\
        +
        \big\| \widehat{f}\LSS(\potY{0} \cond A=0,Z,\bX) - f^*(\potY{0} \cond A=0,Z,\bX)
        \big\|_{P,2}
        \end{array}
    \right] \ .
    \numeq 
    \label{eq-proof-mu L2P}
\end{align*}

\item[\HT{vii}] For $a \in \{0,1\}$, we have 
\begin{align*}
    &
    \big\| \widehat{f}\LSS(Z=1 \cond A=a,\potY{0},\bX) - f^*(Z=1 \cond A=a,\potY{0},\bX)\big\|_{P,2} 
    \\
    &
    \lesssim 
    \left[
        \begin{array}{l}
        \big\| \widehat{f}\LSS(A=1 \cond Z,\bX) - f^*(A=1 \cond Z,\bX)
        \big\|_{P,2}
        \\
        +
        \big\| \widehat{f}\LSS(\potY{0} \cond A=0,Z,\bX) - f^*(\potY{0} \cond A=0,Z,\bX)
        \big\|_{P,2}
        \end{array}
    \right] \ .
    \numeq 
    \label{eq-proof-z given y rate}
\end{align*}

\item[\HT{viii}] 
\begin{align*}
    &
    \big\| f^*(\potY{0} \cond \bX) - \widehat{f}\LSS(\potY{0} \cond \bX) \big\|_{P,2}
    \\
    &
    \lesssim 
    \left[
        \begin{array}{l}
        \big\| \widehat{f}\LSS(A=1 \cond Z,\bX) - f^*(A=1 \cond Z,\bX)
        \big\|_{P,2}
        \\
        +
        \big\| \widehat{f}\LSS(\potY{0} \cond A=0,Z,\bX) - f^*(\potY{0} \cond A=0,Z,\bX)
        \big\|_{P,2}
        \end{array}
    \right] \ .
    \numeq 
    \label{eq-proof-y0 rate}
\end{align*}

\item[\HT{ix}]
\begin{align*}
    &
    \big\|
    \what (\potY{0},Z,\bX)
    -
    \omega^*(\potY{0},Z,\bX)
    \big\|_{P,2}
    \\
    &
    \lesssim 
    \left[ 
     \begin{array}{l}  
     \big\| \widehat{f}\LSS(Z=1 \cond \bX)
     -
     f^*(Z=1 \cond \bX)
     \big\|_{P,2}
     \\
     +
     \big\| \widehat{f}\LSS(A=1 \cond Z,\bX)
     -
     f^*(A=1 \cond Z,\bX)
     \big\|_{P,2}
     \\
     +
     \big\| 
     \widehat{f}\LSS(\potY{0} \cond A=0, Z,\bX)
     -
     f^* (\potY{0} \cond A=0, Z,\bX)
     \big\|_{P,2}
     \end{array}
     \right]  \ .
     \numeq 
     \label{eq-proof-omega hat}
\end{align*}
and
\begin{align*}
    &
    \big\| 
    \widehat{\mu}\LSS (Z,\bX \con \what) - 
    \mu^*(Z,\bX \con \omega^*)
    \big\|_{P,2} 
    \\
    &
    \lesssim
    \left[ 
     \begin{array}{l}  
     \big\| \widehat{f}\LSS(Z=1 \cond \bX)
     -
     f^*(Z=1 \cond \bX)
     \big\|_{P,2}
     \\
     +
     \big\| \widehat{f}\LSS(A=1 \cond Z,\bX)
     -
     f^*(A=1 \cond Z,\bX)
     \big\|_{P,2}
     \\
     +
     \big\| 
     \widehat{f}\LSS(\potY{0} \cond A=0, Z,\bX)
     -
     f^* (\potY{0} \cond A=0, Z,\bX)
     \big\|_{P,2}
     \end{array}
     \right] 
        \ .
    \numeq 
    \label{eq-proof-omega mu hat}
\end{align*}

\item[\HT{x}] For a uniformly bounded function $\mathcal{G}^{(0)} \equiv \mathcal{G}(\potY{0},\bX)$, let $\what$ be a function satisfying Theorems \ref{thm-IF} and \ref{thm-IF binary Z} with the estimated density, i.e., 
\begin{align*}
    \widehat{\EXP}\LSS
    \bigg[ 
    \begin{array}{l}    
    \big\{ \mathcal{G}^{(0)} - \muhat(Z,\bX \con \mathcal{G}) \big\} \\
    - \big\{ \what(\potY{0},Z,\bX) - \muhat (Z,\bX \con \what) \big\}
    \end{array}
    \, \bigg| \, A=1, \potY{0}, \bX \bigg] = 0
    \ .
\end{align*}
Then, we have
\begin{align*}
    &
    \bigg\|
    \EXPk \bigg[ 
        \begin{array}{l}          
     \big\{ \mathcal{G}^{(0)} - \what(\potY{0},Z,\bX) \} \\
     -
     \big\{\muhat(Z,\bX \con \mathcal{G})  - \muhat(Z,\bX \con \what) 
     \big\} 
        \end{array}
        \, \bigg| \, A=1,\potY{0},\bX
        \bigg] 
    \bigg\|_{P,2}
    \\
    &
    \lesssim 
    \left[
        \begin{array}{l}
        \big\| \widehat{f}\LSS(A=1 \cond Z,\bX) - f^*(A=1 \cond Z,\bX)
        \big\|_{P,2}
        \\
        +
        \big\| \widehat{f}\LSS(\potY{0} \cond A=0,Z,\bX) - f^*(\potY{0} \cond A=0,Z,\bX)
        \big\|_{P,2}
        \end{array}
    \right] \ ,
    \numeq 
    \label{eq-proof-GH convergence}
\end{align*}
and
\begin{align*}
    &
    \bigg\|
    \EXPk \bigg[ 
        \begin{array}{l}          
     \big\{ \mathcal{G}^{(0)} - \what(\potY{0},Z,\bX) \} \\
     -
     \big\{\muhat(Z,\bX \con \mathcal{G})  - \muhat(Z,\bX \con \what) 
     \big\} 
        \end{array}
        \, \bigg| \, A=1,Z,\bX
        \bigg] 
    \bigg\|_{P,2}
    \\
    &
    \lesssim 
    \left[
        \begin{array}{l}
        \big\| \widehat{f}\LSS(A=1 \cond Z,\bX) - f^*(A=1 \cond Z,\bX)
        \big\|_{P,2}
        \\
        +
        \big\| \widehat{f}\LSS(\potY{0} \cond A=0,Z,\bX) - f^*(\potY{0} \cond A=0,Z,\bX)
        \big\|_{P,2}
        \end{array}
    \right] \ ,
    \numeq 
    \label{eq-proof-GH convergence given Z}
\end{align*}

\item[\HT{xi}] For a uniformly bounded function $\nu$, let $\widetilde{\nu}\LSS$ be 
\begin{align*}    
&
\widetilde{\nu}\LSS(\potY{0},Z,\bX)
\\
&
=
\nu(\potY{0},Z,\bX)
-
\widehat{\EXP}\LSS \big\{ \nu(\potY{0},Z,\bX) \cond Z, \bX \big\}
\\
&
\qquad \qquad 
-
\widehat{\EXP}\LSS \big\{ \nu(\potY{0},Z,\bX) \cond \potY{0}, \bX \big\}
+
\widehat{\EXP}\LSS \big\{ \nu(\potY{0},Z,\bX) \cond \bX \big\} \ .
\end{align*}
Then, we have 
\begin{align*}
    &
    \big\|
    \EXPk \big\{ 
    \what(\potY{0},Z,\bX)  \big\}
    \big\| 
    \numeq 
    \label{eq-proof-omega mean rate} 
    \\
    &
    \lesssim 
    \left[
        \begin{array}{l}
        \big\| \widehat{f}\LSS(A=1 \cond Z,\bX) - f^*(A=1 \cond Z,\bX)
        \big\|_{P,2}
        \\
        +
        \big\| \widehat{f}\LSS(\potY{0} \cond A=0,Z,\bX) - f^*(\potY{0} \cond A=0,Z,\bX)
        \big\|_{P,2}
        \end{array}
    \right]
    \\
    &
    \quad \quad 
    \times
    \big\|  \widehat{f}\LSS(Z=1 \cond \bX) - f^*(Z=1 \cond \bX) \big\|_{P,2} 
    \ .
\end{align*}

\end{itemize}
\end{lemma}

\begin{proof}

For notational brevity, all integrals with respect to $y$ are taken over $\suppYX$. That is, we write $\int g(y,a,z,\bX) \, dy \equiv \int_{\suppYX} g(y,a,z,\bX) \, dy$. 

\vspace*{0.25cm}

\noindent \HL{i}  From straightforward algebra, we find 
    \begin{align*}
    & \gammahat(Z,\bX) - \gamma^*(Z,\bX)
    \\
    & =
    \frac{ \fhat{}\LSS(A=1 \cond Z,\bX) }
    { \fhat{}\LSS(A=0 \cond Z,\bX) }
    -
    \frac{ f^*(A=1 \cond Z,\bX) }
    { f^*(A=0 \cond Z,\bX) }
    \\
    & =
    \frac{ \fhat{}\LSS(A=1 \cond Z,\bX) 
    f^*(A=0 \cond Z,\bX) 
    - f^*(A=1 \cond Z,\bX) 
    \fhat{}\LSS(A=0 \cond Z,\bX) }
    { \fhat{}\LSS(A=0 \cond Z,\bX) 
    f^*(A=0 \cond Z,\bX) } 
    \\ 
    &
    = 
    \frac{  
        \fhat{}\LSS(A=1 \cond Z,\bX) -
        f^*(A=1 \cond Z,\bX) 
        }{\fhat{}\LSS(A=0 \cond Z,\bX) f^*(A=0 \cond Z,\bX)}  \ .
\end{align*}
This implies 
    \begin{align*}
        &
        \big\| \widehat{\gamma}\LSS(Z,\bX)  - \gamma^*(Z,\bX)
        \big\|
        \asymp
        \big\| \widehat{f}\LSS(A=1 \cond Z,\bX) - f^*(A=1 \cond Z,\bX) \big\| \ .
    \end{align*}
Taking the expectation, we obtain 
\begin{align*}
    \big\| \widehat{\gamma}\LSS(Z,\bX)  - \gamma^*(Z,\bX)
        \big\|_{P,2}
        \asymp
        \big\| \widehat{f}\LSS(A=1 \cond Z,\bX) - f^*(A=1 \cond Z,\bX) \big\|_{P,2} \ . 
\end{align*}

\vspace*{0.25cm}

\noindent \HL{ii} We first fix $z_0 \in \suppZ$ and $\bX \in \suppX$ in order to consider the mapping $\Psi$ and the associated fixed point. For notational brevity, we suppress the dependence on $z_0$. Let $\Psi(g \con \vartheta )$ denote the mappings in \eqref{eq-Psi} of the main paper with parameter $\vartheta \equiv \{\fy,\fa\}$:
\begin{align*}
     \Psi(g \con \vartheta )
     =
     \frac{ \{ \mathcal{C}_1(\bX \con \vartheta ) g(y,\bX) + \mathcal{C}_2(y,\bX \con \vartheta ) \} R_D(\bX \con g, \vartheta )  }{R_Y(y, \bX \con \vartheta )} \ ,
\end{align*}
where
\begin{align*}
    \mathcal{C}_1(\bX \con \vartheta ) 
    &     
    = \frac{ f(A=1 \cond Z=z_0,\bX) }{ f(A=1 \cond Z=z_1,\bX)} \ ,
    \\
    \mathcal{C}_2(y,\bX \con \vartheta ) 
    & 
    = 
    \frac{ 
    \displaystyle{ 
    \left\{
    \begin{array}{l}
    f(A=0 \cond Z=z_0,\bX)
    f(Y=y \cond A=0,Z=z_0,\bX)
    \\
    -
    f(A=0 \cond Z=z_1,\bX)
    f(Y=y \cond A=0,Z=z_1,\bX)
    \end{array} 
    \right\}
    }
    }{f(A=1 \cond Z=z_1,\bX)} 
    \ ,
    \\
    R_D(\bX \con g, \vartheta )
    & = 
    \int g(t,\bX) R_Y(t,\bX \con \vartheta ) \, dt
    \ , 
    \\
    R_Y(y, \bX \con \vartheta )     
    & =
    \frac{ f(Y=y \cond A=0,Z=z_1,\bX) }{ f(Y=y \cond A=0,Z=z_0,\bX) } \ .
\end{align*}
Let $g^\star(\vartheta )$ be the fixed point of $\Psi( \cdot \con \vartheta )$, i.e., $g(\vartheta)=\Psi( g \con \vartheta )$. Let $\Psi^*$ and $\widehat{\Psi}\LSS$ be the mappings evaluated at the true parameter $\vartheta ^*$ and the estimated parameter $\widehat{\vartheta} \LSS$, respectively. Accordingly, let $g^\star$ and $\widehat{g}\LSSstar$ denote the corresponding fixed points of $\Psi^*$ and $\widehat{\Psi}\LSS$. 

By \eqref{eq-proof-gamma}, \eqref{eq-proof-D}, \eqref{eq-proof-z0}, we have
\begin{align*}
    g^\star(y,\bX) 
    &
    = \frac{\alpha^*(y,\bX) \ff{}^*(Y=y \cond A=0,Z=z_0,\bX) }{ \EXP \big\{ \alpha^*(y,\bX) \cond A=0,Z=z_0,\bX \big\} }
    \\
    &
    = \frac{\alpha^*(y,\bX) \beta^*(z_0,\bX) \ff{}^*(Y=y \cond A=0,Z=z_0,\bX) }{ \gamma^*(z_0,\bX) }  \ ,
\end{align*}
which yields
\begin{align*}    
    \alpha^*(y,\bX) \beta^*(z,\bX)
    &
    =
    \frac{g^\star(y,\bX) \gamma^*(z,\bX)
    }{f^*(Y=y \cond A=0,Z=z,\bX)} \ ,
    \\
    \widehat{\alpha}\LSS(y,\bX) \widehat{\beta}\LSS (z,\bX)
    &
    =
    \frac{\widehat{g}\LSSstar(y,\bX) \widehat{\gamma}\LSS(z,\bX)
    }{\widehat{f}\LSS(Y=y \cond A=0,Z=z,\bX)} \ .
\end{align*}
Therefore, we obtain the following result for $z \in \{0,1\}$:
\begin{align*}
    &
    \big\|
    \widehat{\alpha}\LSS(\potY{0},\bX) \widehat{\beta}\LSS (z,\bX)
    -
    \alpha^*(\potY{0},\bX) 
    \beta^*(z,\bX) 
    \big\|_{Y,2}
    \\
    &
    \lesssim
    \big\| \widehat{g}\LSSstar(\potY{0},\bX) - g^\star(\potY{0},\bX) \big\|
    \\
    &
    \quad 
    +
    \underbrace{
    \left[ 
    \begin{array}{l}            
    \| \widehat{f}\LSS(\potY{0} \cond A=0,Z=1,\bX)
    -
    f^*(\potY{0} \cond A=0,Z=1,\bX)\|_{Y,2}
    \\
    +
    \| \widehat{f}\LSS(\potY{0} \cond A=0,Z=0,\bX)
    -
    f^*(\potY{0} \cond A=0,Z=0,\bX) \|_{Y,2}
    \\
             + \| \widehat{f}\LSS(A=1 \cond Z=1,\bX) - f^*(A=1 \cond Z=1,\bX) \|
             \\
             + \| \widehat{f}\LSS(A=1 \cond Z=0,\bX) - f^*(A=1 \cond Z=0,\bX) \|
        \end{array}
    \right] }_{ \equiv \mathfrak{R}(\widehat{\vartheta}\LSS - \vartheta^*) }
    \ .
    \numeq 
    \label{eq-proof-alpha beta estimation}
\end{align*}
Therefore, to characterize the convergence rate of $\widehat{\alpha}\LSS \widehat{\beta}\LSS$, it is necessary to first determine the convergence rate of $\widehat{g}\LSSstar$.

We first show that $g^\star(\vartheta )$ is a smooth mapping with respect to $\vartheta $. In order to do so, we establish $\Psi$ is Fr\'echet differentiable with respect to $\vartheta $ and $g$. We begin by considering differentiability with respect to $\vartheta $. Let $D_{\vartheta} \Psi(g \con \vartheta )[ \delta  ]$ denote the pathwise derivative of $\Psi(g \con \vartheta )$ with respect to $\vartheta $, for a fixed $g$, in the direction of $\delta  \equiv \{\texttt{f}_Y,\texttt{f}_A\}$, i.e.,  
\begin{align*}
    &
    D_{\vartheta} \Psi(g \con \vartheta )[ \delta  ]
    \\
    & 
    = \frac{\partial}{\partial t} \Psi(g \con \vartheta  + t \delta  ) \bigg|_{t=0} 
    \\
    & 
    = 
    \frac{ D_{\vartheta} \mathcal{C}_1(\bX \con \vartheta )[\delta] g(y,\bX) + D_{\vartheta}  \mathcal{C}_2(y,\bX \con \vartheta ) [\delta]
    }{R_Y(y, \bX \con \vartheta )}    
    R_D(\bX \con g, \vartheta ) 
    \\
    &
    \quad \quad
    +
    \frac{ \{ \mathcal{C}_1(\bX \con \vartheta ) g(y,\bX) + \mathcal{C}_2(y,\bX \con \vartheta ) \} 
    }{R_Y(y, \bX \con \vartheta )}    
    \int g(t,\bX) D_{\vartheta} R_Y(t,\bX \con \vartheta )[\delta] \, dt
    \\
    &
    \quad \quad
    -
    \frac{ \{ \mathcal{C}_1(\bX \con \vartheta ) g(y,\bX) + \mathcal{C}_2(y,\bX \con \vartheta ) \} R_D(\bX \con g, \vartheta ) 
    }{R_Y(y, \bX \con \vartheta )} 
    \frac{ D_{\vartheta} R_Y(y,\bX \con \vartheta )[\delta] }{R_Y(y,\bX \con \vartheta )}
    \\
    &
    =
    \Psi(g \con \vartheta )
    \left[ 
        \begin{array}{l}             
        \displaystyle{ \frac{ D_{\vartheta} \mathcal{C}_1(\bX \con \vartheta )[\delta ] g(y,\bX) + D_{\vartheta}  \mathcal{C}_2(y,\bX \con \vartheta ) [\delta ]  } { \mathcal{C}_1(\bX \con \vartheta ) g(y,\bX) + \mathcal{C}_2(y,\bX \con \vartheta )  } 
        }
        \\[0.4cm]
        \displaystyle{ 
        +
        \frac{ \int g(t,\bX) D_{\vartheta} R_Y(t,\bX \con \vartheta ) [\delta ] \, dt }{ \int g(t,\bX) R_Y(t,\bX \con \vartheta ) \, dt }
        -
        \frac{ D_{\vartheta} R_Y(y,\bX \con \vartheta )[\delta ] }{R_Y(y,\bX \con \vartheta )}
        }
        \end{array}
    \right] \ .
\end{align*}
Therefore, if $\mathcal{C}_1, \mathcal{C}_2$, and $R_Y$ are Fr\'echet differentiable with respect to $\vartheta $, then by standard functional calculus, $\Psi$ is also Fr\'echet differentiable, since all denominators are bounded away from zero on the parameter space.

We detail the proof of Fr\'echet differentiability for $\mathcal{C}_1$.  The proofs for $\mathcal{C}_2$ and $R_Y$ follow by the same arguments and are therefore omitted. The pathwise derivative of $\mathcal{C}_1$ is given by
\begin{align*}
    D_{\vartheta} \mathcal{C}_1(\bX \con \vartheta ) [\delta ]
    =
    \frac{
    \bigg[ 
        \begin{array}{l}
    \texttt{f}(A=1 \cond Z=z_0,\bX) f(A=1 \cond Z=z_1,\bX) \\
    -  f(A=1 \cond Z=z_0,\bX) \texttt{f} (A=1 \cond Z=z_1,\bX) 
        \end{array}
    \bigg]
    }{ f_A^2(A=1 \cond Z=z_1,\bX)} \ .
\end{align*}
Therefore, we find
\begin{align*}
    &
    \mathcal{C}_1(\bX \con \vartheta  + \delta )
    -
    \mathcal{C}_1(\bX \con \vartheta )
    - 
    D_{\vartheta} \mathcal{C}_1(\bX \con \vartheta )[\delta ]
    \\
    &
    =
    \frac{
    \left[ 
        \begin{array}{l}
            \{ f (A=1 \cond Z=z_0,\bX) + \texttt{f} (A=1 \cond Z=z_0,\bX) \} 
        f^2(A=1 \cond Z=z_1,\bX)
        \\
        -
        \{ f (A=1 \cond Z=z_1,\bX) + \texttt{f} (A=1 \cond Z=z_1,\bX) \} 
        f (A=1 \cond Z=z_1,\bX)
        f (A=1 \cond Z=z_0,\bX)
        \\
        - 
        \{ f (A=1 \cond Z=z_1,\bX) + \texttt{f} (A=1 \cond Z=z_1,\bX) \} 
        \texttt{f}(A=1 \cond Z=z_0,\bX) f(A=1 \cond Z=z_1,\bX) 
        \\
        +  
        \{ f (A=1 \cond Z=z_1,\bX) + \texttt{f} (A=1 \cond Z=z_1,\bX) \} 
        f(A=1 \cond Z=z_0,\bX) \texttt{f}(A=1 \cond Z=z_1,\bX) 
        \end{array}
    \right] 
    }{
        \{ f (A=1 \cond Z=z_1,\bX) + \texttt{f} (A=1 \cond Z=z_1,\bX) \} 
        f^2(A=1 \cond Z=z_1,\bX)
    }
    \\
    &
    =
    \frac{ \texttt{f}(A=1 \cond Z=z_1,\bX) 
     \left\{
     \begin{array}{l}
     f(A=1 \cond Z=z_0,\bX) \texttt{f}(A=1 \cond Z=z_1,\bX) \\
     - f(A=1 \cond Z=z_1,\bX) \texttt{f}(A=1 \cond Z=z_0,\bX) 
     \end{array}
     \right\}
    }{
        \{ f (A=1 \cond Z=z_1,\bX) + \texttt{f} (A=1 \cond Z=z_1,\bX) \} 
        f^2(A=1 \cond Z=z_1,\bX)
    }  \ .
\end{align*}
Consequently, we obtain
\begin{align*}
    \lim_{t \rightarrow 0}
    \frac{ \big| \mathcal{C}_1(\bX \con \vartheta  + t \delta )
    -
    \mathcal{C}_1(\bX \con \vartheta )
    - 
    D_{\vartheta} \mathcal{C}_1(\bX \con \vartheta ) [t \delta ] \big| }{ t }
    = 0 \ , \quad \forall \delta  \ ,
\end{align*}
which establishes that  $\mathcal{C}_1$ is Fr\'echet differentiable with respect to $\vartheta $. 

Likewise, we also find
\begin{align*}
    \big|
    D_{\vartheta} \mathcal{C}_1(\bX \con \vartheta ) [\delta ]
    \big|
    &
    =
    \frac{
    \bigg| 
        \begin{array}{l}
    \texttt{f}(A=1 \cond Z=z_0,\bX) f(A=1 \cond Z=z_1,\bX) \\
    -  f(A=1 \cond Z=z_0,\bX) \texttt{f} (A=1 \cond Z=z_1,\bX) 
        \end{array}
    \bigg|
    }{ f^2(A=1 \cond Z=z_1,\bX)} 
    \\
    &
    \lesssim | \texttt{f}(A=1 \cond Z=z_0,\bX) | +  | \texttt{f}(A=1 \cond Z=z_1,\bX) | \ .
\end{align*}
The inequality follows directly from the boundedness assumption.

By the same reason, $\mathcal{C}_2$ and $R_Y$ are also Fr\'echet differentiable. This in turn shows that $\Psi$ is Fr\'echet differentiable with respect to $\vartheta $. We denote the Fr\'echet derivative operator with respect to $\vartheta $ by $D_{\vartheta} \Psi(g \con \vartheta ): \delta  \mapsto D_{\vartheta} \Psi(g \con \vartheta ) [ \delta ] $ where $\delta $ is the perturbation direction in $\vartheta $. Moreover, we have the bound 
\begin{align*}
    &
    \big\| D_{\vartheta} \Psi^*(g^\star) [ \delta ] \big\|_{Y,2}
    \\
    &
    \lesssim
    \left[ 
    \begin{array}{l}            
    \| \texttt{f}(\potY{0} \cond A=0,Z=1,\bX) \|_{Y,2}
    +
    \| \texttt{f}(\potY{0} \cond A=0,Z=0,\bX) \|_{Y,2}
    \\
             + \| \texttt{f}(A=1 \cond Z=1,\bX) \|
             + \| \texttt{f}(A=1 \cond Z=0,\bX) \|
        \end{array}
    \right]
    \\
    &
    \equiv 
    \mathfrak{R} (\delta )
    \ ,
    \numeq \label{eq-proof-r2YA}
\end{align*}
where the inequality follows from the boundedness conditions.

Next, let $D_{g}\Psi(g \con \vartheta ) [h]$ denote the pathwise derivative of $\Psi(g \con \vartheta )$ with respect to $g$, for a fixed $\vartheta $, in the direction of $h$, i.e., 
\begin{align*}
    &
    D_{g} \Psi(g \con \vartheta ) [h]
    \\
    &
    = \frac{\partial}{\partial t} \Psi(g + th \con \vartheta ) \bigg|_{t=0}
    \\
    &
    =
    \underbrace{ \frac{ \mathcal{C}_1 (\bX \con \vartheta ) R_D (\bX \con g, \vartheta ) }{R_Y(y,\bX \con \vartheta )} }_{\equiv \mathcal{A}(y,\bX \con g, \vartheta  )}
    h(y,\bX) 
    +
    \underbrace{
    \frac{ \mathcal{C}_1 (\bX \con \vartheta ) g(y,\bX) + \mathcal{C}_2(y,\bX \con \vartheta )}{R_Y(y,\bX \con \vartheta )}}_{\equiv \mathcal{B}(y,\bX \con g, \vartheta  )}
     R_D(\bX \con h, \vartheta  ) \ .
\end{align*}

Define $\Delta(g \con \vartheta ) [h] \equiv \Psi(g+h \con \vartheta ) - \Psi(g \con \vartheta ) - D_{g} \Psi(g \con \vartheta ) [h]$, which is expressed by
\begin{align*}
    &
    \Delta(g \con \vartheta ) [h]
    \\
    &
    = \Psi(g+h \con \vartheta ) - \Psi(g \con \vartheta ) - D_{g} \Psi(g \con \vartheta ) [h]
    \\
    &
    =
    \frac{ \mathcal{C}_1(\bX \con \vartheta ) \{ g(y,\bX) + h(y,\bX) \} + \mathcal{C}_2(y,\bX \con \vartheta )}{R_Y(y,\bX \con \vartheta ) }
    R_D(\bX \con g + h, \vartheta ) 
    \\
    &
    \quad \quad 
    -
    \frac{ \mathcal{C}_1(\bX \con \vartheta ) g(y,\bX) + \mathcal{C}_2(y,\bX \con \vartheta )}{R_Y(y,\bX \con \vartheta ) }
    R_D(\bX \con g, \vartheta ) 
    \\
    &
    \quad \quad 
    -
    \frac{ \mathcal{C}_1(\bX \con \vartheta ) h(y,\bX) 
     }{R_Y(y,\bX \con \vartheta ) }
     R_D(\bX \con g , \vartheta )
    -
    \frac{ \mathcal{C}_1(\bX \con \vartheta ) g(y,\bX) + \mathcal{C}_2(y,\bX \con \vartheta )}{R_Y(y,\bX \con \vartheta ) }
    R_D(\bX \con h, \vartheta ) 
    \\  
    &
    =
    \frac{ \mathcal{C}_1(\bX \con \vartheta ) h(y,\bX)
    }{R_Y(y,\bX \con \vartheta )}
    R_D(\bX \con h, \vartheta )
     \ . 
\end{align*} 
Note that 
\begin{align*}
     R_D(\bX \con h, \vartheta )
    \lesssim \big\| h(\potY{0},\bX)  \big\|_{Y,2}
    \ .
\end{align*}
Therefore, we establish
\begin{align*}
    \big| \Delta(g \con \vartheta ) [h] \big|
    \lesssim 
    \big\| h(\potY{0},\bX)  \big\|_{Y,2}
    \times \big| h(y,\bX) \big| \ ,
\end{align*}
which in turn implies
\begin{align*}
    \big\| \Delta(g \con \vartheta ) [h] \big\|_{Y,2}
    \lesssim 
    \big\| h(\potY{0},\bX)  \big\|_{Y,2}^2 \ . 
\end{align*}
Consequently, we find
\begin{align*}
    \lim_{t \rightarrow 0} 
    \frac{ \big\| \Delta(g \con \vartheta )  [th] \big\|_{Y,2} }
    { |t| }
    = 
    0 \ , \quad \forall h \ ,
\end{align*}
which establishes that  $\Psi$ is Fr\'echet differentiable with respect to $g$. We denote the Fr\'echet derivative operator with respect to $g$ by $D_{g} \Psi(g \con \vartheta ): h \mapsto D_{g} \Psi(g \con \vartheta ) [h]$ where $h$ is the perturbation direction in $g$.

Now, consider the mapping $\Lambda(g, \vartheta ) = g - \Psi(g \con \vartheta )$. The Fr\'echet derivative operator of $\Lambda$ with respect to $\vartheta $ at $(g^\star,\vartheta ^*)$ is given by the operator $D_{\vartheta} \Lambda(g,\vartheta ) = - D_{\vartheta} \Psi^*(g^\star): \delta  \mapsto - D_{\vartheta}\Psi^*(g^\star)[\delta ]$. Likewise, the Fr\'echet derivative operator of $\Lambda$ with respect to $g$ at $(g^\star,\vartheta ^*)$ is given by the operator $D_{g} \Lambda(g,\vartheta ) = I - D_{g} \Psi^*(g^\star): h \mapsto h - D_{g}\Psi^*(g^\star)[h]$. We will show that the second Fr\'echet derivative operator is invertible. Note that
\begin{align*} 
    &
    D_{g} \Lambda(g^\star,\vartheta ^*)[h]
    \\
    &
    =
    h(y,\bX) - D_{g} \Psi^*(g^\star)[h]
    \\
    & =
    \big\{ 1 - \mathcal{A}^*(y,\bX \con g^\star) \big\} h(y,\bX) - \mathcal{B}^*(y,\bX \con g^\star) R_D^*(\bX \con h)
    \\
    &
    \equiv
    b(y,\bX) \ .
    \numeq 
    \label{eq-proof-mapping denom}
\end{align*}

By combining \eqref{eq-proof-RYrange} and \eqref{eq-proof-RDrange} with the definition of $\mathcal{C}_1(\bX \con \vartheta )$, we obtain a similar result to \eqref{eq-proof-kappa}:
\begin{align*}
    &
    \mathcal{A}^*(y,\bX \con g^\star)
    \\
    &
    \in 
    \left\{ 
    \begin{array}{ll}
    \Big[ \frac{ 1/\beta^*(z_1,\bX) + c_\alpha }{ 1/\beta^*(z_0,\bX) + c_\alpha }
    , 
     \frac{ 1/\beta^*(z_1,\bX) + C_\alpha }{ 1/\beta^*(z_0,\bX) + C_\alpha }
     \Big] 
     \subseteq (0,1)
     &
     \text{ if $\beta^*(z_1,\bX) > \beta^*(z_0,\bX)$}
     \\[0.25cm]
     \Big[ \frac{ 1/\beta^*(z_1,\bX) + C_\alpha }{ 1/\beta^*(z_0,\bX) + C_\alpha }
    , 
     \frac{ 1/\beta^*(z_1,\bX) + c_\alpha }{ 1/\beta^*(z_0,\bX) + c_\alpha }
     \Big] 
     \subseteq (1,\infty)
     &
     \text{ if $\beta^*(z_1,\bX) < \beta^*(z_0,\bX)$} 
    \end{array}
     \right. .
     \numeq \label{eq-proof-A star}
\end{align*} 
This implies that $\mathcal{A}^*(y,\bX \con g^\star)  \neq 1 $ for all $y \in \suppX$.

From \eqref{eq-proof-A star}, the following expression is well-defined, with the denominator bounded away from zero:
\begin{align*}
    h(y,\bX) = \frac{ b(y,\bX ) + \mathcal{B}^*(y,\bX \con g^\star) R_D^*(\bX \con h) } { 1- \mathcal{A}^*(y,\bX \con g^\star) } \ .
        \numeq
    \label{eq-proof-delta 1}
\end{align*}   
Multiplying both sides by $R_Y^*$ and integrating with respect to $y$, we obtain
\begin{align*}
    &
    R_D^*(\bX \con h)
    \\
    &
    = 
    \int \frac{ b(y,\bX) }{ 1 - \mathcal{A}^*(y,\bX \con g^\star) } R_Y^*(y,\bX) \, dy 
    + 
    R_D^*(\bX \con h) \int \frac{ \mathcal{B}^*(y,\bX \con g^\star) }{ 1 - \mathcal{A}^*(y,\bX \con g^\star) } R_Y^*(y,\bX) \, dy  \ ,
\end{align*}
meaning that
\begin{align*}
    &
    \underbrace{
    \bigg\{ 1- \int \frac{ \mathcal{B}^*(y,\bX \con g^\star) }{ 1 - \mathcal{A}^*(y,\bX \con g^\star) } R_Y^*(y,\bX) \, dy \bigg\}
    }_{\equiv \mathcal{C}_3^*(\bX)}
    R_D^*(\bX \con h)  \\
    &
    = 
    \bigg\{ \int \frac{ b(y,\bX) }{ 1 - \mathcal{A}^*(y,\bX \con g^\star) } R_Y^*(y,\bX) \, dy \bigg\} \ .
    \numeq
    \label{eq-proof-delta 2}
\end{align*}

From the definition of $g^\star$, we obtain
\begin{align*}
    \frac{ \mathcal{B}^*(y,\bX \con g^\star) }{ 1 - \mathcal{A}^*(y,\bX \con g^\star) } R_Y^*(y,\bX) 
    &
    = \frac{  \mathcal{C}_1^* (\bX) g^\star(y,\bX) + \mathcal{C}_2^* (y,\bX) }{ 1 - \mathcal{A}^*(y,\bX \con g^\star) } 
     \\
    &
    = \frac{ g^\star(y,\bX) }{ 1 - \mathcal{A}^*(y,\bX \con g^\star) }
     \frac{ R_Y^*(y,\bX) }{R_D^*(\bX \con g^\star) } \ .
\end{align*}  
Integrating with respect to $y$, we obtain
\begin{align*}
    \mathcal{C}_3^*(\bX) 
    =
    1-
    \frac{1}{R_D^*(\bX \con g^\star)}
    \int \frac{ g^\star(y,\bX) R_Y^*(y,\bX) }{1-\mathcal{A}^*(y,\bX \con g^\star)} \, dy \ .
    \numeq 
    \label{eq-proof-C3}
\end{align*}

If $z_0 = \argmin_{z \in \{0,1\}} f(A=1 \cond Z=z,\bX)$, we have $\mathcal{A}^*(y,\bX \con g^\star) \in (0,1)$ from \eqref{eq-proof-A star}. Then, the right hand side of \eqref{eq-proof-C3} becomes
\begin{align*}
    \mathcal{C}_3^*(\bX) 
    & =
    1-
    \frac{1}{R_D^*(\bX \con g^\star)}
    \int \frac{ g^\star(y,\bX) R_Y^*(y,\bX) }{1-\mathcal{A}^*(y,\bX \con g^\star)} \, dy
    \\
    & <
    1-
    \frac{1}{R_D^*(\bX \con g^\star)}
    \int g^\star(y,\bX) R_Y^*(y,\bX) \, dy
    \\
    & = 0 \ .
\end{align*}
If $z_0 = \argmax_{z \in \{0,1\}} f(A=1 \cond Z=z,\bX)$, we have $\mathcal{A}^*(y,\bX \con g^\star) \in (1,\infty)$ from \eqref{eq-proof-A star}. Then, the right hand side of \eqref{eq-proof-C3} becomes
\begin{align*}
    \mathcal{C}_3^*(\bX) 
    & =
    1-
    \frac{1}{R_D^*(\bX \con g^\star)}
    \int \underbrace{ \frac{ g^\star(y,\bX) R_Y^*(y,\bX) }{1-\mathcal{A}^*(y,\bX \con g^\star)} }_{<0} \, dy > 1 \ .
\end{align*}
Consequently, we establish $\mathcal{C}_3^*(\bX) \neq 0$. In particular, $1/\mathcal{C}_3^*(\bX)$ is uniformly bounded. 

Combined with \eqref{eq-proof-delta 1} and \eqref{eq-proof-delta 2} with $\mathcal{C}_3^* \neq 0$, it follows that 
\begin{align*}
    h(y,\bX) 
    &
    = \frac{ b(y,\bX) + {  \frac{ \mathcal{B}^*(y,\bX \con g^\star) }{ \mathcal{C}_3^*(\bX) } \int \frac{ b(t,\bX) }{ 1 - \mathcal{A}^*(t,\bX \con g^\star) } R_Y^*(t,\bX) \, dt }} {1- \mathcal{A}^*(y,\bX \con g^\star) }  \ .
\end{align*}
Therefore, the mapping $D_{g} \Lambda(g^\star,\vartheta ^*): h \mapsto D_{g} \Lambda(g^\star,\vartheta ^*)[h]$ in \eqref{eq-proof-mapping denom} is invertible in the following sense: for any given function $b$, there exists a direction $h$ such that $D_{g} \Lambda(g^\star,\vartheta ^*)[h] = b$. In addition, we also obtain
\begin{align*}
     \bigg|
     \int \frac{ b(t,\bX) }{ 1 - \mathcal{A}^*(t,\bX \con g^\star) } R_Y^*(t,\bX) \, dt
     \bigg|
     &
     \lesssim
     \int \big| b(t,\bX) \big| \, dt
     \leq  
     \bigg\{ \int b^2(t,\bX) \, dt  \bigg\}^{1/2} \ .
\end{align*}
Therefore, the $L^2(\suppYX)$-norm of $h$ is bounded in terms of the $L^2(\suppYX)$-norm of $b$ as follows:
\begin{align*}
    \big\| h(\potY{0},\bX) \big\|_{Y,2}^2
    &
    \lesssim
    \int  \frac{ b^2(y,\bX) +  \int b^2 (t,\bX) \, dt } { \{ 1- \mathcal{A}^*(y,\bX \con g^\star) \}^2 } \, dy
    \\
    &
    \lesssim 
    \int b^2(y,\bX) \, dy
    =
    \big\| b(\potY{0},\bX) \big\|_{Y,2}^2
    \ .
    \numeq 
    \label{eq-proof-delta}
\end{align*}

From the established results, specifically, Fr\'echet differentiability of $\Lambda$ and invertible derivative operator $D_{g} \Lambda(g^\star,\vartheta ^*)$, the implicit function theorem applies. Consequently, we may write
\begin{align*}
    & 
    g^\star(\widehat{\vartheta} \LSS)
    -
    g^\star(\vartheta ^*)
    \\
    &
    =
    D_{\vartheta} g^\star(\vartheta ^*) 
    [ \widehat{\vartheta} \LSS - \vartheta ^* ]
    + 
    o_P \big( \| \widehat{\vartheta} \LSS - \vartheta ^* \|_{Y,2} \big) 
    \\
    &
    =
    -
    \big\{ D_{g} \Lambda(g^\star,\vartheta ^*) \big\}^{-1}
    \circ 
    D_{\vartheta} \Lambda(g^\star,\vartheta ^*) 
    [ \widehat{\vartheta} \LSS - \vartheta ^* ]
    + 
    o_P \big( \| \widehat{\vartheta} \LSS - \vartheta ^* \|_{Y,2} \big)  \ .
\end{align*}
Let $ \widehat{b}\LSS \equiv D_{\vartheta} \Lambda(g^\star,\vartheta ^*) [ \widehat{\vartheta} \LSS - \vartheta ^* ]$, and $\widehat{h}\LSS = \big\{ D_{g} \Lambda(g^\star,\vartheta ^*) \big\}^{-1}[ \widehat{b}\LSS ]$. Then, we find
\begin{align*}
    &
    \big\| g^\star(\widehat{\vartheta} \LSS) - g^\star(\vartheta ^*) \big\|_{Y,2}
    \\
    &
    \lesssim
    \big\| \widehat{h}\LSS \big\|
    +
    o_{P} \big( \| \widehat{\vartheta} \LSS - \vartheta ^* \|_{Y,2} \big) 
    \\
    &
    \stackrel{\eqref{eq-proof-delta}}{ \lesssim }
    \big\| \widehat{b}\LSS \big\|
    +
    o_{P} \big( \| \widehat{\vartheta} \LSS - \vartheta ^* \|_{Y,2} \big) 
    \\
    &
    \stackrel{\eqref{eq-proof-r2YA}}{ \lesssim }
    \mathfrak{R}( \widehat{\vartheta} \LSS - \vartheta ^* ) 
    +
    o_{P} \big( \| \widehat{\vartheta} \LSS - \vartheta ^* \|_{Y,2} \big) 
    \ .
    \numeq
    \label{eq-proof-g estimation} 
\end{align*}

Combining \eqref{eq-proof-alpha beta estimation} and \eqref{eq-proof-g estimation}, we establish the following result for $z \in \{0,1\}$:
\begin{align*}
    &
    \big\|
    \widehat{\alpha}\LSS(\potY{0},\bX) \widehat{\beta}\LSS (z,\bX)
    -
    \alpha^*(\potY{0},\bX) 
    \beta^*(z,\bX) 
    \big\|_{Y,2}
    \\
    &
    \lesssim
    \left[ 
    \begin{array}{l}            
    \| \widehat{f}\LSS(\potY{0} \cond A=0,Z=1,\bX)
    -
    f^*(\potY{0} \cond A=0,Z=1,\bX)\|_{Y,2}
    \\
    +
    \| \widehat{f}\LSS(\potY{0} \cond A=0,Z=0,\bX)
    -
    f^*(\potY{0} \cond A=0,Z=0,\bX) \|_{Y,2}
    \\
    + \| \widehat{f}\LSS(A=1 \cond Z=1,\bX) - f^*(A=1 \cond Z=1,\bX) \|
             \\
             + \| \widehat{f}\LSS(A=1 \cond Z=0,\bX) - f^*(A=1 \cond Z=0,\bX) \|
        \end{array}
    \right] \ .
\end{align*}

\vspace*{0.25cm}

\noindent \HL{iii} Under the boundedness condition, $f^*(\potY{0} \cond \bX)$ is uniformly bounded from Lemma \ref{lemma-boundedness}. Therefore, for a function $g(\potY{0},Z,\bX)$, we have
\begin{align*}
    \big\|
    g(\potY{0},z,\bX)
    \big\|_{Y,2}^2
    = 
    \int 
    g^2(y,z,\bX) \, dy
    \asymp
    \int 
    g^2(y,z,\bX) f^*(\potY{0}=y \cond \bX)  \, dy \ .
\end{align*}
This, in turn, results in
\begin{align*}
    &
    \big\|
    g(\potY{0},Z,\bX)
    \big\|_{P,2}^2
    \\
    &
    = 
    \iint  
    g^2(y,z,\bX) f^*(\potY{0}=y \cond \bX)
    f^*(Z=z \cond \bX) 
    f^*(\bX)
    \, 
    d(y,z,\bX)
    \\
    &
    \asymp 
    \int 
    \big\|
    g(\potY{0},z,\bX)
    \big\|_{Y,2}^2
    f^*(Z=z \cond \bX) 
    f^*(\bX) 
    d \bX
    \\
    &
    \leq 
    \int 
    \big\{
    \big\|
    g(\potY{0},1,\bX)
    \big\|_{Y,2}^2
    +
    \big\|
    g(\potY{0},0,\bX)
    \big\|_{Y,2}^2
    \big\}
    f^*(\bX) 
    d \bX \ .
\end{align*}
We also find that  
\begin{align*}
    &
    \int 
    \big\{
    \| g(\potY{0}, 1,\bX)\|_{Y,2}
    +
    \| g(\potY{0}, 0,\bX)\|_{Y,2}
     \big\}^2
    \, f^*(\bX) \, d\bX
    \\
    &
    \lesssim
    \int 
    \big\{
    \| g(\potY{0}, 1,\bX)\|_{Y,2}^2
    +
    \| g(\potY{0}, 0,\bX)\|_{Y,2}^2
     \big\} 
    \, f^*(\bX) \, d\bX
    \\
    &
    \lesssim
    \int 
    \big\{
    \| g(\potY{0}, 1,\bX)\|_{Y,2}^2 
    f^*(Z=1 \cond \bX)
    +
    \| g(\potY{0}, 0,\bX)\|_{Y,2}^2
    f^*(Z=0 \cond \bX)
     \big\} 
    \, f^*(\bX) \, d\bX
    \\
    &
    \lesssim
    \int 
    \| g(\potY{0}, Z,\bX)\|_{Y,2}^2  
    f^*(Z \cond \bX) f^*(\bX) \, d(Z,\bX)
    \\
    &
    \asymp
    \int 
    \| g(\potY{0}, Z,\bX)\|^2  
    f^*(\potY{0} \cond \bX)
    f^*(Z \cond \bX) f^*(\bX) \, d(\potY{0},Z,\bX)
    \\
    &
    =
    \big \| g(\potY{0},Z,\bX) \big \|_{P,2}^2 \ .
\end{align*} 
Thus, we establish
\begin{align*}
   \big \| g(\potY{0},Z,\bX) \big \|_{P,2}^2 
   \asymp 
    \int 
    \big\{
    \big\|
    g(\potY{0},1,\bX)
    \big\|_{Y,2}^2
    +
    \big\|
    g(\potY{0},0,\bX)
    \big\|_{Y,2}^2
    \big\}
    f^*(\bX) 
    d \bX \ .
    \numeq 
    \label{eq-proof-asymptotic equi}
\end{align*}

We now take $g = \widehat{\alpha}\LSS\widehat{\beta}\LSS 
    -
    \alpha^*
    \beta^*$. From \eqref{eq-proof-alpha beta wrt y}, we obtain 
\begin{align*}
    &
    \left[ 
        \begin{array}{l}
             \big\|
    \widehat{\alpha}\LSS(\potY{0},\bX) \widehat{\beta}\LSS (1,\bX)
    -
    \alpha^*(\potY{0},\bX) 
    \beta^*(1,\bX) 
    \big\|_{Y,2}
             \\
             +
             \big\|
    \widehat{\alpha}\LSS(\potY{0},\bX) \widehat{\beta}\LSS (0,\bX)
    -
    \alpha^*(\potY{0},\bX) 
    \beta^*(0,\bX) 
    \big\|_{Y,2}
        \end{array}
    \right] 
    \\
    &
    \lesssim
    \left[ 
    \begin{array}{l}            
    \| \widehat{f}\LSS(\potY{0} \cond A=0,Z=1,\bX)
    -
    f^*(\potY{0} \cond A=0,Z=1,\bX)\|_{Y,2}
    \\
    +
    \| \widehat{f}\LSS(\potY{0} \cond A=0,Z=0,\bX)
    -
    f^*(\potY{0} \cond A=0,Z=0,\bX) \|_{Y,2}
    \\
             + \| \widehat{f}\LSS(A=1 \cond Z=1,\bX) - f^*(A=1 \cond Z=1,\bX) \|
             \\
             + \| \widehat{f}\LSS(A=1 \cond Z=0,\bX) - f^*(A=1 \cond Z=0,\bX) \|
        \end{array}
    \right] \ .
\end{align*}
Multiplying $f^*(\bX)$ both hand sides and integrating over $\bX$, and using \eqref{eq-proof-asymptotic equi}, we obtain
\begin{align*}
    &
    \big\|
    \widehat{\alpha}\LSS(\potY{0},\bX) \widehat{\beta}\LSS (Z,\bX)
    -
    \alpha^*(\potY{0},\bX) 
    \beta^*(Z,\bX) 
    \big\|_{P,2}
    \\
    &
    \lesssim 
    \left[ 
    \begin{array}{l}            
    \| \widehat{f}\LSS(\potY{0} \cond A=0,Z,\bX)
    -
    f^*(\potY{0} \cond A=0,Z,\bX)\|_{P,2}
    \\
     + \| \widehat{f}\LSS(A=1 \cond Z,\bX) - f^*(A=1 \cond Z,\bX) \|_{P,2}
        \end{array}
    \right]
    \ .
\end{align*}

\vspace*{0.25cm}

\noindent \HL{iv} Note that
\begin{align*}
    &
    \big\| \normhat(Z,\bX) - \norm^*(Z,\bX)
    \big\|^2
    \\
    &
    =
    \bigg\|
    \int \pihat (y,Z,\bX) 
    \widehat{f}\LSS(Y=y \cond A=0,Z,\bX) -
    \pi^* (y,Z,\bX) 
    f^*(Y=y \cond A=0,Z,\bX)
    \, dy
    \bigg\|^2
    \\
    &
    \lesssim
    \int 
    \left[ 
        \begin{array}{l}             
    \| \pihat (y,Z,\bX) 
    -
     \pi^* (y,Z,\bX) \|^2
     \\
             +
             \| \widehat{f}\LSS(Y=y \cond A=0, Z,\bX)
    -
     f^*(Y=y \cond A=0,Z,\bX) \|^2 
        \end{array}
    \right] f^*(\potY{0}=y \cond Z, \bX) \, dy \ .
\end{align*}
Multiplying both sides by $f^*(Z,\bX)$ and integrating with respect to $(Z,\bX)$, we obtain
\begin{align*}
    &
    \big\|
    \normhat (Z,\bX) 
    -
    \norm^* (Z,\bX)  
    \big\|_{P,2}
    \\
    &
    \lesssim 
    \left[ 
    \begin{array}{l}            
    \| \widehat{f}\LSS(\potY{0} \cond A=0,Z,\bX)
    -
    f^*(\potY{0} \cond A=0,Z,\bX)\|_{P,2}
    \\
     + \| \widehat{f}\LSS(A=1 \cond Z,\bX) - f^*(A=1 \cond Z,\bX) \|_{P,2}
        \end{array}
    \right]
    \ . 
\end{align*}

\vspace*{0.25cm}

\noindent \HL{v} We first present the result for $\alpha$:
\begin{align*}
    &
    \| 
    \alphahat_{\norm}(\potY{0},\bX)
    -
    \alpha_{\norm}^*(\potY{0},\bX)
    \|_{P,2}
    \\
    &
    \lesssim
    \| \normhat(Z,\bX) - \norm^*(Z,\bX) \|_{P,2}
    +
    \| \pihat (\potY{0},Z,\bX) - \pi^*(\potY{0},Z,\bX) \|_{P,2}
    \\
    &
    \stackrel{\eqref{eq-proof-alpha beta L2P},\eqref{eq-proof-normalD L2P}}{
    \lesssim 
    }
    \left[ 
    \begin{array}{l}            
    \| \widehat{f}\LSS(\potY{0} \cond A=0,Z,\bX)
    -
    f^*(\potY{0} \cond A=0,Z,\bX)\|_{P,2}
    \\
     + \| \widehat{f}\LSS(A=1 \cond Z,\bX) - f^*(A=1 \cond Z,\bX) \|_{P,2}
        \end{array}
    \right]
    \ .
\end{align*}

In turn, the result for $\beta$ is established as follows:
\begin{align*}
    &
    \| 
    \betahat_{\norm}(Z,\bX)
    -
    \beta_{\norm}^*(Z,\bX)
    \|_{P,2}
    \\
    &
    \lesssim
    \| 
    \alphahat_{\norm}(\potY{0},\bX)
    -
    \alpha_{\norm}^*(\potY{0},\bX)
    \|_{P,2}
    +
    \| \pihat (\potY{0},Z,\bX) - \pi^*(\potY{0},Z,\bX) \|_{P,2}
    \\
    &
    \stackrel{\eqref{eq-proof-normalD L2P},\eqref{eq-proof-alpha L2P}}{
    \lesssim 
    }
    \left[ 
    \begin{array}{l}            
    \| \widehat{f}\LSS(\potY{0} \cond A=0,Z,\bX)
    -
    f^*(\potY{0} \cond A=0,Z,\bX)\|_{P,2}
    \\
     + \| \widehat{f}\LSS(A=1 \cond Z,\bX) - f^*(A=1 \cond Z,\bX) \|_{P,2}
        \end{array}
    \right]
    \ .
\end{align*}

\vspace*{0.25cm}

\noindent \HL{vi} From the definitions of $\mu^*$ and $\muhat$, we have
\begin{align*}
    & 
    \muhat(Z,\bX \con v) - \mu^*(Z,\bX \con v)
    \\
    &
    =  
\frac{
\widehat{\EXP}\LSS \big\{ v  \alphahat(\potY{0},\bX) \cond A=0,Z,\bX \big\}
}{
\widehat{\EXP}\LSS \big\{  \alphahat(\potY{0},\bX) \cond A=0,Z,\bX \big\}
} -
\frac{ \EXP \big\{   v \alpha^*(\potY{0},\bX) \cond A=0, Z,\bX \big\} }{ \EXP \big\{  \alpha^*(\potY{0},\bX) \cond A=0, Z,\bX \big\} }
    \\
    &
    =  
\frac{
\widehat{\EXP}\LSS \big\{ v  \hat{\pi}\LSS(\potY{0},Z,\bX) \cond A=0,Z,\bX \big\}
}{
\widehat{\EXP}\LSS \big\{  \hat{\pi}\LSS(\potY{0},Z,\bX) \cond A=0,Z,\bX \big\}
} -
\frac{ \EXP \big\{   v \pi^*(\potY{0},Z,\bX) \cond A=0, Z,\bX \big\} }{ \EXP \big\{  \pi^*(\potY{0},Z,\bX) \cond A=0, Z,\bX \big\} }
\ .
\end{align*}
Then, we find the following result for $Z \in \{0,1\}$:
\begin{align*}  
&
    \big\| \muhat(Z,\bX \con v) - \mu^*(Z,\bX \con v)
    \big\|
    \\  
&
\lesssim 
\Bigg\|
\begin{array}{l}
\EXP \big\{   v(\potY{0},Z,\bX) \pi^*(\potY{0},Z,\bX) \cond A=0, Z,\bX \big\}
\widehat{\EXP}\LSS \big\{ \pihat(\potY{0},Z,\bX) \cond A=0,Z,\bX \big\}
\\[0.1cm]
-
\widehat{\EXP}\LSS \big\{ v(\potY{0},Z,\bX)  \pihat(\potY{0},Z,\bX) \cond A=0,Z,\bX \big\}
\EXP \big\{  \pi^*(\potY{0},Z,\bX) \cond A=0, Z,\bX \big\}
\end{array} 
\Bigg\|
\\
&
\lesssim 
\Bigg\|
\int v(y,Z,\bX) 
\Bigg\{ 
\begin{array}{l}
\pihat(y,Z,\bX) \widehat{f}\LSS(Y=y \cond A=0,Z,\bX)
\\[0.1cm]
-
\pi^*(y,Z,\bX) f^*(Y=y \cond A=0,Z,\bX)
\end{array}
\Bigg\}
\, dy
\Bigg\|
\\
&
\qquad +
\Bigg\|
\int
\Bigg\{ 
\begin{array}{l}
\pihat(y,Z,\bX) \widehat{f}\LSS(Y=y \cond A=0,Z,\bX)
\\[0.1cm]
-
\pi^*(y,Z,\bX) f^*(Y=y \cond A=0,Z,\bX)
\end{array}
\Bigg\}
\, dy
\Bigg\|
\\
&
\lesssim 
\int \Big\|
\pihat(y,Z,\bX) \widehat{f}\LSS(Y=y \cond A=0,Z,\bX)
-
\pi^*(y,Z,\bX) f^*(Y=y \cond A=0,Z,\bX)
\Big\| 
\, dy
\\
&
\lesssim 
\big\| 
\alphahat(\potY{0},\bX) \betahat(Z,\bX) 
- \alpha^*(\potY{0},\bX) \beta^*(Z,\bX)
\big\|_{Y,2}
\\
&
\qquad 
+
\big\| \widehat{f}\LSS(\potY{0} \cond A=0,Z,\bX) - f^*(\potY{0} \cond A=0,Z,\bX) \big\|_{Y,2}  
\ .
\end{align*}
By \eqref{eq-proof-asymptotic equi}, we have   
\begin{align*}
    &
    \big\| \muhat(Z,\bX \con v) - \mu^*(Z,\bX \con v)
    \big\|_{P,2}
    \\
    & 
    \leq 
    \big\| \alphahat(\potY{0},\bX) \betahat(Z,\bX) 
- \alpha^*(\potY{0},\bX) \beta^*(Z,\bX)
\big\|_{P,2}
\\
&
\qquad 
+
\big\| \widehat{f}\LSS(\potY{0} \cond A=0,Z,\bX) - f^*(\potY{0} \cond A=0,Z,\bX) \big\|_{P,2} 
\\
    &
   \stackrel{\eqref{eq-proof-alpha beta L2P}}{ \lesssim }
    \left[
        \begin{array}{l}
        \big\| \widehat{f}\LSS(A=1 \cond Z,\bX) - f^*(A=1 \cond Z,\bX)
        \big\|_{P,2}
        \\
        +
        \big\| \widehat{f}\LSS(\potY{0} \cond A=0,Z,\bX) - f^*(\potY{0} \cond A=0,Z,\bX)
        \big\|_{P,2}
        \end{array}
        \right]
\ . 
\end{align*}

\vspace*{0.25cm}

\noindent \HL{vii}
We first prove the case of $A=0$.
It is trivial to show that
\begin{align*}
    & f(Z=1 \cond A=0,\potY{0},\bX)
    \\
    & =
    \frac{ f(Y,A=0,Z=1,\bX) }{ \sum_{z=0}^{1} f(Y,A=0,Z=z,\bX) }
    \\
    &
	=
    \frac{ f(Y \cond A=0,Z=1,\bX) f(A=0 \cond Z=1,\bX) }{ 
    \sum_{z=0}^{1} f(Y \cond A=0,Z=z,\bX) f(A=0 \cond Z=z,\bX) }
    \ .
\end{align*}  
Therefore, we find that
\begin{align*}
	&
	\big\|
    f^*(Z=1 \cond A=0,\potY{0},\bX)
    -
    \widehat{f}\LSS( Z=1 \cond A=0,\potY{0},\bX) 
    \big\| 
    \\
    &
    \lesssim     
        \left[ 
     \begin{array}{l}  
     \big\| \widehat{f}\LSS(A=1 \cond Z,\bX)
     -
     f^*(A=1 \cond Z,\bX)
     \big\|_{P,2}
     \\
     +
     \big\| 
     \widehat{f}\LSS(\potY{0} \cond A=0, Z,\bX)
     -
     f^* (\potY{0} \cond A=0, Z,\bX)
     \big\|_{P,2}
     \end{array}
     \right] \ .
\end{align*}

Next, we prove the case of $A=1$. Note that $f(Z=1 \cond A=1,\potY{0},\bX)$ is represented as
\begin{align*}
    & f(Z=1 \cond A=1,\potY{0},\bX)
    \\
    & =
    \frac{ f(\potY{0},A=1,Z=1,\bX) }{ \sum_{z=0}^{1} f(\potY{0},A=1,Z=z,\bX) }
    \\
    &
    \stackrel{\eqref{eq-proof-basis gamma}}{=}
    \frac{ \pi(\potY{0},1,\bX) f(\potY{0} \cond A=0,Z=1,\bX) / \gamma(1,\bX) }{ 
    \sum_{z=0}^{1} \pi(\potY{0},z,\bX) f(\potY{0} \cond A=0,Z=z,\bX) / \gamma(z,\bX) }
    \ .
\end{align*}  
Consequently, we establish that
\begin{align*}
    &
    \big\|
    f^*(Z=1 \cond A=1,\potY{0},\bX)
    -
    \widehat{f}\LSS( Z=1 \cond A=1,\potY{0},\bX) 
    \big\| 
    \\
     &
     \lesssim 
     \Bigg\| \sum_{z=0}^{1}
     \bigg\{
     \begin{array}{l} 
     \pi^*(\potY{0},z,\bX) f^*(\potY{0} \cond A=0,Z=z,\bX) / \gamma^*(z,\bX)
     \\[0.1cm]
     -
     \pihat(\potY{0},z,\bX) \widehat{f}\LSS(\potY{0} \cond A=0,Z=z,\bX)  / \gammahat(z,\bX) 
     \end{array} 
     \bigg\}
     \Bigg\|
     \\
     &
     \qquad 
     +
    \Bigg\|
    \begin{array}{l}
    \pi^*(\potY{0},1,\bX) f^*(\potY{0} \cond A=0,Z=1,\bX) / \gamma^*(1,\bX)
    \\[0.1cm]
     -
     \pihat(\potY{0},1,\bX) \widehat{f}\LSS(\potY{0} \cond A=0,Z=1,\bX) / \gammahat(1,\bX) 
    \end{array}
     \Bigg\|
     \\
     &
     \lesssim
     \sum_{z=0}^{1}
     \left[ 
     \begin{array}{l}          
     \big\| \pihat (\potY{0},z,\bX)
     -
     \pi^*(\potY{0},z,\bX) 
     \big\|
     \\
     +
     \big\| \widehat{f}\LSS(A=1 \cond Z=z,\bX)
     -
     f^*(A=1 \cond Z=z,\bX)
     \big\|
     \\
     +
     \big\| 
     \widehat{f}\LSS(\potY{0} \cond A=0, Z=z,\bX)
     -
     f^* (\potY{0} \cond A=0, Z=z,\bX)
     \big\| 
     \end{array}
     \right] \ .
\end{align*} 
Multiplying both sides by $f^*(\potY{0},Z,\bX)$  and integrating, we obtain
\begin{align*}
    &
    \big\|
    f^*(Z=1 \cond A=1,\potY{0},\bX)
    -
    \widehat{f}\LSS( Z=1 \cond A=1,\potY{0},\bX) 
    \big\|_{P,2}
    \\
    &
    \lesssim 
    \left[ 
     \begin{array}{l}          
     \big\| \pihat(\potY{0},Z,\bX)
     -
     \pi^*(\potY{0},Z,\bX) 
     \big\|_{P,2}
     \\
     +
     \big\| \widehat{f}\LSS(A=1 \cond Z,\bX)
     -
     f^*(A=1 \cond Z,\bX)
     \big\|_{P,2}
     \\
     +
     \big\| 
     \widehat{f}\LSS(\potY{0} \cond A=0, Z,\bX)
     -
     f^* (\potY{0} \cond A=0, Z,\bX)
     \big\|_{P,2}
     \end{array}
     \right]
     \\
    &
    \stackrel{\eqref{eq-proof-alpha beta L2P}}{\lesssim }
    \left[ 
     \begin{array}{l}  
     \big\| \widehat{f}\LSS(A=1 \cond Z,\bX)
     -
     f^*(A=1 \cond Z,\bX)
     \big\|_{P,2}
     \\
     +
     \big\| 
     \widehat{f}\LSS(\potY{0} \cond A=0, Z,\bX)
     -
     f^* (\potY{0} \cond A=0, Z,\bX)
     \big\|_{P,2}
     \end{array}
     \right]
     \ .
\end{align*}

\vspace*{0.25cm}

\noindent \HL{viii} From \eqref{eq-proof-Y0Z-gamma}, we obtain
\begin{align*}
    &
    \big\|     
    \widehat{f}\LSS(\potY{0} \cond \bX) 
    -
    f^*(\potY{0} \cond \bX) 
    \big\|_{P,2}
    \\
    &
    \stackrel{\eqref{eq-proof-Y0Z-gamma}}{=}
    \Bigg\|
    \begin{array}{l}
         \big\{ 1+ \pihat(\potY{0},Z,\bX) \big\}
    \fhat{}\LSS(A=0 \cond Z,\bX)
    \fhat{}\LSS(\potY{0} \cond A=0,Z,\bX)
    \\
    \qquad 
    -
         \big\{ 1+ \pi^*(\potY{0},Z,\bX) \big\}
    f^*(A=0 \cond Z,\bX)
    f^*(\potY{0} \cond A=0,Z,\bX)
    \end{array} 
    \Bigg\|_{P,2}
    \\
    &
    \lesssim
    \left[
    \begin{array}{l}         
    \big\| \pihat(\potY{0},Z,\bX) -  \pi^*(\potY{0},Z,\bX)  \big\|_{P,2}
    \\
    +
    \big\| \fhat{}\LSS(A=0 \cond Z,\bX) - f^*(A=0 \cond Z,\bX) \big\|_{P,2}
    \\
    +
    \big\|  \fhat{}\LSS(\potY{0} \cond A=0,Z,\bX)
    -
    f^*(\potY{0} \cond A=0,Z,\bX)
    \big\|_{P,2}
    \end{array}
    \right] \\
    &
    \stackrel{\eqref{eq-proof-alpha beta L2P}}{\lesssim }
    \left[ 
     \begin{array}{l}  
     \big\| \widehat{f}\LSS(A=1 \cond Z,\bX)
     -
     f^*(A=1 \cond Z,\bX)
     \big\|_{P,2}
     \\
     +
     \big\| 
     \widehat{f}\LSS(\potY{0} \cond A=0, Z,\bX)
     -
     f^* (\potY{0} \cond A=0, Z,\bX)
     \big\|_{P,2}
     \end{array}
     \right] 
      \ .
\end{align*}

\vspace*{0.25cm}

\noindent \HL{ix} In the proof of Theorem \ref{thm-IF binary Z}, specifically in \eqref{eq-proof-omega star}, we find that $\omega (\cdot \con \theta)$ is a differentiable functional in terms of other nuisance functions $\theta = \{ \fy,\fa,\fz \}$. Therefore, there exists $\theta' \equiv q \widehat{\theta}\LSS + (1-q) \theta^*$ for some $q \in (0,1)$ satisfying 
\begin{align*}
    \what (\potY{0},Z,\bX)
    &
    =
    \omega(\potY{0},Z,\bX \con \widehat{\theta}\LSS)
    \\
    &
    =
    \omega(\potY{0},Z,\bX \con {\theta}^*)
    +
    \nabla_{\theta}\T \omega(\potY{0},Z,\bX \con \theta') 
    \big\{ \widehat{\theta}\LSS - \theta^* \big\}
    \\
    &
    =
    \omega^*(\potY{0},Z,\bX)
    +
    \nabla_{\theta}\T \omega(\potY{0},Z,\bX \con \theta') 
    \big\{ \widehat{\theta}\LSS - \theta^* \big\}
    \ .
\end{align*}
As a result, we find
\begin{align*} 
    \big\|
    \what (\potY{0},Z,\bX)
    -
    \omega^*(\potY{0},Z,\bX)
    \big\|^2  
    &
    =
    \big\| \nabla_{\theta} \omega(\potY{0},Z,\bX \con \theta') \big\|^2
    \big\| \widehat{\theta}\LSS - \theta^* \big\|^2 
    \\
    &
    \lesssim
    \big\| \widehat{\theta}\LSS - \theta^* \big\|^2 \ ,
\end{align*}
where we used the fact that $\sup_{\theta'} \big\| \nabla_{\theta} \omega(\potY{0},Z,X \con \theta') \big\|^2$ is bounded under the assumptions. 
Therefore, we get
\begin{align*}
    &
    \big\|
    \what (\potY{0},Z,\bX)
    -
    \omega^*(\potY{0},Z,\bX)
    \big\|_{P,2}
    \\
    &
    \lesssim 
    \left[ 
     \begin{array}{l}  
     \big\| \widehat{f}\LSS(Z=1 \cond \bX)
     -
     f^*(Z=1 \cond \bX)
     \big\|_{P,2}
     \\
     +
     \big\| \widehat{f}\LSS(A=1 \cond Z,\bX)
     -
     f^*(A=1 \cond Z,\bX)
     \big\|_{P,2}
     \\
     +
     \big\| 
     \widehat{f}\LSS(\potY{0} \cond A=0, Z,\bX)
     -
     f^* (\potY{0} \cond A=0, Z,\bX)
     \big\|_{P,2}
     \end{array}
     \right]  \ .
\end{align*}
 
We also establish
\begin{align*}
    &
    \big\|
    \mu^*(Z,\bX \con \what)
    - \mu^*(Z,\bX \con \omega^*)
    \big\|
    \\
    &
    \leq 
    \int \big\| \what(y,Z,\bX) - \omega^*(y,Z,\bX) \big\| f^*(\potY{0}=y \cond A=1,Z,\bX) \, dy 
    \\
    &
    \lesssim 
    \big\| \what(y,1,\bX) - \omega^*(y,1,\bX) \big\|_{Y,2}
    +
    \big\| \what(y,0,\bX) - \omega^*(y,0,\bX) \big\|_{Y,2} \ .
\end{align*}
From \eqref{eq-proof-asymptotic equi}, we obtain
\begin{align*}
    &
    \big\|
    \mu^*(Z,\bX \con \what)
    - \mu^*(Z,\bX \con \omega^*)
    \big\|_{P,2}
    \lesssim
     \big\| \what(y,Z,\bX) - \omega^*(y,Z,\bX) \big\|_{P,2} \ .
     \numeq
     \label{eq-proof-aux1}
\end{align*}
Consequently, we have
\begin{align*}
    &
    \big\| 
    \widehat{\mu}\LSS (Z,\bX \con \what) - 
    \mu^*(Z,\bX \con \omega^*)
    \big\|_{P,2}
    \\
    &
    \leq 
    \big\| 
    \widehat{\mu}\LSS (Z,\bX \con \what) 
    - \mu^*(Z,\bX \con \what)
    \big\|_{P,2}
    + 
    \big\| 
    \mu^*(Z,\bX \con \what)
    - \mu^*(Z,\bX \con \omega^*)
    \big\|_{P,2}
    \\
    &
    \stackrel{\eqref{eq-proof-mu L2P},\eqref{eq-proof-aux1}}{
    \lesssim }
    \left[ 
    \begin{array}{l}            
    \| \widehat{f}\LSS(\potY{0} \cond A=0,Z,\bX)
    -
    f^*(\potY{0} \cond A=0,Z,\bX)\|_{P,2}
    \\
     + \| \widehat{f}\LSS(A=1 \cond Z,\bX) - f^*(A=1 \cond Z,\bX) \|_{P,2}
     \\
    +
    \big\| \what(y,Z,\bX) - \omega^*(y,Z,\bX) \big\|_{P,2}     
        \end{array}
    \right]
    \\
    &
    \stackrel{\eqref{eq-proof-omega hat}}{\lesssim}
    \left[ 
     \begin{array}{l}  
     \big\| \widehat{f}\LSS(Z=1 \cond \bX)
     -
     f^*(Z=1 \cond \bX)
     \big\|_{P,2}
     \\
     +
     \big\| \widehat{f}\LSS(A=1 \cond Z,\bX)
     -
     f^*(A=1 \cond Z,\bX)
     \big\|_{P,2}
     \\
     +
     \big\| 
     \widehat{f}\LSS(\potY{0} \cond A=0, Z,\bX)
     -
     f^* (\potY{0} \cond A=0, Z,\bX)
     \big\|_{P,2}
     \end{array}
     \right] 
        \ .
\end{align*}

\vspace*{0.25cm}

\noindent \HL{x} Let $
    \widehat{G}\LSS(\potY{0},Z,\bX) = \mathcal{G}^{(0)} - \widehat{\mu}
    \LSS(Z,\bX \con \mathcal{G}) $. Then, from straightforward algebra, we establish
\begin{align*}
   &
   \EXPk \big\{ \widehat{G}\LSS(\potY{0},Z,\bX) - \muhat(Z,\bX \con \widehat{G}\LSS ) \, \big| \, A=1, \potY{0}, \bX \big\}
    \\
    &
    =
    \sum_{z=0}^{1}
     \big\{
     \widehat{G}\LSS(\potY{0},z,\bX)  
     -\muhat(z,\bX \con \widehat{G}\LSS)
     \big\}   
     f^*(Z=z \cond A=1, \potY{0},\bX)
     \\
     &
     = 
     \sum_{z=0}^{1}
     \big\{
     \widehat{G}\LSS(\potY{0},z,\bX)  
     -\muhat(z,\bX \con \widehat{G}\LSS)
     \big\}   
     f^*(Z=z \cond A=1, \potY{0},\bX)
     \\
     &
     \qquad 
     -
     \underbrace{
     \widehat{\EXP} \LSS 
     \big\{
     \widehat{G}\LSS  
     -\muhat(Z,\bX \con \widehat{G}\LSS) 
     \cond A=1,\potY{0},\bX 
     \big\}    
     }_{=0}
     \\
    &
    =
    \sum_{z=0}^{1}
    \left[ 
        \begin{array}{l}
     \big\{
     \widehat{G}\LSS(\potY{0},z,\bX)  
     -\muhat(z,\bX \con \widehat{G}\LSS)
     \big\}
     \\
     \quad \times 
     \big\{
     f^*(Z=z \cond A=1, \potY{0},\bX) 
     -
     \widehat{f}\LSS (Z=z \cond A=1, \potY{0},\bX) 
     \big\}
        \end{array}
    \right]
     \ .
\end{align*}
In turn, this implies that 
\begin{align*}
    &
    \big\|
    \EXPk \big\{ \widehat{G}\LSS(\potY{0},Z,\bX) - \muhat(Z,\bX \con \widehat{G}\LSS ) \, \big| \, A=1, \potY{0}, \bX \big\}
    \big\|_{P,2}
    \\
    &    
    \lesssim  
    \sum_{z=0}^{1}
    \big\|
    f^*(Z=z \cond A=1,\potY{0},\bX)
    -
    \widehat{f}\LSS(Z=z \cond A=1,\potY{0},\bX)
    \big\|_{P,2}
    \\
    &
    \stackrel{\eqref{eq-proof-z given y rate}}{
    \lesssim 
    }
    \left[
        \begin{array}{l}
        \big\| \widehat{f}\LSS(A=1 \cond Z,\bX) - f^*(A=1 \cond Z,\bX)
        \big\|_{P,2}
        \\
        +
        \big\| \widehat{f}\LSS(\potY{0} \cond A=0,Z,\bX) - f^*(\potY{0} \cond A=0,Z,\bX)
        \big\|_{P,2}
        \end{array}
    \right]  \ .
\end{align*} 
This proves \eqref{eq-proof-GH convergence}.

Next, we show \eqref{eq-proof-GH convergence given Z}. Note that
\begin{align*}
    &
    \big\|
    \EXPk \big\{ 
     \widehat{G}\LSS(\potY{0},Z,\bX) 
     -\muhat(Z,\bX \con \widehat{G}\LSS)  
        \cond A=1,Z,\bX
        \big\}  
        \big\|_{P,2}
    \\
    &
    =
    \big\|
     \mu^*(Z,\bX \con \widehat{G}\LSS)  
        -
        \muhat(Z,\bX \con \widehat{G}\LSS)  
        \big\|_{P,2}
        \\
        &
        \stackrel{\eqref{eq-proof-mu L2P}}{\lesssim}
        \left[ 
    \begin{array}{l}            
    \| \widehat{f}\LSS(\potY{0} \cond A=0,Z,\bX)
    -
    f^*(\potY{0} \cond A=0,Z,\bX)\|_{P,2}
    \\
     + \| \widehat{f}\LSS(A=1 \cond Z,\bX) - f^*(A=1 \cond Z,\bX) \|_{P,2}
        \end{array}
    \right] \ .
\end{align*}

\vspace*{0.25cm}

\noindent \HL{xi} From the form of $\widetilde{\nu}\LSS$, we have 
\begin{align*}
    &
    \EXPk \big\{ 
    \widetilde{\nu}\LSS (\potY{0},Z,\bX)  \big\}
    \\
    &
    =
    \EXPk \big[
    \EXPk \big\{ 
    \widetilde{\nu}\LSS
    (\potY{0},Z,\bX) \cond \bX \big\}
    \big]
    \\
    &
    =
    \EXPk \Bigg[
    \begin{array}{l}
         \EXPk \big\{ 
    \nu \cond \bX \big\}
    -
    \EXPk \big[
    \widehat{\EXP}\LSS \big\{ \nu  \cond \potY{0},\bX \big\}
    \cond Z, \bX \big]
    \\
    -
    \EXPk \big[
    \widehat{\EXP}\LSS \big\{ \nu \cond Z,\bX \big\}
    \cond \potY{0},\bX \big]
    +
    \widehat{\EXP}\LSS \big\{ \nu \cond \bX \big\}
    \end{array} 
    \Bigg] \ .
    \numeq \label{eq-proof-omega 1}
\end{align*}
We establish that
\begin{align*}
    \EXPk \big\{ \nu  \cond \bX \big\} 
    &
    = 
    \iint \nu(y,z,\bX) \, 
    f^*(\potY{0}=y, Z=z \cond \bX) 
    \, d(yz)
    \\
    &
    = 
    \iint \nu(y,z,\bX) \, 
    f^*(\potY{0}=y \cond \bX) 
    f^*(Z=z \cond \bX) 
    \, d(yz)
    \ .
\end{align*}
Likewise, the other representations are equal to
\begin{align*}
    &
    \EXPk \big[
    \widehat{\EXP}\LSS \big\{ \nu   \cond \potY{0},\bX \big\}
    \cond Z, \bX \big] 
    \\
    &
    =
    \iint 
    \nu(y,z,\bX) 
    \widehat{f}\LSS(Z=z \cond  \bX)
    f^*(\potY{0}=y \cond \bX) 
    \, d(yz)
     \ ,
    \\
    &
    \EXPk \big[
    \widehat{\EXP}\LSS \big\{ \nu   \cond Z,\bX \big\}
    \cond \potY{0}, \bX \big] 
    \\
    &
    =
    \iint 
    \nu(y,z,\bX) 
    \widehat{f}\LSS(\potY{0}=y \cond \bX)
    f^*(Z=z \cond \bX)
    \, d(yz)
    \ .
\end{align*}
Therefore, \eqref{eq-proof-omega 1} is equal to 
\begin{align*}
    &
    \EXPk \big\{ 
    \widetilde{\nu}\LSS(\potY{0},Z,\bX) \big\}
    \\      
    &
    =
    \EXPk
    \left[ 
    \begin{array}{l}
     \iint \nu(y,z,\bX) \, 
    f^*(\potY{0}=y \cond \bX) 
    f^*(Z=z \cond \bX)
    \, d(yz)
    \\
    -
    \iint 
    \nu(y,z,\bX) \, 
    \widehat{f}\LSS(\potY{0}=y \cond \bX) 
    f^*(Z=z \cond \bX)
    d(yz)
    \\
    -
     \iint 
    \nu(y,z,\bX) \, 
    f^*(\potY{0}=y \cond \bX) 
    \widehat{f}\LSS(Z=z \cond \bX)
    \, d(yz) 
    \\
    +
    \iint 
    \nu(y,z,\bX) \, 
    \widehat{f}\LSS(\potY{0}=y \cond \bX) 
    \widehat{f}\LSS(Z=z \cond \bX)
    \, d (yz)
    \end{array}
    \right]  
    \\  
    &
    =
    \EXPk
    \bigg[
    \iint \nu(y,z,\bX) 
    \bigg[ 
        \begin{array}{l}
    \big\{ f^*(\potY{0}=y \cond \bX) - \widehat{f}\LSS(\potY{0}=y \cond \bX) \big\}
    \\
    \times 
    \big\{ f^*(Z=z \cond \bX) - 
    \widehat{f}\LSS(Z=z \cond \bX) \big\}
        \end{array}
    \bigg]
    \, d(yz)
    \bigg] 
     \ .
\end{align*}
Consequently,
\begin{align*}
    &
    \big\|
    \EXPk \big\{ 
    \widetilde{\nu}\LSS(\potY{0},Z,\bX)  \big\}
    \big\| 
    \\
    &
    \lesssim 
    \big\|  \widehat{f}\LSS(\potY{0} \cond \bX) - f^*(\potY{0} \cond \bX) \big\|_{P,2} 
    \big\|  \widehat{f}\LSS(Z \cond \bX) - f^*(Z \cond \bX) 
    \big\|_{P,2} 
    \\
    &
    \stackrel{\eqref{eq-proof-y0 rate}}{ \lesssim } 
    \left[
        \begin{array}{l}
        \big\| \widehat{f}\LSS(A=1 \cond Z,\bX) - f^*(A=1 \cond Z,\bX)
        \big\|_{P,2}
        \\
        +
        \big\| \widehat{f}\LSS(\potY{0} \cond A=0,Z,\bX) - f^*(\potY{0} \cond A=0,Z,\bX)
        \big\|_{P,2}
        \end{array}
    \right]
    \\
    &
    \quad \quad \times 
    \big\|  \widehat{f}\LSS(Z \cond \bX) - f^*(Z \cond \bX) \big\|_{P,2} 
    \ .
\end{align*}

\end{proof}

 \newpage 
 
 \section{Proofs of the Theorems} \label{sec-supp-proof}

We consider general $Z$ with support $\suppZ \subset \R$, unless the binary IV condition is specifically required. We establish some notational conventions. For nonnegative quantities $(Q_0,Q_1,\ldots,Q_J)$, we write $Q_0 \lesssim \sum_{j=1}^{J} Q_j$ to indicate that $Q_0 \leq \sum_{j=1}^{J} C_j Q_j$ for some nonnegative constants $C_1,\ldots,C_J$. Let $Q_1 \asymp Q_2$ indicate $Q_1 \lesssim Q_2$ and $Q_2 \lesssim Q_1$. Finally, we let $\gamma(Z,\bX) = f(A=1 \cond Z,\bX)/f(A=0 \cond Z, \bX)$ denote the treatment odds conditional on $(Z,\bX)$.

\subsection{Proof of Theorem \ref{thm-Psi-binary}}

Theorem \ref{thm-Psi-binary} follows directly from Theorem \ref{thm-Psi}, since it corresponds to the special case where $Z$ is binary. The proof of Theorem \ref{thm-Psi} is presented in the next Section.

\subsection{Proof of Theorem \ref{thm-Psi}}

We fix an arbitrary value of $\bX \in \suppX$.  We omit the dependence on $\theta$ in the notation for brevity.
By Lemma \ref{lemma-beta and gamma}, we find that $\beta(z_1,\bX) \neq \beta(z_0,\bX)$ if $z_1 \in \suppZ_{z_0,\bX}^c$.

We begin by proving the existence of a solution to the fixed-point equation. Given the distribution of the full data $(\potY{0}, A, Z \mid \bX)$, we introduce the following shorthand notation:
\begin{align*}
    \numeq 
    \label{eq-proof-D}
    D(Z,\bX) 
    &
    = \EXP \big\{ \alpha(Y,\bX) \cond 
    A=0,Z,\bX \con \fy \big\} 
    \ , 
    \\
    R_{Y}(y,\bX)
    &
    =
    \frac{ \ff{}(Y=y \cond A=0,Z=z_{1},\bX) }{ \ff{}(Y=y \cond A=0,Z=z_{0},\bX) }
    \ , 
    \\
    R_D(\bX)
    &
    =
    \frac{ D(z_{1}, \bX) } { D(z_{0}, \bX) } \ ,
    \numeq 
    \label{eq-proof-RD}
    \\
    g_{z_{0}}(y, \bX) 
    &
    = \frac{ \alpha(y,\bX) \ff{}(Y=y \cond A=0,Z=z_{0},\bX) }{ D(z_{0},\bX) } \ ,
    \numeq 
    \label{eq-proof-z0}
    \\
    g_{z_{1}}(y, \bX) &    
    = \frac{ \alpha(y,\bX) \ff{}(Y=y \cond A=0,Z=z_{1},\bX) }{ D(z_{1},\bX) }
    = \frac{ 
    g_{z_{0}}(y, \bX)
    R_{Y}(y,\bX) }
    {R_D(\bX)} \  . 
\end{align*}
For $g_{z}(y,\bX)$, we find 
\begin{align*}
    &
    \numeq \label{eq-proof-integral 0}
    g_{z}(y, \bX) > 0 \ , \quad \forall y \in \mathcal{Y}_{\bX} , \ z \in \{z_0,z_1\}
    \\
    &
    \int g_{z}(y, \bX) \, d y = 1 \ ,  \quad  z \in \{z_0,z_1\}
    \ , 
    \numeq \label{eq-proof-integral 1}
    \\
    &
    \int g_{z_{0}}(y, \bX) R_{Y}(y,\bX) \, d y = R_D(\bX) \ .    
    \numeq
    \label{eq-proof-integral 2}
\end{align*}
Of note, the construction of $g_{z}(y,\bX)$ is not sensitive to the boundary condition $\alpha(y_R,\bX)=1$. Specifically, if we take $\alpha'(y,\bX) = \alpha(y,\bX) c(\bX)$ for some function $c(\bX)$, we find
\begin{align*}
    g_{z}(y,\bX \con \alpha', \fy)
    & =
    \frac{ \alpha'(y,\bX) \ff{}(Y=y \cond A=0,Z=z,\bX) }{ \EXP \big\{ \alpha'(Y,\bX) \cond 
    A=0,Z,\bX \con \fy \big\}  }
    \\
    & = 
    \frac{ \alpha(y,\bX) \ff{}(Y=y \cond A=0,Z=z,\bX) }{ D(z, \bX) }
    =
    g_{z} (y,\bX \con \alpha, \fy) \ .
\end{align*}
Therefore,  the boundary condition $\alpha(y_R,\bX)=1$ can be imposed without loss of generality and may be ignored when constructing $g_{z}$. 

We find that $\alpha(y,\bX)$ can be represented as a function proportional to $g_{z_0}(y,\bX)/\ff{}(Y=y \cond A=0,Z=z_0,\bX)$ with $\alpha(y_R,\bX)=1$, implying that
\begin{align*}
    \alpha(y,\bX)
    =
    \frac{ g_{z_0}(y,\bX) }{ g_{z_0}(y_R,\bX) }
    \frac{ \ff{}(Y=y_R \cond A=0,Z=z_0,\bX) }{\ff{}(Y=y \cond A=0,Z=z_0,\bX)}
    \ , \quad y \in \suppYX \ .
    \numeq \label{eq-proof-alpha g}
\end{align*}
In addition, $\beta(z,\bX)$ is expressed as follows from \eqref{eq-proof-gamma}:
\begin{align*}
    \beta(z,\bX) =  
    \frac{  \ff{}(A=1 \cond Z=z,\bX) / \ff{}(A=0 \cond Z=z,\bX) }{ \EXP \{ \alpha(Y,\bX) \cond A=0,Z=z,\bX \con \fy \} } \ , \quad z \in \suppZ \ .
     \numeq \label{eq-proof-beta g}
\end{align*}
Therefore, it now remains to show that $g_{z_0}$ is the unique fixed point of the mapping $\Psi$.

We first show that $g_{z_0}$ is a fixed point of $\Psi$, thereby establishing the existence of a solution to the fixed-point equation. From \eqref{eq-proof-Y0Z-gamma}, we have
\begin{align*}
    1 
    & 
    \stackrel{\text{\HL{IV2}}}{=}
    \frac{
    \ff{}(\potY{0}=y \cond Z=z_1, \bX)
    }{
    \ff{}(\potY{0}=y \cond Z=z_0, \bX)
    }
    \\
    &
    \stackrel{\eqref{eq-proof-Y0Z-gamma}}{=}
    \frac{ \displaystyle{
    \bigg\{ 1+  
    \frac{ 
    \alpha(y,\bX)
    \gamma(z_1,\bX) }{ D(z_1,\bX) }
    \bigg\}
    \frac{
    \ff{}(Y=y \cond A=0,Z=z_1,\bX)}{1+\gamma(z_1,\bX)}
    } }{
    \displaystyle{ 
     \bigg\{ 1+  
    \frac{ 
    \alpha(y,\bX)
    \gamma(z_0,\bX) }{ D(z_0,\bX) }
    \bigg\}
    \frac{
    \ff{}(Y=y \cond A=0,Z=z_0,\bX)}{1+\gamma(z_0,\bX)}
    } }
    \\
    &
    =  
    \frac{ \displaystyle{
    D(z_0,\bX)
    \big\{
    D(z_1,\bX) 
    + 
    \alpha(y,\bX)
    \gamma(z_1,\bX)
    \big\}
    \ff{}( Y=y, A=0 \cond Z=z_1,\bX)
    } }{
    \displaystyle{ 
    D(z_1,\bX)
    \big\{
    D(z_0,\bX) 
    + \alpha(y,\bX)
    \gamma(z_0,\bX)
    \big\}
     \ff{}( Y=y, A=0 \cond Z=z_0,\bX)
    }} \ .
\end{align*}
The last line implies  
\begin{align*}
    &
    \alpha(y,\bX) 
    \bigg\{ 
    \begin{array}{l}
    D(z_0,\bX) \gamma(z_1,\bX) \ff{}( Y=y, A=0 \cond Z=z_1,\bX)
    \\
    - 
    D(z_1,\bX) \gamma(z_0,\bX) \ff{}( Y=y, A=0 \cond Z=z_0,\bX)
    \end{array}
     \bigg\}
     \\
     &
     =
     \bigg\{ 
     \begin{array}{l}
     \ff{}( Y=y, A=0 \cond Z=z_0,\bX) 
     \\
     - \ff{}( Y=y, A=0 \cond Z=z_1,\bX) 
     \end{array}\bigg\}
     D(z_1,\bX)
     D(z_0,\bX) 
     \\
     \Rightarrow \quad 
    &
    \left\{
    \begin{array}{l}
    \displaystyle{
    \gamma(z_1,\bX) \frac{ \ff{}( Y=y, A=0 \cond Z=z_1,\bX) }{ D(z_1,\bX) }
    }
    \\
    \displaystyle{
    - 
    \gamma(z_0,\bX) \frac{ \ff{}( Y=y, A=0 \cond Z=z_0,\bX) }{ D(z_0,\bX) }
    }
    \end{array}
     \right\}
    \alpha(y,\bX) 
     \\
     &
     =
     \ff{}( Y=y,A=0 \cond Z=z_0,\bX) 
     -\ff{}( Y=y,A=0 \cond Z=z_1,\bX)
     \\[0.25cm]
     \Rightarrow
     \quad
     &
    \left\{
    \begin{array}{l}
    \displaystyle{
    \ff{}(A=1 \cond Z=z_1,\bX)
    \frac{ \ff{}(Y=y \cond A=0, Z=z_1,\bX) }{ D(z_1,\bX) }
    }
    \\
    -
    \displaystyle{
    \ff{}(A=1 \cond Z=z_0,\bX)
    \frac{ \ff{}(Y=y \cond A=0, Z=z_0,\bX) }{ D(z_0,\bX) }
    }
    \end{array}
    \right\}
    \alpha(y,\bX) 
     \\
     &
     =
     \left\{
     \begin{array}{l}     
     \ff{}(A=0 \cond Z=z_0,\bX) \ff{}(Y=y \cond A=0,Z=z_0,\bX)
     \\
     - 
     \ff{}(A=0 \cond Z=z_1,\bX) \ff{}(Y=y \cond A=0,Z=z_1,\bX)
     \end{array}
     \right\}
     \\[0.25cm]
     \Rightarrow
     \quad
     & 
    \bigg\{  
    \ff{}(A=1 \cond Z=z_1,\bX) 
    \frac{R_Y(y,\bX)}{R_D(\bX)}
    -
    \ff{}(A=1 \cond Z=z_0,\bX) 
    \bigg\}
    g_{z_0}(y,\bX) 
     \\
     &
     =
     \left\{
     \begin{array}{l}     
     \ff{}(A=0 \cond Z=z_0,\bX) \ff{}(Y=y \cond A=0,Z=z_0,\bX)
     \\
     - 
     \ff{}(A=0 \cond Z=z_1,\bX) \ff{}(Y=y \cond A=0,Z=z_1,\bX)
     \end{array}
     \right\}
     \\
     \Rightarrow \quad 
     &
    \bigg\{
    \begin{array}{l}         
    \ff{}(A=1 \cond Z=1,\bX)
    R_Y(y,\bX)
    \\
    -
    \ff{}(A=1 \cond Z=0,\bX)
    R_D(\bX)
    \end{array}
    \bigg\} 
    g_{z_0}(y,\bX) 
    \\
    &
    =
    \bigg\{
    \begin{array}{l}         
    \ff{}(A=0 \cond Z=z_0,\bX)
    \ff{}(Y=y \cond A=0,Z=z_0,\bX) 
    \\
    -
    \ff{}(A=0 \cond Z=z_1,\bX)
    \ff{}(Y=y \cond A=0,Z=z_1,\bX)
    \end{array}
    \bigg\}
    R_D(\bX) \ .
    \numeq \label{eq-proof-g0definition}
     \end{align*} 
 The last line \eqref{eq-proof-g0definition} can be expressed as
\begin{align}
    &
    g_{z_0}(y,\bX)
    =
    \Psi \big( g_{z_0}(y,\bX) \big)  
    \ ,
    \label{eq-proof-fixed point}
\end{align} 
where $\Psi$ is the mapping defined in the Theorem statement, i.e., 
\begin{align*}
    &
    \Psi\big( g(y,\bX) \big)
    \\
    &
    =
    \frac{
    \left\{
    \begin{array}{l}
    \ff{}(A=1 \cond Z=z_0,\bX)
    g(y,\bX)
    \\
    + \ff{}(A=0 \cond Z=z_0,\bX)
    \ff{}(Y=y \cond A=0,Z=z_0,\bX)
    \\
    -
    \ff{}(A=0 \cond Z=z_1,\bX)
    \ff{}(Y=y \cond A=0,Z=z_1,\bX)
    \end{array} 
    \right\}
    \!
    \displaystyle{
    \int \!
    g(t,\bX)
    R_{Y}(t,\bX) \, d t
    }
    }
    {
    \ff{}(A=1 \cond Z=z_1,\bX)
    R_{Y}(y,\bX)
    }    \ . 
\end{align*}
Therefore, $g_{z_0}$ in \eqref{eq-proof-z0} is the fixed point of the mapping $\Psi$.

We now show that the fixed-point equation \eqref{eq-proof-fixed point} has a unique solution via a contradiction argument. Suppose that there exist two distinct solutions $g' \equiv g_{z_0}$ in \eqref{eq-proof-z0} and $g^\dagger$ solving \eqref{eq-proof-fixed point}. Accordingly, let $(\alpha',\beta')$ and $(\alpha^\dagger, \beta^\dagger)$ be the functions defined by \eqref{eq-proof-alpha g} and \eqref{eq-proof-beta g} associated with $g'$ and $g^\dagger$, respectively. We first note that the fixed point of $\Psi$ must satisfy equation \eqref{eq-proof-integral 1}, i.e., $\int g(y, \bX) \, dy = 1$, as shown below: 
\begin{align*}
    &
    g = \Psi(g)
    \\
    \Rightarrow
    \quad 
    &
    \bigg\{
    \begin{array}{l}         
    \ff{}(A=1 \cond Z=z_1,\bX)
    R_Y(y,\bX)
    \\
    -
    \ff{}(A=1 \cond Z=z_0,\bX)
    \int g(t,\bX) R_Y(t,\bX) \, dt
    \end{array}
    \bigg\} 
    g(y,\bX) 
    \\
    &
    =
    \bigg\{
    \begin{array}{l}         
    \ff{}(A=0 \cond Z=z_0,\bX)
    \ff{}(Y=y \cond A=0,Z=z_0,\bX) 
    \\
    -
    \ff{}(A=0 \cond Z=z_1,\bX)
    \ff{}(Y=y \cond A=0,Z=z_1,\bX)
    \end{array}
    \bigg\}
    \int g(t,\bX) R_Y(t,\bX) \, dt
    \\
    \Rightarrow
    \quad 
    &
    \bigg\{
    \ff{}(A=1 \cond Z=z_1,\bX) 
    -
    \ff{}(A=1 \cond Z=z_0,\bX)
    \int g(y,\bX) \, dy  
    \bigg\} \int g(t,\bX) R_Y(y,\bX) \, dt
    \\
    &
    =
    \big\{
    \ff{}(A=0 \cond Z=z_0,\bX)
    -
    \ff{}(A=0 \cond Z=z_1,\bX)
    \big\}
    \int g(t,\bX) R_Y(y,\bX) \, dt
    \\
    \Rightarrow
    \quad 
    &
    \ff{}(A=1 \cond Z=z_0,\bX) 
    \int g(y,\bX) \, dy  
    =
    \ff{}(A=1 \cond Z=z_0,\bX) 
    \\
    \Rightarrow
    \quad 
    &
    \int g(y,\bX) \, dy  
    =
    1 \ .
\end{align*}

Under the assumption of multiple solutions, we consider the following two possible scenarios:

\vspace*{0.25cm}

\noindent \textbf{(Case 1)}: Suppose that $R_D$ in \eqref{eq-proof-integral 2} is the same for both $g'$ and $g^\dagger$, i.e., 
    \begin{align*}
        R_{D}(\bX)
        =
        \int g'(y, \bX) R_{Y}(y,\bX) \, d y = 
        \int g^\dagger(y, \bX) R_{Y}(y,\bX) \, d y \ .
    \end{align*}
Since both $g'$ and $g^\dagger$ satisfy \eqref{eq-proof-fixed point}, they must each satisfy \eqref{eq-proof-g0definition} as well, i.e.,
\begin{align*}
    &
    \bigg\{
    \begin{array}{l}         
    \ff{}(A=1 \cond Z=z_1,\bX)
    R_Y(y,\bX)
    \\
    -
    \ff{}(A=1 \cond Z=z_0,\bX)
    R_D(\bX)
    \end{array}
    \bigg\} 
    g'(y,\bX) 
    \\
    &
    =
    \bigg\{
    \begin{array}{l}         
    \ff{}(A=0 \cond Z=z_0,\bX)
    \ff{}(Y=y \cond A=0,Z=z_0,\bX) 
    \\
    -
    \ff{}(A=0 \cond Z=z_1,\bX)
    \ff{}(Y=y \cond A=0,Z=z_1,\bX)
    \end{array}
    \bigg\}
    R_D(\bX)
    \\
    &
    =
    \bigg\{
    \begin{array}{l}         
    \ff{}(A=1 \cond Z=z_1,\bX)
    R_Y(y,\bX)
    \\
    -
    \ff{}(A=1 \cond Z=z_0,\bX)
    R_D(\bX)
    \end{array}
    \bigg\} 
    g^\dagger(y,\bX) \ .
    \numeq 
    \label{eq-proof-unique RD}
\end{align*}
Note that the common coefficient of $g'$ and $g^\dagger$ is given by
\begin{align*}
    &
    \ff{}(A=1 \cond Z=z_1,\bX)
    R_Y(y,\bX)
    -
    \ff{}(A=1 \cond Z=z_0,\bX)
    R_D(\bX)
    \\
    &
    =
    \ff{}(A=1 \cond Z=z_1,\bX)
    \frac{ \ff{}(Y=y \cond A=0,Z=z_{1},\bX) }{ \ff{}(Y=y \cond A=0,Z=z_{0},\bX) }
    \\
    &
    \quad \quad 
    -
    \ff{}(A=1 \cond Z=z_0,\bX)
    \frac{ \EXP\{ \alpha' (Y,\bX) \cond A=0,Z=z_1,\bX \con \fy \} }
    { \EXP\{ \alpha'(Y,\bX) \cond A=0,Z=z_0,\bX \con \fy \} } 
    \\
    &
    \stackrel{\eqref{eq-proof-gamma}}{=}
    \ff{}(A=1 \cond Z=z_1,\bX)
    \frac{ \ff{}(Y=y \cond A=0,Z=z_{1},\bX) }{ \ff{}(Y=y \cond A=0,Z=z_{0},\bX) }
    \\
    &
    \quad \quad 
    -
    \ff{}(A=1 \cond Z=z_0,\bX)
    \frac{ \ff{}(A=1 \cond Z=z_1,\bX) / \ff{}(A=0 \cond Z=z_1,\bX) }
    { \ff{}(A=1 \cond Z=z_0,\bX) / \ff{}(A=0 \cond Z=z_0,\bX) } 
    \frac{ \beta'(z_0,\bX) }
    { \beta'(z_1,\bX) } 
    \\
    &
    =
    \ff{}(A=1 \cond Z=z_1,\bX)
    \bigg\{
    \frac{ \ff{}(Y=y \cond A=0,Z=z_{1},\bX) }{ \ff{}(Y=y \cond A=0,Z=z_{0},\bX) }
    -
    \frac{\ff{}(A=0 \cond Z=z_0,\bX)}{\ff{}(A=0 \cond Z=z_1,\bX)}
    \frac{\beta'(z_0,\bX)}{\beta'(z_1,\bX)}
    \bigg\} 
    \\
    &
    =
    \ff{}(A=1 \cond Z=z_0,\bX)
    \bigg\{
    \frac{ \ff{}(Y=y,A=0 \cond Z=z_{1},\bX) }{ \ff{}(Y=y, A=0 \cond Z=z_{0},\bX) }
    -
    \frac{\beta'(z_0,\bX)}{\beta'(z_1,\bX)}
    \bigg\} 
    \\
    &
    \stackrel{\eqref{eq-proof-Y0AZ}}{=} 
    \ff{}(A=1 \cond Z=z_0,\bX)
    \bigg\{
    \frac{ 1+\alpha'(y,\bX) \beta'(z_0,\bX) }{ 1+\alpha'(y,\bX) \beta'(z_1,\bX) }
    -
    \frac{\beta'(z_0,\bX)}{\beta'(z_1,\bX)}
    \bigg\} 
    \\
    &
    =
    \ff{}(A=1 \cond Z=z_0,\bX)
    \frac{ \beta'(z_1,\bX) - \beta'(z_0,\bX) }
    { 
    \{ 1+\alpha'(y,\bX) \beta'(z_1,\bX) \} \beta'(z_1,\bX)
    } 
    \neq 0 \ .  
\end{align*}
The last equality holds from Lemma \ref{lemma-beta and gamma}, implying that $\beta'(z_1,\bX) \neq \beta'(z_0, \bX)$ under the relevance assumption \HL{IV3}.  Therefore, \eqref{eq-proof-unique RD} implies that $g' = g^\dagger$, contradicting the assumption that the two solutions are distinct. Therefore, the case where $g'$ and $g^\dagger$ yield the same $R_D$ is impossible. Hence, $R_D$ associated with $g'$ and $g^\dagger$ must differ.
 
\vspace*{0.25cm}

\noindent \textbf{(Case 2)}: From the observation in Case 1, it suffices to consider the scenario in which $R_D$ in \eqref{eq-proof-integral 2} differs for $g'$ and $g^\dagger$, denoted by $R_{D}'(\bX)$ and $R_{D}^\dagger(\bX)$, respectively. Also, let $D'(Z,\bX)$ and $D^\dagger(Z,\bX)$ be the expected values of the odds ratio function in \eqref{eq-proof-D}, i.e.,
\begin{align*}
    D'(Z,\bX)
    &
    = \EXP \big\{ \alpha' (y,\bX) \cond A=0,Z,\bX \con \fy \big\}
    \\
    D^\dagger (Z,\bX)
    &
    = \EXP \big\{ \alpha^\dagger (y,\bX) \cond A=0,Z,\bX \con \fy \big\}
\end{align*} 
Note that $R_D'(\bX) = D'(z_1,\bX)/D'(z_0,\bX)$ and $R_D^\dagger(\bX) = D^\dagger(z_1,\bX)/D^\dagger(z_0,\bX)$.

From \eqref{eq-proof-Y0Z-gamma} and \HL{IV2}, we find the induced density of $\potY{0} \cond (Z, \bX)$ is given by 
\begin{align*}
    &
    \bigg\{ 1 + \frac{ \alpha^{\#}(y,\bX) }{ D^{\#}(z_1,\bX) }
    \frac{\ff{}(A=1 \cond Z=z_1,\bX)}{\ff{}(A=0 \cond Z=z_1,\bX)}
    \bigg\}
    \ff{} (Y=y,A=0 \cond Z=z_1, \bX)
    \\
    &
    \stackrel{\eqref{eq-proof-Y0Z-gamma}}{=}
    \ff{}^{\#} (\potY{0}=y \cond Z=z_1, \bX)
    \\
    &
    \stackrel{\text{\HL{IV2}}}{=}
    \ff{}^{\#} (\potY{0}=y \cond Z=z_0, \bX)
    \\
    &
    \stackrel{\eqref{eq-proof-Y0Z-gamma}}{=}
    \bigg\{ 1 + \frac{ \alpha^{\#}(y,\bX) }{ D^{\#}(z_0,\bX) }
    \frac{\ff{}(A=1 \cond Z=z_0,\bX)}{\ff{}(A=0 \cond Z=z_0,\bX)}
    \bigg\}
    \ff{} (Y=y,A=0 \cond Z=z_0, \bX)
    \ , \quad \# \in \{ ' , \dagger \} \ .
\end{align*}
By letting $\gamma_{z} = \ff{}(A=1 \cond Z=z,\bX)/ \ff{}(A=0 \cond Z=z,\bX)$, we have 
\begin{align*}
&
    \frac{
    \displaystyle{
     1 + \frac{ \alpha'(y,\bX) }{ D'(z_1,\bX) } \gamma_{z_1}
    }
    }{
    \displaystyle{
     1 + \frac{ \alpha'(y,\bX) }{ D'({z_0},\bX) }
     \gamma_{z_0}
    }
    }
    =
    \frac{ f(Y=y,A=0 \cond Z=z_0,\bX) }{ f(Y=y,A=0 \cond Z=z_1,\bX) }
    = 
    \frac{
    \displaystyle{
     1 + \frac{ \alpha^\dagger(y,\bX) }{ D^\dagger({z_1},\bX) } 
     \gamma_{z_1}
    }
    }{
    \displaystyle{
     1 + \frac{ \alpha^\dagger(y,\bX) }{ D^\dagger({z_0},\bX) }
     \gamma_{z_0}
    }
    }
    \\
    \Rightarrow \quad &    
     \bigg\{  1 + \frac{ \alpha'(y,\bX) }{ D'({z_1},\bX) } 
     \gamma_{z_1}
     \bigg\}
    \bigg\{ 
    1 + \frac{ \alpha^\dagger(y,\bX) }{ D^\dagger({z_0},\bX) }
     \gamma_{z_0}
     \bigg\}
     =
     \bigg\{ 1 + \frac{ \alpha'(y,\bX) }{ D'({z_0},\bX) }
     \gamma_{z_0} \bigg\}
     \bigg\{1 + \frac{ \alpha^\dagger(y,\bX) }{ D^\dagger({z_1},\bX) } 
     \gamma_{z_1}
     \bigg\}
 \\
    \Rightarrow \quad &    
    \frac{ \alpha'(y,\bX) }{ D'({z_1},\bX) } \gamma_{z_1}
    +
    \frac{ \alpha^\dagger(y,\bX) }{ D^\dagger({z_0},\bX) }
     \gamma_{z_0}
     +
     \frac{ \alpha'(y,\bX) }{ D'({z_1},\bX) } \frac{ \alpha^\dagger(y,\bX) }{ D^\dagger({z_0},\bX) }
    \gamma_1     
     \gamma_0
     \\
     &
     \quad \quad 
     =  
     \frac{ \alpha^\dagger(y,\bX) }{ D^\dagger({z_1},\bX) } 
     \gamma_{z_1}
     +
      \frac{ \alpha'(y,\bX) }{ D'({z_0},\bX) }
     \gamma_{z_0}
     +
      \frac{ \alpha'(y,\bX) }{ D'({z_0},\bX) } 
     \frac{ \alpha^\dagger(y,\bX) }{ D^\dagger({z_1},\bX) } 
     \gamma_{z_0}
     \gamma_{z_1}  \ .
     \numeq
     \label{eq-proof-unique 1} 
\end{align*} 
Let $\bar{\alpha}^{\#}$ for ${\#} \in \{' , \dagger \}$ be
\begin{align*}
\bar{\alpha}^{\#} (y, \bX)
= 
\frac{ \alpha^{\#} (y, \bX)}{ D^{\#}(z_0, \bX)  } \ ,
\end{align*}
which satisfies $\EXP \{ \bar{\alpha}^{\#} (y, \bX) \cond A=0,Z=z_0, \bX \con \fy \} = 1$. Then, noting that  $R_{D}^{\#} (\bX) = D^{\#}(z_1,\bX)/ D^{\#}(z_0,\bX)$, expression \eqref{eq-proof-unique 1} becomes
\begin{align*}
    &
    \frac{ \bar{\alpha}'(y,\bX) }{ R_D'(\bX) } \gamma_{z_1}
    +
    \bar{\alpha}^\dagger(y,\bX)
     \gamma_{z_0}
     +
     \frac{ 
     \bar{\alpha}'(y,\bX)
     \bar{\alpha}^\dagger(y,\bX)
     }{ R_D'(\bX) }
     \gamma_{z_0} 
     \gamma_{z_1}
     \\
     &
     =  
     \frac{ \bar{\alpha}^\dagger(y,\bX) }{ R_D^\dagger(\bX) } 
     \gamma_{z_1}
     +
    \bar{\alpha}'(y,\bX)
     \gamma_{z_0}
     +
     \frac{ 
     \bar{\alpha}'(y,\bX)
     \bar{\alpha}^\dagger(y,\bX)
     }{ R_D^\dagger(\bX) }
     \gamma_{z_0} 
     \gamma_{z_1} \ .
\end{align*}
We integrate both sides after multiplying by $\ff{}(Y=y \cond A=0, Z=z_0,\bX)$, which results in 
\begin{align*}
    &
    \frac{ \gamma_{z_1} }{ R_D'(\bX) } 
     +
     \frac{      
     M_{12}
     \gamma_{z_1} 
     \gamma_{z_0}
     }{ R_D'(\bX) }
     =  
     \frac{ \gamma_{z_1} }{ R_D^\dagger(\bX) } 
     +
     \frac{ 
     M_{12}
     \gamma_{z_0}
     \gamma_{z_1}
     }{ R_D^\dagger(\bX) } 
     \\
     \Rightarrow \quad 
     &
     \gamma_{z_1}
     \big(
     M_{12}
     \gamma_{z_0}
     +
     1
     \big)
     \bigg\{ 
     \frac{1}{ R_D'(\bX) }
     -
     \frac{1}{ R_D^\dagger(\bX) }
     \bigg\} 
     =
     0 \ .
     \numeq \label{eq-proof-unique 2}
\end{align*}
where $M_{12} =  \EXP \{ \bar{\alpha}'(y,\bX) \bar{\alpha}^\dagger(y,\bX) \cond A=0, Z=z_0,\bX \con \fy \} $. Since $R_D'(\bX) \neq R_D^\dagger(\bX)$ by the assumption, \eqref{eq-proof-unique 2} implies $M_{12} = - 1 / \gamma_{z_0}<0$. This, in turn, implies that $\alpha'(y,\bX) $ and $\alpha^\dagger(y,\bX)$ must have opposite signs for some $y$. Without loss of generality, suppose that $\alpha^\dagger(y,\bX)$ is negative for some $y$. By \eqref{eq-proof-alpha g}, this implies that $g^\dagger(y,\bX)$ and $g^\dagger(y_R,\bX)$ have opposite signs, so at least one of them is negative. This violates the positivity condition of $g$, i.e., \eqref{eq-proof-integral 0}, and therefore violates the defining property of $g$. Consequently, there cannot exist multiple nonnegative solutions to the fixed-point equation \eqref{eq-proof-fixed point}.

Finally, we show that the choices of $z_1 \in \suppZ_{z_0,\bX}^c$ and $z_0\in \suppZ$ do not affect Results (i) and (ii) of the Theorem, respectively. Consider two mappings $\Psi(g \con z_1^{(1)})$ and $\Psi(g \con z_1^{(2)})$, corresponding to different choices of $( z_1^{(1)}, z_1^{(2)} ) \in \suppZ_{z_0,\bX}^c$. Since the respective fixed points are unique, denoted by $g_{z_0}^{(1)} = \Psi(g_{z_0}^{(1)} \con z_1^{(1)})$ and $g_{z_0}^{(2)} = \Psi(g_{z_0}^{(2)} \con z_1^{(2)})$, we have
\begin{align*}
    g_{z_0}^{(1)}(y,\bX)
    \stackrel{\eqref{eq-proof-z0}}{=} \frac{ \alpha(y,\bX) \ff{}(Y=y \cond A=0,Z=z_{0},\bX) }{ D(z_{0},\bX) }
    \stackrel{\eqref{eq-proof-z0}}{=}
    g_{z_0}^{(2)}(y,\bX) \ .
\end{align*}
This implies that the choice of $z_1$ does not affect Result (i). Similarly, from \eqref{eq-proof-z0}, $\alpha(y,\bX)$ is always proportional to $g_{z_0}(y,\bX)/f(Y=y \cond A=0,Z=z_0,\bX)$ regardless of the choice of $z_0$, i.e.,
\begin{align*} 
    \alpha(y,\bX)
    \stackrel{\eqref{eq-proof-z0}}{=}
    \frac{ g_{z_0} (y,\bX) }{ g_{z_0}  (y_R,\bX) }
    \frac{ \ff{}(Y=y_R \cond A=0,Z=z_0,\bX) }{\ff{}(Y=y \cond A=0,Z=z_0,\bX)}
     \ , \quad \forall z_0 \in \suppZ \ .
\end{align*}
This implies that the choice of $z_0$ does not affect Result (ii).

This concludes the proof.

\subsection{Proof of Theorem \ref{thm-global contraction-binary}}

Theorem \ref{thm-global contraction-binary} follows directly from Theorem \ref{thm-global contraction}, since it corresponds to the special case where $Z$ is binary. The proof of Theorem \ref{thm-global contraction} is presented in the next Section.

\subsection{Proof of Theorem \ref{thm-global contraction}} \label{sec-supp-global contraction}

We fix an arbitrary value of $\bX \in \suppX$. We omit the dependence on $\theta$ in the notation for brevity. By Lemma \ref{lemma-beta and gamma}, we find that $z_0 \in \arginf_{z \in \suppZ} \beta(z,\bX)$ if $z_0 \in \suppZinf$. Also, from the same Lemma, we also obtain $\beta(z_1,\bX) > \beta(z_0,\bX)$ if $z_1 \in \suppZinf^c$. 

We introduce the following shorthand notation:
\begin{align*}
    \mathcal{C}_1(\bX) 
    &     
    = \frac{ f(A=1 \cond Z=z_0,\bX) }{ f(A=1 \cond Z=z_1,\bX)} \ ,
    \\
    \mathcal{C}_2(y,\bX) 
    & 
    = 
    \frac{ 
    \displaystyle{ 
    \left\{
    \begin{array}{l}
    f(A=0 \cond Z=z_0,\bX)
    f(Y=y \cond A=0,Z=z_0,\bX)
    \\
    -
    f(A=0 \cond Z=z_1,\bX)
    f(Y=y \cond A=0,Z=z_1,\bX)
    \end{array} 
    \right\}
    }
    }{f(A=1 \cond Z=z_1,\bX)} 
    \\
    & 
    = 
    \frac{ f(Y=y,A=0 \cond Z=z_0,\bX) - f(Y=y,A=0 \cond Z=z_1,\bX) }{f(A=1 \cond Z=z_1,\bX)}
    \ , 
    \\
    \lambda(y,\bX) 
    & = 
    \frac{ R_D(\bX) \mathcal{C}_1(\bX) }{  R_Y(y,\bX)     } \ .
\end{align*}
Note that $\mathcal{C}_1(\bX) > 0$ and $\mathcal{C}_2(y,\bX)>0$ for all $y \in \mathcal{Y}_{\bX}$ from \eqref{eq-proof-fya monotone}. We also remark that $\alpha$ and $\beta$ are uniformly bounded from Lemma \ref{lemma-boundedness}, i.e., for some constants $0<c_\alpha \leq C_{\alpha} < \infty$ and $0<c_\beta \leq C_{\beta} < \infty$, we have 
\begin{align*}
    &
    \alpha(y,\bX) \in [c_{\alpha},C_{\alpha}]
    \ , \quad 
    \beta(z,\bX) \in [c_{\beta},C_{\beta}]
    \ , \quad \forall (y,z,\bX) \ .
\end{align*}

For later use, we derive a bound for $\lambda$. From \eqref{eq-proof-Y0AZ}, we have  
    \begin{align*}
    & 
    R_Y(y,\bX)
    \numeq \label{eq-proof-RYrange}
    \\
    &
    =
    \frac{  1 + \alpha(y,\bX) \beta(z_0,\bX)  }{ 1 + \alpha(y,\bX) \beta(z_1,\bX) }
    \frac{ \ff{}(A=0 \cond Z=z_0, \bX) }
    { \ff{}(A=0 \cond Z=z_1, \bX)  }
    \\
    &
    \in 
    \Bigg[ 
        \frac{ 1+ C_{\alpha} \beta(z_0,\bX) }{ 1+ C_{\alpha} \beta(z_1,\bX)}\frac{ \ff{}(A=0 \cond Z=z_0, \bX) }
    { \ff{}(A=0 \cond Z=z_1, \bX)  }
    ,
    \frac{ 1+ c_\alpha \beta(z_0,\bX) }{ 1+ c_\alpha \beta(z_1,\bX)}\frac{ \ff{}(A=0 \cond Z=z_0, \bX) }
    { \ff{}(A=0 \cond Z=z_1, \bX)  }
    \Bigg] \ .
\end{align*}
From the definition of $R_D$ in \eqref{eq-proof-RD} combined with \eqref{eq-proof-gamma}, we have
\begin{align*}
    R_D(\bX)
    =
    \frac{ \EXP \big\{ \alpha(Y,\bX) \cond A=0,Z=z_1,\bX \con \fy \big\} }{ 
    \EXP \big\{ \alpha(Y,\bX) \cond A=0,Z=z_0,\bX \con \fy \big\} 
    }
    =
    \frac{ \gamma(z_1,\bX)  }{ 
    \gamma(z_0,\bX) 
    }
    \frac{ \beta(z_0,\bX)  }{ 
    \beta(z_1,\bX) 
    } \ .
    \numeq \label{eq-proof-RDrange}
\end{align*}
Therefore, by combining \eqref{eq-proof-RYrange} and \eqref{eq-proof-RDrange} with the definition of $\mathcal{C}_1(\bX)$, we establish
\begin{align*}
    &
    \lambda(y,\bX)
    \\
    &
    \leq 
    \frac{ 1+ C_\alpha \beta(z_1,\bX)}{ 1+ C_\alpha \beta(z_0,\bX) }
    \frac{ \ff{}(A=0 \cond Z=z_1, \bX)  }
    { \ff{}(A=0 \cond Z=z_0, \bX) }
    \frac{ \gamma(z_1,\bX)  }{ 
    \gamma(z_0,\bX) 
    }
    \frac{ \beta(z_0,\bX)  }{ 
    \beta(z_1,\bX) 
    } 
    \frac{ f(A=1 \cond Z=z_0,\bX) }{ f(A=1 \cond Z=z_1,\bX)} 
    \\
    &
    \stackrel{(*)}{=}
    \frac{ 1+ C_\alpha \beta(z_1,\bX)}{ 1+ C_\alpha \beta(z_0,\bX) }
    \frac{ \beta(z_0,\bX)  }{ 
    \beta(z_1,\bX) 
    } 
    \\
    &
    =
    \frac{ 1/\beta(z_1,\bX) + C_\alpha }{ 1/\beta(z_0,\bX) + C_\alpha }
    \quad \quad ( \equiv \kappa(\bX) )
    \numeq 
    \label{eq-proof-kappa}
    \\
    &
    < 1
    \ .
\end{align*} 
The equality with $(*)$ follows from the definition of $\gamma$. The last line holds because $\beta(z_1,\bX) > \beta(z_0,\bX)$. We define $\kappa(\bX)$ be the quantity in \eqref{eq-proof-kappa}, which we will use later.

Define the operator $\Phi: \mathcal{L}_+^\infty(\suppYX) \rightarrow \mathcal{L}_+^\infty(\suppYX)$ by
\begin{align*}
    \Phi(h(y,\bX)) 
    \equiv
    \frac{ \mathcal{C}_1(\bX) h (y,\bX) + \mathcal{C}_2(y , \bX) }{ R_{Y}(y,\bX) }  \ . 
\end{align*}
From \eqref{eq-Psi}, $g^\star$ satisfies
\begin{align*}
    &
    g^{\star}(y,\bX) 
    =
    \frac{ \mathcal{C}_1(\bX) g^{\star} (y,\bX) + \mathcal{C}_2(y, \bX) }{ R_{Y}(y,\bX) }  
    R_D(\bX)  
    \\
    &
    \Leftrightarrow
    \quad 
    \Phi\big( g^{\star}(y,\bX)  \big)
    = R_D(\bX) g^{\star}(y,\bX) \ .
    \numeq 
    \label{eq-proof-g recursive}
\end{align*}

Let the operator $\overline{h}$ be the normalization step, i.e.,
\begin{align*}
    \overline{h}(y,\bX) = \frac{ h(y,\bX) }{\int h(t,\bX) \, dt } \ . 
\end{align*}

Given an initial value $\overline{h}^{(0)}$ satisfying $\int \overline{h}^{(0)}(y,\bX) \, dy =1$, consider the following iterative updates:
\begin{align*}
    &
    h^{(j+1)} \leftarrow \Phi( \overline{h}^{(j)}(y,\bX) ) \ ,
    &&
    \overline{h}^{(j+1)} 
    \leftarrow 
    \text{normalization of } h^{(j+1)} \ .
    \numeq
    \label{eq-proof-Phi normalization}
\end{align*}
Note that $\overline{h}^{(j)}$ coincides with the iteration generated by $\overline{\Psi}$. In particular, the update $\overline{h}^{(j+1)} = \overline{\Psi}(\overline{h}^{(j)})$ is equivalent to the two-step procedure consisting of applying $\Phi$ and then normalizing, as described in \eqref{eq-proof-Phi normalization}. 

Given the $j$th iterate $\overline{h}^{(j)}$, we have
\begin{align*}
    \frac{ h^{(j+1)}(y,\bX) }
    {
    R_D(\bX) g^\star(y,\bX)
    }
    & =
    \frac{ 
    \Phi(\overline{h}^{(j)}(y,\bX)) 
    }{ 
    \Phi(g^{\star}(y,\bX)) 
    }
    \\
    &
    =
    \frac{ \mathcal{C}_1(\bX) \overline{h}^{(j)}  (y,\bX) + \mathcal{C}_2(y , \bX) }{ 
    \mathcal{C}_1(\bX) g^\star (y,\bX) + \mathcal{C}_2(y , \bX)
    }
    \\
    &
    =
    1+ 
    \mathcal{C}_1(\bX) 
    \frac{ \overline{h}^{(j)} (y,\bX) - g^\star(y,\bX) }{ 
    \mathcal{C}_1(\bX) g^\star (y,\bX) + \mathcal{C}_2(y , \bX)
    } 
    \\
    &
    \stackrel{\eqref{eq-proof-g recursive}}{=}
    1+ 
    \frac{1}{g^{\star}(y,\bX)}
    \underbrace{
    \frac{ R_D(\bX) \mathcal{C}_1(\bX) }{ R_Y(y,\bX)     } 
    }_{=\lambda(y,\bX)}
    \big\{  \overline{h}^{(j)} (y,\bX) - g^\star(y,\bX) \big\}
    \ .
\end{align*} 
This implies that
\begin{align*}
    \frac{ h^{(j+1)}(y,\bX) }
    {
    R_D(\bX) g^\star(y,\bX)
    }
    - 1
    =
    \lambda(y,\bX)
    \bigg\{  \frac{ \overline{h}^{(j)}(y,\bX)}{g^\star(y,\bX)}-1 \bigg\} \ .
    \numeq 
    \label{eq-proof-iteration}
\end{align*}
We define the relative ratio between $\overline{h}^{(j)}$ and $g^\star$, along with the normalizing constant of $h^{(j)}$:
\begin{align*}
    &
    r^{(j)}(y,\bX) = \frac{\overline{h}^{(j)}(y,\bX)}{g^\star(y,\bX)} \ , 
    &&
    N^{(j)}(\bX) = \int h^{(j)} (y,\bX) \, dy  
      \ .
\end{align*} 
Then, \eqref{eq-proof-iteration} can be expressed as
\begin{align*}
    r^{(j+1)} (y,\bX)
    = 
    \frac{ R_D(\bX) }{ N^{(j+1)}(\bX) }
    \big\{ \lambda(y,\bX) r^{(j)}(y,\bX) + 1 - \lambda(y,\bX) \big\} 
\end{align*}

Let $M^{(j)} = \sup_{y} r^{(j)} (y,\bX)$ and $m^{(j)} = \inf_{y} r^{(j)}(y,\bX)$. Then, we observe
\begin{align*}
    \frac{M^{(j+1)}}{m^{(j+1)}}
    =
    \frac{ \sup_{y} r^{(j+1)} (y,\bX) }{ \inf_{y} r^{(j+1)} (y,\bX) }
    =
    \frac{ \sup_{y} \big\{ \lambda(y,\bX) r^{(j)}(y,\bX) + 1 - \lambda(y,\bX) \big\} }{ \inf_{y} \big\{ \lambda(y,\bX) r^{(j)}(y,\bX) + 1 - \lambda(y,\bX) \big\}  }
    \ . 
    \numeq 
    \label{eq-proof-Mm ratio}
\end{align*}

Suppose that $M^{(j)} = m^{(j)}$ holds. This condition implies that $\overline{h}^{(j)}(y,\bX) \propto g^\star(y,\bX)$. Moreover, since $\int \overline{h}^{(j)}(y,\bX) \, dy = \int g^\star(y,\bX) \, dy = 1$, if follows that $h^{(j)}=g^\star$. Consequently, it suffices to show that $M^{(j)}/m^{(j)}$ converges to 1 as $j \rightarrow \infty$, which then would imply $\overline{h}^{(\infty)}(y,\bX) = g^\star(y,\bX)$ for all $y \in \suppYX$.

In order to proceed, we consider three cases according to the values of $M^{(j)}$ and $m^{(j)}$:

\vspace*{0.25cm}

\noindent \textbf{(Case 1)}: $M^{(j)} \leq 1$\\
If $M^{(j)} \leq 1$, then necessarily $M^{(j)}= 1$, since:
\begin{align*}
    \overline{h}^{(j)}(y,\bX) 
    \leq 
    M^{(j)} g^\star(y,\bX)
    \quad
    \Rightarrow
    \quad
    1= \int \overline{h}^{(j)}(y,\bX)  \, dy 
    \leq 
    M^{(j)} \int g^\star(y,\bX) \, dy
    =
    M^{(j)} \ .
\end{align*}
Therefore, we have
\begin{align*}
     g^\star(y,\bX) -  \overline{h}^{(j)}(y,\bX)  \geq 0 \ , \quad
     \int \big\{ g^\star(y,\bX) -  \overline{h}^{(j)}(y,\bX) \big\} \, dy = 0 \ .
\end{align*}
A nonnegative function whose integral is zero must vanish almost everywhere; hence, this means that $g^\star = \overline{h}^{(j)}$, indicating that convergence has been achieved at the $j$th iteration.

\vspace*{0.25cm}

\noindent \textbf{(Case 2)}: $m^{(j)} \geq 1$\\
Analogous to (Case 1), one can establish that $m^{(j)} \geq 1$ implies $g^\star = \overline{h}^{(j)}$, indicating that convergence has been achieved at the $j$th iteration.

\vspace*{0.25cm}

\noindent \textbf{(Case 3)}: $m^{(j)} < 1 < M^{(j)}$\\
For nonnegative functions $h_1$ and $h_2$, if follows that $\sup (h_1 h_2) \leq (\sup h_1 )( \sup  h_2) $ and $\inf (h_1  h_2)  \geq (\inf h_1) (\inf h_2)$. Consequently, an upper bound of \eqref{eq-proof-Mm ratio} can be obtained as follows:
\begin{align*}
    \frac{M^{(j+1)}}{m^{(j+1)}}
    \leq
    \frac{ \sup_{y} \big\{ \lambda(y,\bX) M^{(j)} + 1 - \lambda(y,\bX) \big\} }{ \inf_{y} \big\{ \lambda(y,\bX) m^{(j)} + 1 - \lambda(y,\bX) \big\}  }
    \leq
    \frac{ \kappa(\bX) M^{(j)} + 1 - \kappa(\bX)  }{ \kappa(\bX) m^{(j)} + 1 - \kappa(\bX) }
     \ ,
\end{align*}
where the second inequality holds from $m^{(j)}< 1<M^{(j)}$ and the definition of $\kappa(\bX)$ in \eqref{eq-proof-kappa}. From algebraic manipulation, we obtain
\begin{align*}
    \frac{M^{(j+1)}}{m^{(j+1)}}
    -
    1
    \leq 
    \frac{ \kappa(\bX) m^{(j)} }{ \kappa(\bX) m^{(j)} + 1 - \kappa(\bX) }
    \bigg\{ 
    \frac{M^{(j)}}{m^{(j)}}
    -
    1
    \bigg\}
    \leq 
    \kappa(\bX)
    \bigg\{ 
    \frac{M^{(j)}}{m^{(j)}}
    -
    1
    \bigg\}
     \ . 
     \numeq 
     \label{eq-proof-Mm iteration}
\end{align*}
The second inequality follows from the fact that the mapping $m \mapsto (\kappa m)/(\kappa m + 1 - \kappa)$ takes values in the interval $[0,\kappa]$ for $m \in [0,1]$. Therefore, \eqref{eq-proof-Mm iteration} implies that
\begin{align*}
    \bigg\{ 
    \frac{M^{(j+1)}}{m^{(j+1)}}
    -
    1
    \bigg\}
    \leq \{ \kappa(\bX) \}^{j} 
     \bigg\{ 
    \frac{M^{(0)}}{m^{(0)}}
    -
    1
    \bigg\}
    \rightarrow 0 \quad \text{ as \quad $j \rightarrow \infty$} \ ,
    \numeq 
     \label{eq-proof-Exp Convergence}
\end{align*}
where we used $\kappa(\bX) < 1$. This implies that the limiting function $\overline{h}^{(\infty)}(y,\bX)$ satisfies $M^{(\infty)}/m^{(\infty)}=1$, meaning that $\overline{h}^{(\infty)}(y,\bX) = g^\star(y,\bX)$ for all $y \in \suppYX$. Also, \eqref{eq-proof-Exp Convergence} implies that the convergence of the iterative update is exponentially fast with respect to the divergence function $\mathfrak{d}$ defined by:
\begin{align*}
    \mathfrak{d}(h)
    =
    \frac{ \sup_{y} h(y,\bX)/g^\star(y,\bX) }
    {\inf_{y} h(y,\bX)/g^\star(y,\bX)} - 1 \ .
\end{align*}
Note that $\mathfrak{d}(h)=0$ implies that $h$ and $g^\star$ are proportional to each other. In particular, for normalized functions (i.e., $\int h(y,\bX) \, dy = 1$), $\mathfrak{d}(h)=0$ implies that $h$ and $g^\star$ are identical.

This concludes the proof.

\subsection{Proof of Theorem \ref{thm-IF}} \label{sec-proof-IF}

In the proof, we show a more general result by characterizing the IFs for $ \tau^* \equiv \tau_1^* - \tau_0^*$ where $\tau_{a}^* \equiv  \EXP \{ \mathcal{G}(\potY{a},\bX)  \cond A=1 \}$; here, $\mathcal{G}(\cdot)$ is a fixed, uniformly bounded function.  For convenience, we denote $\mathcal{G}^{(a)} = \mathcal{G}(\potY{a},\bX)$ and $\mathcal{G} = \mathcal{G}(Y,\bX)$.

Before proceeding, we introduce additional notation. Let $\bO^{F} = (\potY{1},\potY{0},A,Z,\bX)$ and $\bO= (Y,A,Z,\bX)$ denote the full and observed data. Throughout, we assume the positivity assumption $f(A=0 \cond \potY{0},Z,\bX) \geq c_{A} > 0 $ for some constant $c_{A}$.  Let $\mathcal{M}^F$ and $\mathcal{M}$ denote regular models of $\bO^F$ and $\bO$ of the form:
\begin{align*}
& \mathcal{M}^F 
=
\big\{ P(\bO^F) \cond \text{\HL{A1}-\HL{A2} and  \HL{IV1}-\HL{IV4} hold} \big\} \ ,
\\
& \mathcal{M} 
=
\big\{ P(\bO) \cond \text{\HL{A1}-\HL{A2} and  \HL{IV1}-\HL{IV4} hold} \big\} \ .
\end{align*}
Consider parametric submodels of $\mathcal{M}^F$ and $\mathcal{M}$ parametrized by an one-dimensional parameter $\eta$:
\begin{align*}					
&
\mathcal{M}\ETA^{F}
=
\big\{ P\ETA (\bO^F) \cond 
\text{\HL{A1}-\HL{A2} and  \HL{IV1}-\HL{IV4} hold}
\big\} \ ,
\\
&
\mathcal{M}\ETA
=
\big\{ P\ETA (\bO) \cond 
\text{\HL{A1}-\HL{A2} and  \HL{IV1}-\HL{IV4} hold}
\big\} \ . 
\end{align*} 
Let $\EXP\ETA\{ h(\bO^F) \big\}$ and $\EXP\ETA\{ h(\bO) \big\}$ denote expectations with respect to the distribution $P\ETA(\bO^F)$ and $P\ETA(\bO)$, respectively. Likewise, let $f\ETA (\bO)$ denote the density of the parametric submodel $P\ETA(\bO)$, with $f\ETA (\bO^F)$ defined analogously. We suppose that the full data and observed data laws, $P^*(\bO^F)$ and $P^*(\bO)$, are attained at $\eta^*$, i.e., $P^*(\bO^F) = P_{\eta^*}(\bO^F)$ and $P^*(\bO) = P_{\eta^*}(\bO)$. Let $\alpha\ETA$, $\beta\ETA$, and $\omega\ETA$ be the nuisance function of the IF in \eqref{eq-IF} at the submodel $P\ETA$. For any function $h\ETA$, let $\nabla_{\eta} h\ETA$ denote its derivative with respect to $\eta$. Then, the score function of $P\ETA$ is denoted by $s\ETA = \nabla_\eta f\ETA/f\ETA$, and its value at $\eta^*$ is denoted $s^*$. Following the same notation as the score function, let $s\ETA(\cdot \con \alpha) = \nabla\ETA \alpha\ETA / \alpha\ETA $ and $s\ETA(\cdot \con \beta) = \nabla \ETA \beta\ETA/\beta\ETA$. Finally, for a random variable $\bV \subseteq \bO^F$ or $\bV \subseteq \bO$, define $\mathcal{L}^2(\bV) \equiv \{ f \cond  \EXP \{f^2(\bV) \} < \infty \}$ be the Hilbert space of square-integrable functions of $\bV$, equipped with the inner product $\langle f_1, f_2 \rangle = \EXP \{ f_1(\bV) f_2(\bV) \}$. Likewise, define $\mathcal{L}_0^2(\bV) \subseteq \mathcal{L}^2(\bV)$ be the Hilbert space of square-integrable, mean-zero functions of $\bV$.

At the submodel $P\ETA$, we define the ATT and the function $\uncInfFt$, a part of \eqref{eq-IF}, as follows:
\begin{align*}
    &
    \tau\ETA = 
    \EXP\ETA \big\{ \uncInfFt_{\eta}(\bO) \big\}
    \ , 
    \\
    &
  \uncInfFt\ETA(\bO)
  =
  \frac{ 
   \left[ 
    \begin{array}{l}
    \{ A - (1-A) \alpha\ETA(Y,\bX ) \beta\ETA(Z,\bX) \}
     \big\{ \mathcal{G} - \mu\ETA(Z,\bX \con \mathcal{G}) \big\}     
    \end{array}
    \right]
    }{\EXP\ETA(A)}
    \ . 
\end{align*}
We also define the function $\augInfFt$ indexed by $v \in \mathcal{L}^2(\potY{0},Z,\bX)$ as follows:
\begin{align*}
    &
    \augInfFt\ETA(\bO \con v)
    \\
    &
    =
  \frac{ 
   \left[ 
    \begin{array}{l}
       (1-A) \alpha\ETA(Y,\bX) \beta\ETA(Z,\bX) \big\{ \widetilde{v}\ETA(Y,Z,\bX) - \mu\ETA(Z,\bX \con \widetilde{v}\ETA) \big\}
  \\
  +
    A
    \mu\ETA(Z,\bX \con \widetilde{v}\ETA)
    + 
    (1-A) \widetilde{v}\ETA(Y,Z,\bX) 
    \end{array}
    \right]
    }{\EXP\ETA(A)} 
    \\
    &
    =
  \frac{ 
   (1-A) \big\{ 1+ \alpha\ETA(Y,\bX) \beta\ETA(Z,\bX) \big\}
       \big\{ \widetilde{v}\ETA(Y,Z,\bX) - \mu\ETA(Z,\bX \con \widetilde{v}\ETA) \big\}
  + 
    \mu\ETA(Z,\bX \con \widetilde{v}\ETA) 
    }{\EXP\ETA(A)} 
    \ ,
\end{align*}
where $\widetilde{v}\ETA$ is defined by the following operator of a function $v \in \mathcal{L}^2(\potY{0},Z,\bX)$:
\begin{align*}
    &
    \widetilde{v}\ETA(\potY{0},Z,\bX)
    =
     v(\potY{0},Z,\bX)-v\ETA^\dagger(\potY{0},Z,\bX)
     \ , 
     \numeq
     \label{eq-proof-tilde operator}
     \\
     &
     v\ETA^\dagger(\potY{0},Z,\bX) = 
     \EXP\ETA(v \cond \potY{0},\bX) + \EXP\ETA(v  \cond Z,\bX) - \EXP\ETA(v \cond \bX) \ .
\end{align*}

From the construction of $\widetilde{v}\ETA$, we find that $\EXP\ETA \{ \augInfFt\ETA(\bO \con v) \} = 0$ for any $v \in \mathcal{L}^2(\potY{0},Z,\bX)$ as follows:
\begin{align*}
    &
    \EXP\ETA \{ \augInfFt\ETA(\bO \con v) \}
    \\
    &
    =
    \EXP\ETA 
    \bigg[
    \begin{array}{l}
    (1-A) \alpha\ETA(Y,\bX) \beta\ETA(Z,\bX) \big\{ \widetilde{v}\ETA(Y,Z,\bX) - \mu\ETA(Z,\bX \con \widetilde{v}\ETA) \big\}
    \\
    +
    A
    \mu\ETA(Z,\bX \con \widetilde{v}\ETA)
    + 
    (1-A) \widetilde{v}\ETA(Y,Z,\bX) 
    \end{array}
    \bigg]
    \\
    &
    =
    \EXP\ETA 
    \Big[ A \big\{ \mu\ETA(Z,\bX \con \widetilde{v}\ETA)  - \mu\ETA(Z,\bX \con \widetilde{v}\ETA) \big\}
    +    \widetilde{v}\ETA(\potY{0},Z,\bX)   \Big] 
    \\
    &
    =
    \EXP\ETA 
    \big\{ 
    v(\potY{0},Z,\bX)   - v\ETA^\dagger(\potY{0},Z,\bX)  
    \big\} 
    \\
    &
    =
    \EXP\ETA 
    \big[ 
    \EXP\ETA 
    \big\{ 
    v(\potY{0},Z,\bX)   - v\ETA^\dagger(\potY{0},Z,\bX)  
    \cond \bX
    \big\}
    \big] 
    \\
    &
    =
    \EXP\ETA 
    \bigg[ 
    \begin{array}{l}
    \EXP\ETA 
    \big\{ 
    v(\potY{0},Z,\bX)  \cond \bX \big\}
    - 
    \EXP\ETA 
    \big[ 
    \EXP\ETA \big\{ v(\potY{0},Z,\bX)  \cond \potY{0},\bX \big\}
    \cond \bX \big]
    \\
    -
    \EXP\ETA 
    \big[ 
    \EXP\ETA \big\{ v(\potY{0},Z,\bX)  \cond Z,\bX \big\}
    \cond \bX \big]
    +
    \EXP\ETA 
    \big\{ 
    v(\potY{0},Z,\bX)  \cond \bX \big\} 
    \end{array}
    \bigg] 
    \\
    &
    =
    \EXP\ETA 
    \bigg[ 
    \begin{array}{l}
    \EXP\ETA 
    \big\{ 
    v(\potY{0},Z,\bX)  \cond \bX \big\}
    - 
    \EXP\ETA 
    \big\{ 
    v(\potY{0},Z,\bX)  \cond \bX \big\} 
    \\
    -
    \EXP\ETA 
    \big\{ 
    v(\potY{0},Z,\bX)  \cond \bX \big\} 
    +
    \EXP\ETA 
    \big\{ 
    v(\potY{0},Z,\bX)  \cond \bX \big\} 
    \end{array}
    \bigg] 
    \\
    &
    = 
    0 \ .
\end{align*}

For a fixed $v \in \mathcal{L}^2(\potY{0},Z,\bX)$, we define $\InfFt$ as follows:
\begin{align*}    
    &
    \InfFt\ETA(\bO \con \tau, v) 
    =
    \uncInfFt\ETA(\bO) 
    -
    \frac{A \tau}{\EXP\ETA(A)}
    +
    \augInfFt\ETA(\bO \con v) \ .
\end{align*}
Therefore, $\tau\ETA$ satisfies the following moment condition for any $v \in \mathcal{L}^2(\potY{0},Z,\bX)$:
\begin{align*}
    \EXP\ETA \big\{ \InfFt\ETA(\bO \con \tau\ETA, v) \big\} = 0 \ .
    \numeq 
    \label{eq-proof-IF equation}
\end{align*}
Taking the gradient of \eqref{eq-proof-IF equation} with respect to $\eta$ yields
\begin{align*}
    0 
    &
    = \nabla\ETA \EXP\ETA \big\{ \InfFt\ETA(\bO \con \tau\ETA, v) \big\}
    \\
    &
    =
    \EXP\ETA \big\{ s\ETA(\bO)
    \InfFt\ETA(\bO \con \tau\ETA, v)
    \big\}
    +
    \EXP\ETA \big\{ \nabla\ETA
    \InfFt\ETA(\bO \con \tau\ETA, v)
    \big\}
    \\
    &
    =
    \EXP\ETA \big\{ s\ETA(\bO)
    \InfFt\ETA(\bO \con \tau\ETA, v)
    \big\}
    \\
    &
    \quad \quad 
    +
    \underbrace{
    \EXP\ETA \bigg\{ \frac{\partial}{\partial \tau}
    \InfFt\ETA(\bO \con \tau, v)
    \bigg\}
    }_{=-1}
    \nabla\ETA \tau\ETA
    +
    \EXP\ETA \big\{ \nabla\ETA
    \InfFt\ETA(\bO \con \tau, v)
    |_{\tau=\tau\ETA}
    \big\}  \ .
\end{align*}
This implies
\begin{align*}
    & 
    \nabla\ETA \tau\ETA
    \\
    &
    =
    \EXP\ETA \big\{ s\ETA(\bO)
    \InfFt\ETA(\bO \con \tau\ETA, v)
    +
    \nabla\ETA
    \InfFt\ETA(\bO \con \tau, v)
    |_{\tau=\tau\ETA}
    \big\}  
    \\
    &
    =
    \EXP\ETA \Bigg[ s\ETA(\bO)
    \InfFt\ETA(\bO \con \tau\ETA, v)    
    +
    \nabla\ETA
    \bigg\{
    \uncInfFt\ETA(\bO) 
    -
    \frac{A \tau}{\EXP\ETA(A)}
    +
    \augInfFt\ETA(\bO \con v)
    \bigg\}
    \bigg|_{\tau=\tau_\eta} 
    \Bigg]
    \ .
    \numeq 
    \label{eq-proof-pathwise diff}
\end{align*}

We now find a representation of \eqref{eq-proof-pathwise diff} in order to show that $\tau^*$ is a differentiable parameter \citep{Newey1990}. To facilitate this derivation, we first introduce some useful preliminary results.
\begin{align}
    &
    \EXP\ETA \big\{ \mathcal{G} -  \mu\ETA(Z,\bX \con \mathcal{G}) \cond A=1 \big\} = \tau\ETA  \ ,
    \label{eq-proof-identity1}
    \\
    &
    \EXP\ETA \big\{ \mathcal{G}^{(0)} - \mu\ETA(Z,\bX \con \mathcal{G}) \cond A=1, Z, \bX \big\} = 0 \ ,
    \label{eq-proof-identity2}
    \\
    &
    \EXP\ETA \big\{ (1-A)  
    \alpha\ETA(Y,\bX)
    \beta\ETA(Z,\bX)  - A \cond Z,\bX \big\} = 0 \ ,
    \label{eq-proof-identity3}    
    \\
    &
    \EXP\ETA \big[ (1-A) \big\{ 1+ \alpha\ETA(Y,\bX) \beta\ETA(Z,\bX)  \big\} v(Y,Z,\bX) \big] = 
    \EXP\ETA \big\{ v(\potY{0},Z,\bX)\big\} \ ,
    \label{eq-proof-identity4}
    \\
    &  
    \nabla\ETA \EXP\ETA \big\{   \widetilde{v}\ETA(\potY{0},Z,\bX)
    \big\}  
    = 0  
    \nonumber \\
    &
    \Rightarrow
    \quad        
    \EXP\ETA \big\{ 
    \nabla\ETA \widetilde{v}\ETA(\potY{0},Z,\bX)
    \big\}  
    +
    \EXP\ETA \big\{ 
    s\ETA(\potY{0},Z,\bX)
    \widetilde{v}\ETA(\potY{0},Z,\bX )
    \big\} 
    = 0
    \nonumber
    \\
    &
    \Rightarrow
    \quad        
    \EXP\ETA \big\{ 
    \nabla\ETA \widetilde{v}\ETA(\potY{0},Z,\bX)
    \big\}  
    =
    -
    \EXP\ETA \big\{ 
    s\ETA(\potY{0},Z,\bX)
    \widetilde{v}\ETA(\potY{0},Z,\bX )
    \big\} 
    = 0 \ ,
    \label{eq-proof-identity5}
    \\
    &
    \EXP\ETA \big[ \EXP\ETA \big\{ v(\potY{0},Z,\bX) \cond Z,\bX \big\} \cond \potY{0},\bX \big]  
    \nonumber 
    \\
    &
    =
    \EXP\ETA \big[ \EXP\ETA \big\{ v(\potY{0},Z,\bX) \cond Z,\bX \big\} \cond \bX \big] 
    =
    \EXP\ETA \big\{ v(\potY{0},Z,\bX) \cond \bX \big\} \ .
    \label{eq-proof-identity6} 
\end{align}
Except for \eqref{eq-proof-identity5}, the other results can be established by straightforward algebra. 
Result \eqref{eq-proof-identity5} follows from the decomposition
\begin{align*}
    & 
    s\ETA(\potY{0},Z,\bX) 
    =
    s\ETA(\potY{0} \cond \bX) 
    +
    s\ETA(Z \cond \bX) 
    +
    s\ETA(\bX) \ , 
    \\
    &
    \EXP\ETA \big\{
    s\ETA(\potY{0} \cond \bX) 
    \cond \bX \big\}
    =
    \EXP\ETA \big\{
    s\ETA(Z \cond \bX) 
    \big\}
    =
    \EXP\ETA \big\{
    s\ETA(\bX)
    \big\}
    =
    0 \ , 
\end{align*}
and
\begin{align*}
    &
    \EXP\ETA \big\{ 
    s\ETA(\potY{0},Z,\bX)
    \widetilde{v}\ETA(\potY{0},Z,\bX )
    \big\}
    \\
    &
    =
    \EXP\ETA 
    \big[
    \big\{
        \begin{array}{l}
             s\ETA(\potY{0} \cond \bX)  
    +
    s\ETA(Z \cond \bX)  
    +
    s\ETA(\bX)    
        \end{array}
    \big\}
    \big\{
        \begin{array}{l}
             v(\potY{0},Z,\bX) 
             -
             v\ETA^\dagger(\potY{0},Z,\bX)
        \end{array}
    \big\}
    \big]
    \\
    &
    =
    \EXP\ETA 
    \left[  
        \begin{array}{l}
             s\ETA(\potY{0} \cond \bX) 
             \EXP\ETA\{ v - v\ETA^\dagger \cond \potY{0},\bX \}
             \\
    +
    s\ETA(Z \cond \bX) 
    \EXP\ETA\{ v - v\ETA^\dagger \cond Z,\bX \} 
    +
    s\ETA(\bX)    
    \EXP\ETA\{ v - v\ETA^\dagger \cond \bX \}
        \end{array} 
    \right]
    \\
    &
    =
    0 \ .
\end{align*}
The last equality holds because, for all $v \in \mathcal{L}^2(\potY{0},Z,\bX)$, we obtain $\EXP\ETA\{ v - v\ETA^\dagger \cond \potY{0},\bX \} =
\EXP\ETA\{ v - v\ETA^\dagger \cond Z,\bX \}  = \EXP\ETA\{ v - v\ETA^\dagger \cond \bX \}  = 0$ from \eqref{eq-proof-identity5}.

We now return to derive an alternative representation of \eqref{eq-proof-pathwise diff}. From a straightforward calculation, the first term of \eqref{eq-proof-pathwise diff} is expressed as:
\begin{align*}
    &
    \EXP\ETA \big\{ s\ETA (\bO) \InfFt\ETA (\bO \con \tau\ETA,v) \big\}
    \\
    &
    =
    \frac{ \EXP\ETA \left[ s\ETA(\bO) 
    \left[ 
    \begin{array}{l}
         (1-A) \alpha\ETA(Y,\bX) \beta\ETA(Z,\bX) 
         \big\{ \mathcal{G} - \mu\ETA(Z,\bX \con \mathcal{G} ) \big\} \\
         + A \big\{\mu\ETA(Z,\bX  \con \mathcal{G}) - \tau\ETA \big\}
         \\
         +
         (1-A) \alpha\ETA(Y,\bX ) \beta\ETA(Z,\bX)   
    \big\{ 
        \widetilde{v}\ETA(Y,Z,\bX ) 
        -
        \mu\ETA(Z,\bX \con \widetilde{v}\ETA)
    \big\}
    \\
    +
    (1-A)  \widetilde{v}\ETA(Y,Z,\bX) 
    +
    A
    \mu\ETA(Z,\bX \con \widetilde{v}\ETA)
    \end{array}
    \right]
    \right]
    }{ \EXP\ETA(A) }
    \numeq \label{eq-proof-IF piece1} \ .
\end{align*}

From direct algebra, the second term of \eqref{eq-proof-pathwise diff} is represented as follows:
\begin{align*}
    &
    \nabla\ETA
    \bigg\{
    \uncInfFt\ETA(\bO) 
    -
    \frac{A \tau}{\EXP\ETA(A)}
    \bigg\} 
    \\
    &
    =
    - \frac{ (1-A) \nabla\ETA \alpha\ETA(Y,\bX) \beta\ETA(Z,\bX) \big\{ \mathcal{G} - \mu\ETA(Z,\bX \con \mathcal{G}) \big\} }{ \EXP\ETA(A) }
    \\
    &
    \hspace*{1cm}
    - 
    \frac{ (1-A) \alpha\ETA(Y,\bX) 
    \nabla\ETA \beta\ETA(Z,\bX) 
    \big\{ \mathcal{G} - \mu\ETA(Z,\bX \con \mathcal{G}) \big\} }{ \EXP\ETA(A) }
    \\
    &
    \hspace*{1cm}
    -
    \frac{ \big\{ A -  (1-A)
    \alpha\ETA(Y,\bX)
    \beta\ETA(Z,\bX) 
    \big\}
    \nabla\ETA \mu\ETA(Z,\bX \con \mathcal{G})
    }{ \EXP\ETA(A) }
    \\
    &
    \hspace*{1cm}
    -
    \frac{ \big\{ A-  (1-A) 
    \alpha\ETA(Y,\bX) 
    \beta\ETA(Z,\bX) \big\}
    \big\{  \mathcal{G} - \mu\ETA(Z,\bX \con \mathcal{G}) \big\}  - A \tau }{ \EXP\ETA(A) }
    \frac{ \nabla\ETA \EXP\ETA(A) } { \EXP\ETA(A) } \ .
\end{align*}
By \eqref{eq-proof-identity1}-\eqref{eq-proof-identity3}, we establish
\begin{align*}
    &
    \EXP\ETA
    \Bigg[ \nabla\ETA
    \bigg\{
    \uncInfFt\ETA(\bO) 
    -
    \frac{A \tau}{\EXP\ETA(A)} 
    \bigg\}    \bigg|_{\tau=\tau_\eta} 
    \Bigg]
    \\
    &
    = 
    - 
    \frac{ 
    \EXP\ETA \big[ 
    s\ETA(Y,\bX \con \alpha)
    (1-A) 
    \alpha\ETA(Y,\bX)
    \beta\ETA(Z,\bX) \big\{ \mathcal{G} - \mu\ETA(Z,\bX \con \mathcal{G}) \big\} 
    \big]
    }{ \EXP\ETA(A) } \ .
    \numeq \label{eq-proof-IF piece2}
\end{align*}

The last term of \eqref{eq-proof-pathwise diff} is given by 
\begin{align*}
    &
    \nabla\ETA
    \augInfFt\ETA(\bO \con v)
    \\
    & 
    =
    \frac{
    \left[
    \begin{array}{l}
         (1-A) \nabla\ETA \alpha\ETA(Y,\bX) \beta\ETA(Z,\bX)
         \{ 
    \widetilde{v}\ETA (Y,Z,\bX)
    - \mu\ETA(Z,\bX \con \widetilde{v}\ETA)
    \}
    \\ 
    \hspace*{0.5cm}
    +
    (1-A) \alpha\ETA(Y,\bX) \nabla\ETA \beta\ETA(Z,\bX)
    \{ 
    \widetilde{v}\ETA (Y,Z,\bX)
    - \mu\ETA(Z,\bX \con \widetilde{v}\ETA)
    \}
    \\ 
    \hspace*{0.5cm}
    +
    (1-A) \big\{ 1+ \alpha\ETA(Y,\bX) \beta\ETA(Z,\bX) \big\} 
    \big\{
    \nabla\ETA\widetilde{v}(Y,Z,\bX)
    -
    \nabla\ETA
    \mu\ETA(Z,\bX \con \widetilde{v}) 
    \big\}
    \\ 
    \hspace*{0.5cm}
    +
    \nabla\ETA
    \mu\ETA(Z,\bX \con \widetilde{v})
    \end{array}
    \right]
    }{\EXP\ETA(A)}
    \\
    &
    \qquad 
    -
    \frac{ 
    \left[ 
        \begin{array}{l}
   (1-A) \big\{ 1+ \alpha\ETA(Y,\bX) \beta\ETA(Z,\bX) \big\}
       \big\{ \widetilde{v}\ETA(Y,Z,\bX) - \mu\ETA(Z,\bX \con \widetilde{v}\ETA) \big\}
       \\
    \hspace*{0.5cm}
    +
    \mu\ETA(Z,\bX \con \widetilde{v}\ETA) 
        \end{array}
    \right]
    }{\EXP\ETA(A)}
    \frac{\nabla\ETA \EXP\ETA(A)}{\EXP\ETA(A)} \ .
\end{align*} 
From some algebra, we obtain
\begin{align*}
    &
    \EXP\ETA \big\{
    (1-A) \alpha\ETA(Y,\bX) \nabla\ETA \beta\ETA(Z,\bX)
    \widetilde{v}\ETA(Y,Z,\bX)
    \big\}
    \\
    &
    =
    \EXP\ETA \big\{
    (1-A) s\ETA(Z,\bX \con \beta)
    \alpha\ETA(Y,\bX)  \beta\ETA(Z,\bX)
    \widetilde{v}\ETA(Y,Z,\bX)
    \big\}
    \\
    &
    =
    \EXP\ETA \big\{
    A s\ETA(Z,\bX \con \beta)
    \widetilde{v}\ETA(\potY{0},Z,\bX)
    \big\}
    \\
    &
    =
    \EXP\ETA \big[
    A s\ETA(Z,\bX \con \beta)
    \EXP\ETA \big\{ 
    \widetilde{v}\ETA(\potY{0},Z,\bX) \cond A=1,Z,\bX \big\}
    \big]
    \\
    &
    =
    \EXP\ETA \big\{
    A s\ETA(Z,\bX \con \beta)
    \mu\ETA(Z,\bX \con \widetilde{v}\ETA)
    \big\}
    \\
    &
    =
    \EXP\ETA \big\{
    (1-A) \alpha\ETA(Y,\bX) 
    \nabla\ETA \beta(Z,\bX)
    \mu\ETA(Z,\bX \con \widetilde{v}\ETA)
    \big\} \ .
\end{align*} 
By \eqref{eq-proof-identity4}, \eqref{eq-proof-identity5}, and the above identities, we establish
\begin{align*}
    &
    \EXP\ETA 
    \big\{
    \nabla\ETA
    \augInfFt\ETA(\bO \con v)
    \big\}
    \\
    &
    =
    \frac{\EXP\ETA \big[ 
    (1-A) \nabla\ETA \alpha\ETA(Y,\bX) 
    \beta\ETA(Z,\bX)
    \big\{ \widetilde{v}\ETA(Y,Z,\bX) - \mu\ETA(Z,\bX \con \widetilde{v}\ETA) \big\}
    \big] }{\EXP\ETA(A)}
    \\
    &
    =
    \frac{\EXP\ETA \big[ 
    s\ETA(Y,\bX \con \alpha)
    (1-A) 
    \alpha\ETA(Y,\bX) 
    \beta\ETA(Z,\bX)
    \big\{ \widetilde{v}\ETA(Y,Z,\bX) - \mu\ETA(Z,\bX \con \widetilde{v}\ETA) \big\}
    \big] }{\EXP\ETA(A)}
     \ .
    \numeq \label{eq-proof-IF piece3}
\end{align*} 
Combining \eqref{eq-proof-IF piece2} and \eqref{eq-proof-IF piece3}, we have
\begin{align*}
    &
    \EXP\ETA \Bigg[ 
    \nabla\ETA
    \bigg\{
    \uncInfFt\ETA(\bO) 
    -
    \frac{A \tau}{\EXP\ETA(A)}
    +
    \augInfFt\ETA(\bO \con v)
    \bigg\}
    \bigg|_{\tau=\tau_\eta} 
    \Bigg]
    \\
    &
    =
    -
    \frac{     
    \EXP\ETA \bigg[ 
    \begin{array}{l}
    s\ETA (Y,\bX \con \alpha)
    (1-A) 
    \alpha\ETA (Y,\bX)
    \beta\ETA (Z,\bX)
    \\
    \times 
    \big\{
    \mathcal{G} - \mu\ETA(Z,\bX \con \mathcal{G})
    -
    \widetilde{v}\ETA(Y,Z,\bX)
    +
    \mu\ETA (Z,\bX \con \widetilde{v}\ETA)
    \big\}
    \end{array}
    \bigg]
    }{\EXP\ETA (A)}  \ .
    \numeq \label{eq-proof-IF piece23}
\end{align*}

We now derive $s\ETA (Y,\bX \con \alpha)$. By \eqref{eq-proof-identity3} and \eqref{eq-proof-identity6}, we obtain the following result for any $u \in \mathcal{L}^2(\potY{0},Z,\bX)$:
\begin{align*}
    &
    \EXP\ETA \Bigg[ 
    \begin{array}{l}         
    (1-A) \alpha\ETA(Y,\bX) \beta\ETA(Z,\bX) \big\{ \widetilde{u}\ETA (\potY{0},Z,\bX) - \mu\ETA(Z,\bX \con \widetilde{u}\ETA) \big\}
    \\
    +
    A
    \mu\ETA(Z,\bX \con \widetilde{u}\ETA)
    +
    (1-A) \widetilde{u}\ETA (\potY{0},Z,\bX)
    \end{array}
    \, \Bigg| \, \potY{0}, \bX \Bigg]
    \numeq
    \label{eq-proof-alpha score 1}
    \\
    &
    =
    \EXP \Bigg\{
        \begin{array}{l}
    A \widetilde{u}\ETA(\potY{0},Z,\bX)
    -
    A
    \mu\ETA(Z,\bX \con \widetilde{u}\ETA)
    \\
    +
    A
    \mu\ETA(Z,\bX \con \widetilde{u}\ETA)
    +
    (1-A) 
    \widetilde{u}\ETA(\potY{0},Z,\bX)
        \end{array}
    \, \Bigg| \, \potY{0}, \bX \Bigg\}
    \\
    &
    =
    \EXP \big\{
    \widetilde{u}\ETA (\potY{0},Z,\bX)
    \, \big| \, \potY{0}, \bX \big\}
    \\ 
    &
    = 0 \ ,
\end{align*}
where $\widetilde{u}_{\eta}$ is defined analogously to \eqref{eq-proof-tilde operator}. 

We take the pathwise derivative of \eqref{eq-proof-alpha score 1} with respect to $\eta$. To streamline the derivation, we first establish an intermediate result. For any $g(\potY{0},Z,\bX)$, we have
\begin{align*}
    &
    \nabla\ETA 
    \EXP\ETA \big\{ (1-A) g(\potY{0},Z,\bX) \cond \potY{0},\bX \big\}
    \\
    &
    =
    \iint (1-a) g(\potY{0},z,\bX) \nabla\ETA f\ETA(A=a,Z=z \cond \potY{0},\bX)
    \, d(a,z)
    \\
    &
    =
    \iint \left\{ 
    \begin{array}{l}
    (1-a) g(\potY{0},z,\bX) \\
    \times s\ETA(A=a,Z=z \cond \potY{0},\bX) 
    \end{array}
    \right\} 
    f\ETA(A=a,Z=z \cond \potY{0},\bX)
    \, d(a,z)
    \\
    &
    =
    \iint \left\{ 
    \begin{array}{l}
    (1-a) g(\potY{0},z,\bX) \\
    \times s\ETA(Y,A=a,Z=z,\bX) 
    \end{array}
    \right\} 
    f\ETA(A=a,Z=z \cond \potY{0},\bX)
    \, d(a,z)
    \\
    &
    \qquad -
    s\ETA (\potY{0},\bX)
    \iint 
    (1-a) g(\potY{0},z,\bX)
    f\ETA(A=a,Z=z \cond \potY{0},\bX)
    \, d(a,z) 
    \\
    &
    =
    \EXP\ETA \big\{ s\ETA(\bO) (1-A) g(\potY{0},Z,\bX)  \cond \potY{0},\bX \big\}
    \\
    &
    \qquad 
    -
    s\ETA(\potY{0},\bX) \EXP\ETA \big\{ (1-A) g(\potY{0},Z,\bX)  \cond \potY{0},\bX \big\} \ .
    \numeq
    \label{eq-proof-alpha score 2}
\end{align*}
Likewise, for any $h(Z,\bX)$, we have
\begin{align*}
    &
    \nabla\ETA 
    \EXP\ETA \big\{ A h(Z,\bX) \cond \potY{0},\bX \big\}
    \\
    &
    =
    \iint a h(z,\bX) \nabla\ETA f\ETA(A=a,Z=z \cond \potY{0},\bX)
    \, d(a,z)
    \\
    &
    =
    \iint 
    a h(z,\bX) s\ETA(A=a,Z=z \cond \potY{0},\bX) 
    f\ETA(A=a,Z=z \cond \potY{0},\bX)
    \, d(a,z)
    \\
    &
    =
    \iint 
    a h(z,\bX) s\ETA(\potY{0},A=a,Z=z,\bX) 
    f\ETA(A=a,Z=z \cond \potY{0},\bX)
    \, d(a,z)
    \\
    &
    \qquad -
    s\ETA (\potY{0},\bX)
    \iint 
    a h(z,\bX)
    f\ETA(A=a,Z=z \cond \potY{0},\bX)
    \, d(a,z)  
    \\
    &
    =
    \iint 
    a h(z,\bX) s\ETA(A=a,Z=z,\bX) 
    f\ETA(A=a,Z=z \cond \potY{0},\bX)
    \, d(a,z)
    \\
    &
    \qquad -
    s\ETA (\potY{0},\bX) \EXP\ETA \big\{ A h(Z,\bX) \cond \potY{0},\bX \big\}  
    \\
    &
    =
    \iint 
    a h(z,\bX) s\ETA(Y,A=a,Z=z,\bX) 
    f\ETA(A=a,Z=z \cond \potY{0},\bX)
    \, d(a,z)
    \\
    &
    \qquad -
    s\ETA (\potY{0},\bX) \EXP\ETA \big\{ A h(Z,\bX) \cond \potY{0},\bX \big\}  
    \\
    &
    =
    \EXP\ETA \big\{ s\ETA(\bO) A h(Z,\bX)  \cond \potY{0},\bX \big\}
    -
    s\ETA (\potY{0},\bX) \EXP\ETA \big\{ A h(Z,\bX) \cond \potY{0},\bX \big\}  \ .
    \numeq
    \label{eq-proof-alpha score 3}
\end{align*}
Combining \eqref{eq-proof-alpha score 2} and \eqref{eq-proof-alpha score 3}, we find
\begin{align*}
    &
    \nabla\ETA 
    \EXP\ETA \big\{ (1-A) g(\potY{0},Z,\bX) + A h(Z,\bX)  \cond \potY{0},\bX \big\}
    \\
    &
    =
    \EXP\ETA \big[ s\ETA(\bO) \{ (1-A) g(\potY{0},Z,\bX) + A h(Z,\bX) \}  \cond \potY{0},\bX \big]
    \\
    &
    \qquad 
    -
    s\ETA (\potY{0},\bX) \EXP\ETA \big\{ (1-A) g(\potY{0},Z,\bX) + A h(Z,\bX) \cond \potY{0},\bX \big\} \ .
    \numeq
    \label{eq-proof-alpha score 4}
\end{align*}
Replacing $(1-A)g(\potY{0},Z,\bX) + A h(Z,\bX)$ in \eqref{eq-proof-alpha score 4} with the function in \eqref{eq-proof-alpha score 1}, we find
\begin{align*}
    & 
    \nabla\ETA 
    \EXP\ETA \Bigg[ 
    \begin{array}{l}         
    (1-A) \alpha_t(Y,\bX) \beta_t(Z,\bX) \big\{ \widetilde{u}_t (\potY{0},Z,\bX) - \mu_t(Z,\bX \con \widetilde{u}_t) \big\}
    \\
    +
    A
    \mu_t(Z,\bX \con \widetilde{u}_t)
    +
    (1-A) \widetilde{u}_t (\potY{0},Z,\bX)
    \end{array}
    \, \Bigg| \, \potY{0}, \bX \Bigg]
    \Bigg|_{t = \eta}
    \\
    & 
    =
    \EXP\ETA 
    \Bigg[  
    s\ETA(\bO) 
    \Bigg[
    \begin{array}{l}
    (1-A) \alpha\ETA(Y,\bX) 
    \beta\ETA(Z,\bX)
    \big\{ \widetilde{u}\ETA(Y,Z,\bX) - \mu\ETA(Z,\bX \con \widetilde{u}\ETA) \big\} 
    \\
    +
    A
    \mu\ETA(Z,\bX \con \widetilde{u}\ETA)
    +
    (1-A) \widetilde{u}\ETA(Y,Z,\bX)
    \end{array}
    \Bigg] 
    \,  
    \Bigg|\,   \potY{0},\bX  
    \Bigg]
    \\
    &\!\!
    \qquad
    - s\ETA(\potY{0},\bX) 
    \underbrace{
    \EXP\ETA 
    \Bigg[  
    \begin{array}{l}
    (1-A) \alpha\ETA(Y,\bX) 
    \beta\ETA(Z,\bX)
    \big\{ \widetilde{u}\ETA(Y,Z,\bX) - \mu\ETA(Z,\bX \con \widetilde{u}\ETA) \big\} 
    \\
    +
    A
    \mu\ETA(Z,\bX \con \widetilde{u}\ETA)
    +
    (1-A) \widetilde{u}\ETA(Y,Z,\bX)
    \end{array}  
    \,  
    \Bigg|\,   \potY{0},\bX  
    \Bigg]
    }_{=0 \text{ from } \eqref{eq-proof-alpha score 1}}
    \\
    & 
    =
    \EXP\ETA 
    \Bigg[  
    s\ETA(\bO) 
    \Bigg[
    \begin{array}{l} 
    (1-A) \alpha\ETA(Y,\bX) 
    \beta\ETA(Z,\bX)
    \big\{ \widetilde{u}\ETA(Y,Z,\bX) - \mu\ETA(Z,\bX \con \widetilde{u}\ETA) \big\} 
    \\ 
    +
    A
    \mu\ETA(Z,\bX \con \widetilde{u}\ETA)
    +
    (1-A) \widetilde{u}\ETA(Y,Z,\bX)
    \end{array}
    \Bigg] 
    \, 
    \Bigg|\,   \potY{0},\bX  
    \Bigg] .
    \numeq
    \label{eq-proof-alpha score 5}
\end{align*}

We now take the pathwise derivative of \eqref{eq-proof-alpha score 1} with respect to $\eta$, which yields
\begin{align*}
    0
    & 
    =
    \nabla\ETA 
    \EXP\ETA \Bigg[ 
    \begin{array}{l}
    (1-A) \alpha\ETA(Y,\bX) 
    \beta\ETA(Z,\bX)
    \big\{ \widetilde{u}\ETA(Y,Z,\bX) - \mu\ETA(Z,\bX \con \widetilde{u}\ETA) \big\}
    \\
    +
    A
    \mu\ETA(Z,\bX \con \widetilde{u}\ETA)
    +
    (1-A) \widetilde{u}\ETA(Y,Z,\bX)
    \end{array}
    \, \Bigg| \, \potY{0},\bX 
    \Bigg] 
    \numeq \label{eq-proof-alpha00}
    \\
    & 
    \stackrel{\eqref{eq-proof-alpha score 5}}{=}
    \EXP\ETA \Bigg[ 
    s\ETA(\bO)
    \Bigg[
    \begin{array}{l}
    (1-A) \alpha\ETA(Y,\bX) 
    \beta\ETA(Z,\bX)
    \big\{ \widetilde{u}\ETA(Y,Z,\bX) - \mu\ETA(Z,\bX \con \widetilde{u}\ETA) \big\}
    \\
    +
    A
    \mu\ETA(Z,\bX \con \widetilde{u}\ETA)
    +
    (1-A) \widetilde{u}\ETA(Y,Z,\bX)
    \end{array}
    \Bigg]
    \, \Bigg| \, \potY{0},\bX 
    \Bigg]  
    \\
    &
    \hspace*{0.5cm}
    +
    \EXP\ETA \Bigg[  
    s\ETA(Y,\bX \con \alpha) 
    \bigg[
    \begin{array}{l}    
    (1-A) \alpha\ETA(Y,\bX) 
    \beta\ETA(Z,\bX) \\
    \times \big\{ 
    \widetilde{u}\ETA(Y,Z,\bX) - \mu\ETA(Z,\bX \con \widetilde{u}\ETA) 
    \big\} 
    \end{array}
    \bigg]
    \, \Bigg| \, \potY{0},\bX
    \Bigg] 
    \\
    &
    \hspace*{0.5cm}
    +
    \EXP\ETA \Bigg[  
    s\ETA(Z,\bX \con \beta) 
    \bigg[
    \begin{array}{l}    
    (1-A) \alpha\ETA(Y,\bX) 
    \beta\ETA(Z,\bX) \\
    \times \big\{ 
    \widetilde{u}\ETA(Y,Z,\bX) - \mu\ETA(Z,\bX \con \widetilde{u}\ETA) 
    \big\} 
    \end{array}
    \bigg]
    \, \Bigg| \, \potY{0},\bX
    \Bigg] 
    \numeq \label{eq-proof-alpha 1}
    \\
    &
    \hspace*{0.5cm}
    +
    \EXP\ETA \left[  
    \begin{array}{l}
         (1-A) \alpha\ETA(Y,\bX) 
    \beta\ETA(Z,\bX)
    \\
    \quad 
    \times 
    \big\{ 
    \nabla\ETA  \widetilde{u}\ETA(Y,Z,\bX) - \nabla\ETA  \mu\ETA(Z,\bX \con \widetilde{u}\ETA) \big\} 
    \\
    + 
    \big\{ A \nabla\ETA \mu\ETA(Z,\bX \con \widetilde{u}\ETA) 
    + (1-A) \nabla\ETA \widetilde{u}\ETA(Y,Z,\bX) \big\}
    \end{array} 
    \, \Bigg| \, \potY{0}, \bX
    \right]  \ . 
    \numeq \label{eq-proof-alpha 2}
\end{align*}
Note that \eqref{eq-proof-alpha 1} and \eqref{eq-proof-alpha 2} are given by:
\begin{align*}
    &
    \EXP\ETA \Big[  
    s\ETA(Z,\bX \con \beta)
    (1-A) \alpha\ETA(Y,\bX) 
    \beta\ETA(Z,\bX) \big\{ \widetilde{u}\ETA(Y,Z,\bX) - \mu\ETA(Z,\bX \con \widetilde{u}\ETA) \big\} 
    \, \Big| \, \potY{0}, \bX
    \Big] 
    \\
    &
    =
    \EXP\ETA \Bigg[  
    s\ETA(Z,\bX \con \beta) 
    \EXP\ETA \Bigg[  
    \begin{array}{l}
    (1-A) \alpha\ETA(Y,\bX) 
    \beta\ETA(Z,\bX) \\
    \times \big\{ \widetilde{u}\ETA(Y,Z,\bX) - \mu\ETA(Z,\bX \con \widetilde{u}\ETA) \big\} 
    \end{array}
    \, \Bigg| \, \potY{0}, Z, \bX
    \Bigg]
    \, \Bigg| \, \potY{0}, \bX
    \Bigg] 
    \\
    &
    =
    \EXP\ETA \Big[  
    s\ETA(Z,\bX \con \beta) 
    \EXP\ETA \Big[  
    A 
    \big\{ \widetilde{u}\ETA(\potY{0},Z,\bX) - \mu\ETA(Z,\bX \con \widetilde{u}\ETA) \big\} 
    \, \Big| \, \potY{0}, Z, \bX
    \Big]
    \, \Big| \, \potY{0}, \bX
    \Big] 
    \\
    &
    =
    \EXP\ETA \Big[  
    s\ETA(Z,\bX \con \beta) 
    A
    \big\{ \widetilde{u}\ETA(\potY{0},Z,\bX) - \mu\ETA(Z,\bX \con \widetilde{u}\ETA) \big\} 
    \, \Big| \, \potY{0}, \bX
    \Big] \ ,
\end{align*}
and
\begin{align*}
    &
    \EXP\ETA \left[  
    \begin{array}{l}
         (1-A) \alpha\ETA(Y,\bX) 
    \beta\ETA(Z,\bX) 
    \big\{ 
    \nabla\ETA  \widetilde{u}\ETA(Y,Z,\bX) - \nabla\ETA  \mu\ETA(Z,\bX \con \widetilde{u}\ETA) \big\} 
    \\
    + 
    \big\{ A \nabla\ETA \mu\ETA(Z,\bX \con \widetilde{u}\ETA) 
    + (1-A) \nabla\ETA \widetilde{u}\ETA(Y,Z,\bX) \big\} 
    \end{array} 
    \, \Bigg| \, \potY{0}, \bX
    \right] 
    \\
    &
    =
    \EXP\ETA \big\{ \nabla\ETA \widetilde{u}\ETA(\potY{0},Z,\bX) \cond \potY{0}, \bX \big\}   \ .
\end{align*} 
Rearranging \eqref{eq-proof-alpha00} by using the above identities, we obtain the following representation of $s\ETA (\potY{0},\bX \con \alpha)$:
\begin{align*}
&
s\ETA (\potY{0},\bX \con \alpha)
\numeq 
\label{eq-proof-alpha gradient}
\\
&
=
- 
\frac{
\EXP\ETA
\left[
\begin{array}{l}
     s\ETA(\bO) \bigg[
     \begin{array}{l}
    (1-A) \alpha\ETA(Y,\bX) 
    \beta\ETA(Z,\bX) \big\{ \widetilde{u}\ETA(Y,Z,\bX) - \mu\ETA(Z,\bX \con \widetilde{u}\ETA) \big\}
    \\
    +
    A
    \mu\ETA(Z,\bX \con \widetilde{u}\ETA)
    +
    (1-A) \widetilde{u}\ETA(Y,Z,\bX)
     \end{array}
    \bigg] 
    \\
    +
    s\ETA(Z,\bX \con \beta)
    A 
    \big\{ \widetilde{u}\ETA(\potY{0},Z,\bX) - \mu\ETA(Z,\bX \con \widetilde{u}\ETA) \big\}
    +
    \nabla\ETA \widetilde{u}\ETA(\potY{0},Z,\bX)
\end{array}
    \, \Bigg| \, \potY{0}, \bX
\right]
}
{ \EXP \big[ (1-A) \alpha\ETA(Y,\bX) \beta\ETA(Z,\bX)
\big\{ \widetilde{u}\ETA(Y,Z,\bX) - \mu\ETA(Z,\bX \con \widetilde{u}\ETA) \big\}
\cond \potY{0}, \bX \big]
}
 \ . 
\end{align*}

In \eqref{eq-proof-IF piece23}, we substitute $s\ETA (\potY{0},\bX \con \alpha)$ using the expression in \eqref{eq-proof-alpha gradient} after taking $u = \nu - v$ where $\nu$ satisfies \HL{$\omega$-i} and  \HL{$\omega$-ii}, i.e., 
    \begin{align*}
    &
    u(\potY{0},Z,\bX)
    =
    \nu(\potY{0},Z,\bX) - v(\potY{0},Z,\bX) 
    \ , 
    \numeq \label{eq-proof-u nu v}
    \\
    &
    \EXP\ETA \big[ A \big\{ \mathcal{G}^{(0)} - \mu\ETA(Z,\bX \con \mathcal{G}) 
    \big\}  \cond  \potY{0},\bX \big] 
    \nonumber
    \\
    &
    =
    \EXP\ETA \big[ A \big\{  \widetilde{\nu}\ETA (\potY{0},Z,\bX) - \mu\ETA (Z,\bX \con \widetilde{\nu}\ETA) \big\} \cond  \potY{0},\bX \big]
    \ .
    \numeq \label{eq-proof-nu}
    \end{align*} 
We then obtain
\begin{align*}
    & 
    \eqref{eq-proof-IF piece23}
    \\
    &
    =
    -
    \frac{     
    \EXP\ETA \bigg[ 
    \begin{array}{l}
    s\ETA (Y,\bX \con \alpha)
    (1-A) 
    \alpha\ETA (Y,\bX)
    \beta\ETA (Z,\bX)
    \\
    \times 
    \big\{
    \mathcal{G} - \mu\ETA(Z,\bX \con \mathcal{G})
    -
    \widetilde{v}\ETA(Y,Z,\bX)
    +
    \mu\ETA (Z,\bX \con \widetilde{v}\ETA)
    \big\}
    \end{array}
    \bigg]
    }{\EXP\ETA (A)} 
    \\
    &
    =
    -
    \frac{     
    \EXP\ETA \big[ 
    s\ETA (\potY{0},\bX \con \alpha)
    A
    \big\{
    \mathcal{G}^{(0)} - \mu\ETA(Z,\bX \con \mathcal{G})
    -
    \widetilde{v}\ETA(\potY{0},Z,\bX)
    +
    \mu\ETA (Z,\bX \con \widetilde{v}\ETA)
    \big\} 
    \big]
    }{\EXP\ETA (A)} 
    \\
    &
    =
    -
    \frac{     
    \EXP\ETA 
    \bigg[ 
    s\ETA (\potY{0},\bX \con \alpha)
    \EXP\ETA 
    \bigg[ 
    A
    \bigg\{
    \begin{array}{l}
    \mathcal{G}^{(0)} - \mu\ETA(Z,\bX \con \mathcal{G})
    \\
    -
    \widetilde{v}\ETA(\potY{0},Z,\bX)
    -
    \mu\ETA (Z,\bX \con \widetilde{v}\ETA)
    \end{array}
    \bigg\}
    \, \bigg| \, \potY{0},\bX
    \bigg] 
    \bigg]
    }{\EXP\ETA (A)} 
    \\
    &
    \stackrel{\eqref{eq-proof-nu}}{=}
    -
    \frac{     
    \EXP\ETA \bigg[ 
    s\ETA (\potY{0},\bX \con \alpha)
    \EXP\ETA 
    \bigg[ 
    A
    \bigg\{
    \begin{array}{l}
    \widetilde{\nu}\ETA(\potY{0},Z,\bX)
    -
    \widetilde{v}\ETA(\potY{0},Z,\bX)
    \\
    -
    \mu\ETA (Z,\bX \con \widetilde{\nu}\ETA)
    +
    \mu\ETA (Z,\bX \con \widetilde{v}\ETA)
    \end{array}
    \bigg\}
    \, \bigg| \, \potY{0},\bX
    \bigg] 
    \bigg]
    }{\EXP\ETA (A)} 
    \\
    &
    \stackrel{\eqref{eq-proof-u nu v}}{=}
    -
    \frac{     
    \EXP\ETA \big[ 
    s\ETA (\potY{0},\bX \con \alpha)
    \EXP\ETA 
    \big[ 
    A
    \big\{
    \widetilde{u}\ETA(\potY{0},Z,\bX) 
    -
    \mu\ETA (Z,\bX \con \widetilde{u}\ETA)
    \big\}
    \cond \potY{0},\bX 
    \big] 
    \big]
    }{\EXP\ETA (A)} 
    \\
    &
    =
    -
    \frac{     
    \EXP\ETA \big[ 
    s\ETA (\potY{0},\bX \con \alpha)
    \EXP\ETA \big[
    (1-A) 
    \alpha\ETA(Y, \bX)
    \beta\ETA(Z, \bX)
    \big\{
    \widetilde{u}\ETA(Y,Z,\bX) 
    -
    \mu\ETA (Z,\bX \con \widetilde{u}\ETA)
    \big\}
    \cond \potY{0},\bX
    \big] 
    \big]
    }{\EXP\ETA (A)} 
    \\
    &
    \stackrel{\eqref{eq-proof-alpha gradient}}{=}
    \frac{
    \displaystyle{    \EXP\ETA 
    \left[ 
        \begin{array}{l}
\frac{
\displaystyle{
\EXP\ETA
\left[
\begin{array}{l}
     s\ETA(\bO) \bigg[
     \begin{array}{l}
    (1-A) \alpha\ETA(Y,\bX) 
    \beta\ETA(Z,\bX) \big\{ \widetilde{u}\ETA(Y,Z,\bX) - \mu\ETA(Z,\bX \con \widetilde{u}\ETA) \big\}
    \\
    +
    A
    \mu\ETA(Z,\bX \con \widetilde{u}\ETA)
    +
    (1-A) \widetilde{u}\ETA(Y,Z,\bX)
     \end{array}
    \bigg] 
    \\
    +
    s\ETA(Z,\bX \con \beta)
    A 
    \big\{ \widetilde{u}\ETA(\potY{0},Z,\bX) - \mu\ETA(Z,\bX \con \widetilde{u}\ETA) \big\}
    +
    \nabla\ETA \widetilde{u}\ETA(\potY{0},Z,\bX)
\end{array}
    \, \Bigg| \, \potY{0}, \bX
\right]
}
}
{
\displaystyle{\EXP \big[ (1-A) \alpha\ETA(Y,\bX) \beta\ETA(Z,\bX)
\big\{ \widetilde{u}\ETA(Y,Z,\bX) - \mu\ETA(Z,\bX \con \widetilde{u}\ETA) \big\}
\cond \potY{0}, \bX \big]
}
}
\\[0.5cm]
\quad 
\times
\EXP\ETA \big[
    (1-A) 
    \alpha\ETA(Y, \bX)
    \beta\ETA(Z, \bX)
    \big\{
    \widetilde{u}\ETA(Y,Z,\bX) 
    -
    \mu\ETA (Z,\bX \con \widetilde{u}\ETA)
    \big\}
    \, \big| \, \potY{0}, \bX \big] 
        \end{array}
    \right]}
    }{\EXP\ETA(A)}
\\
 &
    = 
    \frac{ 
    \displaystyle{
    \EXP\ETA 
    \left[ 
    \EXP\ETA
    \left[
    \begin{array}{l}
     s\ETA(\bO) \bigg[
     \begin{array}{l}
    (1-A) \alpha\ETA(Y,\bX) 
    \beta\ETA(Z,\bX) \big\{ \widetilde{u}\ETA(Y,Z,\bX) - \mu\ETA(Z,\bX \con \widetilde{u}\ETA) \big\}
    \\
    +
    A
    \mu\ETA(Z,\bX \con \widetilde{u}\ETA)
    +
    (1-A) \widetilde{u}\ETA(Y,Z,\bX)
     \end{array}
    \bigg] 
    \\
    +
    s\ETA(Z,\bX \con \beta)
    A 
    \big\{ \widetilde{u}\ETA(\potY{0},Z,\bX) - \mu\ETA(Z,\bX \con \widetilde{u}\ETA) \big\}
    +
    \nabla\ETA \widetilde{u}\ETA(\potY{0},Z,\bX)
\end{array}
    \, \Bigg| \, \potY{0}, \bX
\right]
\right]}
    }{\EXP\ETA(A)}
\\
    &
    =
    \frac{\displaystyle{
    \EXP\ETA
    \left[
    \begin{array}{l}
     s\ETA(\bO) \bigg[
     \begin{array}{l}
    (1-A) \alpha\ETA(Y,\bX) 
    \beta\ETA(Z,\bX) \big\{ \widetilde{u}\ETA(Y,Z,\bX) - \mu\ETA(Z,\bX \con \widetilde{u}\ETA) \big\}
    \\
    +
    A
    \mu\ETA(Z,\bX \con \widetilde{u}\ETA)
    +
    (1-A) \widetilde{u}\ETA(Y,Z,\bX)
     \end{array}
    \bigg] 
    \\
    +
    s\ETA(Z,\bX \con \beta)
    A 
    \big\{ \widetilde{u}\ETA(\potY{0},Z,\bX) - \mu\ETA(Z,\bX \con \widetilde{u}\ETA) \big\}
    +
    \nabla\ETA \widetilde{u}\ETA(\potY{0},Z,\bX)
\end{array} 
\right]}}{\EXP\ETA(A)}
\\
    &
    \stackrel{\eqref{eq-proof-identity3}, \eqref{eq-proof-identity5}}{=}
    \frac{\displaystyle{\EXP\ETA
    \left[
    \begin{array}{l}
     s\ETA(\bO) \bigg[
     \begin{array}{l}
    (1-A) \alpha\ETA(Y,\bX) 
    \beta\ETA(Z,\bX) \big\{ \widetilde{u}\ETA(Y,Z,\bX) - \mu\ETA(Z,\bX \con \widetilde{u}\ETA) \big\}
    \\
    +
    A
    \mu\ETA(Z,\bX \con \widetilde{u}\ETA)
    +
    (1-A) \widetilde{u}\ETA(Y,Z,\bX)
     \end{array}
    \bigg]  
\end{array} 
\right] }}{\EXP\ETA(A)}
    \ .
\numeq 
\label{eq-proof-IF piece23 final}
\end{align*}

Combining \eqref{eq-proof-IF piece1} with \eqref{eq-proof-IF piece23 final}, we obtain the representation of \eqref{eq-proof-pathwise diff}:
\begin{align*}
    & 
    \nabla\ETA \tau\ETA
    \\
    & 
    =
    \eqref{eq-proof-IF piece1} + \eqref{eq-proof-IF piece23 final}
    \\
    &
    =
    \frac{  \!\! \EXP\ETA  \!\! \left[  \! s\ETA(\bO) \!\!
    \left[ 
    \begin{array}{l}
         (1-A) \alpha\ETA(Y,\bX) \beta\ETA(Z,X) 
         \big\{ \mathcal{G} - \mu\ETA(Z,\bX \con \mathcal{G}) \big\} \\
         + A \big\{\mu\ETA(Z,\bX \con \mathcal{G}) - \tau\ETA \big\}
         \\
         +
         (1-A) \alpha\ETA(Y,\bX ) \beta\ETA(Z,\bX)   
    \big\{ 
        \widetilde{v}\ETA(Y,Z,\bX ) 
        -
        \mu\ETA(Z,\bX \con \widetilde{v}\ETA)
    \big\}
    \\
         +
         (1-A) \alpha\ETA(Y,\bX ) \beta\ETA(Z,\bX)   
    \big\{ 
        \widetilde{u}\ETA(Y,Z,\bX ) 
        -
        \mu\ETA(Z,\bX \con \widetilde{u}\ETA)
    \big\}
    \\
    +
    (1-A)  \widetilde{v}\ETA(Y,Z,\bX) 
    +
    A
    \mu\ETA(Z,\bX \con \widetilde{v}\ETA)
    \\
    +
    (1-A)  \widetilde{u}\ETA(Y,Z,\bX) 
    +
    A
    \mu\ETA(Z,\bX \con \widetilde{u}\ETA)
    \end{array}
    \!
    \right]
    \!
    \right]
    \!\!
    }{ \EXP\ETA(A) }
    \\
    &
    \stackrel{\eqref{eq-proof-u nu v}}{=}
     \frac{  \!\! \EXP\ETA  \!\! \left[  \! s\ETA(\bO) \!\!
    \left[ 
    \begin{array}{l}
         (1-A) \alpha\ETA(Y,\bX) \beta\ETA(Z,X) 
         \big\{ \mathcal{G} - \mu\ETA(Z,\bX \con \mathcal{G}) \big\} \\
         + A \big\{\mu\ETA(Z,\bX \con \mathcal{G}) - \tau\ETA \big\}
         \\
         +
         (1-A) \alpha\ETA(Y,\bX ) \beta\ETA(Z,\bX)   
    \big\{ 
        \widetilde{\nu}\ETA(Y,Z,\bX ) 
        -
        \mu\ETA(Z,\bX \con \widetilde{\nu}\ETA)
    \big\} 
    \\
    +
    (1-A)  \widetilde{\nu}\ETA(Y,Z,\bX) 
    +
    A
    \mu\ETA(Z,\bX \con \widetilde{\nu}\ETA) 
    \end{array}
    \!
    \right]
    \!
    \right]
    \!\!
    }{ \EXP\ETA(A) }  
    \ .
    \numeq 
\label{eq-proof-pathwise diff end}
\end{align*}
Since $\widetilde{\nu}_{\eta^*} = \nu - \nu_{\eta^*}^\dagger = \omega$ defined in Theorem \ref{thm-IF}, \eqref{eq-proof-pathwise diff end} implies 
\begin{align*}
    & 
    \nabla\ETA \tau\ETA \big|_{\eta = \eta^*}
    =
    \EXP
    \big\{
        s^*(\bO) 
        \InfFt^*(\bO)
    \big\} \ ,
\end{align*}
where $\InfFt^*(\bO)$ is given in \eqref{eq-IF} with $\mathcal{G}^{(0)}$ being used in the definition of $\tau^*$. Therefore, $\tau^*$ is a differentiable parameter \citep{Newey1990}, and $\InfFt^*(\bO)$ serves as its IF. 

Lastly, we find any IF for $\tau^*$ in model $\model$ must be of the form given in \eqref{eq-IF}. From semiparametric efficiency theory, the full collection of influence functions for $\tau^*$ in model $\model$ can be expressed as 
\begin{align*}
    \Big\{ 
    \InfFt^*(\bO)
    +
    t(\bO)
    \, \Big| \,
    t(\bO) \in \mathcal{T}^\perp
    \Big\} \ ,
    \numeq
    \label{eq-proof-entire IF}
\end{align*}
where $\InfFt^*(\bO)$ is an IF for $\tau^*$ in model $\model$ and $\mathcal{T}^\perp$ denotes the orthocomplement of the tangent space $\mathcal{T}$ of model $\model$. Therefore, it remains to characterize $\mathcal{T}^\perp$, which we derive by using the general framework outlined in \citet{RR1997} and \citet{Robins2000_Sensitivity}.

Let $\mathcal{T}^F$ denote the tangent space of the model for the full data $\bO^F$. This space is decomposed as $\mathcal{T}^F = \mathcal{T}_1^F \oplus \mathcal{T}_2^F \oplus \mathcal{T}_3^F$, where $\mathcal{T}_1^F$, $\mathcal{T}_2^F$, and $\mathcal{T}_3^F$ are the tangent space for $(\potY{0},Z,\bX)$, $A \cond (\potY{0},Z,\bX)$, and $\potY{1} \cond (\potY{0},A,Z,\bX)$, respectively. We first provide the closed-form representation of $\mathcal{T}_{j}^F$ for $j \in \{1,2,3\}$. From \HL{IV2}, we have
\begin{align*}
    & \mathcal{T}_1^F
    =
    \left\{
    \!
    \begin{array}{l}
         S(\potY{0} \cond \bX) + S(Z \cond \bX) + S(\bX)
    \end{array}
    \!
    \left| \,
    \begin{array}{l}
        \EXP \{ S(\potY{0} \cond \bX) \cond \bX \}
    \\
    =
    \EXP \{ S(Z \cond \bX) \cond \bX \}
    \\
    =
    \EXP \{ S(\bX) \}
    =
    0
    \end{array}        
    \right\}    
    \right.
    \!
    \cap 
    \mathcal{L}_0^2(\potY{0},Z,\bX) \ .
\end{align*}

We next characterize $\mathcal{T}_2^F$. From \HL{IV4}, the conditional probability $f\ETA(A=1 \cond \potY{0},Z,\bX)$,  at the submodel $P_\eta$ is represented as:
\begin{align*}
    &
    f\ETA(A=1 \cond \potY{0},Z,\bX)
    =
    \frac{ \exp\{ \mathfrak{a}\ETA(\potY{0},\bX) + \mathfrak{b}\ETA(Z,\bX) \} }{ 1+\exp\{ \mathfrak{a}\ETA(\potY{0},\bX) + \mathfrak{b}\ETA(Z,\bX) \} } \ ,
\end{align*}
where $\mathfrak{a}\ETA = \log(\alpha\ETA) $ and $\mathfrak{b}\ETA = \log(\beta\ETA) $. This implies that the corresponding score function is
\begin{align*}
    &
    s\ETA (A \cond \potY{0},Z,\bX)
    =
    \big\{ A - f\ETA(A=1 \cond \potY{0},Z,\bX) \big\}
    \big\{
    \nabla_{\eta} \mathfrak{a}\ETA(\potY{0},\bX)
    +
    \nabla_{\eta} \mathfrak{b}\ETA(Z,\bX)
    \big\} \ .
\end{align*}
Consequently, taking the closed linear span of this set yields the closed-form representation of $\mathcal{T}_2^F$:
\begin{align*}
    \mathcal{T}_2^F
    & 
    =
    \left\{
    \!
        \begin{array}{l}
        \big\{ A - f^*(A=1 \cond \potY{0},Z,\bX) \big\} 
        \\
        \times 
        \big\{ \mathfrak{a}(\potY{0},\bX) + \mathfrak{b}(Z,\bX) \big\} 
        \end{array}
        \!
        \left|
        \, 
        \begin{array}{l}             
        \forall \mathfrak{a} \in \mathcal{L}^2(\potY{0},\bX)
        \\ \text{s.t. } \mathfrak{a}(y_R,\bX)=0 
        \\[0.1cm]
        \forall \mathfrak{b} \in \mathcal{L}^2(Z,\bX) 
        \end{array}
    \!
    \right\} 
    \right. 
    \! \cap  \mathcal{L}_0^2(\potY{0},A,Z,\bX)
    \ .
    \numeq 
    \label{eq-proof-tangent2 full}
\end{align*}

Lastly, $\mathcal{T}_3^F$ can be expressed as
\begin{align*}
    & \mathcal{T}_3^F
    =
    \Big\{ 
        S(\bO^F)
        \, \Big| \, 
        \EXP \{ S(\bO^F) 
        \cond 
        \potY{0},A,Z,\bX
        \} = 0
    \Big\} 
    \cap 
    \mathcal{L}_0^2(\bO^F) \ .
\end{align*}
It is straightforward to show that the orthocomplement of $\mathcal{T}_3^{F}$, denoted by $ \mathcal{T}_3^{F,\perp} $, is equal to
\begin{align*}
    \mathcal{T}_3^{F,\perp} 
    &
    =
    \Big\{ b(\bO^F) \, \Big| \, \EXP \big\{ b(\bO^F) S(\bO^F) \big\} = 0 \ , \ \forall S \in \mathcal{T}_3^F \Big\}  
    \cap \mathcal{L}_0^2(\bO^F)
    \\
    &
    = 
    \Big\{ b(\potY{0},A,Z,\bX) \, \Big| \, \forall b \in \mathcal{L}_0^2(\potY{0},A,Z,\bX) \Big\} 
    \cap \mathcal{L}_0^2(\bO^F) 
    \ .
    \numeq 
    \label{eq-proof-T3Fp}
\end{align*}

By \citet{RR1997}, the observed data tangent space $\mathcal{T}$ can be written as
\begin{align*}
    \mathcal{T} = \text{closure} \big( \mathcal{T}_1 + \mathcal{T}_2 + \mathcal{T}_3 \big)
    \ , \quad 
    \mathcal{T}_j = \text{closure}\big( \mathcal{R}(\varphi \circ \Pi_j) \big) \ .
\end{align*}
Here, $\mathcal{R}(\cdot)$ denotes the range of the operator
$\varphi: \mathcal{L}_0^2(\bO^F) \rightarrow \mathcal{L}_0^2(\bO)$ with $\varphi(f(\bO^F))= \EXP \{ f(\bO^F)\cond \bO \}$, $\Pi_j: \mathcal{L}_0^2(\bO^F) \rightarrow \mathcal{T}_j^F$ for $j \in \{1,2,3\}$ is the Hilbert space projection, and $\text{closure}(\mathcal{S})$ denotes the closed linear span of a set $\mathcal{S}$. We also denote $\bO_{1}$ and $\bO_{0}$ to denote the observed-data components associated with $\potY{1}$ and $\potY{0}$, respectively, i.e., $\bO_{1} = (AY, A,Z,\bX)$ and $\bO_{0} = ((1-A)Y, A,Z,\bX)$.

As shown in \citet{BKRW1998}, the orthocomplement of $\mathcal{T}$ is given by $\mathcal{T}^{\perp}=\mathcal{T}_1^{\perp} \cap \mathcal{T}_2^{\perp} \cap \mathcal{T}_3^{\perp}$. We first characterize $\mathcal{T}_3^\perp$. Following equation (4.6) of \citet{Robins2000_Sensitivity}, we can  establish that 
\begin{align*}
    \mathcal{T}_3^{\perp} 
    & 
    =
    \Big\{ b(\bO) \, \Big| \, b(\bO) \in \mathcal{T}_3^{F,\perp} \Big\}  \cap \mathcal{L}_0^2(\bO)
    \\
    &
    \stackrel{\eqref{eq-proof-T3Fp}}{=}
    \Big\{ b(\bO) \, \Big| \, 
    b(\bO)
    =
    b(\potY{0},A,Z,\bX)
    \Big\}  \cap \mathcal{L}_0^2(\bO)
    \\
    &
    =
    \mathcal{L}_0^2( \bO_{0} ) \ .
\end{align*}
Therefore, we have 
\begin{align*}
    \mathcal{T}^\perp 
    &
    = 
    \big\{ \mathcal{T}_1^\perp \cap \mathcal{L}_0^2(\bO_0) \big\}
    \cap 
    \big\{ \mathcal{T}_2^\perp \cap \mathcal{L}_0^2(\bO_0) \big\}
    \\
    &
    = 
    \underbrace{ \big\{ \mathcal{L}_0^2(\bO_0) \setminus \mathcal{T}_1  \big\} }_{\equiv \mathcal{T}_1^\perp |_0 }
    \cap 
    \underbrace{ \big\{ \mathcal{L}_0^2(\bO_0) \setminus \mathcal{T}_2  \big\} }_{\equiv \mathcal{T}_2^\perp |_0 }
    \\
    &
    =
    \mathcal{T}_1^\perp |_0
    \cap \mathcal{T}_2^\perp |_0
    \ .
    \numeq \label{eq-proof-tangent perp}
\end{align*}

By \citet{BKRW1998} and \citet{TTRR2010}, the orthocomplement of $\mathcal{T}_1^{F}$ within $\mathcal{L}_0^2(\potY{0},Z,\bX)$ is given by 
\begin{align*}
    \mathcal{T}_1^{F, \perp}
    |_{0} 
    &
    \equiv 
    \mathcal{L}_0^2(\potY{0},Z,\bX) \setminus \mathcal{T}_1^{F}
    \\
    &
    =
    \big\{ 
    \widetilde{v}(\potY{0},Z,\bX)
    \cond \forall v \in \mathcal{L}^2(\potY{0},Z,\bX) , \
    \widetilde{v} = v - v^\dagger 
    \big\} \ , 
    \numeq \label{eq-proof-tangent1F}
\end{align*}   
where $v^\dagger = \EXP ( v \cond Z,\bX ) + \EXP ( v \cond \potY{0},\bX ) - \EXP ( v \cond \bX )$.

Propositions A1.3 and A2.3 of \citet{RR1997} established that the orthocomplement of $\mathcal{T}_1$ within $\mathcal{L}_0^2(\bO_0)$ is given by 
\begin{align*}
    \mathcal{T}_1^{\perp}
    |_{0} 
    &
    \equiv
    \mathcal{L}_0^2(\bO_0) \setminus \mathcal{T}_1
    \\
    &
    =
    \Bigg\{ 
    \frac{(1-A) m(\potY{0},Z,\bX)}{f^*(A=0 \cond \potY{0},Z,\bX)}  + T_{car}
    \, \Bigg| \, 
    m(\potY{0},Z,\bX) \in\mathcal{T}_1^{F, \perp}
    |_{0} 
    , \ 
    T_{car} \in \mathcal{T}_{car}
    \Bigg\}
    \ ,
    \\
    \mathcal{T}_{car}
    &
    =
    \Bigg\{ 
    \bigg\{ \frac{ 1-A }{f^*(A=0 \cond \potY{0},Z,\bX)} - 1 \bigg\} k(Z,\bX) \, \Bigg| \, \forall k \in \mathcal{L}_0^2(Z,\bX)
    \Bigg\}
    \ .
\end{align*}
From \eqref{eq-proof-tangent1F}, we have
\begin{align*}
    \mathcal{T}_1^{\perp} |_{0}
    &
    =
    \left\{
    \begin{array}{l}
    \frac{(1-A) \widetilde{v} (\potY{0},Z,\bX) }{f^*(A=0 \cond \potY{0},Z,\bX)} \\
    - 
    \big\{ \frac{ 1-A }{f^*(A=0 \cond \potY{0},Z,\bX)} - 1 \big\} k(Z,\bX)
    \end{array}
    \, \Bigg| \, 
    \begin{array}{l}
    \forall v \in \mathcal{L}^2(\potY{0},Z,\bX): \widetilde{v} = v - v^\dagger
    \\
    \forall k \in \mathcal{L}_0^2(Z,\bX)
    \end{array}
    \right\} \ . 
\end{align*}
In addition, equation (4.6) of \citet{Robins2000_Sensitivity} established that 
\begin{align*}
    & 
    \mathcal{T}_2^{\perp}|_{0}
    \equiv 
    \mathcal{L}_0^2(\bO_0) \setminus \mathcal{T}_2
    =
    \Big\{ b(\bO_0) \, \Big| \, b(\bO_0) \in \mathcal{T}_2^{F,\perp} \Big\}
    \ .
\end{align*} 

By \eqref{eq-proof-tangent perp}, we obtain
\begin{align*}
    \mathcal{T}^{\perp}
    =
    \Big\{
        T_{1} \in \mathcal{T}_1^{\perp} |_{0}
        \, \Big| \, 
        \EXP \big( T_{1} S_2 \big)  = 0 \ , \ S_2 \in \mathcal{T}_{2}^F
    \Big\} \ .
    \numeq
    \label{eq-proof-tagent perp}
\end{align*}
It remains to find a representation of $\mathcal{T}^{\perp}$. 

Consider a generic $T_1 \in \mathcal{T}_1^{\perp} |_{0}$ characterized by $v$ and $k$. In what follows, We will characterize the relationship between $v$ and $k$ used in $T_1$ so that that $T_1$ belongs to $\mathcal{T}^{\perp}$ in \eqref{eq-proof-tagent perp}. In order to do so, we fix an arbitrary  $v \in \mathcal{L}^2(\potY{0},Z,\bX)$ and then determine a specific form of $k$ such that the resulting $T_1 \in \mathcal{T}_1^\perp |_{0}$ becomes orthogonal to every $S_{2} \in \mathcal{T}_2^F$. Since any $S_{2} \in \mathcal{T}_2^F$ is fully characterized by the pair $\mathfrak{a}$ and $\mathfrak{b}$, it suffices to consider the following two cases: (Case 1): $\mathfrak{a}=0$ and $\mathfrak{b}$ is arbitrary; and (Case 2): $\mathfrak{a}$ is arbitrary and $\mathfrak{b}=0$. 

We begin by considering (Case 1). Let $S_{2,\mathfrak{a}=0}$ denote an element of $\mathcal{T}_2^F$ with $\mathfrak{a}=0$ in \eqref{eq-proof-tangent2 full} and arbitrary $\mathfrak{b}$, i.e., 
\begin{align*}
    S_{2,\mathfrak{a}=0}(\potY{0},A,Z,\bX) \equiv
    \big\{ A - f^*(A=1 \cond \potY{0},Z,\bX) \big\} \mathfrak{b}(Z,\bX) \ .
\end{align*}
Then, we find
\begin{align*}
        0 
        &
        =
        \EXP \big( T_1 S_{2,\mathfrak{a}=0}  \big)
        \ , \quad \forall S_{2,\mathfrak{a}=0} \in \mathcal{T}_2^F \text{ with } \mathfrak{a}=0
        \\
        \Leftrightarrow \quad
        0
        &
        =
        \EXP
        \left[ 
        \begin{array}{l}      
        \big[ 
    \frac{(1-A) \widetilde{v}(\potY{0},Z,X) }{f^*(A=0 | \potY{0},Z,\bX)}  
    -
    \big\{ \frac{ 1-A }{f^*(A=0 | \potY{0},Z,\bX)} - 1 \big\} k(Z,\bX)
    \big] 
    \\
    \times 
    \big\{ A - f(A=1 \cond \potY{0},Z,\bX) \big\} 
       \mathfrak{b}(Z,\bX) 
     \end{array}       
        \right]
        \ , \ \forall \mathfrak{b}
        \\
        \Leftrightarrow \quad
        0
        &
        =
        \EXP
        \left[ 
        \begin{array}{l}      
        \big[ 
    \frac{(1-A) \widetilde{v}(\potY{0},Z,X) }{f^*(A=0 | \potY{0},Z,\bX)}  
    -
    \big\{ \frac{ 1-A }{f^*(A=0 | \potY{0},Z,\bX)} - 1 \big\} k(Z,\bX)
    \big] 
    \\
    \times 
    \big\{ A - f^*(A=1 \cond \potY{0},Z,\bX) \big\} 
     \end{array}       
     \, \Bigg| \, Z,\bX
        \right]
        \\
        &
        =
        \EXP
        \left[
        \begin{array}{l}      
    - \widetilde{v}(\potY{0},Z,\bX)
    f^*(A=1 | \potY{0},Z,\bX) \\
    + k(Z,\bX) f^*(A=1 | \potY{0},Z,\bX) 
     \end{array}
        \, \Bigg| \,
        Z,\bX
        \right] 
        \\
        &
        =
        \EXP
        \Big[
    A \big\{ k(Z,\bX) - \widetilde{v}(\potY{0},Z,\bX)
    \big\} \, \Big| \, Z,\bX \Big]
    \\
        &
        =
        f^*(A=1 \cond Z,\bX) 
        \big\{  k(Z,\bX) 
        - \mu(Z,\bX \con \widetilde{v} )
        \big\} \ .
    \end{align*}
    This implies that $k(Z,\bX) = \mu(Z,\bX \con \widetilde{v} )$ yields $\EXP(T_1 S_{2,\mathfrak{a}=0} ) = 0$ for all $S_{2,\mathfrak{a}=0} \in \mathcal{T}_{2}^F$ with $\mathfrak{a}=0$. In other words, any $T_1 \in \mathcal{T}^\perp$ must be of the form:
    \begin{align*}
        T_1 (\potY{0},Z,\bX)
        =
        \frac{(1-A) \widetilde{v} (\potY{0},Z,\bX) }{f^*(A=0 \cond \potY{0},Z,\bX)} 
    - 
    \bigg\{ \frac{ 1-A }{f^*(A=0 \cond \potY{0},Z,\bX)} - 1 \bigg\} \mu(Z,\bX \con \widetilde{v}) \ .
    \end{align*}

    Next, we consider (Case 2). Let $S_{2,\mathfrak{b}=0}$ denote an element of $\mathcal{T}_2^F$ with $\mathfrak{b}=0$ in \eqref{eq-proof-tangent2 full} and arbitrary $\mathfrak{a}$, i.e., 
    \begin{align*}
    S_{2,\mathfrak{b}=0}(\potY{0},A,Z,\bX) \equiv
    \big\{ A - f^*(A=1 \cond \potY{0},Z,\bX) \big\} \mathfrak{a}(\potY{0},\bX) \ .
    \end{align*}
Analogous to the previous derivation, we obtain
\begin{align*}
        0 
        &
        =
        \EXP \big( T_1 S_{2,\mathfrak{b}=0}  \big)
        \ , \quad \forall S_{2,\mathfrak{b}=0} \in \mathcal{T}_2^F \text{ with } \mathfrak{b}=0
        \\
        \Leftrightarrow \quad
        0
        &
        =
        \EXP 
        \left[ 
        \begin{array}{l}      
        \big[ 
    \frac{(1-A) \widetilde{v}(\potY{0},Z,X) }{f^*(A=0 | \potY{0},Z,\bX)}  
    -
    \big\{ \frac{ 1-A }{f^*(A=0 | \potY{0},Z,\bX)} - 1 \big\} k(Z,\bX)
    \big] 
    \\
    \times 
    \big\{ A - f(A=1 \cond \potY{0},Z,\bX) \big\} 
        \mathfrak{a}(\potY{0},\bX) 
     \end{array}       
        \right] \ , \ \forall \mathfrak{a}: \mathfrak{a}(0,\bX) = 0
        \\
        \Leftrightarrow \quad
        0
        &
        =
        \EXP 
        \left[ 
        \begin{array}{l}      
        \big[ 
    \frac{(1-A) \widetilde{v}(\potY{0},Z,X) }{f^*(A=0 | \potY{0},Z,\bX)}  
    -
    \big\{ \frac{ 1-A }{f^*(A=0 | \potY{0},Z,\bX)} - 1 \big\} k(Z,\bX)
    \big] 
    \\
    \times 
    \big\{ A - f(A=1 \cond \potY{0},Z,\bX) \big\} 
        \big\{ \mathfrak{a}(\potY{0},\bX)  + \mathfrak{c} (\bX) \big\} 
     \end{array}       
        \right]  \ , \ \forall \mathfrak{a}: \mathfrak{a}(0,\bX) = 0, \forall \mathfrak{c}
        \\
        \Leftrightarrow \quad
        0
        &
        =
        \EXP 
        \left[ 
        \begin{array}{l}      
        \big[ 
    \frac{(1-A) \widetilde{v}(\potY{0},Z,X) }{f^*(A=0 | \potY{0},Z,\bX)}  
    -
    \big\{ \frac{ 1-A }{f^*(A=0 | \potY{0},Z,\bX)} - 1 \big\} k(Z,\bX)
    \big] 
    \\
    \times 
    \big\{ A - f(A=1 \cond \potY{0},Z,\bX) \big\} 
        \mathfrak{a}'(\potY{0},\bX)
     \end{array}       
        \right]  \ , \ \forall \mathfrak{a}'
        \\
        \Leftrightarrow \quad
        0
        &
        =
        \EXP 
        \left[ 
        \begin{array}{l}      
        \big[ 
    \frac{(1-A) \widetilde{v}(\potY{0},Z,X) }{f^*(A=0 | \potY{0},Z,\bX)}  
    -
    \big\{ \frac{ 1-A }{f^*(A=0 | \potY{0},Z,\bX)} - 1 \big\} k(Z,\bX)
    \big] 
    \\
    \times 
    \big\{ A - f(A=1 \cond \potY{0},Z,\bX) \big\} 
     \end{array}       
        \, \Bigg| \, \potY{0},\bX
        \right]
        \\
        &
        =
        \EXP
        \left[
        \begin{array}{l}      
    - \widetilde{v}(\potY{0},Z,\bX)
    f(A=1 | \potY{0},Z,\bX) \\
    + \mu(Z,\bX \con \widetilde{v})
    f(A=1 | \potY{0},Z,\bX) 
     \end{array}
        \, \Bigg| \,
        \potY{0},\bX
        \right] 
        \\
        &
        =
        \EXP
        \Big[
    A \big\{ \mu(Z,\bX \con \widetilde{v}) - \widetilde{v}(\potY{0},Z,\bX)
    \big\} \, \Big| \, \potY{0},\bX \Big]
    \\
        &
        =
        f(A=1 \cond \potY{0},\bX) 
        \EXP\big\{ \mu(Z,\bX \con \widetilde{v}) - \widetilde{v}(\potY{0},Z,\bX)
        \cond A=1, \potY{0},\bX
    \big\}
        \ .
    \end{align*}
    This implies that $\widetilde{v}$ must satisfy 
    \begin{align*}
        \EXP \big\{ \widetilde{v} - \mu(Z,\bX \con \widetilde{v}) \cond A=1,\potY{0},\bX \big\} = 0
        \ .
        \numeq 
        \label{eq-proof-tangent perp v tilde}
    \end{align*}
    In other words, any $T_1  \in \mathcal{T}^\perp$ must be of the form:
    \begin{align*}
        T_1 (\potY{0},Z,\bX)
        =
        \frac{(1-A) \widetilde{v} (\potY{0},Z,\bX) }{f^*(A=0 \cond \potY{0},Z,\bX)} 
    - 
    \bigg\{ \frac{ 1-A }{f^*(A=0 \cond \potY{0},Z,\bX)} - 1 \bigg\} \mu(Z,\bX \con \widetilde{v}) \ , 
    \end{align*}
    where $\widetilde{v}$ satisfies \eqref{eq-proof-tangent perp v tilde}. Therefore, $\mathcal{T}^\perp$ in \eqref{eq-proof-tagent perp} can be expressed as:
\begin{align*}
    &
    \mathcal{T}^\perp
    \numeq 
    \label{eq-proof-ortho tangent closed}
    \\
    &
    =
     \left\{
    \begin{array}{l}
    \displaystyle{
    \frac{(1-A) \widetilde{v} (\potY{0},Z,\bX) }{f^*(A=0 \cond \potY{0},Z,\bX)}}
    \\
    \displaystyle{
    - 
    \bigg\{ \frac{ 1-A }{f^*(A=0 \cond \potY{0},Z,\bX)} - 1 \bigg\} \mu(Z,\bX \con \widetilde{v})
    }
    \end{array}
    \,
    \left|
    \,
    \begin{array}{l}
        v \in \mathcal{L}^2(\potY{0},Z,\bX) \text{ satisfies \eqref{eq-proof-tangent perp v tilde}}
         \\
         \text{where } 
         \\
         \widetilde{v} = v - v^\dagger \ \text{and} 
         \\
         \mu(Z,\bX \con \widetilde{v}) = \EXP \big\{ \widetilde{v} \cond A=1,Z,\bX \big\}
    \end{array}    
    \right\}
    \right. 
    .
\end{align*}

Now, let $t_v \in \mathcal{T}^\perp$ be any element satisfying \eqref{eq-proof-tangent perp v tilde} corresponding to a function $v$. Then, the set in \eqref{eq-proof-entire IF} is expressed by
  \begin{align*}
  &
  \InfFt^*(\bO)  + t_v(\bO)
  \\
  &
  =
  \frac{ 
   \left[ 
    \begin{array}{l}
    \{ A - (1-A) \alpha^*(Y,\bX) \beta^*(Z,\bX) \}
     \big\{ \mathcal{G} - \mu^*(Z,\bX \con \mathcal{G}) \big\}     
     - A \tau^*
     \\
       +
       (1-A) \alpha^*(Y,\bX) \beta^*(Z,\bX) \big\{ \omega(Y,Z,\bX) - \mu^*(Z,\bX \con \omega) \big\}
  \\
  +
    A
    \mu^*(Z,\bX \con \omega)
    + 
    (1-A) \omega(Y,Z,\bX) 
    \end{array}
    \right]
    }{\EXP(A)}
    \\
  &
  \quad 
  +
  \frac{ 
  \EXP(A)
   \left[ 
    \begin{array}{l}
     (1-A) \alpha^*(Y,\bX) \beta^*(Z,\bX) \big\{ \widetilde{v}(Y,Z,\bX) - \mu^*(Z,\bX \con \widetilde{v}) \big\}
  \\
  +
    A
    \mu^*(Z,\bX \con \widetilde{v})
    + 
    (1-A) \widetilde{v}(Y,Z,\bX) 
    \end{array}
    \right]
    }{\EXP(A)}
    \\
  &
  =
  \frac{ 
   \left[ 
    \begin{array}{l}
    \{ A - (1-A) \alpha^*(Y,\bX) \beta^*(Z,\bX) \}
     \big\{ \mathcal{G} - \mu^*(Z,\bX \con \mathcal{G}) \big\}     
     - A \tau^*
     \\
       +
       (1-A) \alpha^*(Y,\bX) \beta^*(Z,\bX) \big\{ \omega'(Y,Z,\bX) - \mu^*(Z,\bX \con \omega') \big\}
  \\
  +
    A
    \mu^*(Z,\bX \con \omega')
    + 
    (1-A) \omega'(Y,Z,\bX) 
    \end{array}
    \right]
    }{\EXP(A)}
    \ ,
\end{align*}
where $\omega' = \omega + \EXP(A) \widetilde{v}$. Since $\omega'$ satisfies conditions \HL{$\omega$-i} and \HL{$\omega$-ii}, each element of \eqref{eq-proof-entire IF} indeed has the form in \eqref{eq-IF}.

This completes the proof.

\subsection{Proof of Theorem \ref{thm-IF binary Z}} 
\label{sec-proof-IF binary Z}

In the proof, we show a more general result by characterizing the EIF for $ \tau^* \equiv \tau_1^* - \tau_0^*$ where $\tau_{a}^* \equiv  \EXP \{ \mathcal{G}(\potY{a},\bX)  \cond A=1 \}$; here, $\mathcal{G}(\cdot)$ is a fixed, uniformly bounded function.  For convenience, we denote $\mathcal{G}^{(a)} = \mathcal{G}(\potY{a},\bX)$ and $\mathcal{G} = \mathcal{G}(Y,\bX)$.

We use the same notation as in Section \ref{sec-proof-IF}. We also introduce the following shorthand notation for the nuisance functions:
\begin{align*}
    & 
    p_z = f^*(Z=z \cond \bX)
    \ ,
    \\
    &
    \mathfrak{p}_{z}(y) = f^*(Z=z \cond A=1,\potY{0}=y,\bX)
    \ , 
    \\
    &
    \alpha(y) = \alpha^*(y,\bX)
    \ , 
    \\
    &
    \beta_z = \beta^*(z,\bX) 
    \ ,
    \\
    &
    \mu_z = \EXP \{ \mathcal{G}^{(0)} \cond A=1,Z=z,\bX \} \ ,
    \\
    &
    \mu_z (v) = \EXP \{ v(\potY{0},Z,\bX) \cond A=1,Z=z,\bX \} \ , \quad \forall v \in \mathcal{L}^2(\potY{0},Z,\bX) \ .
\end{align*}

By \eqref{eq-proof-Y0AZ}, we can obtain an alternative representation of $\mathfrak{p}_z$:
\begin{align*}
    \mathfrak{p}_{1}(\potY{0})   
    & 
    =  f^*(Z=1 \cond A=1, \potY{0},\bX)
    \\
    & =
    \frac{ 
    \frac{\alpha(\potY{0}) \beta_1}{ 1 + \alpha(\potY{0}) \beta_1 }
    f^*(\potY{0} \cond \bX) p_1
    }{
    \sum_{z=0}^{1}
    \frac{\alpha(\potY{0}) \beta_z}{ 1 + \alpha(\potY{0}) \beta_z }
    f^*(\potY{0} \cond \bX) p_z
    }
    \\  
    &
    =    
    \bigg[
    \frac{ 
    \beta_1
    p_1
    }{ 1 + \alpha(\potY{0}) \beta_1 } 
    +
    \frac{ 
    \beta_0
    p_0
    }{ 1 + \alpha(\potY{0}) \beta_0 } 
    \bigg]^{-1}
    \bigg[
    \frac{ 
    \beta_1
    p_1
    }{ 1 + \alpha(\potY{0}) \beta_1 } 
    \bigg]
    \\
    &
    =
    \frac{
    \big\{ 1 + \alpha(\potY{0}) \beta_0 \big\} \beta_1
    p_1
    }{
    \big\{ 1 + \alpha(\potY{0}) \beta_0 \big\} \beta_1
    p_1
    +
    \big\{ 1 + \alpha(\potY{0}) \beta_1 \big\} \beta_0
    p_0
    }
    \\
    &
    =
    \frac{
    \big\{ 1 + \alpha(\potY{0}) \beta_0 \big\} \beta_1
    p_1
    }{
    \beta_1
    p_1
    +
    \beta_0
    p_0
    +
    \alpha(\potY{0})
    \beta_0
    \beta_1
    }
    \ , 
    \numeq \label{eq-proof-frakp1}
\end{align*}
and
\begin{align*}
    & \mathfrak{p}_0(\potY{0})  
    =
    \frac{
    \big\{ 1 + \alpha(\potY{0}) \beta_1 \big\} \beta_0
    p_0
    }{
    \beta_1
    p_1
    +
    \beta_0
    p_0
    +
    \alpha(\potY{0})
    \beta_0
    \beta_1
    }
    \ .
\end{align*}

From the proof of Lemma \ref{lemma-unique integral equation solution}, specifically the derivation in \eqref{eq-proof-vtilde L representation}, it follows that $\omega = \widetilde{\nu} = \nu - \nu^\dagger$ admits the representation:
\begin{align*}
    \omega (\potY{0},Z,\bX) 
    &
    = 
    \big[ L(\potY{0},\bX) - \EXP \big\{ L(\potY{0},\bX) \cond \bX \big\} \big]
    \big\{ Z - \EXP \big( Z \cond \bX \big) \big\} \ , 
    \numeq \label{eq-proof-omega L representation}
\end{align*}
for a function $L(\potY{0}, \bX) \in \mathcal{L}^2(\potY{0},\bX)$. 

Next, we show that there exists a unique $L^*$ satisfying \HL{$\omega$-ii}. Equation \eqref{eq-proof-omega L representation} implies 
\begin{align*}
    \mu_Z(\omega)
    &
    =
    \EXP \big\{ \omega(\potY{0},Z,\bX) \cond A=1,Z,\bX \big\}
    \\
    &
    = 
    \big\{ Z - \EXP \big( Z \cond \bX \big) \big\}
    \big[ \mu_Z(L) - \EXP\{ L(\potY{0},\bX) \cond \bX\} \big] \ .
\end{align*}
and
\begin{align*}
    &
    A
    \big\{
    \omega (\potY{0},Z,\bX)
    - 
    \mu_Z (\omega)
    \big\}
    =
    A 
    \big\{ Z - \EXP \big( Z \cond \bX \big) \big\}
    \big\{ L(\potY{0},\bX) - \mu_Z(L) \big\}  \ .
\end{align*}

Consequently, we obtain 
\begin{align*}
    &
    \EXP \big\{
    \omega(\potY{0},Z,\bX)
    - 
    \mu_Z(\omega)  \cond A=1,\potY{0},\bX \big\}
    \\
    &
    = 
	\left[ 
		\begin{array}{l}
			- f^*(Z=0 \cond A=1,\potY{0},\bX)
			f^*(Z=1 \cond \bX)
			\big\{ L - \mu_0(L) \big\}
			\\
			+ f^*(Z=1 \cond A=1,\potY{0},\bX)
			f^*(Z=0 \cond \bX)
			\big\{ L - \mu_1(L) \big\}
		\end{array}
	\right]     
    \\
	&
	= 
    \left[ 
    	\begin{array}{l}
        \big\{ \mathfrak{p}_1(\potY{0})
        p_0
        -
        \mathfrak{p}_0(\potY{0})
        p_1 \big\} L
			\\
		  -
			\big\{
            \begin{array}{l}
            \mathfrak{p}_1(\potY{0})
			p_0
			\mu_1(L)
            - 
            \mathfrak{p}_0(\potY{0})
			p_1
			\mu_0(L)
            \end{array}
            \big\}
    	\end{array}
    \right] 
    \\
    &= 
    \left[ 
    	\begin{array}{l}
 	\mathfrak{p}_1(\potY{0}) L - p_1 L  - p_1 \mu_0(L)
 	\\
    +
 	\mathfrak{p}_1(\potY{0})
    \big\{ p_1 \mu_0(L)
    -
    p_0 \mu_1(L)
    \big\}
    	\end{array}
    \right]
    \ .
\end{align*} 
Likewise, we have  
\begin{align*}
    &
    \EXP \big\{
    \mathcal{G}^{(0)}
    - \mu_Z \cond A=1,\potY{0},\bX \big\}
    \\
    &
    = 
    \mathcal{G}^{(0)}
    - \mathfrak{p}_1(\potY{0}) \mu_1
    - \mathfrak{p}_0(\potY{0}) \mu_0
    \\
    &
    = 
    \mathcal{G}^{(0)}
    -
    \mu_0
    + \mathfrak{p}_1(\potY{0})
    \big( \mu_0 - \mu_1  \big) \ .
\end{align*} 
Therefore, $\omega$ satisfies \HL{$\omega$-ii} if and only if $L$ satisfies the following condition:
\begin{align*}
&
\left[ 
\begin{array}{l}
     \mathcal{G}^{(0)} - \mu_0 \\
    + \mathfrak{p}_1(\potY{0}) \big( \mu_0 - \mu_1  \big)
\end{array} 
    \right] 
    =
        \left[ 
    	\begin{array}{l}
 	\mathfrak{p}_1(\potY{0}) L - p_1 L - p_1 
    \mu_0(L)
 	\\
 	+
 	\mathfrak{p}(\potY{0})
 	\big\{  p_1 
    \mu_0(L)
        - p_0 
        \mu_1(L)
 	\big\}
    	\end{array}
    \right] \ .
    \numeq 
    \label{eq-proof-L equation}
\end{align*}
As mentioned in Section \ref{sec-supp-ClosedFormSolution}, \eqref{eq-proof-L equation} can be expressed as a Fredholm integral equation of the second kind with a separable kernel. Moreover, it becomes a just-identified equation when $Z$ is binary, therefore admits a unique solution. 

%

In order to find the actual closed-form representation, we posit the following form for $L$, where $B \equiv B(\bX)$ and $C \equiv C(\bX)$:
\begin{align*}
    &
    L(\potY{0},\bX) 
    = \frac{
    \mathcal{G}^{(0)}
    + \mathfrak{p}_1(\potY{0}) B + C}{\mathfrak{p}_1(\potY{0}) - p_1}  \ .
    \numeq
    \label{eq-proof-closed form L} 
\end{align*}
Under this specification, \eqref{eq-proof-L equation} is expressed as
\begin{align*}
&
\left[ 
\begin{array}{l}
     \mathcal{G}^{(0)} - \mu_0 \\
    + \mathfrak{p}_1(\potY{0}) \big( \mu_0 - \mu_1  \big)
\end{array} 
    \right]
    =
        \left[ 
    	\begin{array}{l}
 	\mathcal{G}^{(0)} - p_1 \mu_0(L) + C
 	\\
 	+
 	\mathfrak{p}_1(\potY{0})
 	\big\{ B + p_1 \mu_0(L)
        - p_0 \mu_1(L)
 	\big\}
    	\end{array}
    \right] \ .
    \numeq 
    \label{eq-proof-L equation 2}
\end{align*} 
By comparing the functional forms on both sides of \eqref{eq-proof-L equation 2}, we obtain
\begin{align*}
    &
    \left\{
    \begin{array}{l}         
    p_1 \mu_0(L) - C
    -
    \mu_0
    = 0
     \\     
    B + p_1 \mu_0(L) - p_0 \mu_1(L) 
    = \mu_0 - \mu_1
    \end{array}
    \right.
     \numeq \label{eq-required identity1} 
    \\
    \Rightarrow \quad
    &
    B + C + \mu_0 - p_0 \mu_1(L)
    =
    \mu_0 - \mu_1
    \\
    \Rightarrow \quad
    &
    p_0 \mu_1(L) - B - C - \mu_1 = 0 \ .
    \numeq \label{eq-required identity2}
\end{align*} 
Consequently, it suffices to find $B$ and $C$ that satisfy \eqref{eq-required identity1} and \eqref{eq-required identity2}.

We first find the expression of $\mathfrak{p}_1(\potY{0}) - p_1$: 
\begin{align*}
    & 
    \mathfrak{p}_1(\potY{0}) - p_1
    \\
    &
    =
    f^*(Z=1 | A=1,\potY{0},\bX) - f^*(Z=1 | \bX)
    \\
    &
    =
    \frac{ f^*(\potY{0},A=1,Z=1\cond \bX) -  f^*(\potY{0},A=1\cond \bX) f^*(Z=1 \cond \bX) }
    { f^*(\potY{0},A=1\cond \bX) }
    \\
    &
    =
    \frac{ f^*(\potY{0},A=1,Z=1\cond \bX) f^*(Z=0 \cond \bX) -  f^*(\potY{0},A=1,Z=0\cond \bX) f^*(Z=1 \cond \bX) }
    { f^*(\potY{0},A=1\cond \bX) }
    \\
    &
    =
    \frac{
    \Big\{
    \frac{ \alpha(\potY{0}) \beta_1 }{1+\alpha(\potY{0}) \beta_1}
    -
    \frac{ \alpha(\potY{0}) \beta_0 }{1+\alpha(\potY{0}) \beta_0}
    \Big\}
    f^*(\potY{0} \cond \bX ) 
    p_1 p_0 
    }
    {
    f^*(\potY{0},A=1 \cond \bX)
    }
    \\
    &
    =
    \frac{\alpha(\potY{0})
    ( \beta_1 - \beta_0 ) }{\{ 1+\alpha(\potY{0}) \beta_1\}\{1+\alpha(\potY{0}) \beta_0\}}
    \frac{ 
    f^*(\potY{0} \cond \bX ) 
    p_1 p_0
    }
    {
    f^*(\potY{0},A=1 \cond \bX)
    }
    \\
    &
    \stackrel{ \eqref{eq-proof-Y0AZ} }{=}
    \frac{ 
    \displaystyle{\frac{\alpha(\potY{0})
    ( \beta_1 - \beta_0 ) }{\{ 1+\alpha(\potY{0}) \beta_1\}\{1+\alpha(\potY{0}) \beta_0\}}
    f^*(\potY{0} \cond \bX ) 
    p_1 p_0 }
    }{
    \displaystyle{
    \sum_{z=0}^{1}
     \frac{\alpha(\potY{0}) \beta_z}{1+\alpha(\potY{0}) \beta_z} f^*(\potY{0} \cond \bX) p_z
     }
    }   
    \\
    &
    =
    \frac{ 
    \displaystyle{    
    \frac{\alpha(\potY{0}) ( \beta_1 - \beta_0 ) }{\{ 1+\alpha(\potY{0}) \beta_1\}\{1+\alpha(\potY{0}) \beta_0 \}} 
    p_1 p_0
    }
    }{
    \displaystyle{
    \frac{\alpha(\potY{0}) \beta_1
    \{ 1+\alpha(\potY{0}) \beta_0 \} 
    p_1
    +
    \alpha(\potY{0}) \beta_0 \{ 1+\alpha(\potY{0}) \beta_1 \} 
    p_0
    }{\{ 1+\alpha(\potY{0}) \beta_1\}     
    \{ 1+\alpha(\potY{0}) \beta_0 \}
    }
    }}
    \\
    &
    =
    \frac{
    ( \beta_1 - \beta_0 )
    p_1
    p_0
    }{
    \beta_1
    \{ 1+ \alpha(\potY{0}) \beta_0 \} 
    p_1
    +
    \beta_0 \{ 1+\alpha(\potY{0}) 
    \beta_1 \} 
    p_0
    }
    \\
    &
    =
    \frac{
    ( \beta_1 - \beta_0 )
    p_1
    p_0
    }{
    \beta_1 p_1
    +
    \beta_0 p_0
    +
    \alpha(\potY{0}) 
    \beta_0
    \beta_1 
    } \ .
    \numeq
    \label{eq-proof-frakp1 minus p1}
\end{align*}
Combining \eqref{eq-proof-frakp1} and \eqref{eq-proof-frakp1 minus p1}, we obtain
\begin{align*}
    \frac{ \mathfrak{p}_1(\potY{0}) }{
    \mathfrak{p}_1(\potY{0}) - p_1
    }
    =
    \frac{ \big\{ 1 + \alpha(\potY{0}) \beta_0 \big\} \beta_1
    p_1 }{ (\beta_1 - \beta_0)
    p_1 p_0 }
    \ .
    \numeq
    \label{eq-proof-closed form L piece 1}
\end{align*}
By substituting \eqref{eq-proof-closed form L piece 1} into \eqref{eq-proof-closed form L}, we find that
\begin{align*}
&
L(\potY{0},\bX) 
\\
&
= \frac{
\mathcal{G}^{(0)}
+ \mathfrak{p}_1(\potY{0}) B + C}{\mathfrak{p}_1(\potY{0}) - p_1}  
\\
&
=
\frac{
\beta_1 p_1
+
\beta_0 p_0
+
\alpha(\potY{0}) 
\beta_0
\beta_1 
}{
( \beta_1 - \beta_0 )
p_1
p_0
}
\big\{ \mathcal{G}^{(0)}
+ C
\big\} 
+
\frac{ \big\{ 1 + \alpha(\potY{0}) \beta_0 \big\} \beta_1
p_1 }{ (\beta_1 - \beta_0)
p_1 p_0 }
B
\\
&
=
\frac{ 
\left[ 
\begin{array}{lll}
\{ \sum_{z=0}^{1} \beta_z p_z \}
\mathcal{G}^{(0)}
&
+ 
&
\beta_0 \beta_1 \mathcal{G}^{(0)}
\alpha(\potY{0}) 
\\[0.0cm]
+ \{ \sum_{z=0}^{1} \beta_z p_z \} C
&
+
&
\beta_0 \beta_1 
\alpha(\potY{0}) C
\\[0.0cm]
+
\beta_1 p_1 B 
&
+ 
&
\beta_0 \beta_1
p_1 \alpha(\potY{0}) B
\end{array}
\right]
}{ (\beta_1 - \beta_0)
p_1 p_0 } \ .
\end{align*}
Consequently, we obtain
\begin{align*}
\mu_z(L)
&
=
\frac{ 
\left[ 
\begin{array}{lll}
\{ \sum_{z=0}^{1} \beta_z p_z \}
\mu_z
&
+ 
&
\beta_0 \beta_1 
\mu_z(\alpha \mathcal{G})
\\[0.0cm]
+ \{ \sum_{z=0}^{1} \beta_z p_z \} C
&
+
&
\beta_0 \beta_1 C
\mu_z(\alpha)
\\[0.0cm]
+
\beta_1 p_1 B 
&
+ 
&
\beta_0 \beta_1
p_1 B
\mu_z(\alpha)
\end{array}
\right]
}{ (\beta_1 - \beta_0)
p_1 p_0 } \ .
\numeq 
\label{eq-proof-muL}
\end{align*} 

By \eqref{eq-required identity1}  and \eqref{eq-proof-muL}, it follows that
\begin{align*}
    0
    &
    =
    p_1 \mu_0(L) - C
    -
    \mu_0
    \\ 
    &
    =
    \frac{ 
\left[ 
\begin{array}{lll}
\{ \sum_{z=0}^{1} \beta_z p_z \}
\mu_0
&
+ 
&
\beta_0 \beta_1 
\mu_0(\alpha \mathcal{G})
\\[0.0cm]
+ \{ \sum_{z=0}^{1} \beta_z p_z \} C
&
+
&
\beta_0 \beta_1 C
\mu_0(\alpha)
\\[0.0cm]
+
\beta_1 p_1 B 
&
+ 
&
\beta_0 \beta_1
p_1 B
\mu_0(\alpha)
\end{array}
\right]
-
(\beta_1 - \beta_0)
p_0 \big( C + \mu_0 \big) 
}{ (\beta_1 - \beta_0)
p_0 }  
\\ 
    &
    =
    \frac{ 
\left[ 
\begin{array}{lll}
\{ 2 \beta_0 p_0 + \beta_1 p_1 - \beta_1 p_0 \} 
\mu_0
&
+ 
&
\beta_0 \beta_1 
\mu_0(\alpha \mathcal{G})
\\[0.0cm]
+ \{ 2 \beta_0 p_0 + \beta_1 p_1 - \beta_1 p_0\} C
&
+
&
\beta_0 \beta_1 C
\mu_0(\alpha)
\\[0.0cm]
+
\beta_1 p_1 B 
&
+ 
&
\beta_0 \beta_1
p_1 B
\mu_0(\alpha)
\end{array}
\right] 
}{ (\beta_1 - \beta_0)
p_0 }   \ .
\end{align*}
This implies
\begin{align*}
    &
    B \cdot \{ \beta_1 p_1 + \beta_0 \beta_1 p_1 \mu_0(\alpha) \}
    +
    C \cdot \{ \beta_1 p_1  + 2 \beta_0 p_0 - \beta_1 p_0 + \beta_0 \beta_1 \mu_0 (\alpha) \}
    \\
    &
    =
    -
    \{ \beta_1 p_1 \mu_0  + 2 \beta_0 p_0 \mu_0 - \beta_1 p_0 \mu_0 + \beta_0 \beta_1 \mu_0(\alpha \mathcal{G}) 
    \} 
    \ .
    \numeq 
    \label{eq-proof-linear system 1}
\end{align*}

By \eqref{eq-required identity1}  and \eqref{eq-proof-muL}, we obtain
\begin{align*}
    0
    &
    =
    p_0 \mu_1(L) - B - C - \mu_1
    \\
    &
    =
    \frac{ 
\left[ 
\begin{array}{lll}
\{ \sum_{z=0}^{1} \beta_z p_z \}
\mu_1
&
+ 
&
\beta_0 \beta_1 
\mu_1(\alpha \mathcal{G})
\\[0.0cm]
+ \{ \sum_{z=0}^{1} \beta_z p_z \} C
&
+
&
\beta_0 \beta_1 C
\mu_1(\alpha)
\\[0.0cm]
+
\beta_1 p_1 B 
&
+ 
&
\beta_0 \beta_1
p_1 B
\mu_1(\alpha)
\end{array}
\right]
-
(\beta_1 - \beta_0)
p_1 \big( B + C + \mu_1 \big) 
}{ (\beta_1 - \beta_0)
p_1 }
\\
    &
    =
    \frac{ 
\left[ 
\begin{array}{lll}
\{ \beta_0 p_0 + \beta_0 p_1 \}
\mu_1
&
+ 
&
\beta_0 \beta_1 
\mu_1(\alpha \mathcal{G})
\\[0.0cm]
+ 
\{ \beta_0 p_0 + \beta_0 p_1 \}
C
&
+
&
\beta_0 \beta_1 C
\mu_1(\alpha)
\\[0.0cm]
+
\beta_0 p_1 B 
&
+ 
&
\beta_0 \beta_1
p_1 B
\mu_1(\alpha)
\end{array}
\right] 
}{ (\beta_1 - \beta_0)
p_1 } 
\\
    &
    =
    \beta_0
    \frac{ 
\left[ 
\begin{array}{lll}
\mu_1
&
+ 
&
\beta_1 
\mu_1(\alpha \mathcal{G})
\\[0.0cm]
+ 
C
&
+
&
\beta_1 C
\mu_1(\alpha)
\\[0.0cm]
+
p_1 B 
&
+ 
&
\beta_1
p_1 B
\mu_1(\alpha)
\end{array}
\right] 
}{ (\beta_1 - \beta_0)
p_1 } 
\ ,
\end{align*}
which implies
\begin{align*}
    &
    B \cdot \{
    p_1 + \beta_1 p_1 \mu_1(\alpha) 
    \}
    + 
    C \cdot 
    \{ 1  + \beta_1 \mu_1(\alpha)  \}
    =
    - \{ \mu_1 + \beta_1 \mu_1(\alpha \mathcal{G}) \} \ .
    \numeq 
    \label{eq-proof-linear system 2}
\end{align*}

Results \eqref{eq-proof-linear system 1} and \eqref{eq-proof-linear system 2} reduce to the following linear system:
\begin{align*}
    \mathcal{A} \begin{bmatrix}
        B \\ C
    \end{bmatrix}
    =
    \mathcal{B}
    \numeq
    \label{eq-proof-linear system full}
\end{align*}
where
\begin{align*} 
    \mathcal{A}
    &
    =
    \begin{bmatrix}
         \beta_1 p_1 + \beta_0 \beta_1 p_1 \mu_0(\alpha)
         &
         \phantom{space}
         &
         \beta_1 p_1  + 2 \beta_0 p_0 - \beta_1 p_0 + \beta_0 \beta_1 \mu_0 (\alpha)
        \\
         p_1 + \beta_1 p_1 \mu_1(\alpha) 
         &
         \phantom{space}
         &
         1  + \beta_1 \mu_1(\alpha)
    \end{bmatrix}
    \in \R^{2 \times 2}
    \ ,
    \\
    \mathcal{B}
    &
    =
    -
    \begin{bmatrix}
        \beta_1 p_1 \mu_0  + 2 \beta_0 p_0 \mu_0 - \beta_1 p_0 \mu_0 + \beta_0 \beta_1 \mu_0(\alpha \mathcal{G}) 
    \\
        \mu_1 + \beta_1 \mu_1(\alpha \mathcal{G})
    \end{bmatrix}
    \in \R^{2 \times 1} \ .
\end{align*}
Note that $\mathcal{A}$ is invertible as follows:
\begin{align*}
    \text{det} (\mathcal{A})
    & =    
    p_1
    \left[ 
        \begin{array}{l}
             \big\{ \beta_1 + \beta_0 \beta_1 \mu_0(\alpha) \big\}
             \big\{ 1  + \beta_1 \mu_1(\alpha) \big\}
             \\
             -
             \big\{
             \beta_1 p_1  + 2 \beta_0 p_0 - \beta_1 p_0 + \beta_0 \beta_1 \mu_0 (\alpha)
             \big\}
             \big\{ 1 + \beta_1 \mu_1(\alpha) \big\}
        \end{array}
    \right]
    \\
    & =    
    p_1
    \big\{ 1 + \beta_1 \mu_1(\alpha) \big\}
    \left[ 
        \begin{array}{l}
             \beta_1 + \beta_0 \beta_1 \mu_0(\alpha)
             \\
             -
             \big\{
             \beta_1 p_1  + 2 \beta_0 p_0 - \beta_1 p_0 + \beta_0 \beta_1 \mu_0 (\alpha)
             \big\}             
        \end{array}
    \right]
    \\
    & =    
    p_1
    \big\{ 1 + \beta_1 \mu_1(\alpha) \big\}
    \big( \beta_1 - \beta_1 p_1  - 2 \beta_0 p_0 + \beta_1 p_0 \big)
    \\
    & =    
    2 p_1 p_0
    \big\{ 1 + \beta_1 \mu_1(\alpha) \big\}
    \big( \beta_1 - \beta_0 \big)
    \\
    &
    \neq 0 \ .
\end{align*}
Of note, $\beta_1 \neq \beta_0$ is guaranteed from \HL{IV3} and Lemma \ref{lemma-beta and gamma}. Therefore, the closed-form expression of $(B,C)$  is given by solving \eqref{eq-proof-linear system full}:
\begin{align*}
    \begin{bmatrix}
        B \\ C
    \end{bmatrix}
    &
    =
    \mathcal{A}^{-1}
    \mathcal{B} 
    \numeq \label{eq-proof-BC}
    \\
    &
    =    
    -
    \frac{
    \displaystyle{ 
    \begin{bmatrix}
         1  + \beta_1 \mu_1(\alpha)
         &
         \phantom{s}
         &
         -\{
         \beta_1 p_1  + 2 \beta_0 p_0 - \beta_1 p_0 + \beta_0 \beta_1 \mu_0 (\alpha)
         \}
        \\
         -\{p_1 + \beta_1 p_1 \mu_1(\alpha) \}
         &
         \phantom{s}
         &
         \beta_1 p_1 + \beta_0 \beta_1 p_1 \mu_0(\alpha)
    \end{bmatrix} 
    }
    }{ 2 p_1 p_0
    \big\{ 1 + \beta_1 \mu_1(\alpha) \big\}
    \big( \beta_1 - \beta_0 \big) }
    \\
    &
    \hspace*{2cm} \times 
    \begin{bmatrix}
        \beta_1 p_1 \mu_0  + 2 \beta_0 p_0 \mu_0 - \beta_1 p_0 \mu_0 + \beta_0 \beta_1 \mu_0(\alpha \mathcal{G}) 
    \\
        \mu_1 + \beta_1 \mu_1(\alpha \mathcal{G})
    \end{bmatrix} \ .
\end{align*}
Consequently, $\omega^*(\potY{0},Z,\bX)$ satisfying \HL{$\omega$-i} and \HL{$\omega$-ii} are expressed by
\begin{align*}
    \omega^*(\potY{0},Z,\bX)
    =
    \big[ L^*(\potY{0},\bX) - \EXP \big\{ L^*(\potY{0},\bX) \cond \bX \big\} \big]
    \big\{ Z - \EXP \big( Z \cond \bX \big) \big\}
    \numeq 
    \label{eq-proof-omega star}
\end{align*}    
where
\begin{align*}
    L^*(\potY{0},\bX) =  \frac{\potY{0} + \EXP\{ Z \cond A=1,\potY{0},\bX\} B^* (\bX) + C^* (\bX)}{ \EXP\{ Z \cond A=1,\potY{0}, \bX \} - \EXP (Z \cond \bX)}  \ ,
\end{align*} 
with $B^*$ and $C^*$ defined in \eqref{eq-proof-BC}. The corresponding IF for $\tau^*$ in model $\model$ is given by
  \begin{align*}
  \InfFt^*(\bO)
  =
  \frac{ 
   \left[ 
    \begin{array}{l}
    \{ A - (1-A) \alpha^*(Y,\bX) \beta^*(Z,\bX) \}
     \big\{ \mathcal{G} - \mu^*(Z,\bX \con \mathcal{G}) \big\}     
     - A \tau^*
     \\
       +
       (1-A) \alpha^*(Y,\bX) \beta^*(Z,\bX) \big\{ \omega^*(Y,Z,\bX) - \mu^*(Z,\bX \con \omega^*) \big\}
  \\
  +
    A
    \mu^*(Z,\bX \con \omega^*)
    + 
    (1-A) \omega^*(Y,Z,\bX) 
    \end{array}
    \right]
    }{\EXP(A)}
     \ .
     \numeq 
     \label{eq-proof-IF omega binary}
\end{align*}

Next, we show that $\omega^*$ in \eqref{eq-proof-omega star} is the unique function that satisfies \HL{$\omega$-i} and \HL{$\omega$-ii}, i.e.,
\begin{align*}
    &
    \EXP \big\{ \mathcal{G}^{(0)}- \mu^*(Z,\bX \con \mathcal{G}) \cond A=1, \potY{0},\bX \big\}
    \\
    &
    =
    \EXP \big\{ \omega^* (\potY{0},Z,\bX) - \mu^*(Z,\bX \con \omega^*) \cond A=1,\potY{0},\bX \big\} \ .
\end{align*}
Such an $\omega^*$ is unique provided the following condition holds for $\widetilde{v} = v - v^\dagger$ with $v \in \mathcal{L}^2(\potY{0},Z,\bX)$:
\begin{align*}
    &
    \EXP \{ \widetilde{v}(\potY{0},Z,\bX) - \mu^* (Z,\bX \con \widetilde{v}) \cond A=1, \potY{0},\bX\}
    =
    0
    \quad \text{implies} \quad
    \widetilde{v}(\potY{0},Z,\bX) =  0 \ .
\end{align*}
Lemma \ref{lemma-unique integral equation solution} ensures that this condition holds, thereby establishing the uniqueness of $\omega^*$ in \eqref{eq-proof-omega star}. Since the IF for $\tau^*$ is unique, it is the EIF for $\tau^*$ in model $\model$.

The fact that the IF in \eqref{eq-proof-IF omega binary} is the EIF for $\tau^*$ in model $\model$ can be additionally justified as follows. It suffices to verify that the IF in \eqref{eq-proof-IF omega binary} lies in the tangent space $\mathcal{T}$ of model $\model$. Under binary $Z$, Lemma \ref{lemma-unique integral equation solution} implies that the unique solution to \eqref{eq-proof-tangent perp v tilde} is $\widetilde{v} = 0$. Therefore, the orthocomplement to $\mathcal{T}$, denoted by $\mathcal{T}^\perp$ in \eqref{eq-proof-ortho tangent closed}, only contains 0, i.e., $\mathcal{T}^\perp = \{0 \}$, meaning that $\mathcal{T} = \mathcal{L}_0^2(\bO)$. Therefore, the IF in \eqref{eq-proof-IF omega binary} belongs to $\mathcal{T}$, meaning that it is the EIF for $\tau^*$ in model $\model$. 

This concludes the proof.

\subsection{Proof of Theorem \ref{thm-AN}} \label{sec-proof-AN}

In the proof, we show a more general result by considering an estimator for $ \tau^* \equiv \EXP \{ \mathcal{G}(\potY{1},\bX) - \mathcal{G}(\potY{0},\bX)  \cond A=1 \}$, where $\mathcal{G}(\cdot)$ is a fixed, uniformly bounded function.  For convenience, we denote $\mathcal{G}^{(a)} = \mathcal{G}(\potY{a},\bX)$ and $\mathcal{G} = \mathcal{G}(Y,\bX)$. 

We use the following shorthand notation:
\begin{align*}
    &
    \pi^*(\potY{0},Z,\bX)
    =
    \alpha^*(\potY{0},\bX) 
    \beta^*(Z,\bX)
     \ ,
    \\
    &
    \pihat(\potY{0},Z,\bX)
    =
    \alphahat(\potY{0},\bX) 
    \betahat(Z,\bX)
     \ ,
    \\
    &
    \EXPhat\{ g(V) \cond W \}
    =
    \int g(v) \widehat{f}\LSS(V=v \cond W) \, dv \ ,
    \\
    &
    \EXPk \{ \widehat{g}\LSS(V) \cond W \}
    =
    \int \widehat{g}\LSS(v) \ftrue{}(V=v \cond W) \, dv \ ,
    \\
    &
    N_{k} = | \SSk | \ ,
    \\
        &
    \AVERk (V) = \frac{1}{N_k} \sum_{i \in \SSk} V_i
    \ , 
    \\
    &
    \AVER (V) = \frac{1}{N} \sum_{i=1}^N V_i
    \ ,
    \\
    &
    \EMPk (V) = \frac{1}{\sqrt{N_k}} \sum_{i \in \SSk} \big\{ V_i - \EXPk(V) \big\} \ .
\end{align*}

From Sections \ref{sec-proof-IF} and \ref{sec-proof-IF binary Z}, the EIF for $\tau^*$ in model $\model$ is given by 
\begin{align*}
    \InfFt^*(\bO) = \frac{ \phi^*(\bO) - A \tau^* }{\EXP(A)} \ ,
\end{align*}
where
\begin{align*}
   \phi^*(\bO) 
    =  
   \left[ 
    \begin{array}{l}
    \{ A - (1-A) \pi^*(Y,Z,\bX) \}
     \big\{ \mathcal{G} - \mu^*(Z,\bX \con \mathcal{G}) \big\}     
     \\
       +
       (1-A) \pi^*(Y,Z,\bX) \big\{ \omega^*(Y,Z,\bX) - \mu^*(Z,\bX \con \omega^*) \big\}
  \\
  +
    A
    \mu^*(Z,\bX \con \omega^*)
    + 
    (1-A) \omega^*(Y,Z,\bX) 
    \end{array}
    \right] \ .
\end{align*} 
The closed form representation of $\omega^*$ can be found in \eqref{eq-proof-omega star}.

The estimator is constructed as
\begin{align*}
    &
    \widehat{\tau} = \sum_{k=1}^{K}  \frac{N_k}{N} \widehat{\tau}\SSS
    \ , \quad 
    \widehat{\tau}\SSS
    =
    \frac{ 
    \AVERk
     \big\{ \widehat{\phi}\LSS(\bO) 
     \big\}
    }{ \AVER(A)} \ ,
    \\
    & \widehat{\phi}\LSS(\bO) 
    =
    \left[ 
    \begin{array}{l}
    \{ A - (1-A) \pihat(Y,Z,\bX)  \}
     \big\{ \mathcal{G} - \widehat{\mu}\LSS(Z,\bX \con \mathcal{G}) \big\}      
     \\
       +
       (1-A) \pihat(Y,Z,\bX) \big\{ \widehat{\omega}\LSS(Y,Z,\bX) - \widehat{\mu}\LSS(Z,\bX \con \widehat{\omega}\LSS) \big\}
  \\
  +
    A
    \widehat{\mu}\LSS(Z,\bX \con \widehat{\omega}\LSS)
    + 
    (1-A) \widehat{\omega}\LSS(Y,Z,\bX) 
    \end{array}
    \right]
    \ .
\end{align*}

From simple algebra, we obtain
\begin{align*}
    &
    \sqrt{N_k}
    \big\{ \widehat{\tau}\SSS - \tau^* \big\}
    \\
    &
    =
    \frac{1}{\sqrt{N_k}} 
    \sum_{i \in \SSk} 
    \bigg\{ 
        \frac{ \widehat{\phi}\LSS(\bO_i)}{\AVER(A)}
        -
        \frac{ A_i \tau^* }{\AVERk(A)}
    \bigg\}
    \\
    &
    =
    \underbrace{	
\frac{ \EXP(A) }{ 2}
\bigg\{
\frac{ 1 }{ \AVER(A) } - \frac{ 1 }{ \AVERk (A) } 
\bigg\}}_{=o_P(1)}
\underbrace{ 
\frac{1}{\sqrt{N_k}}
\sum_{i \in \SSk}
\frac{ \widehat{\phi}\LSS(\bO_i) + A_i \tau^* }{ \EXP(A) }
}_{=O_P(1)}
\\
&
\qquad 
+
\underbrace{
\frac{ \EXP(A) }{2}
\bigg\{
\frac{ 1 }{ \AVER(A) } + \frac{ 1 }{ \AVERk (A) } 
\bigg\} }_{=1+o_P(1)}
\underbrace{
\frac{1}{ \sqrt{N_k} }
\sum_{i \in \SSk}
\frac{ \widehat{\phi}\LSS(\bO_i)  - A_i \tau^* }{ \EXP(A) } }_{ \equiv \text{\HT{AN}} = O_P(1)}
\\
& =
\frac{1}{\sqrt{N_k}}
\sum_{i \in \SSk}
\frac{\widehat{\phi}\LSS (\bO_i) - A_i \tau^* }{ \EXP(A) } +o_P(1) \ .
\end{align*}
The second line holds because $\AVER(A)=\EXP(A)+o_P(1)$ and $ \AVERk(A) = \EXP(A) + o_P(1)$ from the law of large numbers, and  $\widehat{\phi}\LSS (\bO_i)$ is bounded. 
The fourth underbraced term (referred to as \HL{AN} hereafter) will be shown to be asymptotically normal, hence the second and fourth underbraced terms are $O_P(1)$. We further establish that
\begin{align*} 
&
    \sqrt{N_k}
    \big\{ \widehat{\tau}\SSS - \tau^* \big\} 
    \\
& =
\frac{1}{\sqrt{N_k}}
\sum_{i \in \SSk}
\frac{\widehat{\phi}\LSS(\bO_i) - A_i \tau^* }{\EXP(A)} +o_P(1)
\\
& =
\frac{1}{\sqrt{N_k}}
\sum_{i \in \SSk}
\frac{\phi^*(\bO_i) - A_i \tau^* }{\EXP(A)} +o_P(1)
\numeq 
\label{eq-proof-AN}
\\
    &
    \qquad 
    +
    \sqrt{N_k} \EXPk \Big\{ \widehat{\phi}\LSS(\bO) - \phi^*(\bO) \Big\}
    \numeq
    \label{eq-proof-Bias}
    \\
    &
    \qquad 
    +
    \EMPk
    \Big\{ \widehat{\phi}\LSS(\bO) - \phi^*(\bO) \Big\} \ .
    \numeq
    \label{eq-proof-EP}
\\
& =
\frac{1}{\sqrt{N_k}}
\sum_{i \in \SSk}
\InfFt^*(\bO_i) +o_P(1) \ .
\end{align*}
Note that \eqref{eq-proof-AN} is asymptotically normal by the central limit theorem.
Based on the derivation below, we establish that \eqref{eq-proof-Bias} and \eqref{eq-proof-EP} are $o_P(1)$. Therefore, \HL{AN} is asymptotically normal and $O_P(1)$. Furthermore, this in turn yields the main result:
\begin{align*}
\sqrt{N}
\big( 
    \widehat{\tau}
    -
    \tau^*
    \big) 
    =  
    \frac{1}{\sqrt{N}}
    \sum_{i=1}^{N} 
    \InfFt^*(\bO_i)
    + o_P(1)
    \stackrel{D}{\rightarrow}
    N(0,\sigma^2) \ ,
\end{align*}
where $\sigma^2 = \VAR \{ \InfFt^*(\bO) \}$.  

Thus, to complete the proof of asymptotic normality for $\widehat{\tau}$, it remains to show: 
\begin{align*}
    &
    \eqref{eq-proof-Bias}
    =
    \sqrt{N_k} \EXPk \Big\{ \widehat{\phi}\LSS(\bO) - \phi^*(\bO) \Big\}
    =
    o_P(1)
    \ ,
    \\
    &
    \eqref{eq-proof-EP}
    =
    \EMPk
    \Big\{ \widehat{\phi}\LSS(\bO) - \phi^*(\bO) \Big\} 
    =
    o_P(1)
    \ .
\end{align*}

We first characterize the rate of \eqref{eq-proof-Bias}. Note that it can be expressed as 
\begin{align*}
& 
N_k^{-1/2} \eqref{eq-proof-Bias} 
\\
    & =
    \EXPk
    \left[ 
    \begin{array}{l}
    \{ A - (1-A) \pihat(Y,Z,\bX)  \}
     \big\{ \mathcal{G} - \widehat{\mu}\LSS(Z,\bX \con \mathcal{G}) \big\}      
     \\
       +
       (1-A) \pihat(Y,Z,\bX) \big\{ \widehat{\omega}\LSS(Y,Z,\bX) - \widehat{\mu}\LSS(Z,\bX \con \widehat{\omega}\LSS) \big\}
  \\
  +
    A
    \widehat{\mu}\LSS(Z,\bX \con \widehat{\omega}\LSS)
    + 
    (1-A) \widehat{\omega}\LSS(Y,Z,\bX) 
    \end{array}
    \right]
    \\
    & 
    \quad 
    - \EXPk
    \left[ 
    \begin{array}{l}
    \{ A - (1-A) \pi^*(Y,Z,\bX)  \}
     \big\{ \mathcal{G} - \mu^*(Z,\bX \con \mathcal{G}) \big\}      
     \\
       +
       (1-A) \pi^*(Y,Z,\bX) \big\{ \omega^*(Y,Z,\bX) - \mu^*(Z,\bX \con \omega^*) \big\}
  \\
  +
    A
    \mu^*(Z,\bX \con \omega^*)
    + 
    (1-A) \omega^*(Y,Z,\bX) 
    \end{array}
    \right]
    \\
    &
    = - \EXPk\left[ 
        \begin{array}{l}
             (1-A) \pihat(Y,Z,\bX) 
         \big\{ \mathcal{G} - \muhat(Z,\bX \con \mathcal{G}) \big\} \\
         + A \muhat(Z,\bX \con \mathcal{G}) -  A \mu^*(Z,\bX \con \mathcal{G}) 
             \\
             - 
             (1-A) \pihat(Y,Z,\bX)  \big\{ \what(Y,Z,\bX ) - \muhat(Z,\bX \con \what) \big\}
    \\
    - 
    A
    \muhat(Z,\bX \con \what)
    - 
    (1-A) \what(Y,Z,\bX)
        \end{array}
    \right] 
    \\ 
    &
    = - \EXPk\left[ 
    \begin{array}{l}
    (1-A) 
    \big\{ \pihat (Y,Z,\bX) - \pi^* (Y,Z,\bX) \big\}
     \big\{ \mathcal{G} - \muhat(Z,\bX \con \mathcal{G}) \big\} 
     \\
     - 
     (1-A) \big\{ \pihat (Y,Z,\bX) - \pi^* (Y,Z,\bX) \big\}
     \big\{ \what(Y,Z,\bX) - \muhat(Z,\bX \con \what) \big\} 
    \\
    - 
    (1-A) \pi^*(Y,Z,\bX) \what(Y,Z,\bX)
    \\
    - 
    (1-A) \what(Y,Z,\bX)
        \end{array}
    \right]  
    \\ 
    &
    \stackrel{\eqref{eq-proof-identity4}}{=} 
    - \EXPk\left[ 
    \begin{array}{l}
    (1-A) 
    \big\{ 
    \pihat (\potY{0},Z,\bX) 
    - \pi^* (\potY{0},Z,\bX)  
    \big\}
    \\
    \qquad 
    \times
    \bigg[
        \begin{array}{l}          
     \big\{ \mathcal{G}^{(0)} - \what(\potY{0},Z,\bX) \} \\
     -
     \big\{\muhat(Z,\bX \con \mathcal{G})  - \muhat(Z,\bX \con \what) 
     \big\} 
        \end{array}
    \bigg]
    \\ 
    - \what(\potY{0},Z,\bX)
        \end{array}
    \right]  
    \\ 
    &
    \stackrel{\eqref{eq-proof-identity3}}{=} 
    - \EXPk\left[ 
    \begin{array}{l}
    A
    \big\{ 
    \pihat (\potY{0},Z,\bX) / \pi^* (\potY{0},Z,\bX)  
    - 1
    \big\}
    \\
    \qquad 
    \times
    \bigg[
        \begin{array}{l}          
     \big\{ \mathcal{G}^{(0)} - \what(\potY{0},Z,\bX) \} \\
     -
     \big\{\muhat(Z,\bX \con \mathcal{G})  - \muhat(Z,\bX \con \what) 
     \big\} 
        \end{array}
    \bigg]
    \\ 
    - \what(\potY{0},Z,\bX)
        \end{array}
    \right]  
    \ .
    \numeq
    \label{eq-proof-bias of ATT estimator}
\end{align*}

We will use $\alpha_\norm$ and $\beta_\norm$ defined in Lemma \ref{lemma-supp-convergence}-\HL{v}. For any function $h$, we have
\begin{align*}
    &
    \EXPk
    \Bigg[ 
    A
    \bigg\{ 
    \frac{ \pihat (\potY{0},Z,\bX) }{  \pi^* (\potY{0},Z,\bX)  }
    - 1
    \bigg\}
    h(\potY{0},Z,\bX)
    \Bigg]
    \\
    &
    =
    \EXPk
    \Bigg[ 
    A
    \bigg\{ 
    \frac{ \alphahat_\norm(\potY{0},\bX) \betahat_\norm (Z,\bX) }
    { \alpha_\norm^*(\potY{0},\bX) \beta_\norm^* (Z,\bX) }
    - 1
    \bigg\}
    h(\potY{0},Z,\bX)
    \Bigg] 
    \\
    &
    =
    \EXPk
    \Bigg[ 
    A
    \bigg\{ 
    \frac{ \alphahat_\norm(\potY{0},\bX) }
    { \alpha_\norm^*(\potY{0},\bX) }
    - 1
    \bigg\}
    h(\potY{0},Z,\bX)
    \Bigg] 
    \\
    &
    \qquad +
    \EXPk
    \Bigg[ 
    A
    \bigg\{ 
    \frac{  \betahat_\norm (Z,\bX) }
    { \beta_\norm^* (Z,\bX) }
    - 1
    \bigg\}
    h(\potY{0},Z,\bX)
    \Bigg] 
    \\
    &
    \qquad +
    \EXPk
    \Bigg[ 
    A
    \bigg\{ 
    \frac{ \alphahat_\norm(\potY{0},\bX) }
    { \alpha_\norm^*(\potY{0},\bX) }
    - 1
    \bigg\}
    \bigg\{ 
    \frac{  \betahat_\norm (Z,\bX) }
    { \beta_\norm^* (Z,\bX) }
    - 1
    \bigg\}
    h(\potY{0},Z,\bX)
    \Bigg] 
    \\
    &
    =
    \EXPk
    \Bigg[  
    \bigg\{ 
    \frac{ \alphahat_\norm(\potY{0},\bX) }
    { \alpha_\norm^*(\potY{0},\bX) }
    - 1
    \bigg\}
    \EXP \big\{ 
    h(\potY{0},Z,\bX) \cond A=1,\potY{0},\bX \big\}
    \Bigg] 
    \\
    &
    \qquad +
    \EXPk
    \Bigg[  
    \bigg\{ 
    \frac{  \betahat_\norm (Z,\bX) }
    { \beta_\norm^* (Z,\bX) }
    - 1
    \bigg\}
    \EXP \big\{ 
    h(\potY{0},Z,\bX) \cond A=1,Z,\bX \big\}
    \Bigg] 
    \\
    &
    \qquad +
    \EXPk
    \Bigg[ 
    A
    \bigg\{ 
    \frac{ \alphahat_\norm(\potY{0},\bX) }
    { \alpha_\norm^*(\potY{0},\bX) }
    - 1
    \bigg\}
    \bigg\{ 
    \frac{  \betahat_\norm (Z,\bX) }
    { \beta_\norm^* (Z,\bX) }
    - 1
    \bigg\}
    h(\potY{0},Z,\bX)
    \Bigg] 
    \ .
    \numeq
    \label{eq-proof-bias of ATT estimator2}
\end{align*}
Substituting $ h = 
     \big\{ \mathcal{G}^{(0)} - \what \} 
     - \muhat(Z,\bX \con \mathcal{G}-\what) $, equation \eqref{eq-proof-bias of ATT estimator2} becomes
\begin{align*}
    &
    \bigg\|
    \EXPk\left[ 
   A
    \bigg\{ 
    \frac{ \pihat (\potY{0},Z,\bX) }{  \pi^* (\potY{0},Z,\bX)  }
    - 1
    \bigg\} 
    \big[
        \big\{ \mathcal{G}^{(0)} - \what \} 
     - \muhat(Z,\bX \con \mathcal{G}-\what) 
    \big]  
    \right]
    \bigg\|
    \\
    &
    \lesssim
     \EXPk\left[ 
    \bigg\| \frac{\alphahat_{\norm} (\potY{0},\bX)}{\alpha_{\norm}^* (\potY{0},\bX)} - 1 \bigg\| 
    \times 
     \begin{array}{l} 
          \bigg\|
    \EXP \bigg[ 
          \begin{array}{l} \big\{ \mathcal{G}^{(0)} - \what \} 
          \\
     - \muhat(Z,\bX \con \mathcal{G}-\what) 
          \end{array} 
          \, \bigg| \,
          A=1,\potY{0},\bX \bigg]
     \bigg\|
     \end{array}
     \right]
     \\
     &
     \qquad +
      \EXPk\left[ 
    \bigg\| \frac{\betahat_{\norm} (Z,\bX)}{\beta_{\norm}^* (Z,\bX)} - 1 \bigg\| 
    \times 
     \begin{array}{l} 
          \bigg\|
    \EXP \bigg[ 
          \begin{array}{l} \big\{ \mathcal{G}^{(0)} - \what \} 
          \\
     - \muhat(Z,\bX \con \mathcal{G}-\what) 
          \end{array} 
          \, \bigg| \,
          A=1,Z,\bX \bigg]
     \bigg\|
     \end{array}
     \right]
     \\
     &
     \qquad +
      \EXPk\left[  
     \bigg\| \frac{\alphahat_{\norm} (\potY{0},\bX)}{\alpha_{\norm}^* (\potY{0},\bX)} - 1 \bigg\| 
     \bigg\| \frac{\betahat_{\norm} (Z,\bX)}{\beta_{\norm}^* (Z,\bX)} - 1 \bigg\|  
     \bigg\|
     \begin{array}{l}
        \big\{ \mathcal{G}^{(0)} - \what \} 
        \\
     - \muhat(Z,\bX \con \mathcal{G}-\what) 
     \end{array}
    \bigg\|
     \right]
     \\
    &
    \lesssim
    \bigg\| \begin{array}{l}\alphahat_{\norm} (\potY{0},\bX) \\
    - \alpha_{\norm}^*(\potY{0},\bX)
    \end{array}  \bigg\|_{P,2}
    \times 
          \bigg\|
     \begin{array}{l} 
    \EXP \bigg[ 
          \begin{array}{l} \big\{ \mathcal{G}^{(0)} - \what \} 
          \\
     - \muhat(Z,\bX \con \mathcal{G}-\what) 
          \end{array} 
          \, \bigg| \,
          A=1,\potY{0},\bX \bigg]
     \end{array} 
          \bigg\|_{P,2}
     \\
     &
     \qquad +
     \bigg\| \begin{array}{l} \betahat_{\norm} (Z,\bX) \\
     - \beta_{\norm}^*(Z,\bX) \end{array} \bigg\|_{P,2}
    \times 
          \bigg\|
     \begin{array}{l} 
    \EXP \bigg[ 
          \begin{array}{l} \big\{ \mathcal{G}^{(0)} - \what \} 
          \\
     - \muhat(Z,\bX \con \mathcal{G}-\what) 
          \end{array} 
          \, \bigg| \,
          A=1,Z,\bX \bigg]
     \end{array} 
          \bigg\|_{P,2}
     \\
     &
     \qquad +
     \bigg\| \begin{array}{l}\alphahat_{\norm} (\potY{0},\bX) \\
    - \alpha_{\norm}^*(\potY{0},\bX)
    \end{array}  \bigg\|_{P,2}
     \bigg\| \begin{array}{l} \betahat_{\norm} (Z,\bX) \\
     - \beta_{\norm}^*(Z,\bX) \end{array} \bigg\|_{P,2} 
     \\
    &
    \lesssim 
    \left[ 
    \begin{array}{l}            
    \| \widehat{f}\LSS(\potY{0} \cond A=0,Z,\bX)
    -
    f^*(\potY{0} \cond A=0,Z,\bX)\|_{P,2}
    \\
     + \| \widehat{f}\LSS(A=1 \cond Z,\bX) - f^*(A=1 \cond Z,\bX) \|_{P,2}
        \end{array}
    \right]^2 \ .
    \numeq
    \label{eq-proof-bias of ATT estimator3}
\end{align*}
The last line holds from \eqref{eq-proof-alpha L2P}, \eqref{eq-proof-beta L2P}, \eqref{eq-proof-GH convergence}, \eqref{eq-proof-GH convergence given Z}.

In addition, from \eqref{eq-proof-omega mean rate}, we obtain
\begin{align*}
    &
     \Big\| 
    \EXP\LSS \big\{ \what(\potY{0},Z,\bX)
    \big\}
    \Big\|
        \\
    &
    \lesssim 
    \left[
        \begin{array}{l}
        \big\| \widehat{f}\LSS(A=1 \cond Z,\bX) - f^*(A=1 \cond Z,\bX)
        \big\|_{P,2}
        \\
        +
        \big\| \widehat{f}\LSS(\potY{0} \cond A=0,Z,\bX) - f^*(\potY{0} \cond A=0,Z,\bX)
        \big\|_{P,2}
        \end{array}
    \right]
    \\
    &
    \quad \quad 
    \times
    \Big\|  \widehat{f}\LSS(Z=1 \cond \bX) - f^*(Z=1 \cond \bX) \Big\|_{P,2} \ .    
    \numeq
    \label{eq-proof-bias of ATT estimator4}
\end{align*}

Combining \eqref{eq-proof-bias of ATT estimator}, \eqref{eq-proof-bias of ATT estimator3}, and \eqref{eq-proof-bias of ATT estimator4}, we establish
\begin{align*}
     &
     N_k^{-1/2} 
     \Big\| \eqref{eq-proof-Bias} \Big\|
     \\
     &
     \lesssim
    \left[ 
    \begin{array}{l}            
    \| \widehat{f}\LSS(\potY{0} \cond A=0,Z,\bX)
    -
    f^*(\potY{0} \cond A=0,Z,\bX)\|_{P,2}
    \\
     + \| \widehat{f}\LSS(A=1 \cond Z,\bX) - f^*(A=1 \cond Z,\bX) \|_{P,2}
        \end{array}
    \right]
    \\
    &
    \quad \quad 
    \times 
    \left[ 
    \begin{array}{l}            
    \| \widehat{f}\LSS(\potY{0} \cond A=0,Z,\bX)
    -
    f^*(\potY{0} \cond A=0,Z,\bX)\|_{P,2}
    \\
     + \| \widehat{f}\LSS(A=1 \cond Z,\bX) - f^*(A=1 \cond Z,\bX) \|_{P,2}     
    \\
     + \| \widehat{f}\LSS(Z=1 \cond \bX) - f^*(Z=1 \cond \bX) \|_{P,2}
        \end{array}
    \right]
    \\
    &
    = 
    r_Y^{(-k)2}
    +
    r_A^{(-k)2}
    +
    r_Y\LSS 
    r_A \LSS
    +
    r_Y \LSS
    r_Z \LSS
    +
    r_A \LSS
    r_Z \LSS 
    \\
    &
    =
    o_P\big( N^{-1/2} \big)
    \ .
\end{align*}
The second inequality holds from \eqref{eq-proof-mu L2P}, \eqref{eq-proof-GH convergence}, \eqref{eq-proof-omega mean rate}. The last line holds from Assumption \HL{A7}. This completes the proof that \eqref{eq-proof-Bias} is $o_P(1)$.

Next, we characterize the rate of \eqref{eq-proof-EP}. The expectation of \eqref{eq-proof-EP} conditioning on $\SSk$ is 0. The variance of \eqref{eq-proof-EP} is
\begin{align*}
\VAR\LSS \big\{ \eqref{eq-proof-EP} \big\}
\leq 
\VAR\LSS \big\{ \widehat{\phi}\LSS (\bO) - \phi^* (\bO) \big\}
\leq
\EXPk \big[ \big\| 
\widehat{\phi}\LSS (\bO) - \phi^* (\bO) \big\|_2^2 \big] \ .
\end{align*}
Therefore, it suffices to find the rate of $\EXPk \big[ \big\| 
\widehat{\phi}\LSS (\bO) - \phi^* (\bO) \big\|_2^2 \big]$. Note that $\widehat{\phi}\LSS(\bO) - \phi^*(\bO)$ can be explicitly written as:
\begin{align*}
&
\widehat{\phi}\LSS(\bO) - \phi^*(\bO)
\\
        & = 
    \left[ 
    \begin{array}{l}
    \{ A - (1-A) \pihat(Y,Z,\bX)  \}
     \big\{ \mathcal{G} - \widehat{\mu}\LSS(Z,\bX \con \mathcal{G}) \big\}      
     \\
       +
       (1-A) \pihat(Y,Z,\bX) \big\{ \widehat{\omega}\LSS(Y,Z,\bX) - \widehat{\mu}\LSS(Z,\bX \con \widehat{\omega}\LSS) \big\}
  \\
  +
    A
    \widehat{\mu}\LSS(Z,\bX \con \widehat{\omega}\LSS)
    + 
    (1-A) \widehat{\omega}\LSS(Y,Z,\bX) 
    \end{array}
    \right]
    \\
    & 
    \qquad \qquad 
    -  
    \left[ 
    \begin{array}{l}
    \{ A - (1-A) \pi^*(Y,Z,\bX)  \}
     \big\{ \mathcal{G} - \mu^*(Z,\bX \con \mathcal{G}) \big\}      
     \\
       +
       (1-A) \pi^*(Y,Z,\bX) \big\{ \omega^*(Y,Z,\bX) - \mu^*(Z,\bX \con \omega^*) \big\}
  \\
  +
    A
    \mu^*(Z,\bX \con \omega^*)
    + 
    (1-A) \omega^*(Y,Z,\bX) 
    \end{array}
    \right]
    \\
    &
    =
     \left[ 
        \begin{array}{l}
            (1-A) \pi^*(Y,Z,\bX) 
         \big\{ \mathcal{G} - \mu^*(Z,\bX \con \mathcal{G}) - \omega^*(Y,Z,\bX) + \mu^*(Z,\bX \con \omega^*) \big\} 
         \\
         - (1-A) \pihat (Y,Z,\bX) 
         \big\{ \mathcal{G} - \muhat(Z,\bX \con \mathcal{G}) 
         - \what(Y,Z,\bX) + \muhat(Z,\bX \con \what ) \big\} 
         \\
         + A \big\{ \mu^*(Z,\bX \con \mathcal{G}) - \muhat(Z,\bX \con \mathcal{G}) \big\}
         \\
         - A \big\{ \mu^*(Z,\bX \con \omega^*) - \muhat(Z,\bX \con \what) \big\}
         \\
         - (1-A) \big\{ \omega^*(Y,Z,\bX) - \what(Y,Z,\bX) \big\}
        \end{array}
    \right]
    \\
    &
    =
     \left[ 
        \begin{array}{l}
            (1-A) \pi^*(Y,Z,\bX) 
         \big\{ \mathcal{G} - \mu^*(Z,\bX \con \mathcal{G}) - \omega^*(Y,Z,\bX) + \mu^*(Z,\bX \con \omega^*) \big\} 
         \\
         - (1-A) \pi^* (Y,Z,\bX) 
         \big\{ \mathcal{G} - \muhat(Z,\bX \con \mathcal{G}) 
         - \what(Y,Z,\bX) + \muhat(Z,\bX \con \what ) \big\} 
         \\
         + (1-A) \pi^* (Y,Z,\bX) 
         \big\{ \mathcal{G} - \muhat(Z,\bX \con \mathcal{G}) 
         - \what(Y,Z,\bX) + \muhat(Z,\bX \con \what ) \big\} 
         \\
         - (1-A) \pihat (Y,Z,\bX) 
         \big\{ \mathcal{G} - \muhat(Z,\bX \con \mathcal{G}) 
         - \what(Y,Z,\bX) + \muhat(Z,\bX \con \what ) \big\} 
         \\
         + A \big\{ \mu^*(Z,\bX \con \mathcal{G}) - \muhat(Z,\bX \con \mathcal{G}) \big\}
         \\
         - A \big\{ \mu^*(Z,\bX \con \omega^*) - \muhat(Z,\bX \con \what) \big\}
         \\
         - (1-A) \big\{ \omega^*(Y,Z,\bX) - \what(Y,Z,\bX) \big\}
        \end{array}
    \right] 
    \\
    &
    =
     \left[ 
        \begin{array}{l}
            (1-A) \pi^*(\potY{0},Z,\bX) 
            \left[ 
            \begin{array}{l}
                  \muhat(Z,\bX \con \mathcal{G})  - \mu^*(Z,\bX \con \mathcal{G})
                 \\
                 +
                 \what
                 (\potY{0},Z,\bX)
                 -
                 \omega^*(\potY{0},Z,\bX)
                 \\
                 +
                 \mu^*(Z,\bX \con \omega^*)
                 - 
                 \muhat(Z,\bX \con \what) 
            \end{array}
            \right]    
         \\
         + (1-A)
         \left[ 
            \begin{array}{l}
                 \mathcal{G} - \muhat(Z,\bX \con \mathcal{G}) 
                 \\
         - \what(\potY{0},Z,\bX) + \muhat(Z,\bX \con \what ) 
            \end{array}
         \right]   
         \left\{ 
        \begin{array}{l}
         \pi^* (\potY{0},Z,\bX) 
         \\
         -
              \pihat (\potY{0},Z,\bX) 
        \end{array}
         \right\} 
         \\
         + A \big\{ \mu^*(Z,\bX \con \mathcal{G}) - \muhat(Z,\bX \con \mathcal{G}) \big\}
         \\
         - A \big\{ \mu^*(Z,\bX \con \omega^*) - \muhat(Z,\bX \con \what) \big\}
         \\
         - (1-A) \big\{ \omega^*(\potY{0},Z,\bX) - \what(\potY{0},Z,\bX) \big\}
        \end{array}
    \right] 
    \ .
\end{align*}
Consequently, we obtain
\begin{align*}
&
\EXPk \big[ \big\| 
\widehat{\phi}\LSS (\bO) - \phi^* (\bO) \big\|_2^2 \big]
\\
&
\lesssim
\left[ 
\begin{array}{l}
     \big\| 
    \muhat(Z,\bX \con \mathcal{G})
    -
    \mu^*(Z,\bX \con \mathcal{G})
\big\|_{P,2}^2
\\[0.1cm]
+
\big\|
\muhat(Z,\bX \con \what)
    -
    \mu^*(Z,\bX \con \omega^*)
\big\|_{P,2}^2 
\\[0.1cm]
+
\big\| 
    \what(\potY{0},Z,\bX)
    -
    \omega^*(\potY{0},Z,\bX)
\big\|_{P,2}^2
\\[0.1cm]
+
\big\|
\pihat(\potY{0},Z,\bX)
    -
    \pi^*(\potY{0},Z,\bX)
\big\|_{P,2}^2 
\end{array}
\right]
\\
&
= 
o_P(1) \ .
\numeq 
\label{eq-proof-Consistency}
\end{align*}
The last line holds from \eqref{eq-proof-mu L2P}, \eqref{eq-proof-alpha beta L2P}, 
\eqref{eq-proof-omega hat}, \eqref{eq-proof-omega mu hat}.

We finally show that the variance estimator $\widehat{\sigma}^2$ is consistent. Note that
\begin{align*}
    \widehat{\sigma}^2
    =
     \sum_{k=1}^{K} 
     \frac{N_k}{N} \widehat{\sigma}^{2,(k)}
     \ , \quad 
     \widehat{\sigma}^{2,(k)}
     = 
    \AVERk
    \bigg[    
     \bigg\{ 
     \frac{ \widehat{\phi}\LSS(\bO) - A \widehat{\tau}
    }{ \AVER(A)}     
    \bigg\}^2
    \bigg] \ .
\end{align*}
Therefore, it suffices to show that $ 	\widehat{\sigma}^{2,(k)} - \sigma^2 = o_P(1)$, which is expressed as follows:
\begin{align}
&
\widehat{\sigma}^{2,(k)} - \sigma^2
\nonumber
\\
&
=
\big\{ \AVER(A) \big\}^{-2}
\AVER_{\II_k}
\big[
\big\{ \widehat{\phi}\LSS(\bO) - A \widehat{\tau}  \big\}^2
\big] -
\sigma^2
\nonumber
\\
&
=
\big\{ \EXP(A) \big\}^{-2}
\AVER_{\II_k}
\big[
\big\{ \widehat{\phi}\LSS(\bO) - A \widehat{\tau}  \big\}^2
\big] -
\sigma^2 + o_P(1)
\nonumber
\\
&
=
\big\{ \EXP(A) \big\}^{-2}
\Big[
\AVER_{\II_k}
\big[
\big\{ \widehat{\phi}\LSS(\bO) - A \widehat{\tau}  \big\}^2
\big] -
\AVER_{\II_k}
\big[
\big\{ \phi^*(\bO) - A \tau^*  \big\} ^2
\big]
\Big]
\nonumber
\\
& \hspace*{2cm}
+		 
\Big[
\big\{ \EXP(A) \big\}^{-2}
\AVER_{\II_k}
\big[
\big\{ \phi^*(\bO) - A \tau^*  \big\} ^2
\big]
-
\sigma^2
\Big] + o_P(1)
\nonumber
\\
& = 
\big\{ \EXP(A) \big\}^{-2}
\AVER_{\II_k}
\Big[
\big\{ \widehat{\phi}\LSS(\bO) - A \widehat{\tau} \big\}^2
- 
\big\{ \phi^*(\bO) - A \tau^*  \big\}^2 
\Big]
+ o_P(1)
\label{eq-proof-variance1} \ .
\end{align}
The third and fifth lines hold from the law of large numbers.  
Therefore, it is sufficient to show that \eqref{eq-proof-variance1} is also $o_P(1)$. The term in  \eqref{eq-proof-variance1} can be written as
\begin{align*}
&
\frac{1}{N_k}
\sum_{i \in \II_k}
\Big[ 
\big\{
\widehat{\phi}\LSS (\bO_i)
-
A_i \widehat{\tau}
\big\}^2
-
\big\{
\phi^* (\bO_i)
-
A_i
\tau^*
\big\}^2
\Big]
\\
& =
\frac{1}{N_k}
\sum_{i \in \II_k}
\left[
\begin{array}{l}
\big[
\big\{
\widehat{\phi}\LSS (\bO_i)
-
A_i
\widehat{\tau}
\big\}
-
\big\{
\phi^* (\bO_i)
-
A_i
\tau^*
\big\}
\big]
\\
\times 
\big[
\big\{
\widehat{\phi}\LSS (\bO_i)
-
A_i
\widehat{\tau}
\big\}
+
\big\{
\phi^* (\bO_i)
-
A_i
\tau^*
\big\}
\big]
\end{array} 
\right]
    \\
& = 
\frac{1}{N_k}
\sum_{i \in \II_k}
\left[
\begin{array}{l}
\big[
\big\{
\widehat{\phi}\LSS (\bO_i)
-
A_i
\widehat{\tau}
\big\}
-
\big\{
\phi^* (\bO_i)
-
A_i
\tau^*
\big\}
\big] \\
\times \big[
\big\{
\widehat{\phi}\LSS (\bO_i)
-
A_i
\widehat{\tau}
\big\}
-
\big\{
\phi^* (\bO_i)
-
A_i
\tau^*
\big\}
+ 
2 \big\{
\phi^* (\bO_i)
-
A_i
\tau^*
\big\}
\big]
\end{array}
\right]
    \\
& = 
\frac{1}{N_k}
\sum_{i \in \II_k}
\Big[
\big\{
\widehat{\phi}\LSS (\bO_i)
-
A_i
\widehat{\tau}
\big\}
-
\big\{
\phi^* (\bO_i)
-
A_i
\tau^*
\big\}
\Big]^2
\\
&
\quad 
+
\frac{2}{N_k}
\sum_{i \in \II_k}
\Big[
\big[
\big\{
\widehat{\phi}\LSS (\bO_i)
-
A_i
\widehat{\tau}
\big\}
-
\big\{
\phi^* (\bO_i)
-
A_i
\tau^*
\big\}
\big]
\big\{
\phi^* (\bO_i)
-
A_i
\tau^*
\big\}
\Big]			\ .
\end{align*}
Let $\widehat{\Delta}\LSS = 
\big\{
\widehat{\phi}\LSS(\bO)
-
A
\widehat{\tau}
\big\}
-
\big\{
\phi^*(\bO)
-
A
\tau^*
\big\}$. 
From the H\"older's inequality, we find the absolute value of \eqref{eq-proof-variance1} is upper bounded by
\begin{align*}
    \big\| \eqref{eq-proof-variance1} \big\|
    & \precsim
    \AVER_{\mathcal{I}_k} \Big[  \big\{ \widehat{\Delta} \LSS \big\}^2 \Big]
    +
    2
    \AVER_{\mathcal{I}_k} \Big[  \big\{ \widehat{\Delta} \LSS \big\}^2 \Big]
    \cdot 
    \AVER_{\mathcal{I}_k} \Big[  \big\{ \phi^*(\bO) - A \tau^* \big\}^2 \Big] \ .
\end{align*}
Since $\AVER_{\II_k} [  \{  \phi^*(\bO) - A \tau^* \}^2 ] = \{ \EXP(A)\} ^2 \sigma^2 + o_P(1) = O_P(1)$, \eqref{eq-proof-variance1} is $o_P(1)$ if $\AVER_{\II_k} [  \{ \widehat{\Delta} \LSS \}^2 ] = o_P(1)$. From some algebra, we find
\begin{align*}
    \AVER_{\II_k} \Big[  \big\{ \widehat{\Delta} \LSS \big\}^2 \Big]
    & \lesssim
    \frac{1}{N_k}
\sum_{i \in \mathcal{I}_k} 
\big\{
\widehat{\phi} \LSS (\bO_i)
-
\phi^*(\bO_i)
\big\}^2
+
\big( \widehat{\tau}  - \tau^* \big)^2
\\
& =
\EXP\LSS\Big[
\big\{
\widehat{\phi}\LSS(\bO)
-
\phi^*(\bO)
\big\}^2 \Big] 
+
\big( \widehat{\tau}  - \tau^* \big)^2
+ o_P(1)
\\
&
=
o_P(1) \ . 
\end{align*}
The second line holds from the law of large numbers applied to $\big\{ \widehat{\phi}\LSS(\bO)
-
\phi^*(\bO)
\big\}^2$. The last line holds from \eqref{eq-proof-Consistency} and $\widehat{\tau} = \tau^* + o_P(1)$, which is from the asymptotic normality of the estimator. 

This concludes the proof. 

\subsection{Proof of Theorem \ref{thm-falsification}}

The result follows directly from Lemma \ref{lemma-beta and gamma}.

}%

\newpage

\bibliographystyle{apa}
\bibliography{SIV.bib}

\end{document}